\documentclass[aps,prx,reprint,superscriptaddress,nofootinbib]{revtex4-2}
\usepackage{blindtext}
\usepackage{lipsum}
\usepackage{graphics}
\usepackage{amsmath}
\usepackage{bm}
\usepackage{bbm}
\usepackage{graphicx}
\usepackage{graphics}
\usepackage{amssymb}
\usepackage{verbatim}
\usepackage{physics}
\usepackage{float}
\usepackage{dsfont}
\usepackage[normalem]{ulem}
\usepackage[dvipsnames]{xcolor}
\usepackage{framed}
\definecolor{shadecolor}{RGB}{225,225,225}

\usepackage{amsthm}
\usepackage{thmtools}
\usepackage{thm-restate}
\usepackage{amsfonts}
\usepackage{url}
\usepackage{chngcntr}
\counterwithout{equation}{section}


\usepackage[dvipsnames]{xcolor}
\usepackage[normalem]{ulem}
\usepackage{wasysym}
\usepackage[export]{adjustbox}
\usepackage{ifthen}

\usepackage{titlesec,enumitem}
\renewcommand{\theparagraph}{\arabic{paragraph}}
\renewcommand{\thesubparagraph}{\alph{subparagraph}}
\titleformat{\paragraph}[runin]
  {\normalfont\normalsize\bfseries}
  {(\theparagraph) }
  {0pt}
  {}
  [:]
\titlespacing*{\paragraph}
  {0pt}
  {1.5ex plus .5ex minus .5ex}
  {0.5em}

\titleformat{\subparagraph}[runin]
  {\normalfont\normalsize\itshape}
  {(\theparagraph\thesubparagraph) }
  {0pt}
  {}
  []
\titlespacing*{\subparagraph}
  {0pt}
  {1.0ex plus .5ex minus .5ex}
  {0.5em}

\makeatletter
\DeclareFontFamily{OMX}{MnSymbolE}{}
\DeclareSymbolFont{MnLargeSymbols}{OMX}{MnSymbolE}{m}{n}
\SetSymbolFont{MnLargeSymbols}{bold}{OMX}{MnSymbolE}{b}{n}
\DeclareFontShape{OMX}{MnSymbolE}{m}{n}{
    <-6>  MnSymbolE5
   <6-7>  MnSymbolE6
   <7-8>  MnSymbolE7
   <8-9>  MnSymbolE8
   <9-10> MnSymbolE9
  <10-12> MnSymbolE10
  <12->   MnSymbolE12
}{}
\DeclareFontShape{OMX}{MnSymbolE}{b}{n}{
    <-6>  MnSymbolE-Bold5
   <6-7>  MnSymbolE-Bold6
   <7-8>  MnSymbolE-Bold7
   <8-9>  MnSymbolE-Bold8
   <9-10> MnSymbolE-Bold9
  <10-12> MnSymbolE-Bold10
  <12->   MnSymbolE-Bold12
}{}

\let\llangle\@undefined
\let\rrangle\@undefined
\DeclareMathDelimiter{\llangle}{\mathopen}%
                     {MnLargeSymbols}{'164}{MnLargeSymbols}{'164}
\DeclareMathDelimiter{\rrangle}{\mathclose}%
                     {MnLargeSymbols}{'171}{MnLargeSymbols}{'171}
\makeatother

\usepackage{tikz}
\usepackage{tikz-cd}
\usetikzlibrary{arrows}
\usetikzlibrary{snakes}
\usetikzlibrary{intersections}
\usetikzlibrary{shapes.geometric}
\usetikzlibrary{decorations.pathmorphing, patterns,shapes}
\usetikzlibrary{decorations.markings}

\tikzset{
	partial ellipse/.style args={#1:#2:#3}{
		insert path={+ (#1:#3) arc (#1:#2:#3)}
	}
}

\tikzset{
	mid arrow/.style={postaction={decorate,decoration={
				markings,
				mark=at position .575 with {\arrow[#1]{stealth}}
	}}},
	near arrow/.style={postaction={decorate,decoration={
				markings,
				mark=at position .275 with {\arrow[#1]{stealth}}
	}}},
	far arrow/.style={postaction={decorate,decoration={
				markings,
				mark=at position .800 with {\arrow[#1]{stealth}}
	}}},
}

\pgfdeclarepatternformonly{south west lines}{\pgfqpoint{-0pt}{-0pt}}{\pgfqpoint{3pt}{3pt}}{\pgfqpoint{3pt}{3pt}}{
	\pgfsetlinewidth{0.4pt}
	\pgfpathmoveto{\pgfqpoint{0pt}{0pt}}
	\pgfpathlineto{\pgfqpoint{3pt}{3pt}}
	\pgfpathmoveto{\pgfqpoint{2.8pt}{-.2pt}}
	\pgfpathlineto{\pgfqpoint{3.2pt}{.2pt}}
	\pgfpathmoveto{\pgfqpoint{-.2pt}{2.8pt}}
	\pgfpathlineto{\pgfqpoint{.2pt}{3.2pt}}
	\pgfusepath{stroke}}

\newcommand{\exd}{\mathrm{d}}
\newcommand{\Hom}{\mathsf{Hom}}
\newcommand{\unit}{\mathbbm{1}}
\newcommand{\id}{\mathrm{id}}

\newcommand{\mcb}[1]{#1_{\bullet}}
\newcommand{\mcbb}[1]{#1^{\bullet}}

\definecolor{orange(ryb)}{HTML}{FFA500}
\definecolor{lightorange(ryb)}{HTML}{FFB300}
\definecolor{dodgerblue}{HTML}{1E90FF}
\definecolor{steelcityblue}{HTML}{088EBF}
\definecolor{lightdodgerblue}{HTML}{4dbff7}
\definecolor{crimson}{HTML}{FF4C4C}
\definecolor{pinkerton}{HTML}{EC368D}
\definecolor{forest}{HTML}{6DD189}
\definecolor{lightishgray}{HTML}{DFDFDF}
\definecolor{error-red}{HTML}{EFB2B6}

\usepackage[colorlinks=true,citecolor=dodgerblue,linkcolor=dodgerblue,urlcolor=dodgerblue,pdftitle={Chain Maps}]{hyperref}

\def \beq {\begin{equation}}
\def \eeq {\end{equation}}
\def \beqa {\begin{eqnarray}}
\def \eeqa {\end{eqnarray}}
\def \bseq {\begin{subequations}}
\def \eseq {\end{subequations}}
\def \btikz {\begin{tikzpicture}[scale = 1.0, baseline = {([yshift=-.5ex]current bounding box.center)}]}
\def \etikz {\end{tikzpicture}}
\newcommand{\phii}{\varphi}

\newcommand \transpose {\mathsf{T}}

\newcommand \nnb {\nonumber}

\newcommand \td {\tilde}
\newcommand \ZZ {\mathbb{Z}}
\newcommand \FF {\mathbb{F}}

\newcommand \C {\mathcal{C}}
\newcommand \D {\mathcal{D}}

\newcommand{\cupp}{\smallsmile}
\newcommand{\capp}{\smallfrown}

\newcommand \isom {\simeq}
\newcommand \CMH[1] {[\mathbf{C}^{[#1]}]}

\newcommand{\CZLogicalColor}[2]{\overline{\mathsf{\textcolor{orange}{C}\textcolor{dodgerblue}{Z}}}_{\textcolor{orange}{#1}\textcolor{dodgerblue}{#2}}}
\newcommand{\CNOTLogicalColor}[2]{\overline{\mathsf{\textcolor{orange}{C}\textcolor{dodgerblue}{NOT}}}_{\textcolor{orange}{#1}\textcolor{dodgerblue}{#2}}}
\newcommand{\CZLogical}[2]{\overline{\mathsf{CZ}}_{#1#2}}
\newcommand{\CNOTLogical}[2]{\overline{\mathsf{CNOT}}_{#1#2}}

\newtheoremstyle{boldtitle}
  {\topsep}
  {\topsep}
  {\normalfont}
  {}
  {\bfseries}
  {}
  {.5em}
  {\thmname{#1}\thmnumber{ #2}\thmnote{. #3}.}

\theoremstyle{boldtitle}

\newtheorem{prototheorem}{Theorem}[section]
\newtheorem{theorem}[prototheorem]{Theorem}
\newtheorem{lemma}[prototheorem]{Lemma}
\newtheorem{definition}[prototheorem]{Definition}
\newtheorem{conjecture}[prototheorem]{Conjecture}
\newtheorem{cor}[prototheorem]{Corollary}

\definecolor{shadecolor}{RGB}{235,235,235}
\colorlet{theoremshade}{gray!12}
\colorlet{lemmashade}{gray!12}

\newenvironment{theo}{
\colorlet{shadecolor}{theoremshade}\begin{shaded*}\begin{theorem}}
{\end{theorem}\end{shaded*}}

\newenvironment{lem}{\colorlet{shadecolor}{lemmashade}\begin{shaded*}\begin{lemma}}
{\end{lemma}\end{shaded*}}

\newenvironment{defn}{\colorlet{shadecolor}{lemmashade}\begin{shaded*}\begin{definition}}
{\end{definition}\end{shaded*}}

\newenvironment{coro}{\colorlet{shadecolor}{lemmashade}\begin{shaded*}\begin{cor}}
{\end{cor}\end{shaded*}}

\begin{document}

\title{Computing with qLDPC Codes by Climbing the Chain Map Hierarchy}

\author{Rahul Sahay\equalcontrib}
\email{rsahay@g.harvard.edu}
\affiliation{Department of Physics, Harvard University, Cambridge, Massachusetts 02138, USA}

\author{David M. Long\equalcontrib}
\email{dmlong@stanford.edu}
\affiliation{Department of Physics, Stanford University, Stanford, California 94305, USA}

\author{Vedika Khemani}
\email{vkhemani@stanford.edu}
\affiliation{Department of Physics, Stanford University, Stanford, California 94305, USA}

\makeatletter
\DeclareRobustCommand{\equalcontrib}{%
  \frontmatter@footnote{Equal contribution.}
}
\makeatother

\begin{abstract}
We develop a framework for logical computation with qLDPC codes that places logical Pauli, Clifford, and non-Clifford operations on the same footing.
This brings the simple homological description of Pauli logicals to the patchwork landscape of logical Clifford and non-Clifford operations, providing a tool for the discovery of new logical gates. 
In particular, we define the \textit{chain map hierarchy}: a family of chain complexes whose homology classes encode logical unitary and code surgery operations at any level of the Clifford hierarchy, precisely analogous to the chain complex description of Pauli logicals.  
Consequently, intuition for Pauli logicals can be leveraged to discover new logical operations on qLDPC codes. 
For instance, the familiar ability to deform Pauli logicals with stabilizers---i.e. boundaries of the chain complex---becomes a way to search for constant-depth unitary implementations of (non-)Clifford logical gates. 
Using this strategy, we discover constant-depth unitary implementations of the full logical Clifford group on any number of blocks of the 2D toric code, including within a single block, and addressable logical $\mathsf{CCZ}$ gates on any triple of logical qubits on any number of blocks of the 3D toric code. 
Beyond manifold codes, we find addressable non-Clifford gates on codes with many encoded qubits. The chain map hierarchy naturally encompasses and extends other constructions of logical gadgets, for instance providing a universal parameterization of cup products. 
As such, our work provides a unified, useful, and intuitive language for computing with qLDPC codes. 

\end{abstract}

\maketitle

\section{Introduction}
\label{sec:Intro}

\begin{figure*}[!t]
    \centering
    \includegraphics[width=\textwidth]{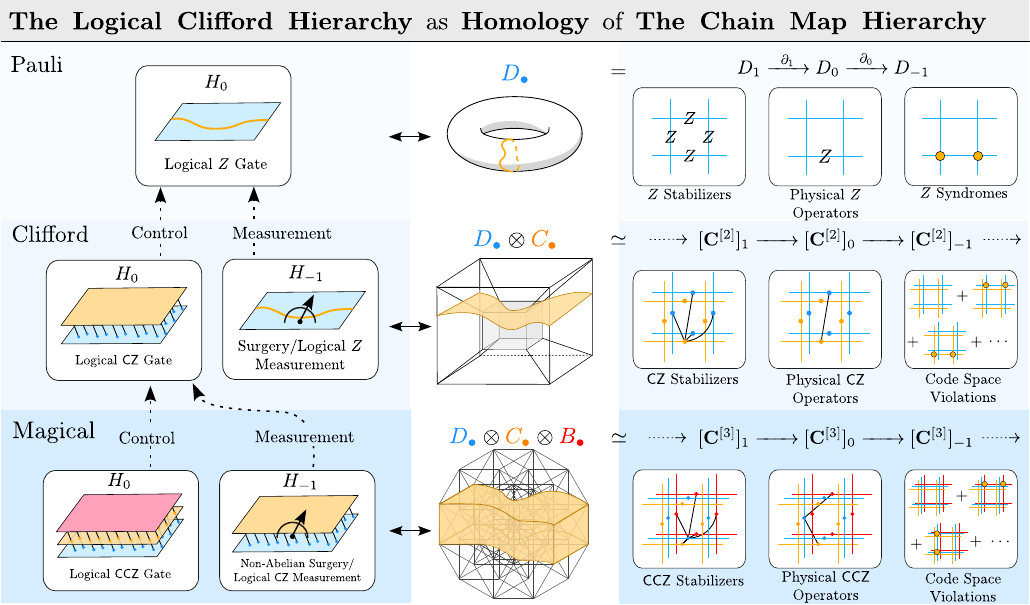}
    \caption{\textbf{Computation with qLDPC Codes and the Chain Map Hierarchy.}
    The relationship between Pauli logical gates and the homology of a chain complex $D_{\bullet}$ (top row) is extended to  Clifford gates (middle row)  and  non-Clifford gates up the Clifford hierarchy (bottom row). 
    In particular, (left column) certain logical operations---both unitary and measurement---from the \(\ell\)th level of the Clifford hierarchy (illustrated here for \(\mathsf{C}^{\ell-1}\mathsf{Z}\) gates) are encoded in the homology classes  of (right column) an auxiliary chain complex \(\CMH{\ell}\).
    We call the collection of these auxiliary chain complexes the \textit{chain map hierarchy}.
    (Middle row) Concretely, physical Clifford gates between codes \(C_{\bullet}\) and \(D_{\bullet}\) are encoded in $\CMH{2}$, which is isomorphic to the tensor product $\CMH{2} \isom D_{\bullet} \otimes C_\bullet$.
    For instance, a tensor like \(d \otimes c\) encodes a \(\mathsf{CZ}\) gate between qubits at \(c \in C_\bullet\) and \(d \in D_\bullet\), and is thought of as a map \(c \mapsto d\).
    Logical Clifford gates are encoded as homologically non-trivial cycles of this chain complex, and can be deformed by boundaries corresponding to logically trivial \textit{Clifford stabilizers}.
    (Bottom row) Higher levels of the hierarchy are iteratively defined as $\CMH{3} \isom D_{\bullet} \otimes C_{\bullet} \otimes B_{\bullet}$, and so on.
    A tensor like \(d \otimes c \otimes b\) encodes, for instance, a \(\mathsf{CCZ}\), and is thought of as a map from \(b\) to the \(\mathsf{CZ}\) encoded by \(d \otimes c\).
    }
    \label{fig:overview}
\end{figure*}

Determining how to reliably store and process quantum information in physical systems composed of unreliable components \cite{von1956probabilistic} is one of the most conceptually rich and technologically pressing fronts in physics. 
In this effort, much has been understood with regards to how to \emph{store} information---that is, in developing schemes for a robust quantum memory~\cite{Terhal2015QECreview,Breuckmann_2021_rev}.
As such, the critical next step is to determine: \emph{How do we manipulate stored information?}

This is especially important given recent advances in quantum low-density parity check (qLDPC) codes, the best of which store many logical qubits into a single code block~\cite{Breuckmann_2021_rev, MacKay_2004, Kovalev_2013, Tillich_2014, Leverrier_2015, Hastings_2021, Breuckmann_2021,panteleev_asymptotically_2022, leverrier2022quantumtannercodes, dinur_good_2023}.
Ultimately, for these logical qubits to be useful for computation, we need to be able to individually and arbitrarily manipulate them with the same fluency with which we control their constituent physical qubits.

However, even for simple qLDPC codes, we understand surprisingly little.
Forgetting more complicated codes such as those that store tens or hundreds of logical qubits, there are substantial gaps in our understanding even for the 2D toric code~\cite{Kitaev_2003}, which has just \emph{two} encoded qubits.
For instance, with current approaches~\cite{horsman2012surface,litinski2019game}, simply performing a logical Hadamard on one of the two qubits involves a complex architecture of ancillas scaffolding the code~\cite{cohen2022lowoverhead}.

This is indicative of the state of the art.
The current approach is frequently to construct examples of codes or logical gadgets that meet a given target---e.g. gadgets at a particular level of the Clifford hierarchy or with a given spacetime overhead---rather than to understand the principles and structures underlying logical computation.
Ideally, we would understand the landscape of \emph{all} codes \emph{and} the logical operations that can be performed with them.
Without such a holistic framework, gaps in our understanding are left behind that rapidly compound to impede future progress. 
Indeed, finding more principled approaches to fault-tolerant computation in qLDPC codes is increasingly the focus of recent work~\cite{huang_2023_homomorphicCNOT, ide2025faulttolerant, williamson2024lowoverhead, zhu2025,
breuckmann2025cupsgatesicohomology, lin2024transversalnoncliffordgatesquantum, zhu2025topological, hsin2025classifyinglogicalgatesquantum, li_poincare_2025, li2026theorycohomologicalinvariantsquantum, haruna2026note, haruna2026homological}.

It is therefore useful to take a step back and ask: \emph{How is it that we have made so much progress on understanding quantum memory?}
A conceptual approach for understanding quantum memories originated from the aforementioned 2D toric code, a topologically ordered system whose properties are derived from the homology of the 2D torus~\cite{Kitaev_2003}.
By generalizing the toric code, it was realized that any Calderbank-Shor-Steane (CSS) code \cite{CalderBankShor, Steane1996} can be understood through a homological formalism based on the mathematical structure of a chain complex \cite{Kitaev1997quantumcomputation,Kitaev_2003, Bombin_2007, Freedman2001}. 

The chain complex formalism has been a cornerstone of progress on quantum memory. 
It repackages the stabilizer description of quantum codes to mirror the structure of the toric code.
Importantly, the chain complex provides a full characterization of the code's Pauli logicals: homology classes label distinct logical operators, while boundaries label stabilizers that relate different physical representatives of the same logical.
Key properties of the codes, such as the code's rate and distance, are derived from this homological description.
Moreover, powerful tools from homological algebra---notably homological products~\cite{Weibel1994HomologicalBook, bravyi2013homologicalproductcodes, Tillich_2014, Hastings_2021,Breuckmann_2021, panteleev_asymptotically_2022}---are responsible for both the best finite instances~\cite{Bravyi_2024, yoder2025tourgross, cain2026shor, bhardwaj2026high, hong2026quantum},  and asymptotic families~\cite{Hastings_2021, Breuckmann_2021, 
Panteleev_2021,
panteleev_asymptotically_2022,leverrier2022quantumtannercodes, dinur_good_2023} of qLDPC codes.

In our work, we generalize the chain complex description of quantum codes---i.e. memory---to an organizing principle for computation with any qLDPC code. 
We define auxiliary chain complexes that, analogous to the Pauli case, provide a full characterization of the code's Clifford and non-Clifford logical gates.
Each logical action is labeled by a distinct homology class, with representative cycles encoding physical gate implementations.
Strikingly, the chain complex boundaries furnish a notion of (\textit{non-})\textit{Clifford stabilizers} that relate distinct physical implementations of the same logical action.
These stabilizer deformations can be used to find constant-depth implementations of logical actions (when they exist).
For instance, returning to the 2D toric code, we use this approach to find a logical \emph{unitary} gate that performs a Hadamard on one of its two qubits, circumventing the need for a complex ancilla architecture.
Indeed, our approach reveals constant-depth unitary implementations of the entire Clifford group in the toric code.

Our formalism encompasses and extends most known constructions of logical unitary gates, and even extends to measurement gadgets.
This includes Clifford gadgets (e.g. code surgery~\cite{horsman2012surface,litinski2019game,Vuillot_2019,cohen2022lowoverhead,cowtan2024ssip,Cowtan2024SurgeryUniversal,swaroop2024universal,williamson2024lowoverhead,cross2024improved,ide2025faulttolerant,zhang2025timeefficient,zhang2025accelerating,zheng2025highrate,baspin2025fast,cowtan_fast_2025,cowtan2025parallel,he2025extractors}, fold-transversal gates~\cite{Kubica2015Unfolding,Moussa2016FoldedSurface,Breuckmann2024foldtransversal}, homomorphic transversal gates~\cite{huang_2023_homomorphicCNOT,xu2024constant}, etc.) and constructions of non-Clifford gadgets (e.g. those based on cup products~\cite{zhu2025,breuckmann2025cupsgatesicohomology, lin2024transversalnoncliffordgatesquantum, zhu2025topological, hsin2025classifyinglogicalgatesquantum,menon2025magic,  tiew2026copycupgatestensorproducts, li_poincare_2025, li2026theorycohomologicalinvariantsquantum} and recently introduced non-Abelian code surgery gadgets~\cite{huang2025generatinglogicalmagicstates, margaritaknots, sajithcodes, zhu_non-abelian_2026, williamson_fast_2026, Christos2026nonabelian}). 
Consequently, our work presents a unified and useful framework for thinking about computation in qLDPC codes.

\subsection{Summary of Results}

\setcounter{secnumdepth}{5}
\paragraph{The Chain Map Hierarchy} 
We develop a chain complex formalism that places Pauli, Clifford, and non-Clifford operations on the same footing.

In the conventional chain complex formalism, \emph{physical} Pauli operators are associated with chains in the complex, while \emph{logical} Paulis correspond to cycles [Fig.~\ref{fig:overview}, top row]. 
Physical Clifford gates (e.g. $\mathsf{CZ}$ gates) acting between two codes map Paulis of one to Paulis of the other, and are thus encoded as linear maps between the code's (co)chain complexes $C^{\bullet}$ and $D_{\bullet}$. 
We show that these maps can themselves be organized as the chains of a chain complex $\CMH{2}$ [Fig.~\ref{fig:overview} (middle row, right)], whose cycles are \emph{chain maps}---maps with well-defined homological action.
The homology classes of $\CMH{2}$ hence label distinct \textit{logical} Clifford actions between these codes, and their representative cycles encode the \textit{full space} of physical implementations of the same logical action. 
The aforementioned Clifford stabilizers are precisely the boundaries of $\CMH{2}$ that deform these cycles/physical implementations.

The chain complex $\CMH{2}$ is simply the tensor product $ \CMH{2} \isom C_{\bullet} \otimes D_{\bullet}$ (closely related to the hypergraph product), analogous to how the space of maps between vector spaces $V$ and $W$ is isomorphic to $W \otimes V^*$. 
Hence, logical Clifford gates between $C_{\bullet}$ and $D_{\bullet}$ are in correspondence with logical Pauli gates of an auxiliary code given by their tensor product!

This construction iteratively extends to gates at any level $\ell$ of the \textit{Clifford hierarchy}~\cite{Gottesman1999Teleportation}.%
\footnote{The Clifford hierarchy is defined recursively: level $\ell = 1$ contains the Pauli operators, while level $\ell$ gates maps Paulis to  level $\ell -1$ gates (e.g. Cliffords live at $\ell = 2$ and non-Cliffords at $\ell > 2$).}
For instance, non-Clifford gates at level $\ell = 3$ (e.g. $\mathsf{CCZ}$)  map Pauli operators in a code $B_{\bullet}$ to Clifford operators between codes $C_{\bullet}$ and $D_{\bullet}$. 
Hence, these gates are captured by a chain complex of maps between $B_{\bullet}$ and $\CMH{2}$, which we dub $\CMH{3} \equiv D_\bullet \otimes C_\bullet \otimes B_\bullet$ [Fig.~\ref{fig:overview} (bottom row)].
This defines a collection of chain complexes encoding gates at increasing levels of the Clifford hierarchy, called the \textit{the chain map hierarchy}.
These complexes characterize inter-block unitary logical gates of $\mathsf{C}^{\ell} \mathsf{Z}$ and $\mathsf{C}^{\ell} \mathsf{X}$ type, certain intra-block gates (e.g. $\mathsf{S}$ gates), and even logical measurement (code surgery) gadgets up the Clifford hierarchy.

\paragraph{Newly Discovered Transversal Gates} The homology classes of the chain map hierarchy characterize the space of physical implementations of logical Clifford and non-Clifford gates.
We explicitly demonstrate how this can be used to find \textit{transversal}---i.e. constant depth%
\footnote{Some authors use transversal to mean strictly depth one, or make a distinction between geometrically local and non-local (``fold transversal''~\cite{Moussa2016FoldedSurface}) gates.
We say a circuit is transversal if it is constant depth with (possibly non-local) few-body gates.
}%
---unitary gates in qLDPC codes, e.g.:
\begin{itemize}
    \item \textit{The Full Clifford Group in the 2D Toric Code:} [Fig.~\ref{fig:applications}(a)] We find implementations of the entire Clifford group on any number of copies of the toric code (including a single copy).
    This includes new addressable $\mathsf{CZ}$ and $\mathsf{CNOT}$ gates acting between the two logical qubits of a single toric code and new addressable $\mathsf{S}$ and Hadamard gates each acting on only one of the qubits of the toric code.

    \item \textit{Addressable Non-Clifford Gates in the 3D Toric Code:} [Fig.~\ref{fig:applications}(b)] We discover addressable $\mathsf{CCZ}$ gates that can act on any triple of logical qubits in any number of copies of the 3D toric code, including an intra-block logical $\mathsf{CCZ}$ gate acting within a single copy.

    \item \textit{Addressable Gates in Non-Manifold Codes:} [Fig.~\ref{fig:applications}(c)] We find \textit{addressable}\footnote{For codes with a number of logical qubits that diverges with increasing system size, ``addressable'' means that a gate acts on an $O(1)$ logical qubits. For codes with a constant number of logical qubits (e.g. the toric code), addressable is taken to mean that the gate acts on the minimal number of logical qubits possible (i.e. a logical $\mathsf{CZ}$ must act on at least two logical qubits, a logical $\mathsf{S}$ must act on at least one logical qubit, etc.).} Clifford and non-Clifford gates in families of fracton models whose number of encoded logical qubits diverge with system size. 

\end{itemize}

Our formalism can also be used to assess when a logical gate is impossible to implement transversally.
While a fuller treatment is left to future work, we make an informal generalization of the Bravyi-K\"onig theorem~\cite{Bravyi2013GatesforLocalStabilizer} to non-local circuits.

\begin{figure}
    \centering
    \includegraphics[width = 247 pt]{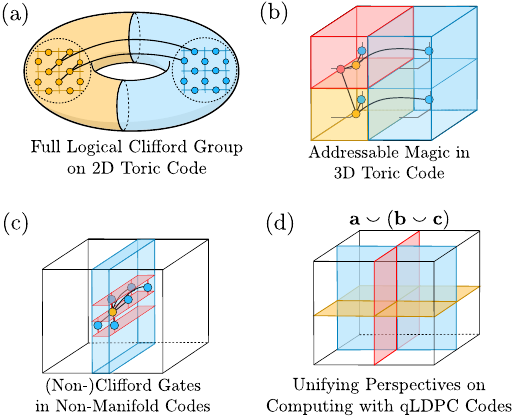}
    \caption{\textbf{Applications.}
    The chain map hierarchy can be used to discover new transversal logical gates.
    For instance, we use it to discover (a)~the full Clifford group on the 2D toric code, (b)~new addressable non-Clifford gates in the 3D toric code, and (c)~addressable Clifford and non-Clifford gates in families of fracton codes with many encoded logical qubits.
    (d)~The chain map hierarchy also unifies several other approaches to finding transversal gates, including cup product constructions and code surgery gadgets for logical measurement.
    }
    \label{fig:applications}
\end{figure}

\paragraph{Unifying Other Perspectives on Computing with qLDPC Codes} Our formalism encompasses and extends to other approaches to computation and to finding logical operations on qLDPC codes.
In particular, logical unitary gates and measurement (code surgery) gadgets appear as different homology classes of the complexes in the chain map hierarchy.
Moreover, approaches to finding transversal gates---e.g.\ via cup products~\cite{breuckmann2025cupsgatesicohomology, hsin2025classifyinglogicalgatesquantum, lin2024transversalnoncliffordgatesquantum}---can be understood and generalized within our formalism.
In particular, we use the chain map hierarchy to develop a universal parameterization of both the space of cup products and cohomology invariants obtained from cup products.
This enables us to find, for instance, cup products in qLDPC codes that live outside of any prior parameterizations of the cup product present in literature. 

The paper is organized as follows. 
Section~\ref{sec:Intuition} provides a minimal and intuitive presentation of some of the key ideas of our work.
Section~\ref{sec:ChainMapHierarchy} then precisely defines and characterizes the chain map hierarchy.
Section~\ref{sec:DiscoveringGates} applies this formalism to discover new transversal gates in several well known qLDPC codes, including the full Clifford group in the 2D toric code.
Section~\ref{sec:UniversalCup} explains how prior approaches to constructing gates---namely cup products---fit within our framework.
We conclude with discussion in Section~\ref{sec:discussion}, which notably includes an informal derivation of a generalization of the Bravyi-K\"onig theorem~\cite{Bravyi2013GatesforLocalStabilizer} to transversal circuits with non-local connectivity.

\makeatletter 
\def\l@subsubsection#1#2{}
\makeatother 
\tableofcontents

\section{Overview of Key Ideas}
\label{sec:Intuition}

Before presenting a complete treatment of our results, we provide an intuitive overview of some key ideas and applications of our work.
This is not intended to be exhaustive and several results of interest are left for later sections---e.g. intra-block gates, relations to cup products and code surgery, and several explicit gate constructions.

The bulk of Section~\ref{subsec:Intuition_UnifiedDescription} describes how physical implementations of logical Clifford gates on two CSS codes are directly encoded in the Pauli logicals of an auxiliary code given by their tensor product.
We illustrate this idea concretely in the case of the 2D toric code.
In particular, we show that the physical implementation of logical $\mathsf{CZ}$ gates on the 2D toric code can be ``read-off'' from the Pauli logicals of the 4D toric code [explicitly see  Fig.~\ref{fig:Cliffordspider} and Section~\ref{subsubsec:informal_chainmapastensorproduct}].
We then naturally generalize this idea to handle gates up the Clifford hierarchy, whose implementations are encoded in a family of auxiliary codes called the chain map hierarchy.

Section~\ref{subsec:Intuition_Application} then highlights a key application of these insights.
We primarily discuss how deformations of the logical Pauli operators of the auxiliary code map back to deformations of the physical implementation of a logical (non-)Clifford gate.
This can then be used as a tool for finding constant-depth implementations of logical gates up the Clifford hierarchy. 

\subsection{Unifying Logical Pauli and (Non-)Clifford Transversal Gates}
\label{subsec:Intuition_UnifiedDescription}

\begin{shaded*}
\noindent
    \textbf{Key Ideas.} The key ideas of this section are:
    \begin{enumerate}
        \item Logical Cliffords can be deformed by local \textit{Clifford stabilizers}.   [Section~\ref{subsubsec:informalCliffordstab}]

        \item Certain Clifford gates can be organized into a chain complex called the \textit{chain map complex} $\CMH{2}$, whose homology labels logical Clifford gates and whose boundaries encode Clifford stabilizers [Section~\ref{subsubsec:intuitive_Chainmapcomplex}].

        \item  If the Cliffords act between codes described by complexes $C_{\bullet}$ and $D_{\bullet}$, then $\CMH{2} \simeq D_{\bullet}\otimes C_{\bullet}$.
        Logical Cliffords can hence be ``read off' from the Pauli logicals of an auxilliary hypergraph product code [Fig.~\ref{fig:Cliffordspider} and Section~\ref{subsubsec:informal_chainmapastensorproduct}].

        \item Each of the above iteratively extends to certain gates up the Clifford hierarchy [Section~\ref{subsubsec:informalCMH}].
    \end{enumerate}

\end{shaded*}

We demonstrate these in the context of the toric code, a CSS stabilizer code defined on a square lattice tiling of a torus. 
As a CSS code, the toric code's stabilizers,  logical Pauli operators, and  syndromes under the action of Pauli operators can be understood through the mathematical concept of a chain complex, $C_\bullet$, and its dual cochain complex, $C^\bullet$ (see Section~\ref{subsec:review} for a detailed review)~\cite{Weibel1994HomologicalBook,hatcher2002algebraic,Kitaev1997quantumcomputation,Kitaev_2003,Bombin_2007}:
\begin{equation} \label{eq-chaincomplex}
        \includegraphics[valign = c, scale = 0.9]{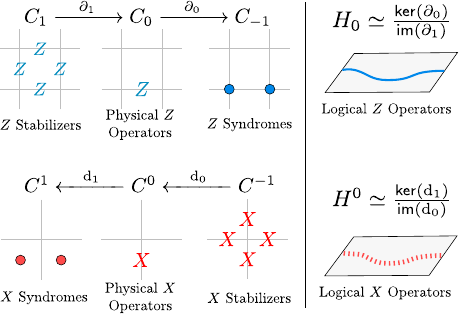}
\end{equation}
Both \(C_\bullet\) and \(C^\bullet\) are sequences of binary vector spaces
connected by matrices, called boundary maps ($\partial_j$)  and coboundary maps ($\exd_j = \partial_j^{\mathsf{T}}$), respectively. 
 Physically,
\(C_\bullet\) organizes the \(Z\)-type Pauli data: \(C_1\) labels the $Z$ stabilizers, \(C_0\) labels physical \(Z\) operators, and
\(C_{-1}\) labels how these operators can violate the code space. Dually, \(C^\bullet\)
organizes the corresponding \(X\)-type Pauli data, with \(C^{-1}\),
\(C^0\), and \(C^1\) labeling \(X\) stabilizers, physical \(X\)
operators, and their syndromes.  For the toric code, the spaces are in correspondence with the plaquettes, edges, and vertices of the lattice,  with the (co)boundary maps encoding their incidence relations.

In a bit more detail, the basis vectors of $C_0$ or $C^0$ are in correspondence with the qubits comprising the CSS code.
Vectors in these spaces are called $0$-chains and $0$-cochains and they describe products of $Z$ and $X$ operators at the qubit locations where the vectors are non-zero.
A similar idea is true for $C_1$ and $C^{-1}$ whose basis vectors label a basis of $Z$ and $X$ stabilizers that generate the stabilizer group.
The boundary maps $\partial_1: C_1 \to C_0$ and $\exd_0: C^{-1} \to C^0$ describe how this basis maps into physical products of $Z$ and $X$ operators.
Finally, the basis vectors of $C_{-1}$ and $C^1$ encode the elementary syndromes of the code; the boundary maps $\partial_0: C_{0} \to C_{-1}$ and $\exd_1: C^0 \to C^1$ encode the syndromes created by products of $Z$ and $X$ operators.
The fact that stabilizers don't violate the code space is encoded in the fact that the (co)boundary maps satisfy the characteristic property that $\partial_{j} \cdot \partial_{j + 1} = 0$ ($\exd_j \cdot \exd_{j - 1}= 0$)---i.e. the boundary of a boundary is zero.

We further note that the Pauli $X$ and $Z$ logical gates are labeled by the the (co)homology groups $H_0 \isom \mathsf{ker}(\partial_0)/\mathsf{im}(\partial_1)$ and $H^0 \isom \mathsf{ker}(\exd_1)/\mathsf{im}(\exd_0)$.
This is because such operators do not violate the code space---i.e. live in $\mathsf{ker}(\partial_0)$ and $\mathsf{ker}(\exd_1)$---and two such operators are treated as equivalent if they can be \textit{deformed} into each other by products of stabilizers---i.e.  $\mathsf{im}(\partial_1)$ and $\mathsf{im}(\exd_0)$.

\subsubsection{Logical Gates as Chain Maps}

We will soon demonstrate that Clifford and non-Clifford operators between codes also form a chain complex, whose zeroth homology group labels \textit{logical} (non-) Clifford gates.
We begin by recalling that logical Clifford gates are conventionally described  as \emph{maps} between chain complexes~\cite{huang_2023_homomorphicCNOT}. This is natural from the defining property of a Clifford unitary \(U\),  which maps every Pauli operator \(P\) to another Pauli operator \(P'\) under conjugation: \(UPU^\dagger=P'\). 

As a concrete example, consider two copies of the toric code, which we will depict in orange and blue, with a transversal Clifford $\mathsf{CZ}$ gate between them: 
\vspace{-4mm}
\begin{equation}\label{eq:globCZ}
    \includegraphics[valign = c]{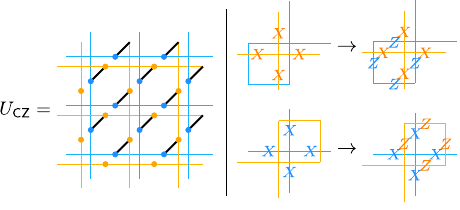}
\end{equation}
Here, each black line indicates a physical $\mathsf{CZ}$ gate, and the action of this gate on stabilizers is shown on the right. Under
conjugation, a two-qubit \(\mathsf{CZ}\) gate leaves \(Z\) unchanged and maps an \(X\) on
either qubit to itself times a \(Z\) on the other.  Thus,
conjugation by the $\mathsf{CZ}$'s above takes  $\textcolor{orange}{X} \to \textcolor{orange}{X}\textcolor{dodgerblue}{Z}$, so each \(\textcolor{orange}{X}\) acquires a $\textcolor{dodgerblue}{Z}$ acting a half diagonal translation backwards, and $\textcolor{dodgerblue}{X}\to \textcolor{dodgerblue}{X}\textcolor{orange}{Z}$, with the $\textcolor{orange}{Z}$ acting a half diagonal translation forwards. Consequently, each \(X\) stabilizer maps to itself times a half-translated \(Z\) stabilizer in the other code.
Since the stabilizer group is preserved,  the circuit preserves
the joint code space and implements a \emph{logical} Clifford gate.
We can use these rules to derive how the logical Pauli operators transform:
\begin{equation} \label{eq-logicalactionglobalCZ}
    \includegraphics[valign = c, scale = 0.9]{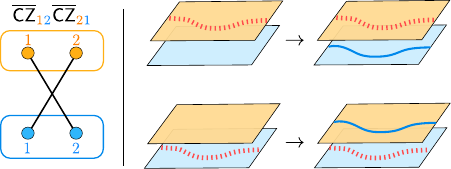},
\end{equation}
where the left shows the resulting logical \(\mathsf{CZ}\) gates between the
four encoded qubits---two logical qubits in each toric-code block, while the right shows how an $X$ Pauli logical of either code is transformed into itself times a $Z$ logical of the other.

Since \(X\)-type Pauli operators are represented by cochains and
\(Z\)-type Pauli operators by chains~\cite{zhu2025}, the additional
\(\textcolor{dodgerblue}{Z}\) acquired by an
\(\textcolor{orange}{X}\) under conjugation with $\mathsf{CZ}$ gate is naturally encoded  
in a map $\phii_{\bullet}: \textcolor{orange}{C^{\bullet}} \to \textcolor{dodgerblue}{C_{\bullet}}$: 
    \begin{equation} \label{eq-chainmapfirst}
    \begin{tikzpicture}[scale=0.7, baseline={([yshift=-.5ex]current bounding box.center)}]
      \node (C1)     at (2, 0)    {\small $\textcolor{orange}{C^{-1}}$};
      \node (C0)     at (4, 0)    {\small $\textcolor{orange}{C^0}$};
      \node (Cm1)    at (6, 0)    {\small $\textcolor{orange}{C^{1}}$};

      \node (D1)     at (2, -1.8) {\small $\textcolor{dodgerblue}{C_{1}}$};
      \node (D0)     at (4, -1.8) {\small $\textcolor{dodgerblue}{C_0}$};
      \node (Dm1)    at (6, -1.8) {\small $\textcolor{dodgerblue}{C_{-1}}$};

      \draw[-stealth, orange(ryb)] (C1)   -- node[above]{\small $\textcolor{orange(ryb)}{\textcolor{orange}{\exd_0}}$} (C0);
      \draw[-stealth, orange(ryb)] (C0)     -- node[above]{\small $\textcolor{orange(ryb)}{\textcolor{orange}{\exd_1}}$} (Cm1);

      \draw[-stealth, dodgerblue] (D1)     -- node[below]{\small $\partial_1$} (D0);
      \draw[-stealth, dodgerblue] (D0)     -- node[below]{\small $\partial_0$} (Dm1);

      \draw[-stealth] (C1)  -- node[left] {\small $\phii_{-1}$}  (D1);
      \draw[-stealth] (C0)  -- node[left] {\small $\phii_{0}$}  (D0);
      \draw[-stealth] (Cm1) -- node[right]{\small $\phii_1$} (Dm1);
    \end{tikzpicture} \quad  U_{\mathsf{CZ}} = \prod_{e} \mathsf{CZ}_{\textcolor{orange}{\mathbf{e}}, \textcolor{dodgerblue}{\phii_0(\textcolor{orange}{\textbf{e}})}}
    \end{equation}
where $\phii_{-1}, \phii_{0},$ and $\phii_{1}$ describe how the $X$ stabilizers, $X$ operators, and $X$-syndromes of the orange toric code map to their $Z$ analogs in the blue toric code.
 Notably, $\phii_0:\textcolor{orange}{C^0} \to \textcolor{dodgerblue}{C_0}$ specifies the pattern of physical $CZ$ gates: given an edge $\textcolor{orange}{\textbf{e}}$ labeling a qubit in the orange code, $\textcolor{dodgerblue}{\phii_0(\textcolor{orange}{\textbf{e}})}$ is the set of edges (qubits) in the blue code coupled to it with $\mathsf{CZ}$ gates; in
Eq.~\eqref{eq:globCZ}, this is the qubit in the blue code displaced by half a diagonal backwards.
Provided that the matrices in $\phii_{\bullet}$ are sparse, the gate above is implemented in constant depth.

To ensure the gates in Eq.~\eqref{eq-chainmapfirst} preserve the code space and hence define a \emph{logical} gate, we require that the chain map sends stabilizers to stabilizers, i.e.\ boundaries to (co)boundaries.
Formally, this is ensured when the  diagram in Eq.~\ref{eq-chainmapfirst} \textit{commutes}, i.e. 
\begin{equation} \label{eq-chainmapcondition}
    \phii_{j} \circ \textcolor{orange}{\text{d}_j} + \textcolor{dodgerblue}{\partial_{-j + 1}} \circ \phii_{j - 1} = 0
\end{equation}
for all $j$.
This is  the defining condition for
\(\phii_\bullet\) to be a \textit{chain map}, reviewed in detail in Sec.~\ref{subsubsec:chainmap}. 
We remark that, while we introduce chain maps in the context of logical $\mathsf{CZ}$ gates, they can also describe gates related to $\mathsf{CZ}$ by the Hadamards on either qubit (e.g. $\mathsf{CNOT}$) by replacing chains with cochains, or vice versa (Sec.~\ref{subsubsec:WhichCliffords}).
Furthermore, we later show how these chain maps can be used to describe gates acting within a single code block (c.f. Section~\ref{sec:IntrablockGates}).

The chain map perspective on Clifford gates is natural but, as described up to this point, has two important shortcomings.
First, it is apparently not clear how to organize the space of logical Clifford circuits.
This is in contrast with logical Paulis, which can be systematically found by computing the homology of the chain complex.
Second, it is not clear whether a logical Clifford gate can be modified while preserving its logical action---again in contrast to Pauli gates.
We now fill in both of these gaps.

\subsubsection{Local Deformations of Logical Clifford Gates} \label{subsubsec:informalCliffordstab}

Having reviewed how chain maps can be used to describe logical $\mathsf{CZ}$ gates, our work introduces the idea that these logical $\mathsf{CZ}$ gates can be locally \textit{deformed} by certain logically trivial \emph{Clifford stabilizers}---analogous to the way that logical Pauli gates are deformable by Pauli stabilizers. This will motivate a chain complex description of Clifford operators.

As an explicit example, consider the following \textit{local, logically trivial, $\mathsf{CZ}$ circuit} between the two toric codes: 
\begin{equation} \label{eq-firstnullhomotopy}
\includegraphics[valign = c, scale = 0.9]{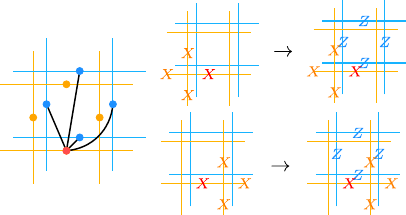}
\end{equation}
The circuit acts on only five qubits: a single orange qubit (marked red) is coupled by
four physical \(\mathsf{CZ}\) gates to the four blue qubits surrounding one
plaquette, and this local pattern can be translated anywhere on the lattice.
It preserves the code space since it maps stabilizers to stabilizers [right of Eq.~\eqref{eq-firstnullhomotopy}], and is logically trivial
because it changes Pauli operators only by stabilizers---in direct analogy to the Pauli stabilizers in the toric code.
Thus, it can be multiplied with a logically non-trivial Clifford gate to obtain a gate with the same logical action, but with a slightly different action on the physical qubits.

A general understanding of these gates can be obtained by noting that it maps an $\textcolor{orange}{X}$ operator  to a $\textcolor{dodgerblue}{Z}$ stabilizer (associated with a plaquette). 
As a result, it can be encoded in a map $Q_0: \textcolor{orange}{C^0} \to \textcolor{dodgerblue}{C_1}$.
Reversing the role of the orange and blue codes gives a distinct Clifford stabilizer that maps an $\textcolor{dodgerblue}{X}$ operator to a $\textcolor{orange}{Z}$ stabilizer,  encoded in a map $Q_1^{\mathsf{T}}: \textcolor{dodgerblue}{C^0} \to \textcolor{orange}{C_1}$.
More generally, Clifford stabilizers are generated from ``left diagonal" maps $Q_i: \textcolor{orange}{C^i} \to \textcolor{dodgerblue}{C_{-i+1}}$, as shown on the left of: 
\begin{equation} \label{eq:informal_null_homotopy}
    \includegraphics[valign = c, scale = 0.95]{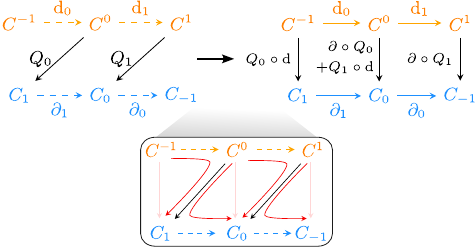}
\end{equation}
In the above equation, the left diagram is not required to be commutative (indicated by the dashed arrows); $Q_0$  and $Q_1$ can be an arbitrary linear maps. 
As an example, in Eq.~\eqref{eq-firstnullhomotopy}, $Q_0$ is the map that is zero on all edges except the one marked in red, which it maps to a plaquette stabilizer, and $Q_1 = 0$. 
Note that the \(Q_i\) do not themselves
directly specify a physical gate: gate patterns are instead encoded by 
``downward'' maps $\textcolor{dodgerblue}{C^0} \to \textcolor{orange}{C_0}$ which map qubits to qubits [cf. Eq.~\ref{eq-chainmapfirst}]. The construction in Eq.~\eqref{eq:informal_null_homotopy} converts an arbitrary collection of left-diagonal  maps $Q_i$  into precisely such a downward {chain map}  [right diagram in Eq.~\eqref{eq:informal_null_homotopy}] by composing with boundary and coboundary maps.  
For instance, its degree zero component,   $\partial \circ Q_0 + Q_1 \circ \text{d}$, arises from the two paths from $\textcolor{orange}{C^0} \to \textcolor{dodgerblue}{C_0}$  in the left diagram of Eq.~\eqref{eq:informal_null_homotopy}.
The resulting chain maps are called \textit{null homotopic}~\cite{Weibel1994HomologicalBook,hatcher2002algebraic}; they have trivial action on (co)homology and, through Eq.~\eqref{eq-chainmapfirst}, therefore define logically trivial Clifford gates. For binary chain complexes, all chain maps with trivial action on (co)homology are null homotopic~\cite{Weibel1994HomologicalBook}.  Section~\ref{subsec:null} provides a detailed treatment of these maps and their associated Clifford stabilizers.

Notice that the construction in Eq.~\eqref{eq:informal_null_homotopy} is reminiscent of a boundary operation: arbitrary maps $Q_i$  are mapped to Clifford stabilizers, similar to how arbitrary  
collections of plaquettes are mapped to Pauli stabilizers in an ordinary chain complex.
We  now make this intuition  precise.

\subsubsection{The Chain Map Complex: Unifying Cliffords and Paulis} \label{subsubsec:intuitive_Chainmapcomplex}

Having discovered the null homotopic $\mathsf{CZ}$ stabilizers of Eqs.~\eqref{eq-firstnullhomotopy} and~\eqref{eq:informal_null_homotopy}, 
we are naturally led to organizing the $\mathsf{CZ}$ operators of a code into a chain complex, in analogy to the chain complex description of Pauli operators in Eq.~\eqref{eq-chaincomplex}.
In particular, we show that the $\mathsf{CZ}$ gates of a qLDPC code are organized into a chain complex $\CMH{2}$, which we call the \textit{chain map complex} (developed formally in Sec.~\ref{subsubsec:chain_map_complex}).
The chain map complex 
comprises a sequence of vector spaces that correspond to the spaces of maps between two other (co)chain complexes~\cite{Eilenberg1966ClosedCategories,Weibel1994HomologicalBook}, which we label  $\textcolor{orange}{C^{\bullet}}$ and $\textcolor{dodgerblue}{D_{\bullet}}$: 
\begin{equation} \label{eq-chainmapcomplex}
    \includegraphics[valign = c, scale = 0.9]{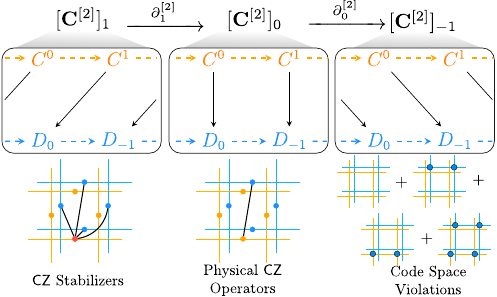}
\end{equation}
Mathematically, $\CMH{2}_{1} \CMH{2}_{0}$, and $\CMH{2}_{-1}$ denote collections of left diagonal ($+1$ over to the left), downwards ($0$ over to the left), and right diagonal maps ($-1$ over to the left) between two chain complexes respectively.
Notably, each of these maps do not need to be chosen such that any of the above diagrams commute.
As a remark, the ``$[2]$'' label in the chain map complex is used to indicate that it encodes operators at the second level of the Clifford hierarchy.
We will later discuss how to build chain complexes $\CMH{\ell}$ describing higher levels.

The physical interpretation of the levels of the chain map complex mirrors the chain complex of Eq.~\eqref{eq-chaincomplex}.
Namely, $\CMH{2}_1$  (maps that run left diagonal) encodes the space of Clifford stabilizers of in Eqs.~\eqref{eq-firstnullhomotopy}~and~\eqref{eq:informal_null_homotopy}, mirroring the Pauli $Z$ stabilizers of $C_1$.
Moreover,  $\CMH{2}_{0}$ encodes the full space of \(\mathsf{CZ}\) circuits (not necessarily logical \(\mathsf{CZ}\) gates) between the two codes, in the same way that $C_0$ in the normal chain complex encodes the space of Pauli $Z$ operators on a code.
Finally, the third space $\CMH{2}_{-1}$  encodes the way in which local $\mathsf{CZ}$ gates violate the code space.
In particular, in contrast to Pauli operators, which generate product states of syndrome violations, Clifford operators generate entangled states of syndrome violations.

As a concrete example, we consider the following vector in $\CMH{2}_0$, which is specified by a collection of maps $\{f_j\}$:
\begin{equation} \label{eq-f0}
    \begin{tikzpicture}[scale=0.7, baseline={([yshift=-.5ex]current bounding box.center)}]
      \node (C1)     at (2, 0)    {\small $\textcolor{orange}{C^{-1}}$};
      \node (C0)     at (4, 0)    {\small $\textcolor{orange}{C^0}$};
      \node (Cm1)    at (6, 0)    {\small $\textcolor{orange}{C^{1}}$};

      \node (D1)     at (2, -1.8) {\small $\textcolor{dodgerblue}{D_{1}}$};
      \node (D0)     at (4, -1.8) {\small $\textcolor{dodgerblue}{D_0}$};
      \node (Dm1)    at (6, -1.8) {\small $\textcolor{dodgerblue}{D_{-1}}$};

      \draw[-stealth, orange(ryb), dashed] (C1)   -- node[above]{\small $\textcolor{orange(ryb)}{\textcolor{orange}{\exd_0}}$} (C0);
      \draw[-stealth, orange(ryb), dashed] (C0)     -- node[above]{\small $\textcolor{orange(ryb)}{\textcolor{orange}{\exd_1}}$} (Cm1);

      \draw[-stealth, dodgerblue, dashed] (D1)     -- node[below]{\small $\partial_1$} (D0);
      \draw[-stealth, dodgerblue, dashed] (D0)     -- node[below]{\small $\partial_0$} (Dm1);

      \draw[-stealth] (C1)  -- node[left] {\small $f_{-1} = 0$}  (D1);
      \draw[-stealth] (C0)  -- node[left] {\small $f_{0}$}  (D0);
      \draw[-stealth] (Cm1) -- node[left]{\small $f_1 = 0$} (Dm1);
    \end{tikzpicture} \quad f_{0}(\textcolor{orange}{\mathbf{e}}) = \begin{cases}
        \textcolor{dodgerblue}{e_1} & \textcolor{orange}{\mathbf{e}} = \textcolor{orange}{\mathbf{e}_0} \\
        0 & \text{otherwise}
    \end{cases} 
\end{equation}
Such a vector is a basis vector of $\CMH{2}_0$ and specifies a single $\mathsf{CZ}$ gate between the qubit associated with $\textcolor{orange}{\mathbf{e}_0}$  in the orange code and the qubit labeled by $\textcolor{dodgerblue}{e_1}$ in the blue code.
It can be checked that the above diagram does not commute, reflective of the fact that a single $\mathsf{CZ}$ gate between two codes is not a logical operator.

The vector spaces $\CMH{2}_j$ above are connected by boundary maps $\partial_{j}^{[2]}$;  for example, $\partial_{1}^{[2]}$ converts left diagonal maps into downward maps in the same way as was done in Eq.~\eqref{eq:informal_null_homotopy}.
Crucially,
elements in the kernel of the boundary map $\partial_0^{[2]}$ are precisely maps that satisfy the chain map condition, Eq.~\eqref{eq-chainmapcondition}, or equivalently maps for which the diagrams above commute.
This tells us that logical $\mathsf{CZ}$ gates correspond to the $0$-cycles of the chain map complex, $\mathsf{ker}(\partial_0^{[2]})$.
Clifford stabilizers, meanwhile, correspond to $0$-boundaries in $\mathsf{im}(\partial_1^{[2]})$, which are automatically $0$-cycles since $\partial_0^{[2]}\partial_1^{[2]}=0$. 
Moreover, two logical Clifford operators are equivalent if they can be deformed into one another by products of Clifford stabilizers. 
Consequently, different logical Clifford gates are labeled by the homology groups $H_0^{[2]} \isom \mathsf{ker}(\partial^{[2]}_0)/\mathsf{im}(\partial_1^{[2]})$!

This tells us that the logical Clifford gates of the code are described by the homology of a particular chain complex.
Since this chain complex can also be used to define a CSS code, this says that logical Cliffords are in correspondences with the logical Paulis of an auxiliary code.

\subsubsection{The Chain Map Complex as a Tensor Product Complex} \label{subsubsec:informal_chainmapastensorproduct}

While this insight is tantalizing, at present, it may seem as though reasoning about the code associated with the chain map complex would be challenging.
However, as we show in more detail in Section~\ref{subsubsec:chain_maps_as_tensors}, the chain map complex for $C^{\bullet}$ and $D_{\bullet}$ is isomorphic to the tensor product $D_{\bullet} \otimes C_{\bullet}$~\cite{Eilenberg1966ClosedCategories,Weibel1994HomologicalBook}:
\begin{equation} \CMH{2} = 
     \cdots
     \CMH{2}_1 \xrightarrow[]{\partial_1^{[2]}} \CMH{2}_0 \xrightarrow[]{\partial_0^{[2]}} \CMH{2}_{-1}
     \cdots \isom D_{\bullet} \otimes C_{\bullet}
\end{equation}
Intuitively, this follows from the fact that, while chains and cochains in a chain complex are vectors, chain maps are matrices.
However, matrices are simply vectors in a tensor product vector space, motivating the fact that the chain map complex is isomorphic to a tensor product complex.

The tensor product of chain complexes is closely related to the hypergraph product of the codes $C_{\bullet}$ and $D_{\bullet}$~\cite{bravyi2013homologicalproductcodes,Tillich_2014} and there is a large body of literature that has built up an understanding of the properties of the tensor product code given its input codes.
As a concrete example, for the 2D toric code, the tensor product complex is simply the 4D toric code with qubits placed on the centers of plaquettes. 
A physical $Z$ operator on a plaquette of the 4D toric code is represented via the tensor product of two edge chains $f_0 = e_0 \otimes e_1$, which encodes a $\mathsf{CZ}$ gate between qubit $e_0$ in one 2D toric code and qubit $e_1$ in the second 2D toric code [c.f. Eq.~\eqref{eq-f0}].
Consequently, logical \(\mathsf{CZ}\) gates of the 2D toric code (cycles in the chain map complex) can be simply ``read-off'' from the Pauli $Z$ logicals of the 4D toric code!
This point is illustrated in Fig.~\ref{fig:Cliffordspider} and will be discussed in more detail shortly [Subsection~\ref{subsec:Intuition_Application}].

\subsubsection{The Chain Map Hierarchy} \label{subsubsec:informalCMH}

With the insight that Pauli $X$ and $Z$ gates are described with chain complexes and that $\mathsf{CZ}$ gates---traditionally represented as maps between chain complexes $C^{\bullet}$ and $D_{\bullet}$---are also described by a chain complex $\CMH{2} \isom D_{\bullet} \otimes C_{\bullet}$,
we are immediately led to a description of higher levels of the Clifford hierarchy within the same language.
The insight will follow from the fact that gates at the $\ell$-th level of the hierarchy send Pauli's to gates at the $\ell-1$th level.

As an example, note that a $\mathsf{\textcolor{red}{C}\textcolor{orange}{C}\textcolor{dodgerblue}{Z}}$ (which lives at $\ell = 3$) acting between three codes $\textcolor{red}{B_{\bullet}}, \textcolor{orange}{C_{\bullet}},$ and $\textcolor{dodgerblue}{D_{\bullet}}$ acts on Pauli $\textcolor{red}{X}$ by taking it  to $\textcolor{red}{X} \times \mathsf{\textcolor{orange}{C}\textcolor{dodgerblue}{Z}}$ (where $\mathsf{\textcolor{orange}{C}\textcolor{dodgerblue}{Z}}$ lives at $\ell = 2$).
Hence, intuitively, a $\mathsf{\textcolor{red}{C}\textcolor{orange}{C}\textcolor{dodgerblue}{Z}}$ can be viewed as a chain map between the chain complex of code $\textcolor{red}{B^{\bullet}}$ into the \textit{chain map complex} of $\mathsf{\textcolor{orange}{C}\textcolor{dodgerblue}{Z}}$ operators, which is isomorphic to  $\textcolor{dodgerblue}{D_{\bullet}}\otimes \textcolor{orange}{C_{\bullet}} $.
However, as we learned from studying $\mathsf{CZ}$ gates, chain maps between two chain complexes themselves form a chain complex.
In particular, this chain complex, denoted $\CMH{3}$ will be isomorphic to $ (\textcolor{dodgerblue}{D_{\bullet}}\otimes \textcolor{orange}{C_{\bullet}}) \otimes\textcolor{red}{B_{\bullet}}$.

Iterating the logic above tells us that logical gates of the form $\mathsf{C}^{\ell-1} \mathsf{Z}$ (or those related by Hadamards on any qubit) are labeled by homology classes of a chain complex $\CMH{\ell}$ isomorphic to a tensor product of all the codes upon which the gate acts.
The collection of chain complexes describing gates at each $\ell$ is dubbed the \textit{chain map hierarchy}. We develop this construction formally in Section~\ref{subsec:ChainMapHeirarchy}.

\subsection{An Application of the Formalism}
\label{subsec:Intuition_Application}

\begin{figure}
    \centering
    \includegraphics[width=247 pt]{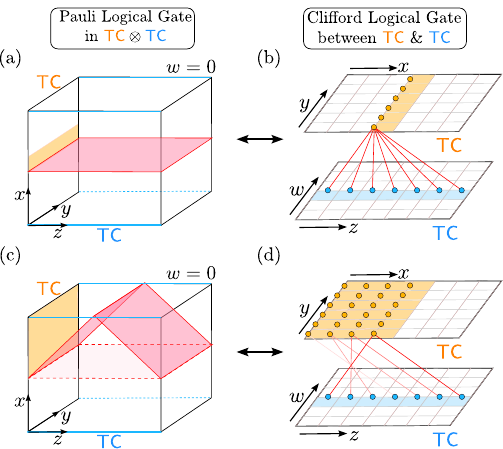}
    \caption{\textbf{Logical Clifford Gates Correspond to Pauli Logical Gates in the Tensor Product.} Logical Clifford gates between two codes---e.g.\ two 2D toric code's $\textcolor{orange}{\mathsf{TC}}$ and $\textcolor{dodgerblue}{\mathsf{TC}}$---can be ``read-off'' from the Pauli logicals of an auxiliary code that is isomorphic to the tensor product of the two codes (in this case, the 4D toric code).
    (a)~Thus, a sheet-like logical \(Z\) in the 4D toric code (shown in a cross-section at fourth coordinate \(w=0\)) encodes the implementation of  (b)~an addressable logical \(\mathsf{CZ}\) gate between two toric codes $\textcolor{orange}{\mathsf{TC}}$ and $\textcolor{dodgerblue}{\mathsf{TC}}$. 
    Explicitly, the 2D sheet is read as a ``plot'' of the gate.
    Every edge along $y$ at a fixed $x$ of $\textcolor{orange}{\mathsf{TC}}$ [every orange qubit in (b)] is connected via a $\mathsf{CZ}$ gate [red lines in (b)] to every edge along $z$ at a fixed $w$ [every blue qubit in (b)] in $\textcolor{dodgerblue}{\mathsf{TC}}$.
    The gate pattern is only illustrated for one orange qubit, but the circuit repeats this pattern along \(y\).
    (c,d) The depth of the corresponding gate is reduced by deforming the logical sheet into a tent shape, so that a qubit at $(x,y)$ in $\textcolor{orange}{\mathsf{TC}}$ only has a $\mathsf{CZ}$ gate with two other qubits in $\textcolor{dodgerblue}{\mathsf{TC}}$.
    }
    \label{fig:Cliffordspider}
\end{figure}

Thus far, we have discussed how Clifford gates between two CSS codes described by chain complexes $C_{\bullet}$ and $D_{\bullet}$ are themselves described by a chain complex isomorphic to $C_{\bullet} \otimes D_{\bullet}$.
We now show how this insight can be used to find transversal implementations of (addressable) gates in qLDPC codes (discussed in depth in Section~\ref{sec:DiscoveringGates}).

\begin{shaded*}
\noindent
    \textbf{Key Idea 2.} Deformations of the logical Pauli operators of the auxiliary code map back to deformations of the physical implementation of a logical (non-)Clifford gate, providing a framework for finding constant-depth realizations of logical gates at any level of the Clifford hierarchy (when they exist).
\end{shaded*}

As a simple illustration of this, we will demonstrate how to easily derive a logical $\mathsf{CZ}$ gate between one logical qubit of an orange toric code $\textcolor{orange}{\mathsf{TC}}$ and one logical qubit of a blue toric code $\textcolor{dodgerblue}{\mathsf{TC}}$.
This is to be contrasted with the gate of Eqs.~\eqref{eq:globCZ}~and~\eqref{eq-logicalactionglobalCZ}, which crucially affects both logical qubits of both toric codes.

To construct this gate, recall that, per our formalism, logical Clifford gates are in correspondence with the Pauli logical gates of the code described by the tensor product of two 2D toric codes.
This code is simply the 4D toric code, whose logical Pauli gates are well understood.
In particular, logical Pauli $Z$ gates of the 4D toric code are given by 2D sheets of $Z$ operators living within the 4D hypercubic lattice.
Any logical action of a gate can be encoded in the logical class of a Pauli, and each distinct logical action (e.g.\ $\CZLogicalColor{1}{2}$, $\CZLogicalColor{2}{1}$ etc.) corresponds to a distinct Pauli.
A canonical representation of one such logical Pauli $Z$ operator in a 3D slice of the 4D lattice is shown in Fig.~\ref{fig:Cliffordspider}(a).

We can view the Pauli logical operator as a \textit{plot} of the Clifford gate acting between the two toric codes.
Specifically, let us label the coordinates of $\textcolor{orange}{\mathsf{TC}}$ by $x$ and $y$ and the coordinates of $\textcolor{dodgerblue}{\mathsf{TC}}$ by $z$ and $w$.
Then, if the logical $Z$ operator of the 4D toric code is supported at coordinates $(x, y, z, w)$, this means that the qubit living at coordinates $(x, y)$ in $\textcolor{orange}{\mathsf{TC}}$ are coupled to the qubits at coordinates $(z, w)$ in $\textcolor{dodgerblue}{\mathsf{TC}}$ via a $\mathsf{CZ}$ gate.
For instance, for the Pauli logical shown in Fig.~\ref{fig:Cliffordspider}(a), it translates to the pattern of Clifford gates shown in Fig.~\ref{fig:Cliffordspider}(b).

While the gate in Fig.~\ref{fig:Cliffordspider}(b) has a depth that scales with code distance, it can easily be checked that the gate enacts an addressable logical $\mathsf{CZ}$ gate between a single qubit of $\textcolor{orange}{\mathsf{TC}}$ and $\textcolor{dodgerblue}{\mathsf{TC}}$, as we ultimately desire.
If we deform the Pauli logical in the 4D toric code, we are able to reduce the depth of this logical Clifford gate between the two toric codes.
In particular, we know that we can, for instance, deform the logical Pauli $Z$ gate in the 4D toric code from a flat sheet to a ``tent'' [Fig.~\ref{fig:Cliffordspider}(c)].
As a consequence, the logical is only supported when $x > 0$.
Moreover, for a fixed $x > 0$ and $y$ there are only two values of $z$ and one value of $w$ such that the logical is supported at $(x, y, z, w)$.
Consequently, the action of the logical gate between the two toric codes takes the form of Fig.~\ref{fig:Cliffordspider}(d).
We note that the resulting gate still takes a large depth to implement as the same values of $z$ and $w$ are ``hit'' for every value of $y$ at fixed $x > 0$.

By further deforming the 4D toric code's Pauli logical with the same ``tenting'' trick in the $w$ direction, we find a constant depth implementation of this addressable $\mathsf{CZ}$ gate.
Such an addressable gate is shown in Fig.~\ref{fig:Addressable_CZ}(a).
Intuitively, we can understand the action of the gate by folding the $\textcolor{dodgerblue}{\mathsf{TC}}$ gate in half and then acting the global $\mathsf{CZ}$ gate of Eq.~\eqref{eq:globCZ} between half of the $\textcolor{orange}{\mathsf{TC}}$ code and both layers of the folded $\textcolor{dodgerblue}{\mathsf{TC}}$.
With this in mind, we can check that the gates above yields an addressable $\mathsf{CZ}$ gate.
In particular, the action of this gate on the Pauli logical operators are given by: 
\begin{equation}
    \includegraphics[valign = c]{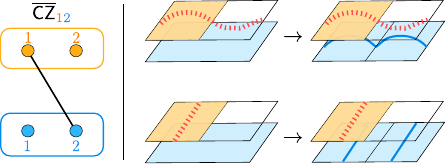}
\end{equation}
Notice that half of the ``horizontal'' $X$ logical in $\textcolor{orange}{\mathsf{TC}}$ gets mapped to a full non-trivial $Z$ logical operator in $\textcolor{dodgerblue}{\mathsf{TC}}$.
Moreover, the ``vertical'' $X$ logical in $\textcolor{orange}{\mathsf{TC}}$ maps to a trivial $Z$ logical operator.
This can be seen by the fact that if the vertical $X$ logical is on the left hand side of $\textcolor{orange}{\mathsf{TC}}$, it is doubled in the bottom, whereas if it is ``cleaned'' to the right side, it does not get acted on by the gate.

We conclude by noting that by lifting the action of this addressable inter-block gate such that the control and target are in the same code, we get an \textit{intra-block} $\mathsf{CZ}$ gate between the two logical qubits of a single toric code.
A depiction of the resulting gate is shown in Fig.~\ref{fig:Addressable_CZ}(b).

\begin{figure}
    \centering
    \includegraphics[width=247 pt]{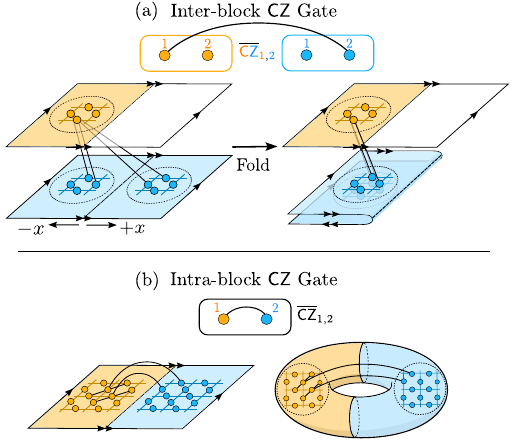}
    \caption{\textbf{An Addressable $\mathsf{CZ}$ Gate for the Toric Code.}
    (a)~Further deforming the logical from Fig.~\ref{fig:Cliffordspider} results in a transversal logical \(\overline{\mathsf{\textcolor{orange}{C}\textcolor{dodgerblue}{Z}}}_{\textcolor{orange}{1},\textcolor{dodgerblue}{2}}\) gate, acting on only two of the four qubits encoded in two blocks of the toric code.
    The pattern of gates is shown for a unit cell, and is repeated everywhere in the colored regions of the code.
    The gate can be made geometrically local by folding the blue toric code.
    (b)~Moving all targets of the physical \(\mathsf{CZ}\) gates to a single code block results in a transversal logical \(\overline{\mathsf{CZ}}_{1,2}\) gate in one copy of the toric code.}
    \label{fig:Addressable_CZ}
\end{figure}

\section{The Chain Map Hierarchy} \label{sec:ChainMapHierarchy}

Equipped with a conceptual overview of the chain map hierarchy, this section is devoted to formally defining it and developing the technical machinery necessary to practically wield it.
The tools developed in this section will be applied in Section~\ref{sec:DiscoveringGates} to discover new logical gates in well-known qLDPC codes.

We start by reviewing how CSS codes are described via chain complexes, and how logical Clifford operations---including unitary gates and code surgery measurement gadgets---are instead described as \emph{maps} between chain complexes (Section~\ref{subsec:review}).
While the content of Section~\ref{subsec:review} is well-known, the section introduces key notation and definitions, and is presented in a manner that naturally leads to our main results, emphasizing specifically how abstract homological data is concretely translated into logical quantum operations.

By holistically considering the entire space of maps between chain complexes, rather than individual examples, one realizes that maps themselves can be assembled into a chain complex---which we call the \emph{chain map complex}.
With this insight, we show how to understand logical Clifford operations on codes in the same chain complex framework used for the codes themselves (Section~\ref{subsec:chainmapcomplexandcliff}).
Thus, we uncover a notion of deformation for logical gates which is precisely analogous to deformation of Pauli logical gates.
This facilitates finding transversal implementations of given logical actions, when possible.

Further, by iteratively forming chain map complexes, this framework can be vastly extended to capture logical \emph{non-Clifford} operations from any level of the Clifford hierarchy.
We call the resulting hierarchy of auxiliary chain complexes the \emph{chain map hiearchy} (Section~\ref{subsec:ChainMapHeirarchy}).
Even commuting gates within a single code block can be captured with this formalism, facilitating optimal usage of qLDPC codes with many encoded qubits in each code block (Section~\ref{sec:IntrablockGates}).

\subsection{Review: Chain Complexes and CSS Codes, Chain Maps and Logical Clifford Operations}
\label{subsec:review}

Throughout this paper, we will take the approach of presenting a relevant mathematical structure, and then concretely translating this structure into quantum data or operations.
As the simplest illustration of this, we start with the well-known connection between chain complexes and CSS error correcting codes~\cite{Kitaev1997quantumcomputation,Kitaev_2003,Freedman2001,Bombin_2007,Haah2013latticequantumcodesexotic}.
In particular, we start by defining a \textit{based} chain complex as a sequence
of vector spaces equipped with privileged bases that will be identified with
the qubits, elementary checks, and syndromes of the CSS code. 

Since we are
ultimately interested in families of codes of increasing size, we implicitly consider families of chain complexes with finite bases whose sizes grow along the family. Thus asymptotic terms such as ``\(O(1)\)", ``sparse", and ``low-density" are understood with respect to increasing code size. We suppress the family index throughout.

\begin{defn}[(Sparse, Based) Chain Complexes]\label{def:SparseBasedChainComplex}
A \textit{chain complex} \(C_{\bullet}\) over the binary field
\(\mathbb{F}_2=\{0,1\}\) is a sequence of \(\mathbb{F}_2\)-vector spaces
connected by linear maps:
    \begin{equation}
         C_{\bullet} \equiv \cdots \xrightarrow[]{\partial_2} C_1 \xrightarrow[]{\partial_1} C_0 \xrightarrow[]{\partial_0} C_{-1} \xrightarrow[]{\partial_{-1}}  \cdots 
    \end{equation}
    where $\partial_{j} \circ \partial_{j+1} = 0$ for all $j$.
    Elements $x \in C_j$ of the degree-$j$ component are called $j$-chains, the $\partial_j$ maps are called boundary maps, and elements of \(\ker(\partial_j)\) are called \(j\)-cycles.

    Associated with a chain complex is a dual cochain complex:
    \begin{equation}\label{eqn:CochainComplex}
        C^{\bullet} \equiv \cdots \xrightarrow[]{\exd_{-1}} C^{-1} \xrightarrow[]{\exd_0} C^0 \xrightarrow[]{\exd_1} C^{1} \xrightarrow[]{\exd_2}  \cdots.
    \end{equation}
    Elements  $\mathbf{x} \in C^j$ are called $j$-cochains and are linear maps $\mathbf{x}: C_j \to \mathbb{F}_2$, i.e. $C^j \isom \Hom(C_j, \mathbb{F}_2)$, where $\Hom(V, W)$ is the space of linear maps from vector space $V$ to $W$. (If \(C_j\) is finite dimensional, \(C^j \isom C_j\).)
    The {coboundary} maps $\exd_j:C^j\to C^{j+1}$ are defined to satisfy $\mathbf{x} \circ \partial_{j+1} = \exd_{j+1} \circ \mathbf{x}$ for all $\mathbf{x} \in C^j$ and hence $\exd_{j+1}\circ \exd_{j} = 0$ for all $j$.
    Elements of \(\ker(\exd_{j+1})\) are called \(j\)-cocycles.

    The chain complex is \textit{based} if each $C_j$ has a privileged finite basis $\mathsf{B}(C_j)$ and \textit{sparse} if, in this basis, the boundary maps $\partial_j$ are sparse as matrices---i.e. have $O(1)$ entries per row and column. If the chain complex is based, the cochain complex has a natural basis $\mathsf{B}(C^j)$ from Kronecker pairing: \(
    \mathbf b(b')=\delta_{b,b'}\), for
    \(b,b'\in\mathsf B(C_j),\) and \(\mathbf b\in\mathsf B(C^j)\) is dual to \(b\).
    In this basis, we have $\exd_j = \partial_j^{\mathsf{T}}$ as matrices, so that \(C^\bullet\) is sparse if and only if \(C_\bullet\) is sparse.
\end{defn}
When obvious from context, we omit indices from the boundary and coboundary maps, simply writing \(\partial\) and \(\exd\), respectively. 
As a matter of notation, we always denote cochains in boldface and chains without boldface, and frequently suppress ``\(\circ\)'' for composition, simply writing \(\exd \mathbf{x}\), \(\partial^2\), etc.

Sparse and based chain complexes can be used to specify qLDPC CSS codes, as in the toric code example of Section~\ref{sec:Intuition} and Eq.~\eqref{eq-chaincomplex}. 
In the toric code, the basis elements $\mathsf B(C_1),\mathsf B(C_0),\mathsf B(C_{-1})$ are associated with the elementary plaquettes, edges and vertices of the 2D square lattice respectively.

In general, the physical Hilbert space is
\((\mathbb C^2)^{\otimes |\mathsf B(C_0)|}\). Thus the number of
physical qubits is \(n=|\mathsf B(C_0)|=\dim C_0\), with each qubit indexed by
a basis element \(e\in \mathsf B(C_0)\). Pauli \(Z\)- and \(X\)-operators are
labeled by \(0\)-chains and \(0\)-cochains, respectively~\cite{Bombin_2007}.
Given a 0-chain $b \in C_0$ and a $0$-cochain $\mathbf{a} \in C^0$, we associate products of $Z$ and $X$ operators to these as:
\begin{equation}\label{eq:Pauli-operators-chains}
    Z_b=\prod_{e\in\mathsf B(C_0)}Z_e^{\mathbf e(b)},
    \qquad
    X_{\mathbf a}=\prod_{e\in\mathsf B(C_0)}
    X_e^{\mathbf a(e)} .
\end{equation}
Here \(\mathbf e\in\mathsf B(C^0)\) is the basis cochain dual to
\(e\in\mathsf B(C_0)\), while \(Z_e\) and \(X_e\) are the physical Pauli
operators on the qubit indexed by \(e\). 
More generally, throughout, in products indexed by basis chains $e$, any boldfaced version of
the same index denotes the corresponding dual cochain. 
Expanding the chain and cochain in their respective
bases as \(b=\sum_e b_e e\) and \(\mathbf a=\sum_e a_e\mathbf e\), with coefficients
\(a_e,b_e\in\mathbb F_2\), the chain-cochain pairing gives
\(\mathbf e(b)=b_e\) and \(\mathbf a(e)=a_e\); these binary coefficients
indicate whether the corresponding physical Pauli acts on qubit \(e\).
More generally, \(\mathbf a(b)=\sum_e a_e b_e \pmod 2\) is the evaluation of the cochain $\mathbf{a}$ on the chain $b$, and this sets the commutation relations between Pauli operators: 
\begin{equation}\label{eqn:ZXCommutatorFromEvaluation}
    Z_{b} X_{\mathbf{a}} = (-1)^{\mathbf{a}(b)} X_{\mathbf{a}} Z_{b}. 
\end{equation}

So far, \(\mathsf B(C_0)\) has specified the physical qubits, while
\(C_0\) and \(C^0\) label Pauli \(Z\)- and \(X\)-operators on those qubits. The CSS code associated with the chain complex is obtained by using
the boundary and coboundary maps to specify the stabilizer generators: 

\begin{theo}[(qLDPC) CSS Codes From (Sparse) Based Chain Complexes]
\label{thm:CSScode}
Let \(C_\bullet\) be a based chain complex over \(\mathbb F_2\), and let
the physical qubits be indexed by the basis \(\mathsf B(C_0)\).
The chosen bases \(\mathsf B(C_1)\) and \(\mathsf B(C^{-1})\) label the
elementary \(Z\)- and \(X\)-type stabilizer generators
\begin{equation}\label{eq-qLDPCstab}
\begin{aligned}
    &Z_{\partial_1 p}
    &&\qquad p\in \mathsf B(C_1),\\
    &X_{\exd_0\mathbf v}
    &&\qquad \mathbf v\in \mathsf B(C^{-1}) .
\end{aligned}
\end{equation}
These generators commute and the subspace of states stabilized by these operators specifies a CSS code. 

If \(C_\bullet\) is sparse, the resulting CSS code is LDPC: the stabilizer generators in
Eq.~\eqref{eq-qLDPCstab} have finite weight, and each physical qubit
participates in only \(O(1)\) such generators.
\end{theo}

\begin{proof}
The proof of the theorem is standard. From
Eq.~\eqref{eqn:ZXCommutatorFromEvaluation}, the commutation between
\(Z_{\partial_1 p}\) and \(X_{\exd_0\mathbf v}\) is set by
\((\exd_0\mathbf v)(\partial_1 p)=\mathbf v(\partial_0\partial_1 p)=0\),
where the last equality follows from \(\partial_0\partial_1=0\).
 The LDPC condition follows from the sparseness of the (co)boundary maps.
 \end{proof}

The chain-complex description of the CSS code also identifies the syndrome
spaces, which record violations of the code space.  Concretely, the syndrome of a \(Z_b\) operator is its boundary
\(\partial_0 b\in C_{-1}\): its commutation with an \(X\)-type stabilizer
\(X_{\exd_0\mathbf v}\) is determined using
Eq.~\eqref{eqn:ZXCommutatorFromEvaluation} by
\(
    (\exd_0\mathbf v)(b)=\mathbf v(\partial_0 b).
\)
Dually, the syndrome of an \(X_{\mathbf a}\) operator is its coboundary
\(\exd_1\mathbf a\in C^1\), since its commutation with a \(Z\)-type
stabilizer \(Z_{\partial_1 p}\) is determined by
\(
    \mathbf a(\partial_1 p)=(\exd_1\mathbf a)(p)
\). 
Thus, as depicted in
Eq.~\eqref{eq-chaincomplex}, \(C_1,C_0,C_{-1}\) encode \(Z\)-checks,
\(Z\)-operators, and violations of \(X\)-checks, while \(C^{-1},C^0,C^1\)
encode \(X\)-checks, \(X\)-operators, and violations of \(Z\)-checks.

The logical Pauli $X$ and $Z$ gates of this CSS code correspond to the homology and cohomology classes of the chain complex:
\begin{coro}[Logical Pauli Gates from (Co)Homology]
    Given a $0$-cycle $b \in \mathsf{ker}(\partial_0)$ and $0$-cocycle $\mathbf{a} \in \mathsf{ker}(\exd_1)$, the operators $Z_b$ and $X_{\mathbf{a}}$ preserve the code space.
    If $\mathbf{a}$ and $b$ live in non-trivial cohomology and homology classes of $H^0$ and $H_0$ respectively, defined by: 
    \begin{equation} H_{j} = \frac{\mathsf{ker}(\partial_{j})}{\mathsf{Im}(\partial_{j + 1})}, \qquad H^{j} = \frac{\mathsf{ker}(\exd_{j + 1})}{\mathsf{Im}(\exd_{j})},
    \end{equation}
    then $Z_b$ and $X_{\mathbf{a}}$  act non-trivially on the code space i.e. they are logical Pauli $Z$ and $X$ gates.

    Thus the number of encoded qubits is
\(
    k = \dim H_0(C) = \dim H^0(C).
\)
\end{coro}

\begin{proof}
The proof is once again standard. Since the anticommutation of $Z_b$ with $ X_{\exd_0 \mathbf{v}}$ is set by $\mathbf{v}(\partial_0 b)$, any $Z_b$ such that $b \in \mathsf{ker}(\partial_0)$, will commute with all $X$-stabilizers and hence preserve the code space.
If these $Z_b$ operators are not stabilizers---i.e. $b \notin \mathsf{Im}(\partial_{1})$---then they correspond to non-trivial logical $Z$ operators.
Hence, logical $Z$ gates are labeled by  non-trivial classes in $\mathsf{ker}(\partial_{0})/\mathsf{Im}(\partial_{1})$ and a similar line of reasoning works for logical $X$ gates.
\end{proof}

\subsubsection{Review of Chain Maps and Logical Clifford Gates} \label{subsubsec:chainmap}

The chain complex formalism  provides a clean way of working with the algebra of Pauli operators on CSS codes and, more importantly, provides a way to homologically compute physical implementations of  logical Pauli gates.
Logical Clifford gates are traditionally described somewhat differently---using \textit{maps} between chain complexes~\cite{huang_2023_homomorphicCNOT}. 
In particular, as we described in Section~\ref{sec:Intuition}, these gates are treated using \textit{chain maps}~\cite{Weibel1994HomologicalBook,hatcher2002algebraic}. 
Keeping with the theme of defining mathematical structure and translating it into quantum operations, we now review chain maps, and then describe how to associate these with physical implementations of logical Clifford gates. 
Formally: 
\begin{defn}[(Sparse) Chain Map]\label{def:SparseChainMap}
    Given two chain complexes \(C^\bullet\) and \(D_\bullet\), a chain map
    \(\phii_\bullet\) is a collection of linear maps
    \(
    \phii_j:C^j\to D_{-j}\),
    i.e. \(\phii_j\in \operatorname{Hom}(C^j,D_{-j})\), 
    such that the following diagram commutes:
    \begin{equation}\label{eqn:ChainMapDiagram}
    \begin{tikzpicture}[scale=0.8, baseline={([yshift=-.5ex]current bounding box.center)}]

      \node (Cdots1) at (0, 0)    {\small $\cdots$};
      \node (C1)     at (2, 0)    {\small $C^{-1}$};
      \node (C0)     at (4, 0)    {\small $C^0$};
      \node (Cm1)    at (6, 0)    {\small $C^{1}$};
      \node (Cdots2) at (8, 0)    {\small $\cdots$};

      \node (Ddots1) at (0, -1.8) {\small $\cdots$};
      \node (D1)     at (2, -1.8) {\small $D_1$};
      \node (D0)     at (4, -1.8) {\small $D_0$};
      \node (Dm1)    at (6, -1.8) {\small $D_{-1}$};
      \node (Ddots2) at (8, -1.8) {\small $\cdots$};

      \draw[-stealth] (Cdots1) -- node[above]{\small $\exd_{-1}$} (C1);
      \draw[-stealth] (C1)     -- node[above]{\small $\exd_0$} (C0);
      \draw[-stealth] (C0)     -- node[above]{\small $\exd_1$} (Cm1);
      \draw[-stealth] (Cm1)    -- node[above]{\small $\exd_{2}$} (Cdots2);

      \draw[-stealth] (Ddots1) -- node[below]{\small $\partial_2$} (D1);
      \draw[-stealth] (D1)     -- node[below]{\small $\partial_1$} (D0);
      \draw[-stealth] (D0)     -- node[below]{\small $\partial_0$} (Dm1);
      \draw[-stealth] (Dm1)    -- node[below]{\small $\partial_{-1}$} (Ddots2);

      \draw[-stealth] (C1)  -- node[left] {\small $\varphi_{-1}$}  (D1);
      \draw[-stealth] (C0)  -- node[left] {\small $\varphi_0$}  (D0);
      \draw[-stealth] (Cm1) -- node[right]{\small $\varphi_{1}$} (Dm1);

    \end{tikzpicture}
    \end{equation}

    The \emph{transpose} of \(\phii_\bullet\), $\phii^{\mathsf{T}}_{\bullet}$,  is a chain map from \(\mcbb{D}\) to \(\mcb{C}\) defined by components as
    \begin{align}
        \phii^{\mathsf{T}}_j : D^{-j} \to C_{j} \qquad
        \mathbf{x} \mapsto \mathbf{x} \circ \phii_j.
    \end{align}
    The chain map condition for this collection of transpose maps is obtained by taking the transpose of Eq.~\eqref{eqn:ChainMapDiagram} which reverses the arrows, interchanges chains and cochains, and replaces each boundary map with its dual, $\partial  \leftrightarrow d$.
    
    If \(\mcbb{C}\) and \(\mcb{D}\) are  based, we say that $\phii_\bullet$ is \emph{sparse} if each of the component maps $\phii_j$ above is sparse as a matrix in the associated bases ---i.e. has \(O(1)\) entries per row and column. As a matrix, each \(\phii^\mathsf{T}_j\) is just the transpose of \(\phii_j\) and hence is sparse if \(\phii\) is sparse.
    
    This definition generalizes to maps between any pair of chain or cochain
    complexes: for \(C_\bullet\to D_\bullet\), \(C^\bullet\to D^\bullet\),
    and \(C_\bullet\to D^\bullet\), the component maps $\varphi_j$
    are respectively from \(C_j\to D_j\), \(C^j\to D^j\) and
    \(C_j\to D^{-j}\), with the corresponding diagrams required to commute and
    transposes obtained by dualizing the diagrams.
\end{defn}

Commutativity of the diagram in Eq.~\eqref{eqn:ChainMapDiagram} is concretely the condition that
\begin{equation}
\label{eq:chainmapcommute}\phii_{j+1}\circ \exd_{j+1} = \partial_{-j}\circ \phii_j \quad \text{for all } j.
\end{equation}
Once again, when obvious from context, we will omit indices from the maps $\phii_j$, simply writing $\phii$ or \(\phii_\bullet\). 
In an expression such as
\(\partial\circ \varphi\) or \(\varphi \circ \exd\), the maps are understood to be
the unique ones whose domains and codomains make the composition well-defined.
With this convention, the commutativity condition of Eq.~\eqref{eq:chainmapcommute} is simply
$\varphi\circ \exd=\partial\circ\varphi.$

The commutativity condition ensures that a chain map induces a well-defined
map on (co)homology. In particular, a chain map
\(\phii_\bullet:C^\bullet\to D_\bullet\) induces maps
\begin{align}
    (\phii_*)_i:H^i(C)&\to H_{-i}(D), \nonumber \\
    [\mathbf c]&\mapsto [\phii_i(\mathbf c)] .
    \label{eqn:PhiiLogicalAction}
\end{align}
Here \([\mathbf c] = \mathbf{c} + \mathsf{Im}(\exd_{i})\) denotes the cohomology class of an
\(i\)-cocycle representative \(\mathbf c\in\ker(\exd_{i+1})\). The map is well defined because
the chain map condition sends cocycles to cycles and coboundaries to
boundaries: if
\(\mathbf c\) is a cocycle, \(\exd\mathbf c=0\), then
\(
\partial\varphi(\mathbf c)=\varphi( \exd\mathbf c)=0,
\)
so \(\varphi(\mathbf c)\) is a cycle. If \(\mathbf b=\exd\mathbf a\) is a
coboundary, then
\(
\varphi(\mathbf b)=\varphi (\exd\mathbf a)=\partial\varphi(\mathbf a),
\)
so \(\varphi(\mathbf b)\) is a boundary. Thus \((\phii_*)_i\) gives a
consistent map from \(H^i(C)\) to \(H_{-i}(D)\). In the gate constructions below, this induced map is the logical object of
interest and the rank of $(\varphi_*)_0$ counts the number of independent logicals on which the gate acts nontrivially. 
Analogous statements hold for
the other chain maps used below; for example, a chain map
\(\psi_\bullet:C^\bullet\to D^\bullet\) induces maps
\((\psi_*)_i:H^i(C)\to H^i(D)\).

As discussed in Section~\ref{sec:Intuition},  induced maps on (co)homology are precisely what is needed to define \textit{logical} Clifford gates because logical Pauli $Z$ and $X$ operators are labeled by homology and cohomology classes respectively. For a logical \(\mathsf{CZ}\) gate coupling codes \(C\) and \(D\), the Heisenberg action
\(X^C \mapsto X^C Z^D\) gives a nontrivial map from cohomology classes of
\(C\) to homology classes of \(D\), which can be induced by a chain map
\(\phii_\bullet:C^\bullet\to D_\bullet\). For a \(\mathsf{CNOT}\) gate, the action
\(X^C \mapsto X^C X^D\) instead gives a nontrivial map from cohomology classes of
\(C\) to cohomology classes of \(D\), which can be induced by a chain map
\(\psi_\bullet:C^\bullet\to D^\bullet\). The corresponding transformations of the \(Z\)-operators are
determined by the transpose maps.
Thus, given such chain maps $\phii_{\bullet}$ and $\psi_{\bullet}$
, we will consider the following product of physical Clifford gates:
\begin{equation}\label{eq:productofCZ}
    U_{\mathsf{CZ}}^{\phii} = \prod_{e \in \mathsf{B}(C_0)} \mathsf{CZ}^{CD}_{e, \phii_0(\mathbf{e})}, \quad
        U_{\mathsf{CNOT}}^{\psi} = \prod_{e \in \mathsf{B}(C_0)} \mathsf{CNOT}^{CD}_{e, \psi_0(\mathbf{e})}. 
\end{equation}
Here $\mathsf{CZ}^{CD}_{e, \phii(\textbf{e})}$ and $\mathsf{CNOT}^{CD}_{e, \psi(\mathbf{e})}$ are products of physical $\mathsf{CZ}$ and $\mathsf{CNOT}$ gates between a qubit at $e$ in code $C$ to the qubits at edges $\phii_0(\mathbf{e})$ and $\psi_0(\mathbf{e})$ in code $D$ respectively.
When no confusion can arise, we suppress the degree-zero subscript in physical
gate expressions and write \(\phii(\mathbf e)\) and \(\psi(\mathbf e)\). 
 In analogy to Eq.~\eqref{eq:Pauli-operators-chains}, we extend the elementary two-qubit gate notation to arbitrary
\(0\)-chains and \(0\)-cochains. For \(a\in C_0\), \(b\in D_0\): 
\begin{equation}\label{eq:CliffordGatesFromChains}
\begin{aligned}
    \mathsf{CZ}^{CD}_{a,b}
    &\equiv
    \prod_{\substack{
        e_C\in\mathsf B(C_0)\\
        e_D\in\mathsf B(D_0)
    }}
    \left(\mathsf{CZ}^{CD}_{e_C,e_D}\right)^{
        \mathbf e_C(a)\mathbf e_D(b)
    },
    \\
    \mathsf{CNOT}^{CD}_{a,\mathbf b}
    &\equiv
    \prod_{\substack{
        e_C\in\mathsf B(C_0)\\
        e_D\in\mathsf B(D_0)
    }}
    \left(\mathsf{CNOT}^{CD}_{e_C,e_D}\right)^{
        \mathbf e_C(a)\mathbf b(e_D)
    } .
\end{aligned}
\end{equation}

Then, one can prove that:
\begin{theo}[(Transversal) 
Logical Gates from (Sparse) Chain Maps]
\label{thm:TransversalGatesFromChainMaps}
Let \(C_\bullet\) and \(D_\bullet\) be sparse, based chain complexes, with associated
CSS codes as in Theorem~\ref{thm:CSScode}.
For chain maps \(\phii_\bullet:C^\bullet\to D_\bullet\) and
\(\psi_\bullet:C^\bullet\to D^\bullet\), the physical operators
\(U_{\mathsf{CZ}}^{\phii}\) and \(U_{\mathsf{CNOT}}^{\psi}\) defined in
Eq.~\eqref{eq:productofCZ} implement the following logical actions:
\begin{align}
    \overline{U}_{\mathsf{CZ}}^{\phii}
    &=
    \prod_{L\in\mathsf B(H_0(C))}
    \overline{\mathsf{CZ}}^{CD}_{
        L,\phii_*(\mathbf L)
    },
    \label{eq:logicalCZ}
    \\
    \overline{U}_{\mathsf{CNOT}}^{\psi}
    &=
    \prod_{L\in\mathsf B(H_0(C))}
    \overline{\mathsf{CNOT}}^{CD}_{
        L,\psi_*(\mathbf L)
    }.
    \label{eqn:UCXLogicalAction}
\end{align}
Here the bar over the gates indicates that these gates act on \emph{logical} qubits, \(\mathsf B(H_0(C))=\{L\}\) is a chosen basis of homology
classes, \(\mathbf L\in H^0(C)\) is the cohomology class dual to
\(L\), 
and \(\phii_*\) and \(\psi_*\) are the maps induced on (co)homology, as in
Eq.~\eqref{eqn:PhiiLogicalAction}.

If, moreover, \(\phii_\bullet\) and \(\psi_\bullet\) are sparse, then
\(U_{\mathsf{CZ}}^{\phii}\) and \(U_{\mathsf{CNOT}}^{\psi}\) are transversal,
i.e., they can be implemented in constant depth.
\end{theo}

The key step in the proof of this theorem is the following equation.
The transformation of a generic $X$ operators in code $C$, $X^C_{\mathbf{a}}$, under the gates of Eq.~\eqref{eq:productofCZ} is:
\begin{equation} \label{eq:phiactsona}
    [ U_{\mathsf{CZ}}^{\phii}, X^C_{\mathbf{a}}]_{\rm grp} =  Z_{\phii(\mathbf{a})}^D,
\end{equation}
where $[ A, B ]_{\rm grp} = B^{\dagger} A^{\dagger} B A $ is the group commutator.
Intuitively, notice that $U^{\phii}_{\mathsf{CZ}}$ maps the $X$ operator associated with $\mathbf{a}$ into the $Z$ operator associated with $\phii(\mathbf{a})$ under the group commutator. This becomes the logical operation of Eq.~\eqref{eq:logicalCZ} when \(X^C_{\mathbf a}\) is a logical because \(\phii\) is a chain map which descends to a well-defined map on (co)homology [Eq.~\eqref{eqn:PhiiLogicalAction}]. 
That is, if a cocycle \(\mathbf a\) represents a logical \(X\)-operator
\([\mathbf a]\in H^0(C)\), then \(\phii_0(\mathbf a)\) is a cycle
representing the corresponding logical \(Z\)-operator
\(\phii_*([\mathbf a])\in H_0(D)\). Thus
Eq.~\eqref{eq:phiactsona} reduces on the logical subspace to the action in
Eq.~\eqref{eq:logicalCZ}.

To prove Eq.~\eqref{eq:phiactsona}, first write
\begin{equation}
    U^\phii_{\mathsf{CZ}}
    =
    \prod_{e \in \mathsf{B}(C_0)}
    \mathsf{CZ}^{CD}_{e,\phii(\mathbf e)}
    =
    \prod_{e \in \mathsf{B}(C_0)}
    \left(Z_e^C\right)^{n^D_{\phii(\mathbf e)}} .
\end{equation}
where $n_e = (1 - Z_e)/2$ is a projector onto the $\ket{1}$ state or,   equivalently, the occupation-number form of the \(Z_e\) spin variable, converting
\(Z_e=+1,-1\) to \(n_e=0,1\) respectively. Likewise, $n_{\phii(\mathbf e)}$ is a sum (modulo 2) of these projectors for each basis element appearing in $\phii(\mathbf e)$.
In the above, we used the fact that $\mathsf{CZ}_{e, e'} = (-1)^{n_e n_e'}$---i.e. it is only $-1$ when both $e$ and $e'$ are in the $\ket{1}$ state---and the fact that $(-1)^{n_e n_e'} = (Z_e)^{n_e'}$.
With this re-writing, it becomes easy to derive Eq.~\eqref{eq:phiactsona}.
Specifically, using the anticommutation relation in Eq.~\eqref{eqn:ZXCommutatorFromEvaluation}, note that:
\begin{align}
     \left(Z_e^C\right)^{n^D_{\phii(\mathbf e)}} X_{\mathbf{a}}^C &=  (-1)^{\mathbf{a}(e) n^D_{\phii(\mathbf e)}}  \times X_{\mathbf{a}}^C \left(Z_e^C\right)^{n^D_{\phii(\mathbf e)}}
\end{align}
Because $(-1)^{n^D} = Z^D$, it follows that
\begin{equation}
    \left[\left(Z_e^C\right)^{n^D_{\phii(\mathbf{e})}} , X_{\mathbf{a}}^C\right]_{\rm grp} =\left(Z_{\phii(\mathbf{e})}^D\right)^{\mathbf{a}(e)}.   
\end{equation}
Taking the product over \(e\in \mathsf{B}(C_0)\), we obtain
\begin{align}
    [U_{\mathsf{CZ}}^{\phii}, X^C_{\mathbf a}]_{\rm grp}
    &=
    \prod_{e\in \mathsf{B}(C_0)}
    \left(Z^D_{\phii(\mathbf e)}\right)^{\mathbf a(e)}
    \nonumber
    =
    Z^D_{\sum_e \mathbf a(e)\phii(\mathbf e)}\\
    &=
    Z^D_{\phii(\mathbf a)},
\end{align}
as desired.

\subsubsection{Review of Chain Maps and Logical Pauli Measurement}
\label{subsubsec:logicalPaulimeasurementreview}

We conclude our review by noting that chain maps are also traditionally used to describe other Clifford operations on codes beyond logical \textit{unitary} Clifford gates.
Specifically, they can also describe the \textit{measurement} of logical Pauli gates via \textit{code surgery}~\cite{horsman2012surface,litinski2019game,cohen2022lowoverhead}---which in the Pauli-based framework for quantum computation acts as a logical Clifford operation.

Logical Pauli measurement via code surgery is a specific type of code switching protocol~\cite{cohen2022lowoverhead,cross2024improved,williamson2024lowoverhead,ide2025faulttolerant,baspin2025fast}.
In particular, to measure the logical Pauli operator of a CSS code, one brings in an ancilla system of qubits and uses measurement to ``switch'' the CSS code into another stabilizer code where the Pauli logical operators to be measured are now in the updated stabilizer group.
By switching back to the original CSS code, one has then measured the desired logical Pauli operator.

In literature, this protocol is typically presented by using a chain map to construct the new system-ancilla stabilizers that are to be measured, implementing the code switching~\cite{cohen2022lowoverhead,cross2024improved,williamson2024lowoverhead,ide2025faulttolerant,baspin2025fast,he2025extractors}.
Our presentation follows a slightly different (but ultimately equivalent) approach which facilitates a comparison to the unitary implementation of the Clifford gate [c.f. Fig.~\ref{fig:overview}]~\cite{verresenefficiently2022, Tantivasadakarn_2024_LRE, tantivasadakarn_shortest_2023, tanti_Hierarchy, Christos2026nonabelian}.

In particular, we show that code surgery can be understood as a two step process.
First, one performs a transversal Clifford gate between a certain ancilla stabilizer state (a hypergraph cluster state) and a CSS code of interest.
Second, logical measurement is performed by projectively measuring all of the qubits in the ancilla state.
Crucially, the transversal Clifford gate used is identical in form to the logical Clifford gate of Eq.~\eqref{eq:productofCZ}.

To establish this formally and make the connection with the unitary gates manifest, we introduce the ancilla stabilizer state of interest:

\begin{defn}[Hypergraph Cluster State]
    The \textit{hypergraph cluster state} associated to a sparse, based chain complex $C_{\bullet}$ is a CSS stabilizer state on a bipartite system of qubits labeled by elements of the bases $\mathsf{B}(C_{-1})$ and $\mathsf{B}(C_0)$.
    The elementary stabilizers of this state are:
    \begin{equation}
        X_{\mathbf{v}} X_{\exd_0 \mathbf{v}}, \qquad Z_{e} Z_{\partial_0 e}
    \end{equation}
    where $\mathbf{v} \in \mathsf{B}(C^{-1})$ and $e \in \mathsf{B}(C_0)$.
\end{defn}
We now show that logical Pauli measurement (via code surgery) follows from a product of Clifford gates between the $\mathsf{B}(C^{-1})$ qubits of the hypergraph state and the qubits  of the CSS code of interest [associated with $\mathsf{B}(D_{0})$], followed by a complete measurement of the hypergraph state.

To make contact with the chain map formalism, it is useful to define the notion of a shifted chain complex since gates will now be performed between qubits in $\mathsf{B}(C^{-1})$ and $\mathsf{B}(D_{0})$, which are shifted relative to one another in their respective chain complexes.
\begin{defn}[Shifted Chain Complex]
    Given a chain complex $C^{\bullet}$, we define the chain complex shifted by $j$, $C[j]^{\bullet}$, to be the chain complex with $C[j]^k \simeq C^{k + j}$ and boundary maps shifted similarly.
    Similarly, \(C[j]_\bullet\) has components \(C[j]_k = C_{k-j}\).
\end{defn}
Then, we specify the product of Clifford gates through the chain map $\phii_{\bullet}:C[-1]^{\bullet} \to D_{\bullet}$,  represented diagrammatically as:
\begin{equation} \label{eq-diagchainmap}
        \begin{tikzpicture}[scale=0.7, baseline={([yshift=-.5ex]current bounding box.center)}]
      \node (Cdots1) at (0, 0)    {\small $\cdots$};
      \node (C1)     at (2, 0)    {\small $C^{-1}$};
      \node (C0)     at (4, 0)    {\small $C^0$};
      \node (Cm1)    at (6, 0)    {\small $C^{1}$};
      \node (Cdots2) at (8, 0)    {\small $\cdots$};
      \node (Ddots1) at (0, -2.2) {\small $\cdots$};
      \node (D1)     at (2, -2.2) {\small $D_1$};
      \node (D0)     at (4, -2.2) {\small $D_0$};
      \node (Dm1)    at (6, -2.2) {\small $D_{-1}$};
      \node (Ddots2) at (8, -2.2) {\small $\cdots$};
      \draw[-stealth] (Cdots1) -- node[above]{\small $\exd_{-1}$} (C1);
      \draw[-stealth] (C1)     -- node[above]{\small $\exd_0$} (C0);
      \draw[-stealth] (C0)     -- node[above]{\small $\exd_1$} (Cm1);
      \draw[-stealth] (Cm1)    -- node[above]{\small $\exd_{2}$} (Cdots2);
      \draw[-stealth] (Ddots1) -- node[below]{\small $\partial_2$} (D1);
      \draw[-stealth] (D1)     -- node[below]{\small $\partial_1$} (D0);
      \draw[-stealth] (D0)     -- node[below]{\small $\partial_0$} (Dm1);
      \draw[-stealth] (Dm1)    -- node[below]{\small $\partial_{-1}$} (Ddots2);
    \draw[-stealth] (C1)     -- node[pos=0.4, left]{\small $\phii_0$}    (D0);
    \draw[-stealth] (C0)     -- node[pos=0.4, left]{\small $\phii_1$} (Dm1);
    \draw[-stealth] (Cdots1) -- node[pos=0.4, left]{\small $\phii_{-1}$}    (D1);
    \draw[-stealth] (Cm1)    -- node[pos=0.4, left]{\small $\phii_{2}$} (Ddots2);
    \end{tikzpicture}
\end{equation}
The shifted arrows ensures that  the gates are performed between the $\mathsf{B}(C_{-1})$ qubits of the hypergraph cluster state and the qubits of the CSS code associated with $D$.
The chain map above defines a product of $\mathsf{CZ}$ gates, $U_{\mathsf{CZ}}^\varphi$, via Eq.~\eqref{eq:productofCZ}.
Then, a quantum circuit that measures a logical Pauli $Z$ operator is given by:
\begin{equation} \label{eq-measurementcircuit}
M^{\phii}_{Z} = \left(\prod_{e \in \mathsf{B}(C_0)} M^C_{e, Z} \right) \times  \left( \prod_{\mathbf{v} \in \mathsf{B}(C^{-1}) } M^{C}_{\mathbf{v}, X} \right) U_{\mathsf{CZ}}^{\phii}
\end{equation}
where $M^C_{j, \alpha}$ is a projective measurement of the qubit $j$ in code $C$ in the $\alpha$-basis. As a reminder, the hypergraph cluster state on code $C$ has qubits on both $\mathsf{B}(C_0)$ and $\mathsf{B}(C_{-1})$, which are fully measured out by $M_Z^\varphi$. 
We have the following theorem:
\begin{shaded*}
\begin{restatable}[Code Surgery from Chain Maps]{theorem}{CodeSurgeryChainMaps} \label{thm:CodeSurgeryChainMaps} 
Let $\phii_{\bullet}:C[-1]^{\bullet} \to D_{\bullet}$ be a chain map.
Then, $M^{\phii}_Z$ of Eq.~\eqref{eq-measurementcircuit} leaves the qubits in $C$ in a product state and, after applying a product of Pauli operators to correct for random measurement outcomes, implements the following logical measurement on code $D$:
\begin{equation}
    \prod_{[\mathbf{c}] \in \mathsf{B}(H^{-1}(C))}
    \overline{M}^D_{\phii_*([\mathbf{c}]), Z}.
\end{equation}
i.e.\ a product of logical Pauli $Z$ measurements in code $D$.
In the above, $\mathsf{B}(H^{-1}(C))$ is a chosen basis of cohomology classes with \(\mathbf{c}\) denoting a chosen cocycle representative of
\([\mathbf{c}]\), and
\(\phii_*:H^{-1}(C)\to H_0(D)\) is the induced map on (co)homology.
\end{restatable}
\end{shaded*}
\begin{proof}[Proof (Informal)] 
The full proof of the above theorem is deferred to Appendix~\ref{subapp:PauliMeasurements}.
A sketch of why the theorem holds can be understood by replacing the measurements in Eq.~\eqref{eq-measurementcircuit} with projections on the $\ket{0} $ state and the $\ket{+}$ state.
To see this, let us note that prior to the circuit of Eq.~\eqref{eq-measurementcircuit}, the hypergraph cluster state has the property that $X^C_{\mathbf{c}}  = +1$ for all  $\mathbf{c} \in \mathsf{ker}(\exd_0)$.
After the unitary gate in Eq.~\eqref{eq-measurementcircuit}, this condition is updated to [Eq.~\eqref{eq:phiactsona}]:
\begin{equation}
    U_{\mathsf{CZ}}^{\phii} X^C_{\mathbf{c}} (U_{\mathsf{CZ}}^{\phii})^{\dagger} =X^C_{\mathbf{c}} Z_{\phii(\mathbf{c})}^D = +1 
\end{equation}
where $\phii(\mathbf{c})$ is a $0$-cycle of the chain complex $D_{\bullet}$, and hence \(Z_{\phii(\mathbf{c})}^D\) is a logical Pauli $Z$ gate of code $D$.
After the projection, the ancilla qubits will be in a product state satisfying $X_{\mathbf{c}}^{C} = +1$ and the qubits of code $D$ will be restored back into the code space of the CSS code associated with $D_{\bullet}$. 
But since $X^C_{\mathbf{c}} Z_{\phii(\mathbf{c})}^D = +1$, this means that we have projected the logical operator of code $D$, $Z_{\phii(\mathbf{c})}^D = +1$.

As a remark, if we replace the physical projections with physical measurements, two things change.
First, the \textit{logical} projection we just derived will be turned into a \textit{logical} measurement.
Second, it will be necessary to apply a pattern of Pauli operators to return \(D\) to a code state.
\end{proof}
The above shows that chain maps can specify a logical Pauli measurement operation via a code surgery.
Moreover, to make a connection with the typical code switching picture of code surgery, we note that after measuring the $\mathsf{B}(C^{-1})$ qubits in Eq.~\eqref{eq-measurementcircuit}, our code is ``switched'' into a stabilizer code typically called the ``mapping cone code''.
In particular, this code is associated with a chain complex known as the \textit{mapping cone} of the chain map  $\phii_{\bullet}$ of Eq.~\eqref{eq-diagchainmap}~\cite{ide2025faulttolerant,Cowtan2024SurgeryUniversal}:
\begin{equation}
        \mathsf{Cone}(\phii_{\bullet}) \equiv \begin{tikzpicture}[scale=0.7, baseline={([yshift=-.5ex]current bounding box.center)}]
      \node (Cdots1) at (0, 0)    {\small $\cdots$};
      \node (C1)     at (2, 0)    {\small $C^{-1}$};
      \node (C0)     at (4, 0)    {\small $C^0$};
      \node (Cm1)    at (6, 0)    {\small $C^{1}$};
      \node (Cdots2) at (8, 0)    {\small $\cdots$};
      \node (Ddots1) at (0, -1.2) {\small $\cdots$};
      \node (D1)     at (2, -1.2) {\small $D_1$};
      \node (D0)     at (4, -1.2) {\small $D_0$};
      \node (Dm1)    at (6, -1.2) {\small $D_{-1}$};
      \node (Ddots2) at (8, -1.2) {\small $\cdots$};
      
      \node at (2, -0.6) {\small $\oplus$};
      \node at (4, -0.6) {\small $\oplus$};
      \node at (6, -0.6) {\small $\oplus$};

      \draw[color = black] (1.55, 0.5) -- (1.55, -1.7) -- (2.45, -1.7) -- (2.45, 0.5) -- cycle;
      \draw[color = black] (3.55, 0.5) -- (3.55, -1.7) -- (4.45, -1.7) -- (4.45, 0.5) -- cycle;
      \draw[color = black] (5.55, 0.5) -- (5.55, -1.7) -- (6.45, -1.7) -- (6.45, 0.5) -- cycle;
      
      \draw[-stealth] (Cdots1) -- node[above]{\small $\exd_{-1}$} (C1);
      \draw[-stealth] (C1)     -- node[above]{\small $\exd_0$} (C0);
      \draw[-stealth] (C0)     -- node[above]{\small $\exd_1$} (Cm1);
      \draw[-stealth] (Cm1)    -- node[above]{\small $\exd_{2}$} (Cdots2);
      \draw[-stealth] (Ddots1) -- node[below]{\small $\partial_2$} (D1);
      \draw[-stealth] (D1)     -- node[below]{\small $\partial_1$} (D0);
      \draw[-stealth] (D0)     -- node[below]{\small $\partial_0$} (Dm1);
      \draw[-stealth] (Dm1)    -- node[below]{\small $\partial_{-1}$} (Ddots2);
    \draw[-stealth] (C1)     --     (D0);
    \draw[-stealth] (C0)     --  (Dm1);
    \draw[-stealth] (Cdots1) --     (D1);
    \draw[-stealth] (Cm1)    --  (Ddots2);
    \end{tikzpicture}
\end{equation}
which specifies stabilizers given by:
\begin{equation} \label{eq:conestab}
 X^C_{\exd_0 \mathbf{v}} Z^D_{\phii_0(\mathbf{v})} \quad Z_{\partial_1 p}^C \qquad Z_{\partial_1 q}^D \qquad X_{\exd_0 \mathbf{w}}^D Z^C_{\phii_1^{\mathsf{T}}(\mathbf{w})} 
\end{equation}
where $\mathbf{v} \in \mathsf{B}(C^{-1})$, $p \in \mathsf{B}(C_1)$, $q \in \mathsf{B}(D_1)$, $\mathbf{w} \in \mathsf{B}(D^{-1})$.
The above stabilizers can be shown to commute.
In particular, one has that:
\begin{shaded*}
\begin{restatable}[The Mapping Cone Code]{lemma}{MappingConeCode} \label{lem:MappingConeCode}
    After measuring the $\mathsf{B}(C^{-1})$ qubits in Eq.~\eqref{eq-measurementcircuit}, the state of the system is in a stabilizer state that is related to the state specified by Eq.~\eqref{eq:conestab} by a product of local Pauli operators.
\end{restatable}
\end{shaded*}

At a remark, in the above, we focused on the logical measurement of the Pauli $Z$ operator.
The same protocol can be adjusted to measure the Pauli $X$ operator by performing the code surgery gadget with $\mathsf{CNOT}$ gates instead of $\mathsf{CZ}$ gates, corresponding to chain maps $\phii_{\bullet}: C^{\bullet}[1] \to D^{\bullet}$.
Concretely, given a chain map $\phii_{\bullet}: C^{\bullet}[1] \to D^{\bullet}$, we perform the quantum circuit:
\begin{equation} \label{eq-measurementcircuitX}
M^{\phii}_{X} = \left(\prod_{e \in \mathsf{B}(C_0)} M^C_{e, Z} \right) \times  \left( \prod_{\mathbf{v} \in \mathsf{B}(C^{-1}) } M^{C}_{\mathbf{v}, X} \right) U_{\mathsf{CNOT}}^{\phii}
\end{equation}
between the hypergraph cluster state and a CSS code of interest to measure a product of Pauli $X$ operators.
Moreover, it is worth remarking that the mapping cone code associated with the logical Pauli $X$ measurement will be a CSS code---in contrast to the stabilizers of Eq.~\eqref{eq:conestab}.

\subsection{The Chain Map Complex, Clifford Stabilizers and Clifford Logic}
\label{subsec:chainmapcomplexandcliff}

In the previous subsection, we reviewed how a CSS code is encoded by a chain complex: physical Pauli operators are chains, Pauli stabilizers are boundaries, and logical Pauli gates are labeled by the resulting homology classes. 
On the other hand, logical Clifford operations, including unitary
Clifford gates and logical Pauli measurements, are conventionally 
described rather differently: as chain \emph{maps} between the chain complexes
associated with CSS codes.

We now place these two descriptions on the same footing.
We use a mathematical ingredient that is standard in homological algebra:  collections of maps between two chain complexes themselves form a chain complex, the internal-hom complex, which we refer to as the \emph{chain map complex}~\cite{Eilenberg1966ClosedCategories,Weibel1994HomologicalBook}. Our contribution is to give this  mathematical object a physical interpretation as a  complex of \emph{operations} on quantum codes.
In this dictionary, physical Clifford gate patterns are chains and certain logically trivial circuits---which we call \emph{Clifford stabilizers}---are boundaries. The homology of the chain map complex then labels logical Clifford operations, in direct analogy with the way the homology of the original code complex labels logical Pauli operations. Since the chain map complex is itself a chain complex, it defines an auxiliary CSS code.
Thus, the chain map complex turns logical Clifford operations on the original codes into ordinary Pauli logicals of this auxiliary CSS code.

We develop this correspondence in stages. In Sec.~\ref{subsec:null}, we first identify Clifford stabilizers, which are local, logically trivial \emph{deformations} of Clifford gates, analogous to deforming Pauli logicals by Pauli stabilizers. We then introduce the full chain map complex and show that its homology labels logical Clifford operations in Sec.~\ref{subsubsec:chain_map_complex}.  Next, we identify the chain map complex with a tensor-product complex in Sec.~\ref{subsubsec:chain_maps_as_tensors}. This makes the auxiliary code concrete and turns the search for useful physical Clifford implementations into the familiar problem of finding and deforming Pauli logical representatives within a homology class, which will be exploited later in Sec.~\ref{sec:DiscoveringGates}
for constructing explicit logical gates. Importantly, because the chain map complex is itself a chain complex describing an auxiliary code, the construction can be iterated. This recursion produces the chain map hierarchy of Sec.~\ref{subsec:ChainMapHeirarchy}, extending the homological chain complex framework to non-Clifford operations arbitrarily high in the Clifford hierarchy.

For most of the discussion, we focus on inter-block \(\mathsf{CZ}\)-type gates
between two or more CSS codes. Generalizations to other Clifford gates and to
intra-block gates are discussed in
Sec.~\ref{subsubsec:WhichCliffords} and
Sec.~\ref{sec:IntrablockGates}, respectively. 

\subsubsection{Clifford Stabilizers and Null Homotopic Chain Maps}
\label{subsec:null}

The starting point is the observation, discussed informally in Section~\ref{sec:Intuition}, that logical Clifford gates admit local, logically trivial
deformations analogous to the familiar deformations of logical Paulis.
We will show that these deformations are generated by \emph{Clifford stabilizer} circuits that preserve the code space and act trivially on the logical qubits.

To define Clifford stabilizers formally, we first introduce the concept of a null-homotopic chain map~\cite{Weibel1994HomologicalBook,hatcher2002algebraic} [cf. Eq.~\eqref{eq:informal_null_homotopy} for the informal sketch], and then discuss how to associate this map with logically trivial Clifford circuits.

\begin{defn}[(Local) Null-Homotopic Chain Maps]
\label{def:NullChainMapChainHomotopy}
A chain map \(\theta_\bullet : C^\bullet \to D_\bullet\) with components \(\theta_j : C^j \to D_{-j}\) is
\emph{null-homotopic} if there exists a collection of maps \(Q_\bullet\), with
components
\(
    Q_{j}:C^j\to D_{1-j},
\)
such that for all integers $j$,
\begin{equation}\label{eq:nullTheta}
    \theta_j
    =
    \partial_{1-j} \circ Q_{j}
    +
    Q_{j+1} \circ \exd_{j+1}.
\end{equation}
Diagrammatically, the maps \(Q_\bullet\) appearing in
Eq.~\eqref{eq:nullTheta} have the form:
\begin{equation} \label{eq:defNull}
\begin{tikzpicture}[scale=0.7, baseline={([yshift=-.5ex]current bounding box.center)}]
  \node (Cdots1) at (0, 0)    {\small $\cdots$};
  \node (C1)     at (2, 0)    {\small $C^{-1}$};
  \node (C0)     at (4, 0)    {\small $C^0$};
  \node (Cm1)    at (6, 0)    {\small $C^{1}$};
  \node (Cdots2) at (8, 0)    {\small $\cdots$};

  \node (Ddots1) at (0, -2.2) {\small $\cdots$};
  \node (D1)     at (2, -2.2) {\small $D_1$};
  \node (D0)     at (4, -2.2) {\small $D_0$};
  \node (Dm1)    at (6, -2.2) {\small $D_{-1}$};
  \node (Ddots2) at (8, -2.2) {\small $\cdots$};

  \draw[-stealth, dashed] (Cdots1) -- node[above]{\small $\exd_{-1}$} (C1);
  \draw[-stealth, dashed] (C1)     -- node[above]{\small $\exd_0$} (C0);
  \draw[-stealth, dashed] (C0)     -- node[above]{\small $\exd_1$} (Cm1);
  \draw[-stealth, dashed] (Cm1)    -- node[above]{\small $\exd_2$} (Cdots2);

  \draw[-stealth, dashed] (Ddots1) -- node[below]{\small $\partial_2$} (D1);
  \draw[-stealth, dashed] (D1)     -- node[below]{\small $\partial_1$} (D0);
  \draw[-stealth, dashed] (D0)     -- node[below]{\small $\partial_0$} (Dm1);
  \draw[-stealth, dashed] (Dm1)    -- node[below]{\small $\partial_{-1}$} (Ddots2);

  \draw[-stealth] (C1)     -- node[pos=0.4, left]{\small $Q_{-1}$}    (Ddots1);
  \draw[-stealth] (C0)     -- node[pos=0.4, left]{\small $Q_0$}    (D1);
  \draw[-stealth] (Cm1)    -- node[pos=0.4, left]{\small $Q_{1}$} (D0);
  \draw[-stealth] (Cdots2) -- node[pos=0.4, left]{\small $Q_{2}$} (Dm1);
\end{tikzpicture}
\end{equation}
The dashed horizontal arrows in Eq.~\eqref{eq:defNull} indicate that the
diagram need not commute.

Two chain maps \(\phii_\bullet\) and \(\psi_\bullet\) are
\emph{chain homotopic} if there is a null-homotopic map
\(\theta_\bullet\) such that
\[
    \phii_\bullet = \psi_\bullet + \theta_\bullet .
\]
This is an equivalence relation. 

If \(\mcbb{C}\) and \(\mcb{D}\) are based, we say that
\(\theta_\bullet\) is \emph{local} (in the LDPC sense) if each of the component maps
\(\theta_j\), written as a matrix in the associated bases, has \(O(1)\) nonzero entries in total. This is to be contrasted with a sparse map, which
may have \(O(1)\) nonzero entries in each row and column.

Although written for maps \(\mcbb C\to \mcb D\), analogous definitions apply for
null homotopies between the other chain maps considered above, such as
\(\mcbb C\to \mcbb D\). 
\end{defn}
The distinction between \(Q_\bullet\) and \(\theta_\bullet\) is important. The maps \(Q_j\) in Eq.~\eqref{eq:defNull} do \emph{not} themselves specify
physical \(\mathsf{CZ}\) circuits. For example, \(Q_0\) maps \(C^0\) to
\(D_1\), whereas a physical \(\mathsf{CZ}\)-gate pattern is specified by a map
from \(C^0\) to the physical qubit space \(D_0\), as in
Eq.~\eqref{eq:productofCZ}. Instead, Eq.~\eqref{eq:nullTheta} converts \(Q_\bullet\) into chain maps
\(\theta_j:C^j\to D_{-j}\) of the form in Eq.~\eqref{eqn:ChainMapDiagram}, whose degree-0 component $\theta_0$ specifies a physical \(\mathsf{CZ}\)-gate pattern via
Eq.~\eqref{eq:productofCZ}.
 The two terms in Eq.~\eqref{eq:nullTheta}
are the two natural ways to use $Q$ to get from \(C^j\) to \(D_{-j}\): first apply \(Q_j\)
and then a boundary in \(D_\bullet\), or first apply a coboundary in
\(C^\bullet\) and then \(Q_{j+1}\).
\begin{equation}   \label{eq-illustrateboundary}
    \includegraphics[valign = c]
 {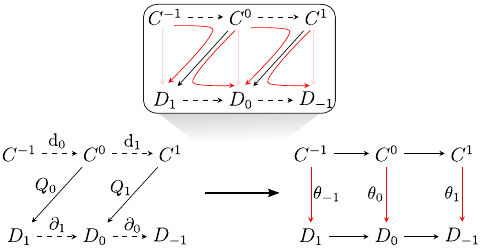}
\end{equation}
Specializing to the physical degree-zero component which determines gates, $\theta_0: C^0\longrightarrow D_0$, gives
\begin{equation}\label{eq:nullThetaZero} \theta_0 = \partial_1\circ Q_0 + Q_1\circ\exd_1, \end{equation}

A few additional remarks are in order.
First, let us explicitly verify that the null homotopic chain maps $\theta_\bullet$ defined in Eq.~\eqref{eq:nullTheta} do indeed form chain maps, i.e. they obey the commutation condition in Eq.~\eqref{eq:chainmapcommute}:
\begin{equation}
    \theta_{j} \circ \exd_{j} =  \partial_{1-j} \circ Q_{j} \circ \exd_j = \partial_{1-j}\circ \theta_{j-1} 
\end{equation}
where we used $ \exd_{j+1}\circ \exd_j= 0$ in the first equality, and  $\partial_{1-j}\circ \partial_{2-j} =0$ in the second. 

Second, null homotopic chain maps all have trivial action on (co)homology.
To see this, consider the action of a null homotopic map \(\theta = \partial \circ Q + Q \circ \exd\) on a cocycle \(\mathbf{c} \in \mathsf{ker}(\exd)\).
Note that \(\theta(\mathbf{c}) = \partial Q(\mathbf{c}) + 0\) is a boundary, and thus in the zero homology class. More generally, null-homotopic maps send every (co)homology
class to the trivial (co)homology class. Conversely, chain homotopies
are precisely the equivalence relations that identify chain maps with the same
homological action~\cite{Weibel1994HomologicalBook}:

\begin{lem}[Trivial Chain Maps as Null Homotopies] \label{lem:homotopyclass} Two chain maps $\phii_{\bullet}$ and $\psi_{\bullet}$ are chain homotopic if and only if they induce the same action on
all (co)homology groups.
As a consequence, any chain map with trivial action on homology is a null homotopy. 
\end{lem}

Thus null homotopies are exactly the deformations that do not change
the induced logical action of a chain map.
(We note that, for this lemma to hold, it is important that we work over \(\FF_2\) or another field.)

Third and finally, to use null homotopies as physical Clifford stabilizer circuits, we require these deformations to be generated locally, mirroring the locality of Pauli stabilizer generators in an LDPC code. The term ``local in the LDPC sense" is meant to emphasize the analog of locality
appropriate to non-local qLDPC codes: the map $\theta_j$ has support on only \(O(1)\)
basis elements \emph{in total},  although those basis elements need not be nearby in any geometric sense. When no confusion is likely, we simply call such maps local.
 For sparse
based complexes, this locality condition follows directly:

\begin{lem}[Local Basis for Null-Homotopic Chain Maps]
\label{lem:local_null_homotopies}
Let \(\mcbb C\) and \(\mcb D\) be sparse, based chain complexes. Then the vector
space of null-homotopic chain maps \( \theta_\bullet: \mcbb C\to \mcb D\) has a basis of \emph{local} null-homotopic chain maps.
\end{lem}

\begin{proof}
For each \(j\), the space of maps \(Q_j\in \Hom(C^j,D_{1-j})\) has an
elementary matrix basis determined by the bases of \(C^j\) and
\(D_{1-j}\). We denote these basis maps by \(Q_j^{\mathbf c,d}\), indexed by
\(\mathbf c\in\mathsf B(C^j)\) and \(d\in\mathsf B(D_{1-j})\), where
\(Q_j^{\mathbf c,d}(\mathbf c)=d\) and  zero on all
other basis elements.
For example, for \(\mathbf e\in\mathsf B(C^0)\) and \(p\in\mathsf B(D_1)\), one basis
map is \(Q_0^{\mathbf e,p}:C^0\to D_1\), defined by
\(Q_0^{\mathbf e,p}(\mathbf e')=\delta_{\mathbf e,\mathbf e'}p\) for
\(\mathbf e'\in\mathsf B(C^0)\) [cf. Eq.~\eqref{eq-firstnullhomotopy}]. 

Since  $\theta_{\bullet}$ in Eq.~\eqref{eq:nullTheta} depends linearly on \(Q\), applying Eq.~\eqref{eq:nullTheta} to these
basis maps spans the space of null-homotopic chain maps. Since \(\mcbb C\) and
\(\mcb D\) are sparse, composing  a basis map of $Q$ with the (co)boundary
maps in Eq.~\eqref{eq:nullTheta} produces only \(O(1)\) nonzero matrix entries. Hence each resulting
null-homotopic chain map is local. Choosing a linearly independent
subset of these spanning maps gives a local basis of null homotopic maps.
\end{proof}

The preceding discussion has a striking consequence. We are now ready to define Clifford stabilizers, the promised analog of Pauli stabilizers, by associating a physical CZ circuit with the chain map $\theta_\bullet$ according to Eq.~\eqref{eq:productofCZ}: 
\begin{defn}[(Local) Clifford Stabilizers]
\label{def:clifford_stabilizers}
Let \(\theta_\bullet:\mcbb C\to\mcb D\) be a null-homotopic chain map. The
physical Clifford circuit
\[
    U_{\mathsf{CZ}}^\theta
    =
    \prod_{e\in \mathsf B(C_0)}
    \mathsf{CZ}^{CD}_{e,\theta_0(\mathbf e)}
\]
is called a
\emph{Clifford stabilizer}. This is  obtained from Eq.~\eqref{eq:productofCZ} with \(\phii=\theta\).

Because \(\theta_\bullet\) is a chain map, \(U_{\mathsf{CZ}}^\theta\) preserves
the code space and has a well-defined action on (co)homology. Because
\(\theta_\bullet\) is null-homotopic, this induced action is trivial.

If \(\theta_\bullet\) is local, then
\(U_{\mathsf{CZ}}^\theta\) is called a \emph{local Clifford
stabilizer}. By Lemma~\ref{lem:local_null_homotopies}, for sparse
based chain complexes, the null-homotopic chain maps admit a local basis;
the corresponding Clifford stabilizers are called
\emph{local Clifford stabilizer generators}:

\begin{equation}\label{eq:elementaryCliffordStabilizers}
\begin{aligned}
    &\mathsf{CZ}^{CD}_{e_C,\partial_1 p_D}
    &&\qquad
    e_C\in\mathsf B(C_0),\quad p_D\in\mathsf B(D_1),
    \\
    &\mathsf{CZ}^{CD}_{\partial_1 p_C,e_D}
    &&\qquad
    p_C\in\mathsf B(C_1),\quad e_D\in\mathsf B(D_0).
\end{aligned}
\end{equation}
\end{defn}

The two generators in Eq.~\eqref{eq:elementaryCliffordStabilizers} arise from
the two contributions to
\(\theta_0=\partial_1Q_0+Q_1\exd_1\) in Eq.~\eqref{eq:nullThetaZero}. An elementary \(Q_0\) maps a qubit
label \(e_C\) to a \(Z\)-stabilizer label \(p_D\), while an elementary \(Q_1\)
maps a \(Z\)-stabilizer label \(p_C\) to a qubit label \(e_D\). In the toric
code, these are the edge-to-plaquette and plaquette-to-edge constructions discussed in Section~\ref{sec:Intuition},
respectively.

Clifford stabilizers thus deform logical Clifford gates in exactly the same
sense that Pauli stabilizers deform logical Pauli operators. Let
\(\phii_\bullet:C^\bullet\to D_\bullet\) be a chain map defining a logical
\(\mathsf{CZ}\) circuit \(U_{\mathsf{CZ}}^\phii\) via Eq.~\eqref{eq:productofCZ}. As
shown later in Theorem~\ref{thm:Cliffordlogic}, every logical
\(\mathsf{CZ}\) action admits at least one physical representative of the form
\(U_{\mathsf{CZ}}^\phii\).
Then, if
\(\theta_\bullet\) is null-homotopic, it can be used to construct a different physical representative which is logically equivalent:
\begin{equation}
\label{eq:deformUCZ}
     U_{\mathsf{CZ}}^\phii
    \sim
    U_{\mathsf{CZ}}^{\phii+\theta}.
\end{equation}
Indeed, \(\phii_\bullet\) and \(\phii_\bullet+\theta_\bullet\) induce the same
maps on (co)homology and hence the same logical action. Thus Clifford
stabilizers generate the deformation freedom within a fixed homology class of
the chain map complex. This allows one to search within that class for simpler
physical representatives of a given logical action, for example representatives of lower depth or
sparser support, although such representatives need not always exist. We
return to this problem in Sec.~\ref{sec:DiscoveringGates}.

\subsubsection{The Chain Map Complex and Clifford Logic}
\label{subsubsec:chain_map_complex}

The preceding subsection identified the Clifford analog of Pauli stabilizer
deformations. In the Pauli case, stabilizers are not isolated objects; rather, they appear as \emph{boundaries} in the chain complex that
also contains physical Pauli operators and their syndrome violations,  with logical Pauli operators labeled
by the resulting homology classes. We now show that Clifford gates admit an analogous chain complex structure.

For the chain complex which organizes Clifford gates, the relevant chains are collections of \emph{maps} between the constituent codes. For clarity, we begin with the
inter-block \(\mathsf{CZ}\) case, involving maps from \(C^\bullet \rightarrow D_\bullet\). We have already encountered several such
maps.  Logical \(\mathsf{CZ}\) gates are specified by a chain map
\(\phii_\bullet\), namely a collection of maps
\(\phii_i:C^i\to D_{-i}\), depicted as downward arrows in
Eq.~\eqref{eqn:ChainMapDiagram}, which satisfy the commutativity condition
Eq.~\eqref{eq:chainmapcommute}.
Clifford stabilizers are specified by null-homotopic chain maps, obtained from a collection of 
``shifted" maps \(Q_\bullet\) of Eq.~\eqref{eq:defNull}, with components
\(Q_i:C^i\to D_{1-i}\), shifted diagonally one step to the left relative to the downward maps. Similarly,
the code-surgery maps in Eq.~\eqref{eq-diagchainmap} have components $C^i\to D_{-1-i}$, shifted by \(-1\), i.e. one step to the right.

This section assembles
these various Clifford operations and constructions into a single homological structure, in which
Clifford stabilizers appear as boundaries and logical Clifford gates are
labeled by homology classes.
To do so, it is natural to consider collections of maps shifted by an integer $j$, with components
\(C^i\to D_{j-i}\).
In fact, these $j$-shifted collections of maps assemble into a chain
complex (See Eq.~\eqref{eq-chainmapcomplex} for a depiction). This is the standard
internal-hom complex of homological algebra~\cite{Eilenberg1966ClosedCategories,Weibel1994HomologicalBook}, which we denote as $[C^\bullet, D_\bullet]$.  We call it the
\emph{chain map complex}, reflecting the fact that its cycles are precisely
(shifted) chain maps.

\begin{defn}[Chain Map Complex/Internal Hom]\label{def:ChainMapComplex}
The \emph{chain map complex} or \emph{internal hom} associated to a pair of chain complexes $\mcbb{C}$ and $\mcb{D}$ is a chain complex 
\([\mcbb{C},\mcb{D}]\):
\begin{equation}
     \cdots \xrightarrow[]{\partial^{[2]}_2} [C,D]_1 \xrightarrow[]{\partial^{[2]}_1} [C,D]_0 \xrightarrow[]{\partial^{[2]}_0} [C,D]_{-1} \xrightarrow[]{\partial^{[2]}_{-1}}  \cdots 
\end{equation}
(suppressing bullets) whose degree-\(j\) component is 
\begin{equation}\label{eqn:CMCComponents}
    [C^{\bullet}, D_{\bullet}]_j = \bigoplus_{i \in \mathbb{Z}} \Hom(C^i, D_{j -i}),
\end{equation}
Thus a $j$-chain, \(f_{j, \bullet}\in [C^\bullet,D_\bullet]_j\) is a \(j\)-shifted
collection of maps, with components \[f_{j,i}:C^i\to D_{j-i}.\]
No commutativity condition is imposed at this stage.

For $f_{j,\bullet} \in [C^{\bullet}, D_{\bullet}]_{j}$ the boundary map $\partial^{[2]}_j: [C^{\bullet}, D_{\bullet}]_{j} \to [C^{\bullet}, D_{\bullet}]_{j-1}$ 
is defined componentwise by
\begin{equation}\label{eqn:CMCBoundaryIndices}
    \bigl(\partial^{[2]}_jf_{j,\bullet}\bigr)_i
    =
    \partial_{j-i}^D\circ f_{j,i}
    +
    f_{j,i+1}\circ \exd^C_{i+1}.
\end{equation}
Equivalently, suppressing indices, \(\partial^{[2]}f=\partial f+f\exd\).

\end{defn}

The boundary maps in Eq.~\eqref{eqn:CMCBoundaryIndices}, $\partial_j^{[2]}: [C,D]_j \to [C,D]_{j-1}$,  convert a \(j\)-shifted collection of maps into a \((j-1)\)-shifted collection by
composing with the boundary maps of \(D_\bullet\) and the coboundary maps of
\(C^\bullet\). One readily verifies that these satisfy the chain complex condition $\partial^{[2]}_{j-1}\circ \partial^{[2]}_j = 0$.%
\footnote{Suppresing component indices and applying two successive boundary maps to
\(f\in [\mathbf{C}^{[2]}]_j\) gives
\begin{equation}
    \partial^{[2]}_{j-1}\circ \partial^{[2]}_j f
    =
    \partial^2\circ f
    +
    \partial\circ f\circ \exd
    +
    \partial\circ f\circ \exd 
    +
    f\circ \exd^2=0.
\end{equation}
This vanishes over \(\mathbb F_2\), since
\(\partial^2=\exd^2=0\) and the two mixed terms add to zero modulo 2.}
The superscript ($^{[2]}$) anticipates that this complex is built from two constituent code complexes and will describe Clifford operations at the second level of the Clifford hierarchy.

To first introduce notation that will extend naturally to operations at all levels of the Clifford hierarchy, as developed in Sec.~\ref{subsec:ChainMapHeirarchy}, we package the two constituent complexes \(C^\bullet\) and \(D_\bullet\) into the ordered length-two tuple $\mathbf{C}^{[2]}=(D,C)$, with the associated chain complex denoted as $[\mathbf{C}^{[2]}]_\bullet = [C^\bullet,D_\bullet]$.
We suppress the chain or cochain orientations of the constituent complexes within $\mathbf{C}^{[2]}$, restoring them when specifying the operations under consideration. For example, $\mathsf{CNOT}$-type operations are instead described by the chain complex
$
[\mathbf{C}^{[2]}] = [C^\bullet,D^\bullet]
$.

The key observation that establishes the link between the chain map \emph{complex} and \emph{chain maps} is that the boundary map \(\partial_0^{[2]}\) exactly measures the failure of a map to send boundaries to boundaries.
Indeed, \(\phii_\bullet \in \CMH{2}_0\) is a \(0\)-cycle when
\begin{equation}
    \partial_0^{[2]}\phii_\bullet = \partial \phii_\bullet + \phii_\bullet \exd = 0.
\end{equation}
Over \(\mathbb F_2\), this is equivalent to
\(\partial\circ\phii_\bullet=\phii_\bullet\circ\exd\), which is the chain map commutation condition in Eq.~\eqref{eq:chainmapcommute}. 
Thus, the \(0\)-cycles are precisely the
downward chain maps of Eq.~\eqref{eqn:ChainMapDiagram}.
Further, \(0\)-boundaries in the chain map complex are exactly the null homotopies \(\theta_\bullet\),
\begin{equation}
     \theta_\bullet = \partial_1^{[2]}Q_\bullet = \partial Q_\bullet + Q_\bullet \exd,
\end{equation}
where \(Q_\bullet \in \CMH{2}_1\) is arbitrary.
Regarding the cycles \(\mathsf{ker}(\partial^{[2]}_0)\) up to addition of boundaries \(\mathsf{im}(\partial^{[2]}_0)\), we see that the 0-homology of the chain map complex \(H_0(\CMH{2}) = \mathsf{ker}(\partial^{[2]}_0)/\mathsf{im}(\partial^{[2]}_0)\) is nothing but the \emph{homotopy classes} of chain maps \(C^\bullet \to D_\bullet\)---classes of chain maps up to homotopy equivalence.
By Lemma~\ref{lem:homotopyclass}, these classes are exactly the same as classes of chain maps with the same induced action on (co)homology.

The chain map complex thus inherits all of the relations between chain maps and logical Clifford gates.
Given a generic \(f_\bullet \in \CMH{2}_0\), the component \(f_0\) specifies a pattern of \(\mathsf{CZ}\) gates by the same rule as Eq.~\eqref{eq:productofCZ}:
\begin{equation}
\label{eq-productofCZ_f}
    U_{\mathsf{CZ}}^f
    =
    \prod_{e\in\mathsf B(C_0)}
    \mathsf{CZ}^{CD}_{e,f_0(\mathbf e)}.
\end{equation}
As a general \(f\) need not commute with the boundary maps, \(U_{\mathsf{CZ}}^f\) is need not be a logical gate.
However, when \(f_\bullet = \phii_\bullet \in \mathsf{ker}(\partial^{[2]})\) is a cycle, \(U_{\mathsf{CZ}}^\phii\) is logical; and if two cycles \(\phii_\bullet\) and \(\psi_\bullet\) differ by a null-homotopy the gates \(U_{\mathsf{CZ}}^\phii\) and \(U_{\mathsf{CZ}}^{\psi}\) have the same logical action.

Other homology groups \(H_j(\CMH{2})\) encode shifted chain maps \(\phii_\bullet : C[j]^\bullet \to D_\bullet\), up to homotopy.
In particular, the logical measurement gadgets associated to chain maps \(\phii_\bullet : C[-1]^\bullet \to D_\bullet\) (Theorem~\ref{thm:CodeSurgeryChainMaps}) are associated to \((-1)\)-cycles of \(\CMH{2}_{-1}\), and homologous cycles encode gadgets that measure the same pattern of logicals.

This relation between the chain map complex and Clifford logic is summarized by the following theorem.

\begin{shaded*}
\begin{restatable}[Clifford Logic and the Homology of the Chain Map Complex]{theorem}{Cliffordlogic}
\label{thm:Cliffordlogic} 
Consider the chain map complex $[\mathbf{C}^{[2]}]_{\bullet} \simeq [C^{\bullet}, D_{\bullet}]$.
Then the following are true:
\begin{enumerate}
    \item \textit{Homology and Logical Clifford Gates:} Each homology class
\([\phii]\in H_0([\mathbf{C}^{[2]}]_\bullet)\) determines a unitary logical
\(\mathsf{CZ}\) action coupling codes \(C\) and \(D\). Any representative
\(\phii_\bullet\) of the class specifies a unitary logical CZ circuit
\(U_{\mathsf{CZ}}^\phii\) of the form
Eq.~\eqref{eq:productofCZ}.  Moreover, if two logical $\mathsf{CZ}$ circuits of this form act differently on logical qubits, they are labeled by different homology classes of $H_0([\mathbf{C}^{[2]}]_{\bullet})$.

    \item \textit{Homology and Logical Pauli Measurement:} Each homology class
\([\phii]\in H_{-1}([\mathbf{C}^{[2]}]_\bullet)\) determines a logical
\(Z\)-Pauli measurement pattern on code \(D\). Any representative
\(\phii_\bullet\) specifies a physical code-surgery gadget
\(M_Z^\phii\) of the form Eq.~\eqref{eq-measurementcircuit}.
    Moreover, if two code surgery gadgets of this form measure different patterns of logical qubits, they are labeled by different homology classes of $H_{-1}([\mathbf{C}^{[2]}]_{\bullet})$.

    \item \textit{Exhaustiveness:} For any logical $\mathsf{CZ}$ action, there is logical $\mathsf{CZ}$ gate of the form $U_{\mathsf{CZ}}^{\phii}$ [Eq.~\eqref{eq:productofCZ}] labeled by a homology class of $H_0([\mathbf{C}^{[2]}]_{\bullet})$ that performs that action.
    Moreover, for any Pauli $Z$ measurement on logical qubits, there is code surgery gadget of the form $M_{Z}^{\phii}$ [Eq.~\eqref{eq-measurementcircuit}] labeled by a homology class of $H_{-1}([\mathbf{C}^{[2]}]_{\bullet})$ that performs that action.
\end{enumerate}
\end{restatable}
\end{shaded*}

The first two statements follow directly from Lemma~\ref{lem:homotopyclass} together with Theorems~\ref{thm:TransversalGatesFromChainMaps} and~\ref{thm:CodeSurgeryChainMaps}.
Exhaustiveness follows from the fact that any linear action on (co)homology is achieved as the induced action of some chain map, with a direct proof being given in Appendix~\ref{subapp:CliffordLogic}.

A few remarks are in order.
First, while the above theorem is phrased in terms of logical $\mathsf{CZ}$ gates and logical Pauli $Z$ measurements, it can be generalized to other classes of gates and $X$ measurements easily as will be discussed shortly in Section~\ref{subsubsec:WhichCliffords}.

Second, the correspondence between homology classes and logical actions is generally not one-to-one. As discussed above, a homology class specifies the induced maps
\(
H^i(C)\to H_{-i}(D)
\)
for all degrees \(i\), whereas the logical \(\mathsf{CZ}\) action depends only on the degree-zero map \(H^0(C)\to H_0(D)\), where the qubits live. Thus, distinct homology classes can differ in their action at other degrees while inducing the same action on the logical qubits. We make this redundancy explicit below using the tensor-product description of the chain map complex.

Third, the exhaustiveness clause of the above theorem states that for any pattern of $\mathsf{CZ}$ gates on \textit{logical} qubits, there is a pattern of $\mathsf{CZ}$ gates of the form of Eq.~\eqref{eq:productofCZ} that one can perform on \textit{physical qubits} to implement this action. 
Crucially, this does not imply that this pattern of gates can necessarily be implemented in constant depth i.e. it does not impose that $\phii$ is sparse (cf. Theorem \ref{thm:TransversalGatesFromChainMaps}).
However, since the pattern of gates on physical qubits are associated with homology classes of $[\mathbf{C}^{[2]}]_{\bullet}$, they can be deformed with Clifford stabilizers as in Eq.~\eqref{eq:deformUCZ} to lower their depth. 
In some cases, this leads to transversal implementations of these Clifford gates and we give examples of this in Section~\ref{sec:DiscoveringGates}.
However, we leave the general problem of determining when these high depth gates realize constant depth implementations to forthcoming work~\cite{Sahay2026GoNoGo}.

The chain map complex has the same structural ingredients as the chain complex of an ordinary CSS code: chains, boundaries, and homology classes, each with a direct physical interpretation. This suggests promoting the analogy one step further and viewing the chain map complex itself as the chain complex of an auxiliary code.

\begin{coro}[Chain Map Complex as a qLDPC CSS Code]
\label{corr:auxiliarycode}
Let $C_\bullet$ and $D_\bullet$ be sparse, based chain complexes. Then the chain map complex
$[\mathbf{C}^{[2]}]
\equiv
[C^\bullet,D_\bullet]$
is itself a sparse, based chain complex and hence defines an auxiliary qLDPC CSS code by Theorem \ref{thm:CSScode}. 
\end{coro}

\begin{proof}
The corollary follows immediately from the preceding results. The chosen bases of $C_\bullet$ and $D_\bullet$ induce the elementary-map basis of $[\mathbf{C}^{[2]}]_\bullet$ (analogous to the discussion below Lemma \ref{lem:local_null_homotopies}), and Eq.~\eqref{eqn:CMCBoundaryIndices} shows that its boundary maps are sparse whenever those of $C_\bullet$ and $D_\bullet$ are sparse. The auxiliary-code interpretation then follows from Theorems~\ref{thm:CSScode} and~\ref{thm:Cliffordlogic}.
\end{proof}

Then, to summarize and synthesize the discussion of this section, the chain map complex is a new auxiliary chain complex encoding \emph{operations} on codes $C$ and $D$:  
physical \(\mathsf{CZ}\)-gate patterns
are chains in the chain map complex, Clifford stabilizers are boundaries, and
logical \(\mathsf{CZ}\) gates and logical Pauli measurements are labeled by the resulting homology classes. The auxiliary chain map complex can itself be viewed as a CSS LDPC code. Then, logical Pauli operators of this auxiliary code encode logical Clifford  operations on codes $C$ and $D$.

\subsubsection{Chain Map Complex as a Tensor Product}
\label{subsubsec:chain_maps_as_tensors}

By identifying logical Clifford gates with homology classes of the chain map
complex, the preceding subsection places them in direct analogy with Pauli
logical operators in an auxiliary CSS code. 
To make this perspective useful, however, we need a more physical picture of
the auxiliary code. As presently formulated, the chain map complex is a rather abstract object about which we have little intuition: its chains
are collections of (shifted) maps between
 two chain complexes, which are much less familiar than the chains of the original CSS chain complex.  The goal of this
subsection is to give an equally concrete picture of the chain map complex, in order to be able to think about it more intuitively and concretely. We will do this by identifying the  chain map complex
with a familiar tensor-product complex.

It is helpful to make an analogy to vector spaces.
Chain complexes resemble vector spaces in many ways---not surprising, as they are made from a sequence of vector spaces.
In this sense, we should think of chain maps as matrices relating vectors.
The statement that chain maps can be organized into a chain complex is the analog of the fact that matrices themselves form a vector space.
In fact, the concrete realization of the space of linear maps between finite dimensional vector spaces \(V \to W\) as a vector space is well known: it is isomorphic to \(W \otimes V^*\), where \(V^*\) is the dual vector space of linear functionals on \(V\).
\begin{equation}\label{eqn:VSpaceHomTensorIsom}
    \Hom(V,W) \isom W \otimes V^*.
\end{equation}
Under this isomorphism, a map which takes \(\ket{v} \to \ket{w}\) is sent to a product \(\ket{w}\bra{v}\).
The isomorphism is then simply the statement that linear maps are spanned by products of kets in \(W\) and bras in \(V^*\).

Returning to chain complexes, the space of linear maps \(C^\bullet \to D_\bullet\) is the chain map complex we have already discussed.
The tensor product of chain complexes is also a familiar concept which has appeared extensively in the quantum error correction literature.
It is closely related to the hypergraph product of CSS codes.
\begin{defn}[Tensor Product of Chain Complex] \label{def:TensorProduct}
  Given chain complexes \(C_\bullet\) and \(D_\bullet\), the tensor product \(C_\bullet \otimes D_\bullet\) is a chain complex with components
  \begin{equation}\label{eqn:TensorComponents}
      (C_\bullet\otimes D_\bullet)_j = \bigoplus_{i \in \ZZ} C_{-i} \otimes D_{i+j},
  \end{equation}
  and boundary map defined by its action on pure tensors
  \begin{equation}\label{eqn:TensorBoundary}
      \partial_\otimes(c \otimes d) = (\partial c) \otimes d + c \otimes (\partial d).
  \end{equation}
\end{defn}

For tensor products of cochain complexes \(C^\bullet\) or \(D^\bullet\), one replaces any \(C_{-i}\) or \(D_{i+j}\) with \(C^i\) or \(D^{-i-j}\), respectively.
This is to maintain the ordering of cochain complexes written in Eq.~\eqref{eqn:CochainComplex}, where coboundary maps go from left to right (see Appendix~\ref{app:CMCDetails} for more details).
For example,
\begin{equation}
  (C^\bullet\otimes D_\bullet)_j = \bigoplus_{i \in \ZZ} C^{i} \otimes D_{i+j}.
\end{equation}

The isomorphism~\eqref{eqn:VSpaceHomTensorIsom} motivates the following theorem, which expresses that maps between chain complexes can be \emph{reshaped} into chains in a tensor product.

\begin{theo}[Chain Map Complex as a  Tensor Product]\label{thm:HomTensorIsom} 
    For based chain complexes \(\mcbb{C}\) and \(\mcb{D}\), there is an isomorphism between the chain map complex \([C^{\bullet}, D_{\bullet}]\) and the tensor product of \(\mcb{D}\) and the dual complex of \(\mcbb{C}\):
    \begin{equation}\label{eqn:CMCToTensor}
         [C^{\bullet}, D_{\bullet}] \isom \mcb{D} \otimes \mcb{C}.
    \end{equation}

\end{theo}

Equation~\eqref{eqn:CMCToTensor} is the chain complex analog of Eq.~\eqref{eqn:VSpaceHomTensorIsom}:
the space of maps between chain complexes is isomorphic to the tensor product of the target complex \(D_\bullet\) and the dual of the control complex \(C^\bullet\).
Dualizing swaps chains and cochains, so that if \(C^\bullet\) appears in the chain map complex, \(C_\bullet\) appears in the tensor product.
Unlike vector spaces, there is usually no isomorphism between a chain complex and its dual, \(C^\bullet \not\isom C_\bullet\), so it is important to keep track of dualizations in Eq.~\eqref{eqn:CMCToTensor}.
For instance, a map \(\psi: C^\bullet \to D^\bullet\) encoding a \(\mathsf{CNOT}\) is a chain in \(D^\bullet \otimes C_\bullet\), which can be a very different space to \(D_\bullet \otimes C_\bullet\), which encodes \(\mathsf{CZ}\) gates.
Indeed, whether or not transversal gates can be found in these complexes may depend on the dualizations: later, in higher levels of the hierarchy, we will find \(\mathsf{CCZ}\) gates in a complex like \(C_\bullet \otimes C_\bullet \otimes C_\bullet\), but one expects to never find logical \(\mathsf{CCNOT}\) gates~\cite{JochymOConnor2018} in \(C^\bullet \otimes C_\bullet \otimes C_\bullet\).

Theorem~\ref{thm:HomTensorIsom} gives a concrete handle on the auxiliary CSS code defined by the chain map
complex. 
Rather than working with an abstract complex of maps, we can
view it as a tensor-product CSS code and use the geometry and deformability of its Pauli logicals
to guide the construction of Clifford gates on the original codes $C, D$. 

The proof of Theorem~\ref{thm:HomTensorIsom} is again standard in homological algebra~\cite{CartanEilenberg}. 
We recall it (in greater generality than needed here---including for non-binary chain complexes relevant for non-qubit codes) in Appendix~\ref{app:CMCDetails}. 

The tensor-product description makes two features of Theorem~\ref{thm:Cliffordlogic} particularly transparent: the exhaustiveness of the logical \(\mathsf{CZ}\) construction, and the fact that distinct homology classes can nevertheless have the same logical action. By the K\"unneth formula~\cite{Weibel1994HomologicalBook},
\begin{equation}\label{eqn:CMCKunneth}
H_0([C^\bullet,D_\bullet])
\cong
\bigoplus_i H_{-i}(D)\otimes H_i(C).
\end{equation}
Because qubits live at the degree-$0$ component of codes $C$ and $D$, the logical \(\mathsf{CZ}\) action depends only on the \(i=0\) sector,
\(H_0(D)\otimes H_0(C)\), which is isomorphic to
\(
\Hom(H^0(C),H_0(D)).
\)
Thus, an element of \(H_0(D)\otimes H_0(C)\) specifies precisely the induced map on the logical \(X\) operators of code \(C\).

In particular, if \(C\) and \(D\) encode \(k_C\) and \(k_D\) logical qubits, respectively, this sector has dimension \(k_Ck_D\). Choosing bases
\(\{L_\beta^C\}_{\beta=1}^{k_C}\) and
\(\{L_\alpha^D\}_{\alpha=1}^{k_D}\) of \(H_0(C)\) and \(H_0(D)\), with dual basis
\(\{\mathbf L_\beta^C\}\) of \(H^0(C)\), the tensor
\(L_\alpha^D\otimes L_\beta^C\) corresponds to the map
\(
\mathbf L_\beta^C\mapsto L_\alpha^D
\)
and hence to a logical \(\mathsf{CZ}\) between that pair of logical qubits. Arbitrary elements of this \(k_Ck_D\)-dimensional sector therefore realize arbitrary patterns of the \(k_Ck_D\) possible pairwise logical \(\mathsf{CZ}\) gates, giving a direct explanation of the exhaustiveness statement in Theorem~\ref{thm:Cliffordlogic}. At the same time, the remaining \(i\neq0\) sectors in Eq.~\eqref{eqn:CMCKunneth} do not affect the action on the encoded qubits. Consequently, distinct homology classes that differ only in these other sectors can implement the same logical \(\mathsf{CZ}\) action. Thus \(H_0([C^\bullet,D_\bullet])\) organizes the full homological action of chain maps, while its \(H_0(D)\otimes H_0(C)\) sector captures their action on the encoded qubits.

As a matter of notation and to make the tensorial perspective manifest, it will sometimes be helpful to keep track of indices when manipulating chains, cochains, and chain maps (especially, when discussing non-Clifford and intra-block gates).
A consistent assignment of indices would be to associate chains with vectors with upper indices $b^{\alpha}$, and cochains with vectors with lower indices $\mathbf{a}_{\alpha}$.%
\footnote{The choice of upper/lower indices here is to make contact with other fields (e.g. relativity) where upper indices are associated with vectors and lower indices are associated with their dual vectors. 
It is rather unfortunate that typically chains $b^{\alpha}$ are associated with chain complexes $C_{\bullet}$, which now have a lower bullet, and cochains $\mathbf{a}_{\alpha}$ are associated with $C^{\bullet}$, which now have an upper bullet.
}
Accordingly, a map \(f:C^\bullet\to D_\bullet\) is represented by a tensor
\(f^{\alpha\beta}\), while a map \(g:C^\bullet\to D^\bullet\) is represented
by the tensor \(g_{\alpha}^{\;\;\beta}\). Their actions map cochains to chains and cochains respectively: 
\begin{equation}\label{eq:chainmaptensor}
\begin{aligned}
    b^{\alpha}
    =
    f^{\alpha\beta}\mathbf a_{\beta},
    \qquad
    \mathbf b_{\alpha}
    =
    g_{\alpha}^{\;\;\beta}\mathbf a_{\beta},
    .
\end{aligned}
\end{equation}
Here and throughout, repeated upper and lower indices are summed according to
the Einstein summation convention. 
 Diagrammatically:
\begin{equation}\label{eqn:tensorcontract}
     \begin{tikzpicture}[scale = 1.2, baseline = {([yshift=-.5ex]current bounding box.center)}]
        \draw[color = black] (0, 0) --node[pos = 0.6, left]{\small $\alpha$} (0, 0.7);
         \draw[fill = lightgray] (0,0) circle (0.2);
         \node at (0, 0) {$b$};
     \end{tikzpicture}\  =\  \begin{tikzpicture}[scale = 1.2, baseline = {([yshift=-.5ex]current bounding box.center)}] 
        \draw[color = black] (0, 0) --node[pos=0.5, right]{\small $\beta$} (0.565, 0.565);
        \draw[color = black] (0,0) --node[pos=0.5, left]{\small $\alpha$} (-0.7, 0.7);
        \draw[fill = lightdodgerblue] (0,0) circle (0.3);
        \node at (-0.0,0) {\small $f$};
        \draw[fill = lightgray] (0.7,0.7) circle (0.2);
        \node at (0.7,0.7) {$\mathbf{a}$};
    \end{tikzpicture} \qquad \begin{tikzpicture}[scale = 1.2, baseline = {([yshift=-.5ex]current bounding box.center)}]
        \draw[color = black] (0, 0) --node[pos = 0.6, left]{\small $\alpha$} (0, -0.7);
         \draw[fill = lightgray] (0,0) circle (0.2);
         \node at (0, 0) {$\mathbf{b}$};
     \end{tikzpicture}\  =\  \begin{tikzpicture}[scale = 1.2, baseline = {([yshift=-.5ex]current bounding box.center)}] 
        \draw[color = black] (0, 0) --node[pos=0.7, right]{\small $\beta$} (0, 0.7);
        \draw[color = black] (0,0) --node[pos=0.7, left]{\small $\alpha$} (0, -0.7);
        \draw[fill = lightdodgerblue] (0,0) circle (0.3);
        \node at (-0.0,0) {\small $g$};
        \draw[fill = lightgray] (0,0.85) circle (0.2);
        \node at (0,0.85) {$\mathbf{a}$};
    \end{tikzpicture}
\end{equation}
In Eq.~\eqref{eqn:tensorcontract}, the legs of the tensor diagrams represent chain complexes, and joining legs corresponds to contraction of a cochain with a chain, providing a graphical calculus for chain maps.
In Eq.~\eqref{eqn:tensorcontract}, we draw legs corresponding to upper indices pointing up and legs corresponding to lower indices pointing down.

\subsubsection{Chains of \texorpdfstring{$[\mathbf{C}^{[2]}]_{\bullet}$}{G} and Physical Clifford Gates}
\label{subsubsec:WhichCliffords}

We conclude our discussion on the chain map complex by summarizing how to describe other Clifford gates.

We briefly recall the association of Paulis to (co)chains [Eq.~\eqref{eq:Pauli-operators-chains}], but now using the tensor notation of Eq.~\eqref{eq:chainmaptensor}.
Given a \(0\)-chain \(b \in C_0\) and a \(0\)-cochain \(\mathbf{a} \in C^0\), we write the associated Pauli $Z$ and $X$ operators as
\begin{equation}
\label{eq-Pauliassoc2}
    Z_b = (-1)^{n_\alpha b^\alpha},
    \qquad
    X_{\mathbf{a}} = (-1)^{m^\alpha \mathbf{a}_\alpha}.
\end{equation}
Here, writing the privileged basis as \(\mathsf{B}(C_0)=\{e_\alpha\}\), the coefficients \(b^\alpha\) and \(\mathbf{a}_\alpha\) are the components of \(b\) and \(\mathbf{a}\) in the corresponding chain and dual cochain bases---\(b=b^\alpha e_\alpha\) and \(\mathbf a=\mathbf a_\alpha \mathbf e^\alpha\)---and \(n_\alpha = (1-Z_{e_\alpha})/2\), \(m^\alpha = (1-X_{e_\alpha})/2\).

We can use a very similar notation now to associate physical $\mathsf{CZ}$ operators with the $0$-chains of the chain map complex.
The \(0\)-chain $f_\bullet \in [C^{\bullet}, D_{\bullet}]_0 \simeq (D\otimes C)_0$ can have its \(f_0 : C^0 \to D_0\) component represented by a matrix $f_0^{\alpha\beta}$, where $\alpha$ indexes the basis $\mathsf B(D_0)$ and $\beta$ indexes the basis $\mathsf B(C_0)$.
The circuit \(U_{\mathsf{CZ}}^f\) specified by \(f_0\) [Eq.~\eqref{eq:productofCZ}] can then, analgous to the Pauli case, be written as an exponential of a tensor contraction:
\begin{equation} \label{eq:CZopassoc}
\begin{tikzpicture}[scale=0.6, baseline={([yshift=-.5ex]current bounding box.center)}]
  \node (Cdots1) at (0.75, 0)    {\small $\ $};
  \node (C1)     at (2, 0)    {\small $C^{-1}$};
  \node (C0)     at (4, 0)    {\small $C^0$};
  \node (Cm1)    at (6, 0)    {\small $C^{1}$};
  \node (Cdots2) at (7.25, 0)    {\small $\ $};

  \node (Ddots1) at (0.75, -2.2) {\small $\ $};
  \node (D1)     at (2, -2.2) {\small $D_1$};
  \node (D0)     at (4, -2.2) {\small $D_0$};
  \node (Dm1)    at (6, -2.2) {\small $D_{-1}$};
  \node (Ddots2) at (7.25, -2.2) {\small $\ $};

  \draw[-stealth, dashed] (Cdots1) -- node[above]{\small $\ $} (C1);
  \draw[-stealth, dashed] (C1)     -- node[above]{\small $\exd_0$} (C0);
  \draw[-stealth, dashed] (C0)     -- node[above]{\small $\exd_1$} (Cm1);
  \draw[-stealth, dashed] (Cm1)    -- node[above]{\small $\ $} (Cdots2);

  \draw[-stealth, dashed] (Ddots1) -- node[below]{\small $\ $} (D1);
  \draw[-stealth, dashed] (D1)     -- node[below]{\small $\partial_1$} (D0);
  \draw[-stealth, dashed] (D0)     -- node[below]{\small $\partial_0$} (Dm1);
  \draw[-stealth, dashed] (Dm1)    -- node[below]{\small $\ $} (Ddots2);

  \draw[-stealth] (C1)  -- node[pos=0.4, left]{\small $f_{-1}$} (D1);
  \draw[-stealth] (C0)  -- node[pos=0.4, left]{\small $f_{0}$}  (D0);
  \draw[-stealth] (Cm1) -- node[pos=0.4, left]{\small $f_{1}$}  (Dm1);
\end{tikzpicture}
\hspace{-2mm}
U_{\mathsf{CZ}}^f
=
(-1)^{n_\alpha^D f_0^{\alpha\beta} n_\beta^C}.
\end{equation}
Equation~\eqref{eq:CZopassoc} is equivalent to the expression in Eq.~\eqref{eq:productofCZ}, which follows from the identity $\mathsf{CZ}_{e, e'} = (Z_e)^{n_{e'}} = (Z_{e'})^{n_e} = (-1)^{n_e n_{e'}}$.
Hence \(
    U_{\mathsf{CZ}}^f
    =
    \prod_{e_2\in\mathsf B(C_0)}
    \mathsf{CZ}^{CD}_{e_2,f_0(\mathbf e_2)}
\), where  $ f_0(\mathbf e_2) = \sum_{e_1\in\mathsf B(D_0)} f_0^{e_1e_2}e_1 \in D_0$. 
Note that the tensor indices are written in tuple order
\((D,C)\), while the physical gate is written in control-first order
\((C,D)\). 

Other Clifford gates are expressed as 0-chains of a chain map complex, but with some constituent chain complexes replaced by cochain complexes.
For instance, \(f \in [C^\bullet, D^\bullet]_0\) encodes a \(\mathsf{CNOT}\) circuit, which will be logical if \(f\) is a cycle.
All such chain maps can be expressed in a tensorial notation and converted to circuits by contracting with \(n_\alpha\) (for upper indices \(f_{\cdots}^{\alpha \cdots}\)) or \(m^\alpha\) (for lower indices \(f^{\cdots}_{\alpha \cdots}\)).
Explicitly, the $0$-chains of the following chain map complexes can be associated with the following Clifford gates:
\begin{equation}
\label{eq:CliffordGateDictionary}
\begin{aligned}
    f^{\alpha \beta}
    \in [C^\bullet,D_\bullet]_0
    &\isom (D_\bullet\otimes C_\bullet)_0
    &&\longrightarrow
    (-1)^{n_{\alpha}^{D}f_0^{\alpha \beta}n_{\beta}^{C}},
    \\[0.3em]
    f_{\alpha}^{\;\;\beta}
    \in [C^\bullet,D^\bullet]_0
    &\isom (D^\bullet\otimes C_\bullet)_0
    &&\longrightarrow
    (-1)^{m_D^{\alpha}(f_0)_{\alpha}^{\;\;\beta}n_{\beta}^{C}},
    \\[0.3em]
    f_{\alpha \beta}
    \in [C_\bullet,D^\bullet]_0
    &\isom (D^\bullet\otimes C^\bullet)_0
    &&\longrightarrow
    (-1)^{m_D^{\alpha}(f_0)_{\alpha\beta}m_C^{\beta}},
    \\[0.3em]
    f^{\alpha}_{\;\;\beta}
    \in [C_\bullet,D_\bullet]_0
    &\isom (D_\bullet\otimes C^\bullet)_0
    &&\longrightarrow
    (-1)^{n_{\alpha}^{D}(f_0)^{\alpha}_{\;\;\beta}m_C^{\beta}}.
\end{aligned}
\end{equation}
The first operation is \(\mathsf{CZ}\). The second is
\(\mathsf{CNOT}\) with control in \(C\) and target in \(D\). The third is
\(\mathsf{CZ}\) conjugated by Hadamards on both qubits. The fourth is
\(\mathsf{CNOT}\) with control in \(D\) and target in \(C\).

\subsection{The Chain Map Hierarchy and Non-Clifford Logic}
\label{subsec:ChainMapHeirarchy}

Having identified Clifford operations with the homology classes of an auxiliary chain complex, the generalization to non-Clifford operations---arbitrarily high up the Clifford hierarchy---is straightforward. 
The construction is recursive and, at this point, almost automatic! 

\subsubsection{Chain Map Hierarchy}

The \emph{chain map hierarchy} organizes operations arbitrarily high up the Clifford hierarchy into chain complexes. 
At a conceptual level, the construction mirrors the recursive structure of the Clifford hierarchy~\cite{Gottesman1999Teleportation}. We index the hierarchy so that Pauli operators form the first level. Clifford gates occupy the second level, and map Paulis to Paulis; third-level gates map Paulis to Cliffords; and, more generally, gates at each higher level map Paulis to operators from the level immediately below. 

\emph{Climbing the chain map hierarchy} means recursively iterating the chain map complex construction of Sec.~\ref{subsubsec:chain_map_complex} to ascend the Clifford hierarchy. At a given level $\ell$ of the Clifford hierarchy,  operations from that level are packaged into a chain (map) complex, which can be associated with an auxiliary code; Pauli logicals of that auxiliary code encode logical operators at that level. Maps into that auxiliary code then define a new chain map complex, whose Pauli logicals encode logical operations at the next level, $\ell+1$.

This recursion is already visible in the passage from \(\mathsf{CZ}\) to \(\mathsf{CCZ}\).
At the Clifford level, a logical \(\mathsf{CZ}\) maps a logical \(X\) in the control code $C_\bullet$ into a logical \(Z\) of code $D_\bullet$
The logical \(Z\) is an ordinary Pauli logical of an ordinary code, and logical \(\mathsf{CZ}\)'s are encoded by maps
\(
C^\bullet \to D_\bullet\). After the Clifford construction, however, the logical \(\mathsf{CZ}\) gate  has itself been repackaged as a  cycle in an auxiliary chain map complex \([\mathbf C^{[2]}] = [C^\bullet, D_\bullet]\). 
This means that, at the next level of the Clifford hierarchy, 
\(\mathsf{CZ}\) can play exactly the same role that \(Z\) played one level below! Correspondingly, \([\mathbf C^{[2]}]_\bullet\), which is a chain complex in its own right, can play the same role as $D_\bullet$ to iterate the chain map complex construction  (Definition \ref{def:ChainMapComplex}) to build an auxiliary code at the next level.  For example, a logical \(\mathsf{CCZ}\)  maps a logical \(X\) in a new control code \(B_\bullet\) into a logical \(\mathsf{CZ}\) on codes $C_\bullet$ and $D_\bullet$, so it is encoded by maps \(B^\bullet\to[\mathbf C^{[2]}]_\bullet\), or equivalently by the iterated chain map complex \([\mathbf C^{[3]}] \equiv  [B^\bullet,[\mathbf C^{[2]}]_\bullet] = [B^\bullet, [C^\bullet,D_\bullet]]\).

To define the chain map hierarchy formally, we fix a collection $\mathcal A$ of chain complexes describing codes of interest. We mathematically define the chain map hierarchy before showing how it encodes physical and logical operations  at arbitrary levels of the Clifford hierarchy.

\begin{defn}[Chain Map Hierarchy]\label{def:ChainMapHierarchy}
    Let $\mathcal{A} = \{C_{(i)}\}$ be some set of chain complexes of interest.
    Then, the \textit{chain map hierarchy} of $\mathcal{A}$ is the set of chain complexes
    \begin{equation}
        \mathcal{C}_{\mathcal{A}} = \bigcup_{\ell = 1}^{\infty} \mathcal{C}_{\mathcal{A}}^{[\ell]},
    \end{equation}
    graded by a positive integer $\ell$.
    The set $\mathcal{C}_{\mathcal{A}}^{[\ell]}$ is  referred to as the $\ell$-th level of the chain map hierarchy.
    
    The levels are defined recursively.
    The first level is $\mathcal{C}_{\mathcal{A}}^{[1]} = \mathcal{A}$.  Higher levels are sets of chain map complexes defined for $\ell \geq 2$ as: 
    \begin{equation}
        \mathcal{C}_{\mathcal{A}}^{[\ell]} = \{ [C^{}, D_{}]\, :\,  C^{} \in \mathcal{A} \text{ and } D_{} \in \mathcal{C}^{[\ell - 1]}_{\mathcal{A}} \}.
    \end{equation}
    Here \([C,D]\) denotes the chain map complex of
Definition~\ref{def:ChainMapComplex}.

    An element $ [\mathbf{C}^{[\ell]}] \in \mathcal{C}_\mathcal{A}^{[\ell]}$ in the $\ell$-th level is specified by an ordered length-\(\ell\) tuple
    \begin{equation}\label{eq:ltuple}
        \mathbf{C}^{[\ell]} = (C_{(i_1)}^{}, C_{(i_2)}^{}, \cdots C_{(i_\ell)}), 
    \end{equation}
    with entries $C_{(i_j)} \in \mathcal{A}$. 
    The associated iterated chain map complex is defined by
\(
    [\mathbf C^{[1]}]
    =
    C_{(i_1)}
\)
and, for \(\ell\geq 2\),
\(
    [\mathbf C^{[\ell]}]
    =
    \bigl[
        C_{(i_\ell)},
        [\mathbf C^{[\ell-1]}]
    \bigr].
\)
Each \([\mathbf C^{[\ell]}]\) is itself a chain complex, with  boundary maps
\[
\partial_j^{[\ell]}:
[\mathbf C^{[\ell]}]_j
\longrightarrow
[\mathbf C^{[\ell]}]_{j-1}.
\]
\end{defn}

In the tuple of chain complexes Eq.~\eqref{eq:ltuple}, we suppress the bullets on \(C^\bullet_{(i_1)}\), etc, and hence also whether the complexes are of chains or cochains.
Bullets will be restored where they are important.

We note that \(\mathcal{C}_{\mathcal A}^{[\ell]}\) is not a single chain complex,
but a collection of chain (map) complexes, each encoding operations at the $\ell$th level of the Clifford hierarchy. 
An individual
element is specified by two pieces of data: the ordered length-\(\ell\) tuple
of code complexes involved,
\(
    \mathbf{C}^{[\ell]}
    =
    \bigl(C_{(i_1)},C_{(i_2)},\ldots,C_{(i_\ell)}\bigr),
\)
with entries in \(\mathcal A\), and the choice of chain/cochain orientations
appropriate to the operation being described. The tuple itself records only
the code blocks participating in the operation. Repetitions are allowed, so a
tuple such as \((C,C,C)\) describes an operation involving three copies of the
same code. Since the tuple does not specify whether an entry is used as a
chain complex or a cochain complex, we suppress bullets and dualizations
inside \(\mathbf{C}^{[\ell]}\), restoring them only when writing the 
chain map complex \([\mathbf{C}^{[\ell]}]\) associated with a particular operation.  By construction, each \([\mathbf{C}^{[\ell]}]_\bullet\) is itself a chain complex: every new level is obtained by applying the chain map complex construction of Definition~\ref{def:ChainMapComplex} to a chain complex from the previous level, with boundary maps inherited recursively. Thus, at fixed level
\(\ell\), the hierarchy contains many distinct chain map complexes, indexed by
the participating code blocks and the operation type.

Specializing for now to \(\mathsf{C}^{\ell-1}\mathsf{Z}\)-type operations, it is useful to see how the construction works in the first few levels. 
At level \(\ell=1\), there is no chain map construction yet: an element is just an ordinary code complex \([\mathbf{C}^{[1]}]=C_{(i_1)}\), whose homology labels Pauli \(Z=\mathsf{C}^{0}\mathsf{Z}\) logicals. 
At level \(\ell=2\), we add a new control code and form maps into the ordinary code complex $C_{(i_1)}$, obtaining the auxiliary code complex \([\mathbf{C}^{[2]}]=[C_{(i_2)}^\bullet,[\mathbf{C}^{[1]}]_\bullet]=[C_{(i_2)}^\bullet,C_{(i_1)\bullet}]\), whose homology labels logical \(\mathsf{CZ}=\mathsf{C}^{1}\mathsf{Z}\)-type operations. 
At level \(\ell=3\), we repeat the same construction, now mapping into the auxiliary complex that encoded \(\mathsf{CZ}\) operations: \([\mathbf{C}^{[3]}]=[C_{(i_3)}^\bullet,[\mathbf{C}^{[2]}]_\bullet]=[C_{(i_3)}^\bullet,[C_{(i_2)}^\bullet,C_{(i_1)\bullet}]]\). Its homology labels logical \(\mathsf{CCZ}=\mathsf{C}^{2}\mathsf{Z}\)-type operations.  These examples illustrate the recursive \(\mathsf{C}^{\ell-1}\mathsf{Z}\)-type construction.  Each new level constructs  a new auxiliary chain map complex by adding one new input code as control, and mapping into the auxiliary complex encoding the previous level:
\[
\begin{array}{rcl@{}l}
Z=\mathsf{C}^{0}\mathsf{Z}
&
\leftrightarrow
&
[\mathbf{C}^{[1]}]=
&
\hbox to .57\columnwidth{\hss $C_{(i_1)\bullet}$}
\\[0.45em]
\mathsf{CZ}=\mathsf{C}^{1}\mathsf{Z}
&
\leftrightarrow
&
[\mathbf{C}^{[2]}]=
&
\hbox to .57\columnwidth{\hss $[C_{(i_2)}^\bullet,C_{(i_1)\bullet}]$}
\\[0.45em]
\mathsf{CCZ}=\mathsf{C}^{2}\mathsf{Z}
&
\leftrightarrow
&
[\mathbf{C}^{[3]}]=
&
\hbox to .57\columnwidth{\hss $[C_{(i_3)}^\bullet,[C_{(i_2)}^\bullet,C_{(i_1)\bullet}]]$}
\\[-0.1em]
&
\vdots
&
&
\\[0.45em]
\mathsf{C}^{\ell-1}\mathsf{Z}
&
\leftrightarrow
&
[\mathbf{C}^{[\ell]}]=
&
\hbox to .57\columnwidth{\hss $[C_{(i_\ell)}^\bullet,[C_{(i_{\ell-1})}^\bullet,\cdots[C_{(i_2)}^\bullet,C_{(i_1)\bullet}]]]$}
\end{array}
\]
where brackets are nested to the right. Note that the nesting order of the chain map complex is opposite to the ordering
of the tuple \(\mathbf C^{[\ell]}=(C_{(i_1)},\ldots,C_{(i_\ell)})\). 

\subsubsection{Tensor Product Structure of the Chain Map Hierarchy}

Having defined the chain map hierarchy, 
we will find it useful to extend Sec.~\ref{subsubsec:chain_maps_as_tensors} to provide a  concrete tensor product
interpretation to the iterated chain map complexes populating the hierarchy. We proceed inductively. At the Clifford level,
Theorem~\ref{thm:HomTensorIsom} identifies the ordinary chain map complex with the tensor product chain complex:
\(
    [C_{(i_2)}^\bullet,C_{(i_1)\bullet}]
    \isom
    C_{(i_1)\bullet}\otimes C_{(i_2)\bullet}.
\)
At the next level, we apply the same theorem again, now with the previous
auxiliary complex playing the role of the target:
\(
    [C_{(i_3)}^\bullet,[C_{(i_2)}^\bullet,C_{(i_1)\bullet}]]
    \isom
    [C_{(i_2)}^\bullet,C_{(i_1)\bullet}]\otimes C_{(i_3)\bullet}
    \isom
    C_{(i_1)\bullet}\otimes C_{(i_2)\bullet}\otimes C_{(i_3)\bullet}.
\)
Iterating this inductive structure arrives at the following general result:
\begin{theo}[Iterated Chain Map Complex as a qLDPC Tensor Product Code]
\label{thm:IteratedChainMapComplex}
Let
$
    \mathbf C^{[\ell]}
    =
    \bigl(C_{(i_1)},\ldots,C_{(i_\ell)}\bigr)
$
be an ordered length-\(\ell\) tuple of sparse, based chain complexes, with
\(\ell\) fixed. Then an iterated chain map complex
$[\mathbf C^{[\ell]}] \in \mathcal{C}_\mathcal{A}^{[\ell]}$, which belongs to the \(\ell\)th level of the chain map hierarchy,
is itself a sparse, based chain complex and
hence defines an auxiliary qLDPC CSS code by
Theorem~\ref{thm:CSScode}. Moreover, it is isomorphic to a repeated tensor
product of the constituent chain complexes:
\begin{equation}\label{eqn:MultiLegHomTensIsom}
    [\mathbf C^{[\ell]}]
    \isom
    \bigotimes_{j=1}^{\ell}
     C_{(i_j)}.
\end{equation}
\end{theo}
In Eq.~\eqref{eqn:MultiLegHomTensIsom}, chain and cochain bullets are
suppressed; their orientations are inherited from the chosen iterated chain
map complex \( [\mathbf C^{[\ell]}]_\bullet \). The theorem follows inductively from
Theorem~\ref{thm:HomTensorIsom} and
Corollary~\ref{corr:auxiliarycode}. 

Thus each iteration of the chain map construction appends one new tensor
factor. Continuing the analogy of reshaping matrices into vectors, chains in an \(\ell\)-th-level
complex may equivalently be viewed as \(\ell\)-leg tensors. Because the constituent chain complexes are based, their tensor product carries
the corresponding product basis. 
More generally, once a chain in an iterated chain map complex is identified
with an \(\ell\)-leg tensor, the same object may be depicted either in
chain-map form or in tensor-product form.

Thus a degree-zero chain  \(f\in [\mathbf{C}^{[\ell]}]_0\) may be written as a tensor
\(
    f^{\alpha_1\alpha_2\cdots\alpha_\ell},
\)
where \(\alpha_j\) labels the leg associated with \(C_{(i_j)}\). Viewed as a
multilinear map, \(f\) acts on \(\ell-1\) cochains to produce a chain,
generalizing Eq.~\eqref{eq:chainmaptensor}:
\begin{equation}\label{eq:chainmap_multitensor}
      b^{\alpha_1}
    =
    f^{\alpha_1\alpha_2\cdots\alpha_\ell}
    \mathbf{a}_{\alpha_2}
    \cdots
    \mathbf{a}_{\alpha_\ell}.  
\end{equation}
Other contractions with combinations of cochains \(\mathbf{a}_{\alpha_j}\), or even other tensors \(g_{\alpha_{j} \alpha_k \cdots}\), produce chain or chain-tensor outputs from other combinations of \(C_{(i_l)}\) inputs.

This tensor-product viewpoint suggests a graphical calculus for the hierarchy.
Chains in ordinary chain complexes are one-leg tensors, chains in chain maps are two-leg tensors, and chains in higher level iterated chain map complex encoding 
higher-level operations are higher-leg tensors. In this sense, the chain map
hierarchy is a tensor-network-like language whose legs are chain complexes
rather than ordinary vector spaces.
The graphical calculus of these tensor networks could generalize the traditional quantum circuit formalism to naturally encode the separation between operations on physical qubits and logical qubits.

\subsubsection{Chains of \texorpdfstring{$[\mathbf{C}^{[\ell]}]_{\bullet}$}{G} and Physical Non-Clifford Gates}
\label{subsubsec:WhichNonCliffords}

We now explain the unitary-gate dictionary at level \(\ell\) of the chain map
hierarchy. The dictionary directly generalizes the Clifford case: for a given
chain map complex \([\mathbf C^{[\ell]}]\), \(0\)-chains specify physical gate
patterns, \(0\)-cycles specify code-space-preserving logical gates at the
\(\ell\)th level of the Clifford hierarchy, and \(0\)-boundaries specify
generalized stabilizer circuits with trivial logical action. Accordingly, gates corresponding to cycles in the same homology class of
\(H_0([\mathbf C^{[\ell]}])\) have the same logical action. 
In this subsection, we develop the first part of this dictionary
by associating \(0\)-chains with physical gates; the interpretation of cycles,
boundaries, and homology will be developed in the subsequent subsections. For
notational simplicity, we spell out the construction for the
\(\mathsf C^{\ell-1}\mathsf Z\)-type branch of the hierarchy. Other gate types
are obtained by changing the chain/cochain orientations, as in
Sec.~\ref{subsubsec:WhichCliffords}.

Let \(f_\bullet\in[\mathbf C^{[\ell]}]_0 = [C_{(i_\ell)}^\bullet,[\mathbf C^{[\ell-1]}]_\bullet]_0\) be a \(0\)-chain. By Definition~\ref{def:ChainMapComplex},  \(f_\bullet \) 
is a {collection} of maps, with components
\(
    f_k:C_{(i_\ell)}^k\to[\mathbf C^{[\ell-1]}]_{-k}
\).
Now, since \([\mathbf C^{[\ell-1]}]\) is itself a chain map complex, these components can be recursively unpacked into maps involving all the component codes \(C_{(i_1)},\ldots,C_{(i_\ell)}\). In the tensor-network interpretation, this gives an $\ell$ leg tensor with one leg for each component code. The physical qubits of each code \(C_{(i)}\) live in degree \(0\) so, as before, to associate $f$ with physical operators, we first
restrict to the degree-zero component
\(f_0:C_{(i_\ell)}^0\to[\mathbf C^{[\ell-1]}]_0\).  
For \(\ell>2\), however,
the target is itself a collection of maps,
\(
[\mathbf C^{[\ell-1]}]_0
=
\bigoplus_r
\operatorname{Hom}
(C_{(i_{\ell-1})}^r,[\mathbf C^{[\ell-2]}]_{-r}),
\)
and includes maps involving chains at nonzero degrees \(r\) and \(-r\), in
addition to the degree-zero map relevant to physical qubits. 
The physical gate is thus determined not by all of \(f_0\), but by the component
obtained by recursively selecting degree zero at every level of the nested
chain map complex. We denote this by the \(\ell\)-leg tensor whose indices all run over degree-\(0\) basis elements, i.e. whose \(j\)-th index $\alpha_j$ indexes $\mathsf B((C_{(i_j)})_0)$, the basis of qubits in code \(C_{(i_j)}\):
\[
    f^{\alpha_1\alpha_2\cdots \alpha_\ell}_{0_10_2\cdots0_\ell}
    .
\]
We henceforth suppress the all-degree-zero subscripts when clear from context. Each nonzero component labels an ordered \(\ell\)-tuple of physical qubits,
one from each constituent code. The remaining components of
\(f_\bullet\) enter the cycle and boundary conditions but do not directly
specify gates between physical qubits.
As before, repeated Greek indices will be summed over the corresponding physical-qubit bases in each code, with $\alpha_j$ indexing \(\mathsf B((C_{(i_j)})_0)\), whereas $e_j\in\mathsf B((C_{(i_j)})_0)$ will denote particular fixed basis elements.

We associate this physical tensor to the diagonal gate
\begin{equation}
\label{eq-ClZ}
    U_{\mathsf{C}^{\ell-1}\mathsf Z}^{f}
    =
    \exp\left(
        i\pi\,
        f^{\alpha_1 \alpha_2\cdots \alpha_\ell}
        n_{\alpha_1}
        n_{\alpha_2}
        \cdots
        n_{\alpha_\ell}
    \right),
\end{equation}
where 
\(
n_{\alpha_j}
\equiv
(1-Z_{e_{\alpha_j}})/2
\)
is the projector onto the \(Z=-1\) eigenspace of the qubit  
\(e_{\alpha_j}\) in  code \(C_{(i_j)}\).
Equivalently, this is a product of physical
\(\mathsf{C}^{\ell-1}\mathsf Z\) gates over all tuples
\((e_1,\ldots,e_\ell)\) for which
\(f^{e_1\cdots e_\ell}\neq 0\), which generalizes Eq.~\eqref{eq-productofCZ_f}:
\begin{equation}\label{eq:productofClZ}
\begin{aligned}
    U_{\mathsf C^{\ell-1}\mathsf Z}^{f}
&=
    \prod_{\substack{
        e_j\in\mathsf B((C_{(i_j)})_{0})\\
        j=1,\ldots,\ell
    }}
    \left(
        \mathsf C^{\ell-1}\mathsf Z_{
            e_\ell,\ldots,e_1
        }
    \right)^{f^{e_1e_2\cdots e_\ell}} 
    \\
    &=
    \prod_{\substack{
        e_j\in\mathsf B(C_{(i_j),0})\\
        j=2,\ldots,\ell
    }}
    \mathsf C^{\ell-1}\mathsf Z_{
        e_\ell,\ldots,e_2,\,
        f(\mathbf e_2,\ldots,\mathbf e_\ell)
    }.
\end{aligned}
\end{equation}
Here \(\mathsf C^{\ell-1}\mathsf Z_{e_\ell,\ldots,e_1}\) denotes the physical \(\ell\)-qubit controlled-\(Z\) gate acting on the qubits \(e_j\) of codes \(C_{(i_j)}\), for \(j=1,\ldots,\ell\). As in the Clifford case, the gate labels are written in control-first order, opposite to the tensor-leg order. In the second line we regard \(f\) as a multilinear map (cf. Eq~\eqref{eq:chainmap_multitensor}) whose value on the basis
cochains of the last \(\ell-1\) codes is a \(0\)-chain of the first code 
\begin{equation}\label{eq:fasmap}
    f(\mathbf e_2,\ldots,\mathbf e_\ell)
    =
    \sum_{e_1\in\mathsf B((C_{(i_1)})_{0})}
    f^{e_1e_2\cdots e_\ell}e_1
    \in (C_{(i_1)})_0    
\end{equation}
so that
\(\mathsf C^{\ell-1}\mathsf Z_{e_\ell,\ldots,e_2,
f(\mathbf e_2,\ldots,\mathbf e_\ell)}\)
denotes the product of multi-controlled-\(Z\) gates obtained by using the
fixed qubits \(e_2,\ldots,e_\ell\) as controls and applying \(Z\) on the
qubits in the support of \(f(\mathbf e_2,\ldots,\mathbf e_\ell)\).

The tensor representation also makes transparent when the corresponding physical gate can be implemented in constant depth. In Eq.~\eqref{eq:productofClZ}, each nonzero component \(f^{e_1\cdots e_\ell}\) labels one physical \(\mathsf C^{\ell-1}\mathsf Z\) gate acting on the corresponding \(\ell\)-tuple of physical qubits. To isolate the gates acting on a \emph{particular} qubit
\(e_j\) of the \(j\)th constituent code, we denote by \(f(\mathbf e_j)\) the \((\ell-1)\)-leg tensor obtained by contracting the \(j\)th leg of \(f\) with the basis cochain \(\mathbf e_j\in\mathsf B(C_{(i_j)}^0)\). In components,
\begin{equation}
\bigl(f(\mathbf e_j)\bigr)^{
\alpha_1\cdots\alpha_{j-1}\alpha_{j+1}\cdots\alpha_\ell}
=
f^{\alpha_1\cdots e_j\cdots\alpha_\ell}.
\end{equation}
Thus, the number of nonzero entries in $f(\mathbf e_j)$ is precisely the number of physical gates acting on qubit $e_j$. This motivates the following notion of sparsity.

\begin{defn}[Sparse Tensor]
\label{def:SparseTensor}
An \(\ell\)-leg tensor
$
    T^{\alpha_1\alpha_2\cdots\alpha_\ell}
$
is \emph{sparse} if, after fixing any one index \(\alpha_j\), only \(O(1)\)
choices of the remaining indices yield a nonzero component
\(T^{\alpha_1\alpha_2\cdots\alpha_\ell}\).
\end{defn}

For \(\ell=2\), this is the usual finite row- and column-weight condition which defines a sparse matrix. 

It is useful to quantitatively relate the sparsity of \(f\) to the circuit
depth of \(U_{\mathsf C^{\ell-1}\mathsf Z}^{f}\).
For a \(0\)-chain
\(f_\bullet\in[\mathbf C^{[\ell]}]_0\), we define its \emph{sparsity degree} of its all-zero component
as
\begin{equation}
\label{eqn:SparsityDegree}
\Delta(f)
=
\max_{1\leq j\leq \ell}\
\max_{e_j\in\mathsf B((C_{(i_j)})_0)}\
\bigl|f(\mathbf e_j)\bigr|,
\end{equation}
where \(|\cdot|\) denotes Hamming weight in the product basis. Thus, \(f\) is
sparse precisely when \(\Delta(f)=O(1)\). By definition, \(\Delta(f)\) is  the maximum
number of physical \(\mathsf C^{\ell-1}\mathsf Z\) gates in
\(U_{\mathsf C^{\ell-1}\mathsf Z}^{f}\) that act on any one physical qubit.
Consequently, if \({\rm depth}(f)\) denotes the minimum circuit depth required to
implement \(U_{\mathsf C^{\ell-1}\mathsf Z}^{f}\), then
\begin{equation}
    \Delta(f)\
    \leq \ {\rm depth}(f)\
    \leq  \ \ell\bigl(\Delta(f)-1\bigr)+1.
\end{equation}
The lower bound follows because a given qubit can participate in at most one
gate in each circuit layer. For the upper bound, each
\(\mathsf C^{\ell-1}\mathsf Z\) gate overlaps with at most
\(\ell(\Delta(f)-1)\) other gates, so the gates can be partitioned into 
\(\ell(\Delta(f)-1)+1\) sets of mutually disjoint gates. Thus, if \(f\) is
sparse, then \(\Delta(f)=O(1)\) and
\(U_{\mathsf C^{\ell-1}\mathsf Z}^{f}\) has a constant-depth, i.e.
transversal, implementation.

We conclude our discussion of physical gates by describing their action on
Pauli \(X\) operators. This action makes manifest the recursive structure of
the Clifford hierarchy: a  $\mathsf C^{\ell-1}\mathsf Z$ gate at level \(\ell\) maps a Pauli \(X\) operator in any of the constituent codes into a gate one level lower in the hierarchy, namely a $\mathsf C^{\ell-2}\mathsf Z$ gate on the remaining $(\ell-1)$ codes, generalizing the familiar Clifford relation for $U_{\mathsf C\mathsf Z}$ in Eq.~\eqref{eq:phiactsona}:
\begin{equation}\label{eq:CMHrecursiveaction}
\left[
U_{\mathsf C^{\ell-1}\mathsf Z}^{f},
X_{\mathbf a}^{C_{(i_j)}}
\right]_{\rm grp}
=
U_{\mathsf C^{\ell-2}\mathsf Z}^{f(\mathbf a)}.
\end{equation}
Here \(\mathbf a\in C_{(i_j)}^0\) is an arbitrary \(0\)-cochain of the
\(j\)th constituent code, specifying a Pauli \(X\) operator. Again viewing \(f\) as
a multilinear map with the \(j\)th tensor leg as an input, contraction with
\(\mathbf a\) produces an \((\ell-1)\)-leg tensor \(f(\mathbf a)\).
Then, the right-hand side of Eq.~\eqref{eq:CMHrecursiveaction} is the physical
\(\mathsf C^{\ell-2}\mathsf Z\) gate specified by \(f(\mathbf a)\) on the
remaining \(\ell-1\) codes via Eq.~\eqref{eq:productofClZ}.

This tensor contraction has a natural interpretation in the nested chain-map
description. To understand this, first consider the distinguished outer leg,
\(j=\ell\). Then, a \(0\)-chain \(f\in[\mathbf C^{[\ell]}]_0 = [C_{(i_\ell)}^\bullet,[\mathbf C^{[\ell-1]}]_\bullet]_0\) has a degree-zero component which is a map
\(
f_0:C_{(i_\ell)}^0\longrightarrow[\mathbf C^{[\ell-1]}]_0, 
\) so \(f(\mathbf a)\) is naturally a \(0\)-chain in the level-\((\ell-1)\)
auxiliary complex. More generally,  the tensor-product isomorphism of
Theorem~\ref{thm:IteratedChainMapComplex} allows the tensor legs to be reshaped so that any  code
\(C_{(i_j)}\) can be treated as the outer input:
\begin{equation}\label{eq-Clreshape}
[\mathbf C^{[\ell]}]
\isom
[C_{(i_j)}^\bullet,
[\mathbf C_{\setminus j}^{[\ell-1]}]_\bullet],
\end{equation}
where
\(
\mathbf C_{\setminus j}^{[\ell-1]}
=
\bigl(
C_{(i_1)},\ldots,C_{(i_{j-1})},
C_{(i_{j+1})},\ldots,C_{(i_\ell)}
\bigr)
\)
is the ordered length $(\ell-1)$ tuple obtained from \(\mathbf C^{[\ell]}\) by removing
\(C_{(i_j)}\). Accordingly,  the same map interpretation applies for arbitrary \(j\), 
so that \(
f(\mathbf a)\in[\mathbf C_{\setminus j}^{[\ell-1]}]_0,
\)
in agreement with its interpretation above as the \((\ell-1)\)-leg tensor
obtained by contracting the \(j\)th leg of \(f\) with \(\mathbf a\).

To derive Eq.~\eqref{eq:CMHrecursiveaction}, first take
\(\mathbf a=\mathbf e_j\) to be a basis cochain selecting a single physical
qubit \(e_j\). Every physical \(\mathsf C^{\ell-1}\mathsf Z\) gate in
Eq.~\eqref{eq:productofClZ} that does not act on \(e_j\) commutes with
\(X_{\mathbf e_j}\) and drops out of the group commutator, while each gate
that does act on \(e_j\) leaves behind a
\(\mathsf C^{\ell-2}\mathsf Z\) gate on its remaining \(\ell-1\) qubits. %
These surviving gates are precisely those encoded by
\(f(\mathbf e_j)\). Taking the product over the support of a general
\(0\)-cochain \(\mathbf a\) gives \(f(\mathbf a)\). 

Finally, we note that for a general \(0\)-chain \(f\), the lower-level operation
\(U_{\mathsf C^{\ell-2}\mathsf Z}^{f(\mathbf a)}\) need not have any special
relation to the code space. We now show that when \(f\) is a \(0\)-cycle, denoted $\phii$,
the chain-map condition ensures that Eq.~\eqref{eq:CMHrecursiveaction}
preserves stabilizers and descends to a well-defined action on logical
operators.

\subsubsection{Chain Maps, Homology, and Logical Non-Clifford Unitary Gates}\label{subsec:nonclifford_logic}

We have now understood how \(0\)-chains of \([\mathbf C^{[\ell]}]\) specify physical \(\mathsf C^{\ell-1}\mathsf Z\)-gate patterns. We next consider the role of $0$-cycles.
As in the Clifford case, the $0$-cycle condition promotes a physical gate pattern to a chain map: the resulting gate preserves the joint code space, descends to a well-defined action on homology, and hence defines a logical operation. We now develop this correspondence. The role of \(0\)-boundaries as generalized stabilizers and deformations of these logical gates is discussed  in the next subsection. 

Once again, we focus our discussion on \(\mathsf C^{\ell-1}\mathsf Z\) type gates, but generalizations to \(\mathsf C^{\ell-1}\mathsf {NOT}\) or to any gates related to Hadamard transformations on any of the affected qubits is straightforward, and proceeds by replacing cochain complexes \(C^\bullet_{(i_j)}\) with chain complexes \(C_\bullet^{(i_j)}\) on the relevant codes, as discussed for the Clifford case in \ref{subsubsec:WhichCliffords}.

We first establish the key algebraic property of a \(0\)-cycle
\(\phii\in\ker(\partial^{[\ell]}_0)\subseteq[\mathbf C^{[\ell]}]_0\):
under its action, \(0\)-cocycles of any constituent code are mapped to \(0\)-cycles one level lower in the hierarchy, while \(0\)-coboundaries are mapped to \(0\)-boundaries. 
This directly follows from iterating the chain-map property from the Clifford case.
In particular, it implies that changing the input cocycle by a coboundary changes its image
only by a boundary, so that  \(\phii\) induces a well-defined map from the
cohomology of a constituent code to the homology of the level-\((\ell-1)\)
auxiliary complex.

To make this precise, fix a constituent code \(C_{(i_j)}\) and use the
reshaped presentation of Eq.~\eqref{eq-Clreshape}. In this presentation,
\(\phii\) is a chain map from \(C_{(i_j)}^\bullet\) to
\([\mathbf C_{\setminus j}^{[\ell-1]}]_\bullet\), and therefore induces
well-defined maps on (co)homology:
\begin{align}
(\phii_*)_k:
H^k(C_{(i_j)})
&\longrightarrow
H_{-k}([\mathbf C_{\setminus j}^{[\ell-1]}]),
\nonumber\\
[\mathbf c]
&\longmapsto
[\phii_k(\mathbf c)],
\label{eq:CMHInducedMap}
\end{align}
where \(\mathbf c\in\ker(\exd)\) is a \(k\)-cocycle
representative of the cohomology class
\([\mathbf c]\in H^k(C_{(i_j)})\).

To establish that the maps in Eq.~\eqref{eq:CMHInducedMap} are well defined,  we simply note that Eq.~\eqref{eqn:CMCBoundaryIndices} implies that
the \(0\)-cycle condition \(\partial^{[\ell]}\phii=0\) is precisely the
chain-map condition from \(C_{(i_j)}^\bullet\) to
\([\mathbf C_{\setminus j}^{[\ell-1]}]_\bullet\): 
\begin{equation}
\label{eq:CMHchainmapcondition}
\partial^{[\ell-1]}_{k}\phii_{-k}
=
\phii_{-k+1}\exd^{C_{(i_j)}}_{-k+1}
\end{equation}
for all degrees $k$. Here 
\(\partial^{[\ell-1]}\) denotes the boundary map of the
level-\((\ell-1)\) complex, \([\mathbf C_{\setminus j}^{[\ell-1]}]_\bullet\), while
\(\exd^{C_{(i_j)}}\) is the coboundary map of the  code \(C_{(i_j)}^\bullet\).
Thus, \(\phii\) maps cycles (boundaries) to cycles (boundaries) and induces a well-defined map on (co)homology.

Just as the physical all-degree-zero component of the map \(f \in [\mathbf{C}^{[\ell]}]_0\) was represented by the tensor
\(f^{\alpha_1\cdots\alpha_\ell}\), with each index $\alpha_j$ running over the
physical-qubit basis, the induced all-degree-zero component of the map on homology admits an
analogous tensor representation, $\phii_*^{\alpha_1\cdots\alpha_\ell}$, with indices now running over logical-qubit
bases. We denote the logical basis of \(H_0(C_{(i_j)})\) by
\(\{L_{\alpha_j}\}\), with dual basis
\(\{\mathbf L_{\alpha_j}\}\) of \(H^0(C_{(i_j)})\). In logical-tensor expressions, repeated Greek indices are summed over the
chosen logical bases, with \(\alpha_j\) indexing
\(\mathsf B(H_0(C_{(i_j)}))\), whereas
\(L_j\in\mathsf B(H_0(C_{(i_j)}))\) denotes a particular fixed logical
basis element of code $C_{(i_j)}$. As in the case of $f$, the tensor-product structure allows us to unpack this induced
map into a multilinear map involving all \(\ell\) constituent codes. For example,
\begin{equation}
\label{eq:phiasmap}
\phii_*(\mathbf L_{2},\ldots,\mathbf L_{\ell})
=
\sum_{L_1\in \mathsf B(H_0(C_{(i_1)}))}
\phii_*^{L_1L_2\cdots L_\ell}
L_{1}
\in H_0(C_{(i_1)}),
\end{equation}
 treats  \(\phii_*\) as a multilinear map which, when evaluated on fixed dual basis classes of the last \(\ell-1\) codes, produces a homology class of the first code,  directly paralleling the corresponding physical multilinear map in
Eq.~\eqref{eq:fasmap}. 

We now translate this homological action into the corresponding logical
action on the constituent codes.

\begin{theo}[ Non-Clifford Unitary Logic and Homology in the Chain Map Hierarchy]
\label{thm:ChainMapHierarchyGates}
Let
\(
\mathbf C^{[\ell]}
=
\bigl(C_{(i_1)},\ldots,C_{(i_\ell)}\bigr)
\)
be an ordered length-\(\ell \) tuple of sparse, based chain complexes defining a chain map complex $[\mathbf{C}^{[\ell]}]$ at the $\ell$-th level of the chain map hierarchy, with 
\(\mathsf C^{\ell-1}\mathsf Z\)-type orientation. Then, the following statements are true:

\begin{enumerate}

\item[(1)] \textit{Homology and Logical Non-Clifford Gates:}
Each homology class
\(
[\phii]\in H_0([\mathbf C^{[\ell]}])
\)
determines a unitary logical \(\mathsf C^{\ell-1}\mathsf Z\)-type action on the
constituent codes in $\mathbf C^{[\ell]}$. Any \(0\)-cycle representative
\(\phii\) of the class specifies a unitary  
 circuit \(U_{\mathsf C^{\ell-1}\mathsf Z}^{\phii}\) of the form in
Eq.~\eqref{eq:productofClZ}, which implements the logical action
\begin{equation}
\label{eq:ChainMapHierarchyLogicalAction}
\overline U_{\mathsf C^{\ell-1}\mathsf Z}^{\phii}
=
\prod_{\substack{
L_j\in\mathsf B(H_0(C_{(i_j)}))\\
j=1,\ldots,\ell
}}
\left(
\overline{\mathsf C^{\ell-1}\mathsf Z}_{
L_\ell,\ldots,L_1}
\right)^{
\phii_*^{L_1L_2\cdots L_\ell}
}.
\end{equation}
where $\phii_*$ is the all-degree-zero component of the induced map on cohomology, Eq.~\eqref{eq:phiasmap}. 

As a consequence, if two physical  $\mathsf{C^{\ell-1}Z}$ circuits of the form in Eq.~\eqref{eq:productofClZ} act differently on logical qubits, they are labeled by different homology classes of $H_0([\mathbf{C}^{[\ell]}]_{\bullet})$.

If, additionally, \(\phii\) is a sparse tensor in the sense of
Definition~\ref{def:SparseTensor}, then
\(U_{\mathsf C^{\ell-1}\mathsf Z}^{\phii}\) is transversal, i.e. it can be
implemented in constant depth.

\item[(2)] \textit{Exhaustiveness:}
For any logical \(\mathsf C^{\ell-1}\mathsf Z\)-type action on the
constituent logical qubits, there is a homology class
\(
[\phii]\in H_0([\mathbf C^{[\ell]}])
\)
with a physical representative of the form
$U_{\mathsf{C}^{\ell-1}\mathsf{Z}}^{\phii} $  [Eq.~\eqref{eq:productofClZ}] that implements that action.

\end{enumerate}
\end{theo}

\begin{proof}
The proof proceeds by induction on the level of the hierarchy.  The Clifford case
\(\ell=2\) was established by
Theorems~\ref{thm:TransversalGatesFromChainMaps}
and~\ref{thm:Cliffordlogic}, and forms the base of the induction.  In particular, note that for $\ell=2$, $[\mathbf{C}^{[2]}]$ is just the usual chain map complex and the gate in Eq.~\eqref{eq:productofClZ} is the usual $\mathsf{CZ}$ gate of Eq.~\eqref{eq-productofCZ_f}.

Suppose the theorem holds at level
\(\ell-1 \geq 1\), and consider the $0$-cycle
\(\phii\in\ker(\partial^{[\ell]}_0)\subseteq[\mathbf C^{[\ell]}]_0\).  We first verify that the circuit
\(U_{\mathsf C^{\ell-1}\mathsf Z}^{\phii}\) is a logical operator by demonstrating that it leaves the  the joint code space invariant, and then we show that it implements the logical action specified in Eq.~\eqref{eq:ChainMapHierarchyLogicalAction}.

To verify that code space is preserved, we evaluate the action of \(U_{\mathsf C^{\ell-1}\mathsf Z}^{\phii}\) on the stabilizers of each constituent code. 
Since the \(U_{\mathsf C^{\ell-1}\mathsf Z}^{\phii}\) circuit is composed of diagonal gates,  it trivially commutes with all \(Z\)-stabilizers. Next, consider
an \(X\)-stabilizer \(X_{\exd\mathbf v}^{C_{(i_j)}}\) of any constituent
code $C_{(i_j)}$, where $\mathbf{v} \in \mathsf B((C_{i_j})_{-1})$. By the recursive physical action of
Eq.~\eqref{eq:CMHrecursiveaction},
\begin{equation}
\left[
U_{\mathsf C^{\ell-1}\mathsf Z}^{\phii},
X_{\exd\mathbf v}^{C_{(i_j)}}
\right]_{\rm grp}
=
U_{\mathsf C^{\ell-2}\mathsf Z}^{\phii(\exd\mathbf v)}.
\end{equation}
Now, because \(\phii\) is a \(0\)-cycle and $\exd\mathbf v$ is a boundary,
\(\phii(\exd\mathbf v)\) is a \(0\)-boundary in the level-\((\ell-1)\)
complex \([\mathbf C_{\setminus j}^{[\ell-1]}]_\bullet\).  Since every boundary is a cycle, by the inductive hypothesis, the  operation $U^{\phii(\exd\mathbf v)}_{\mathsf{C^{\ell-2}Z}}$ implements a logical action, and hence preserves the code space.
Since $\phii(\exd\mathbf v)$ is a boundary, which represents the trivial class in
\(H_0([\mathbf C_{\setminus j}^{[\ell-1]}])\), this logical action is trivial.
Therefore, $U^{\phii}_{\mathsf{C^{\ell} Z}}$ preserves the code space because every $X$ stabilizer is mapped to itself times an an operator that fixes the code space, while every $Z$ stabilizer is left invariant---thus, the simultaneous \(+1\)-eigenspace of the stabilizers is unaltered by the gate. 

We now determine the logical action, proceeding by induction on \(\ell\).
Let \(L_j\in\mathsf B(H_0(C_{(i_j)}))\) be a fixed basis logical of the
\(j\)th constituent code, let
\(\mathbf L_j\in H^0(C_{(i_j)})\) be its dual cohomology class, and choose a cocycle representative
\(\mathbf c_j\in C_{(i_j)}^0\) such that
\(
[\mathbf c_j]=\mathbf L_j
\). By the recursive physical action
Eq.~\eqref{eq:CMHrecursiveaction},
\begin{equation}
\left[
U_{\mathsf C^{\ell-1}\mathsf Z}^{\phii},
X_{\mathbf c_j}^{C_{(i_j)}}
\right]_{\rm grp}
=
U_{\mathsf C^{\ell-2}\mathsf Z}^{\phii(\mathbf c_j)}.
\label{eq:CMHlogicalactioninduction}
\end{equation}
Since \(\phii\) is a \(0\)-cycle and $\mathbf{c}_j$ is a cocycle, \(\phii(\mathbf c_j)\) is a \(0\)-cycle
in the level-\((\ell-1)\) complex
\([\mathbf C_{\setminus j}^{[\ell-1]}]_\bullet\). Its tensor representation in the logical basis is obtained from \(\phii_*^{\alpha_1\cdots\alpha_\ell}\)
by fixing the \(j\)th leg to the chosen logical \(L_j\):
\begin{equation}
\bigl(\phii_*(\mathbf L_j)\bigr)^{
\alpha_1\cdots\alpha_{j-1}\alpha_{j+1}\cdots\alpha_\ell}
=
\phii_*^{
\alpha_1\cdots\alpha_{j-1}L_j\alpha_{j+1}\cdots\alpha_\ell}.
\label{eq:CMHlogicaltensorslice}
\end{equation}
By the inductive hypothesis and Eq.~\eqref{eq:ChainMapHierarchyLogicalAction}, the right-hand side of
Eq.~\eqref{eq:CMHlogicalactioninduction} therefore has logical action
\begin{equation}
\prod_{\substack{
L_r\in\mathsf B(H_0(C_{(i_r)}))\\
r\neq j
}}
\left(
\overline{\mathsf C^{\ell-2}\mathsf Z}_{
L_\ell,\ldots,L_{j+1},L_{j-1}\ldots,L_1}
\right)^{
\phii_*^{L_1\cdots L_j\cdots L_\ell}
}.
\label{eq:CMHlogicalactionslice}
\end{equation}

Now consider the claimed logical gate in
Eq.~\eqref{eq:ChainMapHierarchyLogicalAction}. Its group commutator with
the basis logical \(X_{\mathbf L_j}\) receives a contribution only from
those factors whose \(j\)th logical label is the fixed \(L_j\). Each such
\(\overline{\mathsf C^{\ell-1}\mathsf Z}\) factor leaves behind a
\(\overline{\mathsf C^{\ell-2}\mathsf Z}\) gate on the remaining
\(\ell-1\) logical qubits. Hence its group commutator is exactly
Eq.~\eqref{eq:CMHlogicalactionslice}.

Thus the physical gate and the logical gate claimed in
Eq.~\eqref{eq:ChainMapHierarchyLogicalAction} have the same action on every
basis logical \(X\) operator of every constituent code. Both commute with all
logical \(Z\) operators, so their actions agree on the logical Pauli
generators. This proves Eq.~\eqref{eq:ChainMapHierarchyLogicalAction}.

Finally, exhaustiveness follows directly from the tensor-product description,
in the same way as in the Clifford case discussed below
Theorem~\ref{thm:HomTensorIsom}. By
Theorem~\ref{thm:IteratedChainMapComplex} and the K\"unneth formula,
\begin{equation}\label{eqn:CMHKunneth}
H_0([\mathbf C^{[\ell]}])
\cong
\bigoplus_{k_1+\cdots+k_\ell=0}
\bigotimes_{j=1}^{\ell}
H_{k_j}(C_{(i_j)}).
\end{equation}
In particular, the all-degree-zero sector
\(\bigotimes_{j=1}^{\ell} H_0(C_{(i_j)})
\) contains a basis tensor
\(L_1\otimes\cdots\otimes L_\ell\) for every choice of basis logical in each constituent code,
\(L_j\in\mathsf B(H_0(C_{(i_j)}))\). By
Eq.~\eqref{eq:ChainMapHierarchyLogicalAction}, this basis tensor corresponds
to the logical
\(\mathsf C^{\ell-1}\mathsf Z_{L_\ell,\ldots,L_1}\) gate. Arbitrary elements
of the all-degree-zero sector therefore realize arbitrary products of such
logical gates, proving exhaustiveness.

As in the Clifford case, the full homology
\(H_0([\mathbf C^{[\ell]}])\) generally contains additional sectors with
some \(k_j\neq0\). These sectors do not affect the action on the encoded
qubits, which is determined entirely by the all-degree-zero logical tensor. Consequently, distinct homology classes may
implement the same logical gate.
\end{proof}

Theorem~\ref{thm:ChainMapHierarchyGates} guarantees that a \(0\)-cycle
\(\phii\) defines a well-defined logical gate, but it does not guarantee that
the corresponding physical gate is transversal: this additionally requires a
sparse representative of the homology class. In general, such representatives
need not exist. For example, the Bravyi--K\"onig theorem forbids transversal
realizations of gates arbitrarily high in the Clifford hierarchy for
Euclidean topological codes~\cite{Bravyi2013GatesforLocalStabilizer}, and no error-detecting stabilizer code admits a
transversal \(\mathsf{CCNOT}\) (Toffoli) gate between three identical code blocks~\cite{JochymOConnor2018}. Thus, a
central practical question is whether a given logical gate admits a sparse
representative within its homology class. The freedom to search for such
representatives comes from adding \(0\)-boundaries, which leave the logical
action unchanged while deforming the physical gate pattern. We turn to these
higher-level stabilizer deformations next.

\subsubsection{Non-Clifford Stabilizers and Boundary Deformations}\label{subsubsec:NonCliffordStab}

We now turn to the \(0\)-boundaries of the chain map hierarchy. These play
the same role for non-Clifford gates at higher levels of the Clifford hierarchy that Pauli stabilizers and Clifford
stabilizers played at the first two levels: they specify physical gate
patterns with trivial logical action, and hence generate deformations of a
logical gate within a fixed homology class.

In analogy to Clifford stabilizers, we consider null-homotopic chain maps $\theta$, i.e. \(0\)-boundaries of the level-\(\ell\) chain map complex:
\begin{equation}
\label{eq:CMHNullHomotopy}
\theta
=
\partial^{[\ell]}_1 Q \in [\mathbf C^{[\ell]}]_0,
\qquad
Q\in[\mathbf C^{[\ell]}]_1, 
\end{equation}
directly generalizing Definition~\ref{def:NullChainMapChainHomotopy}.
Since \(\theta\in[\mathbf C^{[\ell]}]_0\), its all-degree-zero component tensor
specifies a physical gate
\(U_{\mathsf C^{\ell-1}\mathsf Z}^{\theta}\) through
Eq.~\eqref{eq:productofClZ}. Moreover, every boundary is a cycle, so this
gate preserves the code space and has a well-defined logical action by
Theorem~\ref{thm:ChainMapHierarchyGates}. We will show that this logical
action is trivial.

To see this, we unpack the null-homotopy relation using the nested
form
\(
[\mathbf C^{[\ell]}]
=
[C_{(i_\ell)}^\bullet,[\mathbf C^{[\ell-1]}]_\bullet].
\)
In this presentation, \(Q\) is a collection of \(1\)-shifted maps, and its
boundary \(\theta=\partial^{[\ell]}Q\) is the corresponding collection of
unshifted maps. Using Eq.~\eqref{eq:nullTheta}, their components are
\begin{align}
Q_k &: C_{(i_\ell)}^k
\longrightarrow
[\mathbf C^{[\ell-1]}]_{1-k},
\label{eq:CMHQComponents}
\\
\theta_k
&=
\partial^{[\ell-1]}_{1-k}Q_k
+
Q_{k+1}\exd_{k+1},
\label{eq:CMHThetaComponents}
\end{align}
where \(\partial^{[\ell-1]}\) denotes the boundary map of the level-\((\ell-1)\)
auxiliary complex \([\mathbf C^{[\ell-1]}]_\bullet\), while \(\exd\) denotes
the coboundary map of the outer constituent code \(C_{(i_\ell)}^\bullet\).

Then, the fact that \(\theta\) is a null-homotopy implies that it sends any cocycle to a boundary (i.e. to the trivial homology class).
Indeed, evaluating
Eq.~\eqref{eq:CMHThetaComponents} on a 0-cocycle \(\mathbf{c} \in C^0_{(i_\ell)}\) gives
\begin{equation}
\label{eq:CMHThetaZero}
\theta_0(\mathbf{c})
=
\partial^{[\ell-1]}_1 Q_0(\mathbf{c})
+
Q_1(\exd_1\mathbf{c}) = \partial^{[\ell-1]}_1 Q_0(\mathbf{c}),
\end{equation}
which is a \(0\)-boundary in the level-\((\ell-1)\) auxiliary complex, just as in the Clifford \(\ell=2\) case.
Thus \(\theta\) induces the trivial map on degree-zero (co)homology.
Equivalently, its logical tensor vanishes, \(\theta_*=0\), so
Theorem~\ref{thm:ChainMapHierarchyGates} implies that
\(U_{\mathsf C^{\ell-1}\mathsf Z}^{\theta}\) acts trivially on the logical
qubits.  Consequently, if \(\phii\) is any
\(0\)-cycle, then \(\phii\) and
\(\phii+\theta\) lie in the same homology class and hence
implement the same logical action:
 We therefore call the gates \(U_{\mathsf C^{\ell-1}\mathsf Z}^{\theta}\)
\emph{non-Clifford stabilizers}, which are the higher-level analog of stabilizer
deformations.

We now derive the physical form of the elementary non-Clifford stabilizer
generators. Recall that the elementary Clifford stabilizers of
Eq.~\eqref{eq:elementaryCliffordStabilizers} include gates of the form
\(
\mathsf{CZ}^{CD}_{e_C,\partial_1 p_D}
\),
which couples a fixed physical qubit \(e_C\) of one code to
support of an elementary Pauli stabilizer \(\partial_1p_D\) of the other. At
higher levels, the same structure persists: an elementary generator fixes
one physical qubit in all but one constituent code, while in the remaining
code the gate ranges over the support of an elementary stabilizer
\(\partial_1p_j\). We now derive this form from the tensor-product structure
of \([\mathbf C^{[\ell]}]\).

Since
\(Q\in[\mathbf C^{[\ell]}]_1\), Theorem~\ref{thm:IteratedChainMapComplex}
identifies the space in which \(Q\) lives as
\begin{equation}
\label{eq:CMHdegreeone}
[\mathbf C^{[\ell]}]_1
\cong
\bigoplus_{k_1+k_2+\cdots+k_\ell=1}\quad
\bigotimes_{j=1}^{\ell}(C_{(i_j)})_{k_j}.
\end{equation}
We use the elementary product basis for this space, induced by
the bases of the constituent chain complexes. Its elements have the
form
\(
c_{k_1}\otimes\cdots\otimes c_{k_\ell},
\)
where \(c_{k_j}\in\mathsf B((C_{(i_j)})_{k_j})\) is a basis \(k_j\)-chain
of the \(j\)th  code.

To obtain the null homotopic chain map
\(\theta=\partial^{[\ell]}Q\), we  apply the boundary to these elementary
tensor components of \(Q\). Under the tensor-product isomorphism,
\(\partial^{[\ell]}\) is simply the usual tensor-product boundary of
Eq.~\eqref{eqn:TensorBoundary}:
\begin{equation}
\label{eq:CMHTensorBoundary}
\partial^{[\ell]}
\left(
c_{k_1}\otimes\cdots\otimes c_{k_\ell}
\right)
=
\sum_{j=1}^{\ell}
c_{k_1}\otimes\cdots\otimes
(\partial c_{k_j})
\otimes\cdots\otimes c_{k_\ell}.
\end{equation}
Thus each term applies the ordinary boundary map to one constituent chain,
lowering the degree of that tensor leg by one while leaving all other legs
unchanged.

The physical gate
\(U_{\mathsf C^{\ell-1}\mathsf Z}^{\theta}\) depends only on the
all-degree-zero component of
\(\theta=\partial^{[\ell]}Q\). Since each term in
Eq.~\eqref{eq:CMHTensorBoundary} lowers exactly one degree \(k_j\) by one,
a component of \(Q\) can contribute to this physical sector only when
\(k_j=1\) for exactly one constituent code $C_{i_j}$, and \(k_r=0\) for every \(r\neq j\).
Thus the relevant elementary tensor components of \(Q\) have exactly one
degree-\(1\) factor and \(\ell-1\) degree-\(0\) factors.

Using the elementary product basis, let \(e_j\in\mathsf B((C_{(i_j)})_0)\) denote a fixed physical qubit of the
\(j\)th constituent code, and let
\(p_j\in\mathsf B((C_{(i_j)})_1)\) label an elementary \(Z\)-stabilizer of the $j$th code with
support \(\partial p_j\). The relevant elementary components of \(Q\) therefore
have the form
\(e_1\otimes\cdots\otimes p_j\otimes\cdots\otimes e_\ell\).
Only the boundary acting on \(p_j\) contributes to the all-degree-zero
physical component, yielding the elementary non-Clifford stabilizer generators:
\label{eq:elementaryNonCliffordStabilizers}
\begin{equation}
\mathsf C^{\ell-1}\mathsf Z_{
e_\ell,\ldots,e_{j+1},\,\partial_1p_j,\,
e_{j-1},\ldots,e_1},
\qquad
j=1,\ldots,\ell .
\end{equation}
Thus an elementary non-Clifford stabilizer fixes one physical qubit in
\(\ell-1\) constituent codes, while the remaining qubit ranges over the
support of an elementary Pauli \(Z\)-stabilizer.   For \(\ell=2\), this reduces to
Eq.~\eqref{eq:elementaryCliffordStabilizers}. For sparse constituent
codes, \(\partial p_j\) has \(O(1)\) support, so these generators are
local in the LDPC sense.

Thus, just as Pauli and Clifford stabilizers allow one to locally deform physical representatives without changing their logical action, non-Clifford stabilizer circuits of the form Eq.~\eqref{eq:elementaryNonCliffordStabilizers} allow us to search within a fixed homology class for sparser representatives of a non-Clifford logical gate. This is the deformation freedom relevant for finding lower-depth, and when possible transversal, implementations.

\subsubsection{Non-Abelian Codes and Non-Abelian Surgery}
\label{subsubsec:NonAbelianSurgery}

We conclude our discussion of the chain map hierarchy by discussing how it can describe the logical measurement of Clifford and non-Clifford operators via the recently described method of non-Abelian code surgery~\cite{Christos2026nonabelian} in a manner directly analogous to the chain map complex.

In particular, let $\mathbf{C}^{[\ell]} = (C_{(i_\ell)}, \cdots, C_{(i_1)})$ be a tuple of chain complexes.
Similar to the case of the standard code surgery gadget, let us suppose that the qubits associated with $C_{(i_\ell)}^{\bullet}$ are  placed in the hypergraph cluster state and the rest of the qubits are in the CSS codes associated with their respective chain complexes.
Moreover, let $\phii_{\bullet}$ be a $-1$-cycle in $[\mathbf{C}^{[\ell]}]_\bullet$---i.e. it maps from $C^{\bullet}_{(i_\ell)}[-1]$ to $[\mathbf{C}^{[\ell-1]}_{\setminus \ell}]_\bullet$.
Then, the non-Abelian code surgery gadget will be associated with the following quantum circuit:
\begin{equation}\label{eq-nonAbelianmeasurementcircuit}
M_{\mathsf{C}^{\ell - 2} \mathsf{Z}}^{\phii}  = \left(\prod_{e \in \mathsf{B}(C^{(i_\ell)}_0)} M^{(i_\ell)}_{e, Z} \right) \times  \left( \prod_{\mathbf{v} \in \mathsf{B}(C^{-1}_{(i_\ell)}) } M^{(i_\ell)}_{\mathbf{v}, X} \right) U_{\mathsf{C^{\ell-1} Z}}^{\phii}
\end{equation}
We now will show that the quantum circuit above measures a (product of) $\mathsf{C}^{\ell-2} \mathsf{Z}$ gates performed between the codes in $\mathbf{C}^{[\ell-1]}_{\setminus \ell}$.
In particular, we prove the following theorem:
\begin{shaded*}
\begin{restatable}[Non-Abelian Surgery and the Chain Map Hierarchy]{theorem}{NonAbelianSurgery}\label{thm:NonAbelianSurgery}
    Let \(\mathbf{C}^{[\ell]} = (C_{(i_\ell)}, \cdots, C_{(i_1)})\) be a tuple of chain complexes defining a chain map complex $[\mathbf{C}]$ at the $\ell$-th level of the chain map hierarchy.
    Then, the following statements are true: 
    \begin{enumerate}
        \item[(1)]\textit{Homology and Logical Non-Abelian Surgery:} Associated with each homology class $[\phii] \in H_{-1}([\mathbf{C}^{[\ell]}])$ is a non-Abelian code surgery gadget that measures a logical $\mathsf{C}^{\ell-2} \mathsf{Z}$ gate between codes $\mathbf{C}^{[\ell-1]}_{\setminus \ell}$.
        The gate is implemented via the quantum circuit in Eq.~\eqref{eq-nonAbelianmeasurementcircuit} and  performs the following logical action:
        \begin{equation}\label{eqn:CMHMeasurementLogical}
            \hspace{7 mm} \prod_{\mathbf{L}_\ell} \overline{M}(\mathsf{C^{\ell-2}Z}_{ \phii_*(\mathbf{L}_\ell)}).
        \end{equation}
        Here, the product over $\mathbf{L}_{\ell} \in \mathsf{B}(H^{-1}(C_{(i_\ell)}))$ runs over a basis of $(-1)$-cohomology classes,
        and $\phii_*$ is the map induced on cohomology by \(\phii\), so that \(\phii_*(\mathbf{L}_\ell) \in H_0([\mathbf{C}^{[\ell-1]}_{\setminus \ell}])\) specifies a logical \(\mathsf{C}^{\ell-2}\mathsf{Z}\) gate.

        As a consequence, if the non-Abelian code surgery gadget associated with any two $(-1)$-cycles $\phii$ act differently on logical qubits, they are in different homology classes.

        \item[(2)]\textit{Exhaustiveness:} For any logical $\mathsf{C^{\ell-2}Z}$ measurement pattern, there is a quantum circuit of the form of Eq.~\eqref{eqn:CMHMeasurementLogical}
        labeled by a homology class of $H_{-1}([\mathbf{C}^{[\ell]}])$ that performs that action.
    \end{enumerate}
\end{restatable}
\end{shaded*}

The proof of the above theorem is relegated to Appendix~\ref{subapp:CMHMeasurements} and is a direct generalization of the logical measurement of Pauli's in the traditional code surgery gadget.

\subsection{Intra-block Gates in the Chain Map Hierarchy}
\label{sec:IntrablockGates}

Thus far, our discussion has primarily centered around logical gates that are performed between code blocks.
We conclude this section by discussing how the chain map hierarchy can further be used to describe a large class of logical gates that act \textit{within} a code block.

When generalizing to the intra-block case, new challenges naturally arise.
At a high level, note that our goal will be to associate chain maps $\phii_{\bullet}$ at some level of the chain map hierarchy to gates $U^{\phii}$ acting on physical qubits.
Our previous associations between chain maps and (non-)Clifford gates focused on multiply-controlled $Z$ or $X$ gates---e.g. a chain map $\phii_{\bullet}: C^{\bullet} \to D_{\bullet}$ described a pattern of  inter-block $\mathsf{CZ}$ gates between qubits $e$ and target $\phii(e)$. 
However, in the intra-block case, ``control'' and ``target'' qubits can be the same qubit.
Consequently, we need to revisit how we associate chain maps to physical gates. 

\subsubsection{Commuting Intra-block Clifford Gates} 

We begin our discussion of intra-block gates by first considering the case where the gates in a block all commute---specifically, are diagonal in some basis---and are Clifford, before returning to non-Clifford and non-commuting generalizations at the end.

Naively, to specify an intra-block gate on code $C$, we would simply need to think about chain maps $\phii_{\bullet}: C^{\bullet} \to C_{\bullet}$.
We will see now that this is, in general, not sufficient.
In particular, in this case, we will need to supplement our chain map description of gates with two extra ingredients.
First, we will need to work with a subcomplex of \textit{symmetric} chain maps (i.e. chain maps $\phii$ that are equal to their transpose $\phii^{\mathsf{T}}$).%
\footnote{Strictly speaking, we do not need that a chain map is symmetric on the level of chains but only require it to be induce a symmetric matrix on (co)homology. 
If this is the case, $\phii$ and $\phii^{\mathsf{T}}$ will be homotopic and our formulas for gates will pick up a correction arising from the null homotopy connecting them.
This is directly analogous to the role that the higher cup product plays in defining the Pontryagin square.
} 
Second, these chain maps will act on a $\mathbb{Z}_4$ ``lift'' of the chain complex describing the code we are interested in. 

The motivation for these two ingredients is straightforward.
In particular, chain maps $\phii_{\bullet}: C^{\bullet} \to C_{\bullet}$ should intuitively be associated to intra-block gates that map $X$ operators to $Z$ operators under the group commutator.
There are two saliant Clifford gates that do this: the $\mathsf{CZ}$ gate and the $\mathsf{S} = e^{i \frac{\pi}{2} n} $ gate, which act on $X$ operators as: 
\begin{equation}
    [\mathsf{CZ}_{e, e'}, X_e]_{\rm grp} = Z_{e'} \qquad [\mathsf{S}_e, X_e]_{\rm grp} = -iZ_{e}
\end{equation}
where $e \neq e'$.
The association of chain maps to $\mathsf{CZ}$ gates is more familiar from our discussion on inter-block gates.
There, an element \(\phii^{e_C,e'_D}=1\) implied that \(X_e^C\) maps under the gate to \(Z_{e'}^D\)---there is a \(\mathsf{CZ}\) between qubits at \(e_C\) and \(e'_D\).
However, if all gates connect between the same code block, then a \(\mathsf{CZ}_{e,e'}\) gate will map both \(X^C_e \mapsto Z^C_{e'}\) and \(X^C_{e'} \mapsto Z^C_{e}\).
For \(\phii\) to encode the correct action on \(X\) operators, it must thus be symmetric, \(\phii^{e,e'} = \phii^{e',e}\).
That is, \(\phii = \phii^\transpose\).

The chain map should also encode \(\mathsf{S}\) gates, corresponding to an \(X\) that maps to a \(Z\) supported on the same site.
However, note that the $\mathsf{S}$ gate does not square to $1$, as is the case with the $\mathsf{CZ}$ gate.
Instead, $\mathsf{S}^4 = 1$.
Consequently, a normal binary chain map $\phii$ mapping between two chain complexes over $\mathbb{F}_2$ would not be able to distinguish between say $\mathsf{S}$ and $\mathsf{S}^{\dagger}$.
This motivates a need to ``lift'' our chain complexes to be defined over $\mathbb{Z}_4$.
We now formally introduce these ingredients and show how they can be used to define intra-block logical Clifford gates.

We start by incorporating the symmetry constraint into our chain complex description.
To do so, let us first recognize that imposing that $\phii_0^{\alpha \beta} = \phii_{0}^{\beta \alpha}$ defines a \textit{linear constraint} on the space of chain maps that are associated with logical gates.
In particular, if we view $\phii_0$ as a matrix, we can rephrase this constraint as a linear equation of the form: 
\begin{equation} \label{eq-symmetricconstraint}
    \mathbb{S}\cdot \phii_0 = \mathbb{S}^{\rho \sigma}_{\alpha \beta} \phii_0^{\alpha \beta} \equiv (\delta^{\rho}_{\alpha} \delta^{\sigma}_{\beta} - \delta^{\rho}_{\beta} \delta^{\sigma}_{\alpha}) \phii^{\alpha \beta}_0 = 0
\end{equation}
where $\delta_{\alpha}^{\beta}$ is the Kronecker delta tensor---i.e. the identity matrix---and the equation above holds for all $\rho$ and $\sigma$.
The equation above is set up to check that $\phii^{\rho \sigma}_0 = \phii^{\sigma \rho}_0$. 

Our goal now is to build a chain complex description of these linearly constrained maps.
Doing so is straightforward: the space of such maps forms a \textit{constrained subcomplex} of the chain map complex.
We define the constrained subcomplex as follows:

\begin{defn}[Constrained Subcomplex] Consider a chain complex $C_{\bullet}$ and let $A = \{A_j: C_0 \to C_0\}$ be a set of  linear maps specifying linear constraints on $C_0$, e.g. $A_j f = 0$ for $f \in C_0$.
Then, the \textit{constrained subcomplex} of $C_{\bullet}$ subject to $A$ is a chain complex $C^A_{\bullet}$:
\begin{equation}
    \cdots \xrightarrow[]{\partial_1} C_1^A \xrightarrow[]{\partial_1} C_0^{A} \xrightarrow[]{\partial_0} C_{-1}^A \xrightarrow[]{} \cdots 
\end{equation}
whose vector spaces are given by: 
\begin{enumerate}
    \item[(1)] $C_0^A = \bigcap_{j} \mathsf{ker}(A_j)$---i.e. $0$-chains satisfing the constraints.
    \item[(2)] $C_{-j}^A = C_{-j}$ for all $j > 0$
    \item[(3)] $C_{j}^A = \partial^{-1}_{j} (C_{j - 1}^A)$ for all $j > 0$
\end{enumerate}
where $\partial^{-1}_j(C_{j - 1})$ specifies the pre-image of $C_{j-1}$ under $\partial_j$.
\end{defn}
Checking that the above forms a chain complex is straightforward.

For the purpose of diagonal Clifford gates, we will be interested in a constrained subcomplex of the chain map complex subject to the constraint $\mathbb{S}$ as defined in Eq.~\eqref{eq-symmetricconstraint}, $[C^{\bullet}, C_{\bullet}]^{\mathbb{S}}$.
As a remark, we note that this particular subcomplex has a a nice interpretation from the tensor product viewpoint.
In particular, $[C^{\bullet}, C_{\bullet}]^{\mathbb{S}}$ is isomorphic to the subcomplex of \emph{symmetric tensors}---i.e. tensors where $\phii^{\alpha \beta} = \phii^{\beta \alpha}$. 
Consequently, it is more descriptive to call the constrained subcomplex $[C^{\bullet}, C_{\bullet}]^{\mathbb{S}}$ the \textbf{symmetric subcomplex} of $[C^{\bullet}, C_{\bullet}]$.

The symmetric subcomplex handles the symmetry constraint alluded to earlier.
All that is left is to introduce the concept of a $\mathbb{Z}_4$-lift of the chain complex, which will allow us to handle the $\mathbb{Z}_4$ nature of the $\mathsf{S}$ gate. 
A $\mathbb{Z}_{2n}$-lift of a chain complex is defined as follows:
\begin{defn}[(Nice, Sparse) Chain Complex Lift] Let $C_{\bullet}$ be a based chain complex associated with a sequence of vector spaces $C_j$ over the field $\mathbb{F}_2$ and boundary maps $\partial_j$.
We say that a chain complex $\widetilde{C}_{\bullet}$ over the ring $\mathbb{Z}_{2n}$ is a $\mathbb{Z}_{2n}$-lift of $C_{\bullet}$ if:
\begin{enumerate}
    \item[(1)] The spaces $\widetilde{C}_j$\footnotemark\ associated with $\widetilde{C}_{\bullet}$ are related to the vector spaces of $C_{\bullet}$ via $C_j = \widetilde{C}_j/2\widetilde{C}_j$. 

    \item[(2)] As matrices, the boundary maps of $\widetilde{C}_j$ satisfy $\partial_j = \widetilde{\partial}_j \bmod 2$.
\end{enumerate}
In other words, there exists a quotient chain map $\pi_{\bullet}: \widetilde{C}_{\bullet} \to C_{\bullet}$ that is surjective and $\mathsf{ker}(\pi_j)  = 2\widetilde{C}_j$. 
We further require that there is a quotient on cochains \(\hat{\pi}^\bullet: \widetilde{C}^\bullet \to C^\bullet\), with the same properties and \(\tilde{\mathbf{a}}(\tilde{b}) = \hat{\pi}(\tilde{\mathbf{a}})(\pi(\tilde{b}))\) for all \(\tilde{\mathbf{a}} \in \widetilde{C}^\bullet\), \(\tilde{b} \in \widetilde{C}_\bullet\).

We say that \(\widetilde{C}_\bullet\) is \emph{nice} if its induced maps on (co)homology are surjective---i.e.\ every (co)cycle in \(C_\bullet\) is the image of a (co)cycle in \(\widetilde{C}_\bullet\).
Moreover, we say that $\widetilde{C}_{\bullet}$ is a \textit{sparse lift} of $C_{\bullet}$ if $\widetilde{C}_{\bullet}$ is a sparse based chain complex.   
\end{defn}\footnotetext[\value{footnote}]{Formally, the spaces $\widetilde{C}_j$ are $\mathbb{Z}_{2n}$-modules---analogs of vector spaces but where the underlying field is replaced by a ring (see Appendix~\ref{app:CMCDetails})}

A lift of a chain complex provides us a formal way of extending a binary chain complex to one defined over, say, $\mathbb{Z}_4$.
While one might imagine that a lift is as simple as mapping the elements $0, 1 \in \mathbb{F}_2$ to the elements $0, 1 \in \mathbb{Z}_4$, this ``naive lift'' doesn't ensure that the characteristic properties of chain complexes hold over $\mathbb{Z}_4$, e.g. $\widetilde{\partial}_j \circ \widetilde{\partial}_{j + 1}$ may not equal to zero modulo $4$ if we use the naive lift.
Further, even if a lift can be found, it is not guaranteed to be either sparse or nice: the matrices \(\tilde{\partial}_j\) may be densely filled with even elements, and some cycles in \(C_\bullet\) need not be quotients of cycles in \(\widetilde{C}_\bullet\).
Nevertheless, it was proven in Ref.~\cite{Freedman2021CodeToManifold} that a nice lift of a binary chain complex to $\mathbb{Z}_{2n}$ (for arbitrary $n$) always exists---though such a lift is still not guaranteed to be sparse.
Sparse lifts are guaranteed to exist for large classes of qLDPC codes, include all hypergraph product codes of classical codes, certain codes achieved through other product constructions (e.g. certain balanced product and fiber-bundle product codes), and even good qLDPC codes~\cite{Freedman2021CodeToManifold}. 
These are not guaranteed to be nice but nice and sparse lifts are conjectured to exist for some of these in Ref.~\cite{Freedman2021CodeToManifold}.

Equipped with a $\mathbb{Z}_4$ lift of a chain complex $C_{\bullet}$, $\widetilde{C}_{\bullet}$, and a symmetric subcomplex $[\widetilde{C}^{\bullet}, \widetilde{C}_{\bullet}]^{\mathbb{S}}$, we can associate its $0$-chains and $0$-cycles with intra-block Clifford gates.
In particular, the way of doing so is 
given a $0$-cycle $\widetilde{\phii}_{\bullet}$ of $[\widetilde{C}^{\bullet}, \widetilde{C}_{\bullet}]^{\mathbb{S}}$, we can define the $\widetilde{\phii}_0$ as the matrix associated with the map $\widetilde{\phii}_0$ and use this to define the following unitary gate: 
\begin{equation} \label{eq-intrablockdiaggate}
    V^{\widetilde{\phii}} = \exp\left(i \frac{\pi}{2} n_{\alpha}\,  \widetilde{\phii}^{\alpha \beta}_0\,   n_{\beta} \right)
\end{equation}
Notice the factor of $1/2$ in the exponent relative to the gate in the inter-block gate [Eq.~\eqref{eq:CZopassoc}]
The above gate is simply a product of $\mathsf{S}$ and $\mathsf{CZ}$ gates and can be re-written in a form that is identical to the ``fold-transversal'' gates of Refs.~\cite{Moussa2016FoldedSurface,Breuckmann2024foldtransversal}.
In particular, the gate can be written as: 
\begin{align}
    V^{\widetilde{\phii}} &= \prod_{e \in \mathsf{B}(C_0)} e^{i \frac{\pi}{2} \widetilde{\phii}^{e, e} n_e} \times \prod_{i < j} e^{i \pi \, n_{e_i}\,  \widetilde{\phii}^{e_i, e_j}\,   n_{e_j} } \\ 
    &= \prod_{e \in \mathsf{B}(C_0)} \mathsf{S}^{\widetilde{\phii}_{e, e}}_e  \times \prod_{i < j} \mathsf{CZ}_{e_i, e_j}^{\widetilde{\phii}_{e_i, e_j}} \label{eqn:TildePhiSCZForm}
\end{align}
We will soon show that the above defines a logical gate on the code $C$.
However, prior to doing so, an important remark is in order.
Let us note that because $\mathsf{CZ}^2 = 1$, only the data of $\widetilde{\phii}$ modulo two is relevant for these gates.
In contrast, the full $\mathbb{Z}_4$ data of $\widetilde{\phii}$ is necessary for defining the $\mathsf{S}$ gate piece.
However, since $\mathsf{S}$ gates are single qubit gates and only the $\mathsf{CZ}$ gates contribute to the depth of the gate, the depth of $V^{\widetilde{\phii}}$ is only determined by the sparsity of $\widetilde{\phii}$ modulo two!
Note that $\widetilde{\phii}$ modulo two defines a chain map $\phii_{\bullet} \in [C^{\bullet}, C_{\bullet}]^{\mathbb{S}}$.
Consequently, a viable strategy for finding constant depth implementations of intra-block diagonal Clifford gates is simply to find a sparse chain map in $[C^{\bullet}, C_{\bullet}]^{\mathbb{S}}$ and then use the non-sparse but nice lift of Ref.~\cite{Freedman2021CodeToManifold} to ``correct'' the diagonal contribution $\widetilde{\phii}^{e, e}$ to an appropriate logical intra-block gate.

We will now prove that when $\widetilde{\phii}$ is a $0$-cycle of the symmetric and $\mathbb{Z}_4$ lifted chain map complex, it  implements a logical Clifford gate on the CSS code associated with $C_{\bullet}$.
In particular, we prove:

\begin{theo}[Intra-Block Diagonal Logical Clifford Gates] Suppose that $\widetilde{C}_{\bullet}$ is a nice $\mathbb{Z}_4$ lift of the binary chain complex $C_{\bullet}$ and let $\widetilde{\phii}_{\bullet}$ be a $0$-cycle of the symmetric subcomplex $[\widetilde{C}^{\bullet}, \widetilde{C}_{\bullet}]^{\mathbb{S}}$ of the chain map complex where $\mathbb{S}$ is defined in Eq.~\eqref{eq-symmetricconstraint}.
Then, the action of $V^{\widetilde{\phii}}$ on operators $X_{\mathbf{a}}$, for $\mathbf{a} \in C^0$, is given by:
\begin{equation}
    [V^{\widetilde{\phii}}, X_{\mathbf{a}}]_{\rm grp} = (- i)^{\tilde{\mathbf{a}}(\widetilde{\phii}(\tilde{\mathbf{a}}))} Z_{\pi(\widetilde{\phii}(\widetilde{\mathbf{a}}))},
\end{equation}
where \(\tilde{\mathbf{a}}\) is any $0$-cocycle of $\widetilde{C}^{\bullet}$ such that \(\hat{\pi} (\widetilde{\mathbf{a}}) =  \mathbf{a}\) (which exists by niceness). 
Consequently, $V^{\widetilde{\phii}}$ is a logical gate acting on code $C$ whose logical action only depends on the homology class of $\widetilde{\phii}_{\bullet}$.
\end{theo}

\begin{proof} The proof follows from direct calculation.
In particular, let us note that: 
\begin{equation} \label{eq:XVX}
    X_{\mathbf{a}} (V^{\widetilde{\phii}})^{\dagger} X_{\mathbf{a}} = \exp\left(- i \frac{\pi}{2} X_{\mathbf{a}}  n_{\alpha}\,  \widetilde{\phii}^{\alpha \beta}\,   n_{\beta} X_{\mathbf{a}} \right)
\end{equation}
To unravel the above equation, it is useful to note that  $X_e n_e X_e = 1 - n_e$ and consequently:
\begin{equation}
    X_{\mathbf{a}} n_{\alpha} X_{\mathbf{a}} \overset{4}{=} \mathbf{a}_{\alpha} + (-1)^{\mathbf{a}_{\alpha}} n_{\alpha}
\end{equation}
where $\overset{4}{=}$ is understood to mean equality mod $4$ and by \(\mathbf{a}_{\alpha} \bmod 4\) we mean the naive lift of \(\mathbf{a}_{\alpha} \in \FF_2\) to \(\ZZ_4\)---directly taking \(1 \in \FF_2\) to \(1 \in \ZZ_4\).
Hence, the exponent of Eq.~\eqref{eq:XVX} can be expanded to: 
\begin{align} \label{eq-AHHHHHHHH}
    X_{\mathbf{a}}  n_{\alpha}\,  \widetilde{\phii}^{\alpha \beta}\,   &n_{\beta} X_{\mathbf{a}} \overset{4}{=} a_{\alpha} \widetilde{\phii}^{\alpha \beta} a_{\beta} + (-1)^{a_{\alpha} + a_{\beta}} n_{\alpha} \widetilde{\phii}^{\alpha \beta} n_{\beta} \nonumber \\
    &+ (-1)^{a_{\alpha}} n_{\alpha} \widetilde{\phii}^{\alpha \beta} a_{\beta} + (-1)^{a_{\beta}} a_{\alpha} \widetilde{\phii}^{\alpha \beta} n_{\beta}.
\end{align}
We treat the first two terms on the right above individually and the last two terms together.
Note that the first term is simply $\mathbf{a}(\widetilde{\phii}(\mathbf{a}))$.
Now, $\mathbf{a}$, taken modulo $4$, need not be a cocycle.
However, we know by niceness that it is related to a $0$-cocycle $\widetilde{\mathbf{a}}$ of $\widetilde{C}^{\bullet}$ as:
\begin{equation}
    \widetilde{\mathbf{a}} = \mathbf{a} + 2\mathbf{y} 
\end{equation}
where $\mathbf{y}$ is some $0$-cochain.
Consequently: 
\begin{align}
    \widetilde{\mathbf{a}}(\widetilde{\phii}(\widetilde{\mathbf{a}})) &\overset{4}{=} \mathbf{a} (\widetilde{\phii}(\mathbf{a})) + 2 \mathbf{y}(\widetilde{\phii}(\mathbf{a})) + 2 \mathbf{a}(\widetilde{\phii}(\mathbf{y}))\\
    &= \mathbf{a} (\widetilde{\phii}(\mathbf{a})) +  4 \mathbf{a}(\widetilde{\phii}(\mathbf{y})) \overset{4}{=}\mathbf{a}( \widetilde{\phii}(\mathbf{a})).
\end{align}

Moreover, the last two terms in Eq.~\eqref{eq-AHHHHHHHH} are equal to one another since $\widetilde{\phii}^{\alpha \beta} = \widetilde{\phii}^{\beta \alpha}$ and hence these terms together are equal to $2(-1)^{\alpha} n_{\alpha} \widetilde{\phii}^{\alpha \beta} a_{\beta}$.
Finally, note that the second term in Eq.~\eqref{eq-AHHHHHHHH} above can be re-written as: 
\begin{align}
    (-1)^{a_{\alpha} + a_{\beta}} n_{\alpha} &\widetilde{\phii}^{\alpha \beta} n_{\beta}  = \sum_{\alpha} \widetilde{\phii}^{\alpha \alpha} n_{\alpha}\\
    &+ 2\sum_{\alpha < \beta} (-1)^{a_{\alpha} + a_{\beta}} n_{\alpha} \widetilde{\phii}^{\alpha \beta} n_{\beta} \nonumber
\end{align}
Taking this expression modulo \(4\), the signs in the second term drop out, \(2(-1)^{a_\alpha + a_\beta} \overset{4}{=} 2\).
Consequently:
\begin{equation} \label{eq:XVX}
    X_{\mathbf{a}} (V^{\widetilde{\phii}})^{\dagger} X_{\mathbf{a}} = (-i)^{\widetilde{\mathbf{a}}(\widetilde{\phii}(\widetilde{\mathbf{a}}))} (-1)^{n_{\alpha} \widetilde{\phii}^{\alpha \beta} a_{\beta}} (V^{\widetilde{\phii}})^{\dagger}
\end{equation}
Now, we note that \((-1)^{n_{\alpha} \widetilde{\phii}^{\alpha \beta} a_{\beta}} = (-1)^{n_{\alpha} \pi(\widetilde{\phii}^{\alpha \beta} \widetilde{a}_{\beta})}\), where \(\pi\) is the modulo \(2\) quotient, and that \(\pi(\widetilde{\phii}(\widetilde{\mathbf{a}}))\) does not depend on the choice of lift \(\widetilde{\mathbf{a}}\).
Thus, the group commutator: 
\begin{equation}
[V^{\widetilde{\phii}}, X_{\mathbf{a}}]_{\rm grp} = X_{\mathbf{a}} (V^{\widetilde{\phii}})^{\dagger} X_{\mathbf{a}} V^{\widetilde{\phii}} = (-i)^{\widetilde{\mathbf{a}}(\widetilde{\phii}(\widetilde{\mathbf{a}}))} Z_{\pi(\widetilde{\phii}(\widetilde{\mathbf{a}}))}    
\end{equation}
Note that on stabilizers $\mathbf{a} = \exd \mathbf{v}$ the phase $(-i)^{\widetilde{\mathbf{a}}(\widetilde{\phii}(\widetilde{\mathbf{a}}))} = 1$ because $\widetilde{\exd \mathbf{v}} (\widetilde{\phii}(\widetilde{\exd \mathbf{v}}) \overset{4}{=} 0$.
As a consequence of this group commutator formula, stabilizers map to stabilizers under $V^{\phii}$ and hence it is a logical action on the CSS code associated with $C_{\bullet}$. 
\end{proof}

We conclude by remarking that $n_{\alpha} \widetilde{\phii}_0^{\alpha \beta} n_{\beta}$ specifies a quadratic polynomial in the number operators $n_{\alpha}$.
While we used a $\mathbb{Z}_4$ lift and a linear chain map (corresponding to a cycle in the tensor product) to specify this quadratic polynomial, a more general understanding of intra-block Clifford gates will likely require one to consider the space of quadratic homological maps between chain complexes~\cite{baues1996homotopy}.

\subsubsection{Inter-Block Gates as Intra-Block Gates}

Two independent blocks of CSS codes $C$ and $D$ can just as well be thought of as one block of the code $C \oplus D$.
As such, we should be able to understand the inter-block $\mathsf{CZ}$ gates in the intra-block formalism developed above.

To do so, let us consider a chain map $\phii_{\bullet}: C^{\bullet} \to D_{\bullet}$, corresponding to a $0$-cycle of the chain map complex $[C^{\bullet}, D_{\bullet}]$.
Our goal is to show that this $\phii_{\bullet}$ can be used to define a $0$-cycle of the symmetric subcomplex $[(C \oplus D)^{\bullet}, (C \oplus D)_{\bullet}]^\mathbb{S}$.
[As we are trying to find a \(\mathsf{CZ}\) gate, we can use an arbitrary nice lift of this cycle to \(\ZZ_4\) without worrying about sparsity, per the remarks below Eq.~\eqref{eqn:TildePhiSCZForm}.]
Doing so is straightforward.
In particular, we define the following map on $C^{\bullet} \oplus D^{\bullet}$: 
\begin{equation}
    \psi_{\bullet}^{\alpha \beta} =  \begin{pmatrix}
        0 & \phii_{\bullet} \\ \phii_{\bullet}^{\mathsf{T}} & 0
    \end{pmatrix}^{\alpha \beta}
\end{equation}
The map above is manifestly symmetric and one can verify that it forms a chain map between $C^{\bullet} \oplus D^{\bullet} \to C_{\bullet} \oplus D_{\bullet}$.
Hence, the above is a $0$-cycle in $[(C \oplus D)^{\bullet}, (C \oplus D)_{\bullet}]^{\mathbb{S}}$.
Finally, one can show that:
\begin{equation}
    V^{\psi} = e^{i \frac{\pi}{2} n_{\alpha}^{C \oplus D}\,  \psi^{\alpha \beta}\, n_{\beta}^{C \oplus D}} = e^{i \pi n_{\alpha}^D \phii^{\alpha \beta} n_{\beta}^C} = U_{\mathsf{CZ}}^{\phii}
\end{equation}
Establishing the inter-block $\mathsf{CZ}$ gates that we have as the intra-block diagonal gates above.

\subsubsection{A Few Remarks on Non-Clifford and Non-Commutative Intra-Block Gates}
\label{subsubsec:IntraBlockClZandCNOT}

We conclude with some brief remarks on non-Clifford and non-commutative intra-block logical gates (i.e. intra-block gates where individual gates don't mutually commute).
While we leave a full treatment of these gates to forthcoming work, we describe some simple ways to define these gates that we will utilize in Section~\ref{sec:DiscoveringGates}.

We saw for intra-block \(\mathsf{CZ}\) gates that only the \(\ZZ_2\) chain map \(\phii_\bullet : C^\bullet \to C_\bullet\) was needed to find a logical gate, and not the more general \(\ZZ_4\) lift.
This simplification extends to higher level of the chain map hierarchy, and indeed, we can also avoid the need to ``correct'' the diagonal part of the gate.
For any \(0\)-cycle \(\phii_\bullet \in \CMH{\ell}_0 = [C^\bullet, [\cdots, C_{\bullet}]]_0\), we define a logical \(\mathsf{C^{\ell-1}Z}\) gate.
In fact, the naive application of Eq.~\eqref{eq:CZopassoc} defines this gate,
\begin{equation}\label{eq-intrablockClZ}
    V^\phii_\ell = \exp\left(i \pi \phii^{\alpha \cdots \gamma} n_\alpha \cdots n_\gamma \right).
\end{equation}

However, a few remarks are in order.
In general, Eq.~\eqref{eq-intrablockClZ} is a circuit both of \(\mathsf{C^{\ell-1}Z}\) gates and also \(\mathsf{C^{\ell-2}Z}\)s, \(\mathsf{C^{\ell-3}Z}\)s, all the way down to \(Z\)s.
Whenever some of the indices in the sum over \(n_\alpha \cdots n_\gamma\) coincide, the number of \(n\) operators in the sum reduces (\(n_e n_e = n_e\)) resulting in a gate from lower in the Clifford hierarchy.
Similarly, \(V^\phii_\ell\)'s logical action includes both \(\mathsf{C^{\ell-1}Z}\) pieces and corrections from lower in the hierarchy.

Further, note that we do not need \(\phii_\bullet\) to be symmetric.
Nonetheless, our previous remarks about the action of any gate on operators \(X_{\mathbf{a}}\) still holds, and one can see that \(V^\phii_\ell\) has such a symmetric action, regardless of asymmetry in \(\phii_\bullet\).
In fact, the logical action of \(V^\phii_\bullet\) is not simply encoded in \(\phii_*\), the induced action of \(\phii_\bullet\) on (co)homology, but rather in its \emph{symmetrization}.
For any \(f_\bullet \in \CMH{\ell}\), we define the symmetrization
\begin{equation}
    f^\mathbb{S}_\bullet = \sum_{\sigma \in S_\ell} \sigma_\bullet \circ f_\bullet \circ \sigma_\bullet^{-1},
\end{equation}
where \(S_\ell\) is the group of permutations on \(\ell\) objects and \(\sigma_\bullet : \CMH{\ell} \xrightarrow{\sim} \CMH{\ell}\) is a chain map corresponding to permutation of the factors in \(\CMH{\ell} \isom C_\bullet^{\otimes\ell}\).
Then the \(\mathsf{C^{\ell-1}Z}\) part of the logical action of \(V^\phii_\ell\) is encoded in the following repeated commutator,
\begin{equation}
    [[[V^\phii_\ell, X_{\mathbf{a}}], \cdots], X_\mathbf{c}]_{\rm grp} = \exp\left[i\pi \phii^{\mathbb{S}}_*(\mathbf{L}_{\mathbf{a}}, \cdots , \mathbf{L}_{\mathbf{c}})\right],
\end{equation}
viewing \(\phii_*^\mathbb{S}\) as a multilinear map from \(\ell\)-tuples of logical classes \(\mathbf{L}_\mathbf{a} \ni \mathbf{a}\) to \(\FF_2\).

These results are summarized by the following theorem.
\begin{shaded*}
\begin{restatable}[Intra-Block Logical Non-Clifford Gates]{theorem}{CCZIntraBlock}
    \label{thm:CCZIntraBlock} 
    Suppose that $\phii_{\bullet} \in \CMH{\ell}_0$ is a $0$-cycle.
    Then, Eq.~\eqref{eq-intrablockClZ} defines an intra-block logical \(\mathsf{C^{\ell-1}Z}\) gate whose action is given by \(\phii^{\mathbb{S}}_*\), the symmetrization of the logical action \(\phii_*\), up to logical gates from lower levels of the Clifford hierarchy.
\end{restatable}
\end{shaded*}

The proof follows from direct calculation and is deferred to Appendix~\ref{subapp:IntraBlockCCZ}.

To conclude, we observe a sufficient condition for the existence of an intra-block $\mathsf{CNOT}$ gate.
Crucially, unlike $\mathsf{C^{\ell} Z}$ gates, different $\mathsf{CNOT}$ gates do not commute with one another and hence a complete theory of these gates would need to account for this non-commutativity by tracking the order of the gates.
Nevertheless, in situations where these $\mathsf{CNOT}$ gates commute, chain maps can still be used to describe intra-block CNOT gates.
In particular, suppose that $\phii_{\bullet} \in [C^{\bullet}, C^{\bullet}]_0$ is a $0$-cycle of the chain map complex such that $\phii_0$ is \textit{block off-diagonal} as a matrix.
In other words, there is a set of controls $\mathsf{B}_{\mathsf{C}}(C_0) \subset \mathsf{B}(C_0)$ and a set of targets $\mathsf{B}_{\mathsf{T}}(C_0) \subset \mathsf{B}(C_0)$ that are mutually disjoint $\mathsf{B}_{\mathsf{C}}(C_0) \cap \mathsf{B}_{\mathsf{T}}(C_0) = \varnothing$ such that $(\phii_0)_{e'}^{\ e} \neq 0$ implies that $e \in \mathsf{B}_C(C_0)$ and $e' \in \mathsf{B}_T(C_0)$.
Put simply, this means that if \(e\) is used as a control for a \(\mathsf{CNOT}\) gate, it is never used as a target, and vice versa.
Then, the following gate: 
\begin{align} \label{eq-intrablockCNOT}
    V^{\phii}_{\mathsf{CNOT}} &= \prod_{e \in \mathsf{B}_C(C_0)} \mathsf{CNOT}_{e, \phii_0(e)}\\
    &= \exp\left(i \pi m^{\alpha} \phii_{\alpha}^{\ \beta} n_{\beta} \right), 
\end{align}
 [which reproduces the naive application of the $\mathsf{CNOT}$ gate of Eq.~\eqref{eq:productofCZ}] yields a logical $\mathsf{CNOT}$ gate.
In going from the first to the second line of Eq.~\eqref{eq-intrablockCNOT}, we used the fact that controls and targets (by construction) don't overlap.
Consequently, we circumvented the need to keep track of the order of these gates and were able to write them simply using a chain map.

With this in mind, we prove the following theorem:
\begin{shaded*}
\begin{restatable}[Intra-Block Logical CNOT Gates]{theorem}{CNOTIntraBlock}
    \label{thm:CNOTIntraBlock}
    Suppose that $\phii_{\bullet}$ is a $0$-cycle of $[C^{\bullet}, C^{\bullet}]$ such that as a matrix $\phii_0$ is block off-diagonal, then the action of Eq.~\eqref{eq-intrablockCNOT} on $X$ and $Z$ operators is given by:
    \begin{equation}
        [V^{\phii}_{\mathsf{CNOT}}, X_{\mathbf{a}}] = X_{\phii(\mathbf{a})} \qquad [V^{\phii}_{\mathsf{CNOT}}, Z_{b}] = Z_{\phii^{\mathsf{T}}(b)}  
    \end{equation}
    for all $\mathbf{a} \in C^0$ and all $b \in C_0$.
    Consequently, Eq.~\eqref{eq-intrablockCNOT} defines a logical gate on $C$, whose action is a product of $\mathsf{CNOT}$s on logical qubits.
\end{restatable}
\end{shaded*}

Once again, the proof follows from direct calculation and is deferred to Appendix~\ref{subapp:IntraBlockCCZ}.

\section{Discovering Transversal Gates Made Simple}
\label{sec:DiscoveringGates}

The chain map hierarchy provides a structured method to search for transversal gates in general qLDPC codes.
In this section, we derive several examples of previously-unknown gates uncovered by the chain map hierarchy.
Most of the codes we consider, including the usual two- and three-dimensional toric code, have been studied extensively in the literature, but nonetheless have transversal gates which had not been found.

While logical gates can be systematically enumerated with the chain map hierarchy, the gates so obtained are not automatically transversal.
For the codes we consider here,
the isomorphism between the chain map complex and tensor product gives an intuitive geometric way to \emph{sparsify} chain maps, when possible, and thus obtain transversal gates.
(The general conditions under which a candidate chain map can be sparsified will be explored in upcoming work~\cite{Sahay2026GoNoGo}.)

While the gates in this section are all novel, our emphasis is primarily on the \emph{method} of employing tensor product structure and sparsification to discover gates (outlined in Sec.~\ref{subsec:GateRecipe}) rather than the particular utility unlocked by the new gates.

In Sec~\ref{subsec:2DToricGate} we derive transversal logical gates for the entire Clifford group in the two-dimensional toric code.
In Sec.~\ref{subsec:3DToricGate} we move to the next level of the chain map hierarchy, and obtain  addressable \(\mathsf{CCZ}\)
gates in the three-dimensional toric code.
Further moving beyond manifold codes, in Sec.~\ref{subsec:ALPGate} we derive many addressable Clifford gates in a fracton code with many encoded qubits---the anisotropic linon-planon (ALP) code.
Finally, in Sec.~\ref{subsec:ALPIsiGate}, we extend this non-manifold example to addressable \(\mathsf{CCZ}\)
gates in the tensor product code of ALP and the repetition code.

As a remark, our results refute some folklore that addressable \(\mathsf{CZ}\) gates, say, require the existence of a global transversal \(\mathsf{CCZ}\) gate, and similar for higher levels of the Clifford hierarchy~\cite{lin2024transversalnoncliffordgatesquantum,JochymOConnor20214dtoric,he2025quantumcodesaddressabletransversal,zhu_non-abelian_2026}.
The toric code, for instance, has no transversal \(\mathsf{CCZ}\), but does have an exhaustive suite of addressable \(\mathsf{CZ}\)s.

\subsection{Recipe for Discovering Gates}
\label{subsec:GateRecipe}

All of the gates we find in this section are constructed with the same basic recipe.
Fixing a collection of codes 
\(\mathbf{C}^{[\ell]} = (..., \textcolor{orange}{C},\textcolor{dodgerblue}{D})\), the recipe is as follows (Fig.~\ref{fig:Recipe}).
\begin{enumerate}
    \item[(1)] \textbf{Specify} a logical action of the gate as an element \((\phii_{*})_0 \in H_0(\textcolor{dodgerblue}{D}) \otimes H_0(\textcolor{orange}{C}) \otimes \cdots\).

    \item[(2)] \textbf{Sparsify} a representative cycle of the logical action using chain map hierarchy stabilizer deformations.

    \item[(3)] \textbf{Synthesize} the physical implementation of the gate and enjoy computing with your qLDPC code. 
    
    \vspace{0.8ex}\emph{Bon appetit}.
\end{enumerate}
\begin{figure}
    \centering
    \includegraphics[width=\linewidth]{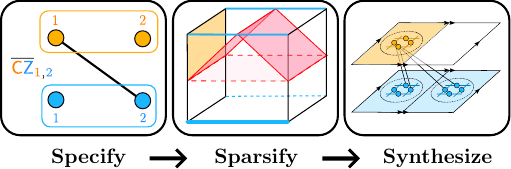}
    \caption{\textbf{Recipe for discovering gates:} To find a transversal logical gate using the chain map complex, first \emph{specify} the desired logical action: a \(\CZLogicalColor{1}{2}\) between one pair of the four encoded qubits in two blocks of the toric code, say.
    Then, \emph{sparsify} a representative of this logical action using the chain map complex stabilizer deformations.
    Finally, \emph{synthesize} the gate by converting the chain map to a physical quantum circuit.}
    \label{fig:Recipe}
\end{figure}

This recipe reiterates the procedure outlined in Secs.~\ref{sec:Intuition} and \ref{sec:ChainMapHierarchy}.
Specifying any desired logical action from a finite level in the chain map hierarchy, that action can be encoded as a logical tensor \((\phii_{*})_0\) (the all-degree-zero component of \(\phii_*\)).
This is the exhaustiveness of Theorem~\ref{thm:ChainMapHierarchyGates}.
Then, picking any representative cycle \(\phii_\bullet\) of this homology class, we get a chain tensor that can be turned into a gate.
However, this tensor will usually be dense, so one needs to sparsify the tensor by deforming with boundaries.
With that done, a physical implementation of the gate can be synthesized with Eq.~\eqref{eq-ClZ}.

The key step in the recipe is (2), sparsifying the gate by deforming with stabilizers.
In full generality, this is algorithmically hard---similar to finding the minimum weight logical of a CSS code.
Nonetheless, step (2) may be easy if all the codes in \(\mathbf{C}^{[\ell]}\) are well-understood, such that the full space of deformations is simple.
In this section, we will focus on the case where each of the codes in \(\mathbf{C}^{[\ell]}\) is a tensor product of well-understood classical codes, \(\textcolor{dodgerblue}{D_\bullet} = \textcolor{dodgerblue}{d^{(1)}_\bullet \otimes d_\bullet^{(2)} \otimes \cdots}\). 

In this case of tensor product codes, the recipe can be unpacked into more explicit steps. 
Gates from the \(\ell\)th level of the chain map hierarchy are constructed as follows.

\paragraph{Specify}
First, choose a target logical action of the gate.

\begin{enumerate}[label=(\alph*),leftmargin=*]
\item \emph{Logical action:} The logical action is specified as a homology class 
\begin{align}\label{eqn:RecipeLogicalAction}
    \hspace{14 mm} (\phii_{*})_0 &= \textcolor{dodgerblue}{L_D} \otimes \textcolor{orange}{L_C} \otimes \cdots \nonumber \\
    &\in 
    H_0(\textcolor{dodgerblue}{D}) \otimes 
    H_0(\textcolor{orange}{C}) \otimes \cdots
\end{align}
as in Eq.~\eqref{eq:phiasmap} and Theorem~\ref{thm:ChainMapHierarchyGates}.
In Eq.~\eqref{eqn:RecipeLogicalAction}, we specify an \emph{addressable} action---\((\phii_{*})_0\) is rank \(1\) as a tensor, and affects the minimum possible combination of logicals from each code. 
More general \((\phii_{*})_0\) are constructed as sums of terms like this.

\item \emph{Extend to homological action:}
To completely specify a chain map, we need not only the all-degree-zero homological component
\(H_0(\textcolor{dodgerblue}{D}) \otimes 
H_0(\textcolor{orange}{C})  \otimes \cdots\),
but also  every other component of \(H_0([\mathbf{C}^{[\ell]}])\) [Eq.~\eqref{eqn:CMHKunneth}].
Arbitrarily choose the other components of
\begin{align}
    \hspace{14 mm} \phii_* &\in \bigoplus_{\textcolor{dodgerblue}{j_D} + \textcolor{orange}{j_C} + \cdots = 0} H_{\textcolor{dodgerblue}{j_D}}(\textcolor{dodgerblue}{D}) \otimes 
    H_{\textcolor{orange}{j_C}}(\textcolor{orange}{C}) \otimes \cdots \nonumber \\
    &\isom H_0([\mathbf{C}^{[\ell]}]),
\end{align}
such that the \(\textcolor{dodgerblue}{j_D} = \textcolor{orange}{j_C}  = \cdots = 0\) component is \((\phii_{*})_0\).
This choice does not affect the logical action, but will affect the actual implementation of the gate, possibly including its sparsity.
For simplicity in the rest of the recipe, we will take all the other components to be \(0\), so that \(\phii_* = (\phii_{*})_0\).

\end{enumerate}

\paragraph{Sparsify}\label{par:sparsify} The sparsification process can be greatly simplified for tensor product codes.
In particular, it is frequently possible to find sparsifications of the full \(\phii_*\) from sparsifications of gates for the constituent tensor factors of each code.
\begin{enumerate}[label=(\alph*),leftmargin=*]
\item \emph{Reshape logical tensor:}
Splitting the legs of \(\phii_*\) into their tensor factors, the legs can be arranged into groups which correspond to gates between the simple codes \(\textcolor{dodgerblue}{d^{(i)}_\bullet}\), \(\textcolor{orange}{c^{(j)}_\bullet}\), and so on.
Indeed, there is a basis of logicals \(\textcolor{dodgerblue}{l_{j}^{di}} \in H_j(\textcolor{dodgerblue}{d^{(i)}})\) such that \(\textcolor{dodgerblue}{L_D}\) is a sum of terms like
\begin{equation}
    \textcolor{dodgerblue}{L_D} = (\textcolor{dodgerblue}{l_{j}^{d1} \otimes l_{j'}^{d2} \otimes \cdots}) + \cdots.
\end{equation}
For simplicity, we assume that only the first term is present---we are addressing a pure tensor logical.
Then the next step of the recipe is to reshape the homological action \(\phii_*\),
\begin{align}
    \phii_* &= 
    (\textcolor{dodgerblue}{l_{j}^{d1} \otimes  \cdots}) \otimes
    (\textcolor{orange}{l_{j'}^{c1} \otimes  \cdots}) \otimes \cdots \nonumber\\
    &\isom 
    (\textcolor{dodgerblue}{l_{j}^{d1}} \otimes
    \textcolor{orange}{l_{j'}^{c1}} \otimes \cdots) \otimes \cdots, \label{eqn:RecipeLogicalReshaping}
\end{align}
such that the regrouped factors
\begin{equation}
    \hspace{7mm}\textcolor{dodgerblue}{l_{j}^{d1}} \otimes
    \textcolor{orange}{l_{j'}^{c1}} \otimes \cdots \in 
    H_j(\textcolor{dodgerblue}{d^{(1)}}) \otimes 
    H_{j'}(\textcolor{orange}{c^{(1)}}) \otimes 
    \cdots
\end{equation}
correspond to chain map hierarchy gates for the factor codes.

\item\label{itm:SparsifyFactors} \emph{Sparsify factors:} 
As the factor codes are assumed to be well-understood, the factors 
\(\textcolor{dodgerblue}{l_{j}^{d1}} \otimes
\textcolor{orange}{l_{j'}^{c1}} \otimes \cdots\)
should be easy to sparsify---it is possible to identify representative cycles 
\begin{equation}
    \phi^{(1)}_\bullet \in \textcolor{dodgerblue}{l_{j}^{d1}} \otimes
    \textcolor{orange}{l_{j'}^{c1}} \otimes \cdots
\end{equation}
of each logical class such that each \(\phi^{(i)}_\bullet\) is sparse.
This can be done by first finding a dense logical representative (which is always possible by picking representative cycles \(\textcolor{dodgerblue}{a} \in \textcolor{dodgerblue}{l^{d1}_j}\), \(\textcolor{orange}{b} \in \textcolor{orange}{l^{c1}_{j'}}\), i.e.\ \([\textcolor{dodgerblue}{a}] = \textcolor{dodgerblue}{l^{d1}_j}\), and forming their tensor product \(\textcolor{dodgerblue}{a} \otimes \textcolor{orange}{b} \otimes \cdots\))  and deforming by boundaries from  \(\textcolor{dodgerblue}{d^{(1)}_\bullet} \otimes \textcolor{orange}{c^{(1)}_\bullet }\otimes \cdots\).
Then, the full chain tensor
\begin{equation}
    \phii_\bullet = \phi_\bullet^{{(1)}} \otimes \phi_\bullet^{{(2)}} \otimes \cdots
\end{equation}
has a sparsity degree [Eq.~\eqref{eqn:SparsityDegree}] which is bounded by the product of degrees from each \(\phi^{(i)}_\bullet\). 
\end{enumerate}

\paragraph{Synthesize physical gate}
Apply Eq.~\eqref{eq-ClZ} (or its \(\mathsf{CNOT}\) equivalents) to define a transversal logical gate from the sparse chain tensor \(\phii_\bullet\).

\paragraph{(Optional) Intra-block gates}
To find intra-block gates, one needs to find a chain tensor \(\tilde{\phii}_\bullet\) with additional properties (Section~\ref{sec:IntrablockGates}).
In this section, we will find transversal intra-block implementations of \(\mathsf{C^{\ell-1}Z}\) gates (which can be produced from any inter-block \(\phii_\bullet\) by symmetrization), \(\mathsf{S}\) gates (which can be found with a \(\ZZ_4\) lift of the chain complex), and \(\mathsf{CNOT}\) gates (which can be found from inter-block gates with disjoint controls and targets).

\vspace{1ex}
In the sparsification step (2\ref{itm:SparsifyFactors}), it is not guaranteed that every candidate logical action \(\phii_*\) will admit some regrouping in Eq.~\eqref{eqn:RecipeLogicalReshaping} that results in easily sparsified factors.
Indeed, if this were always possible, we could find a transversal implementation of a universal gate set for any code by looking high enough in the chain map hierarchy, violating various no-go theorems~\cite{Eastin2009TransversalRestrictions,Bravyi2013GatesforLocalStabilizer,JochymOConnor2018}.
In this section, ``easy to sparsify'' will be a synonym for ``line-like logicals,'' in which case geometric intuition will allow the identification of sparse representatives.
Then, the logical regroupings that work will be those such that each factor has at most one line-like logical \(\textcolor{dodgerblue}{l_{j}^{d1}}\) (say), and every other \(\textcolor{orange}{l_{j'}^{c1}}\), and so on are point-like.
(Because the factor codes \(\textcolor{orange}{c^{(i)}_\bullet}\) are classical, they have point-like \(Z\) logicals.)

\subsection{Full Clifford Group in the 2D Toric Code}
\label{subsec:2DToricGate}

In this section, we apply the recipe of Sec.~\ref{subsec:GateRecipe} to discover new transversal gates in the toric code~\cite{Kitaev_2003} which generate the entire Clifford group.
It is quite remarkable that, despite almost three decades of study, there are several transversal gates in the toric code which were previously unknown, such as an addressable Hadamard affecting just one of the two encoded qubits.
While very recent work~\cite{Albert2026transversalitystructurecliffordcircuits} used a numerical search to find transversal implementations of the Clifford group in small instances of the rotated surface code, it has not been clear until now how to find such gates in the full asymptotic toric code family of increasing code size.

\subsubsection{New Gates and Exhaustiveness}
\label{subsubsec:2TCExhaustiveness}

In the following subsections, we find transversal implementations of all of the following addressable gates.
Logical action is denoted with a bar, and the physical gates in the constant depth circuit realization of gates are also listed.
\begin{itemize}
    \item \emph{Phase gates:} \(\bar{\mathsf{S}}\), as a circuit of \(\mathsf{CZ}\) and \(\mathsf{S} ^{(\dagger)}\);
    \item \emph{Controlled \(Z\):} \(\overline{\mathsf{CZ}}\), as a circuit of \(\mathsf{CZ}\); and
    \item \emph{Controlled not:} \(\overline{\mathsf{CNOT}}\), as a circuit of \(\mathsf{CNOT}\).
\end{itemize}
The single qubit logical gates can act on any of the \(2N\) encoded qubits in \(N\geq 1\) blocks of the toric code, and the two-qubit gates act on any pair, including those within the same code block. 
Further, we prove that, between the gates above and the known
\begin{itemize}
    \item \emph{Global Hadamard} (known): \(\bar{\mathsf{H}}\otimes \bar{\mathsf{H}}\) within a single code block, as a physical \(\mathsf{H}\) on every qubit followed by a reflection (swap) around a diagonal~\cite{Kubica2015Unfolding,Moussa2016FoldedSurface,Breuckmann2024foldtransversal,Moylett2026logicalgatesfloquetcodes},
\end{itemize}
all possible transversal logical gates in \(N\) blocks of the toric code are now exhausted. 
That is, any logical action that can be implemented by transversal gates in the toric code can be implemented by an \(O(1)\) depth circuit composed of the above logical gates.
(The number of blocks \(N\) is treated as a constant for scaling statements.)
This includes, in particular, the addressable Hadamard mentioned in Sec.~\ref{sec:Intro}.
\begin{itemize}
    \item \emph{Addressable Hadamard:} \(\bar{\mathsf{H}}\), as a circuit of \(\mathsf{CZ}\), \(\mathsf{S}^{(\dagger)}\), and \(\mathsf{H}\).
\end{itemize}

The exhaustiveness of this gate set is a consequence of the observation that the above gates generate the entire Clifford group in any number of toric code blocks.
The Bravyi-K\"onig theorem implies that all logical gates formed from geometrically local circuits in the toric code are Cliffords.
We conjecture that Bravyi-K\"onig continues to hold for all finite depth circuits of few-body but possibly non-local gates (we present an informal proof for manifold codes in Section~\ref{sec:discussion}), and in the two-dimensional toric code this statement has already been proved in Ref.~\cite{JochymOConnor2018}.

\begin{lem}\label{lem:CliffordGeneration}
    The transversal logical gates listed above generate the entire Clifford group on any number \(N\) of toric code blocks.
\end{lem}
\begin{proof}
    We compile the \(\bar{\mathsf{H}}\otimes \bar{\mathsf{H}}\), \(\bar{\mathsf{S}}\), and \(\overline{\mathsf{CNOT}}\) gates into a known generating set for the Clifford group: consisting of all \(\bar{\mathsf{H}}\), \(\bar{\mathsf{S}}\), and \(\overline{\mathsf{CNOT}}\) gates. 
    Henceforth, we drop bars on logical gates.
    
    We are only missing single-qubit Hadamards.
    Label two qubits from the same block \(0,1\).
    Then the gate \(\mathsf{H}_0\) is given by
    \begin{equation}
        \mathsf{H}_{0} = e^{-i\pi/4} \mathsf{S}_{0} (\mathsf{H}_{0} \mathsf{H}_{1}) \mathsf{S}_{0} (\mathsf{H}_{0} \mathsf{H}_{1}) \mathsf{S}_{0},
    \end{equation}
    and similarly for \(\mathsf{H}_{1}\).
    This follows from \((\mathsf{H} \mathsf{S})^3 = e^{i\pi/4}\) and \(\mathsf{H}^2 = \unit\).    
\end{proof}

\begin{theo}
    The transversal logical gates above  are exhaustive: they generate all possible logical actions realizable by transversal gates in  \(N = O(1)\) blocks of the two-dimensional toric code.
\end{theo}
\begin{proof}
    Lemma~\ref{lem:CliffordGeneration} shows that all Clifford gates are possible to achieve in \(O(1)\) depth (in code distance) as products of the gates we have identified.
    By \cite[Corollary~10, Example~8]{JochymOConnor2018}, no family of constant locality (that is, \(O(1)\) sparseness), constant depth circuits on the family of two-dimensional toric codes can implement a non-Clifford logical action on sufficiently large distance codes.
    Thus, the available gates generate all possible logical actions of transversal gates.
\end{proof}

The exhaustiveness of our generated gates does not exclude a potentially more efficient (lower depth) realization of these gates.
However, we already have the optimal \(O(1)\) scaling of circuit depth with code distance: all (finitely many at fixed \(N\)) elements of the Clifford group are realized as \(O(1)\) combinations of finite depth circuits.
Our scaling with the number of code blocks \(N\) may not be optimal.

Finally, we comment on the relation of our results to a recent no-go theorem in Ref.~\cite{Chakraborty2026nogo}, which states that it is impossible to realize the full Clifford group with ``fold transversal'' Clifford gates in a stabilizer code with more than one qubit.
We use a less restrictive notion of transversality than Ref.~\cite{Chakraborty2026nogo}, which demands that the entire Clifford group be possible to implement with a depth-1 circuit, while we only demand \(O(1)\) depth in code distance.
As such, the no-go of Ref.~\cite{Chakraborty2026nogo} is not relevant to our setting.

Before proceeding to derive the new gates, we review the chain complex presentation of the toric code, and its classical factor code the Ising model.

\subsubsection{Chain Complex Formalism for the Ising Model and Toric Code}

While the toric code needs little introduction, our calculations in this section are couched in the chain complex perspective on the toric code and its tensor product structure, and in the particular conventions developed in Section~\ref{sec:ChainMapHierarchy}.
In particular, while the toric code is usually presented in terms of a square cellulation of a torus, with chain complex components living in degrees \(2\) (plaquettes), \(1\) (edges with qubits) and \(0\) (vertices), we insist that qubits are always at degree \(0\).
As such, the toric code chain complex should have components at degrees \(1\), \(0\), and \(-1\), so that qubits are still placed in the middle component.
This convention also affects how one interprets the toric code as a tensor product of repetition codes---one of the codes should have components at degrees \(1\) and \(0\), while the other should be at \(0\) and \(-1\).
We use this subsection to review the toric code using these conventions.

The toric code chain complex, \(\mathsf{TC}_\bullet\), is the tensor product of an Ising model (repetition code) complex \(\mathsf{I}_{x\bullet}\) and its dual \(\mathsf{I}_y^\bullet\),
\begin{equation}
    \mathsf{TC}_\bullet \isom \mathsf{I}_{x\bullet} \otimes  \mathsf{I}_y^\bullet.
\end{equation}
The codes \(\mathsf{I}_{x\bullet}\) and \(\mathsf{I}_y^\bullet\) play the same role as \(d^{(1)}_\bullet\) and \(d^{(2)}_\bullet\) from the recipe.
The numerical labels have been replaced with \(x\) and \(y\), as these factors geometrically correspond to spatial axes of the toric code.

To unpack the product construction of the toric code in our convention, we first define the Ising model chain complex,
\begin{equation}
    \mathsf{I}_\bullet = \left( \mathsf{I}_1 \xrightarrow{\partial} \mathsf{I}_0 \rightarrow 0 \right),
\end{equation}
so that the nontrivial components are at degree \(1\) and \(0\).
This is the \(\mathbb{F}_2\) chain complex corresponding to a cyclic graph.
Both of the nonzero components \(\mathsf{I}_1\) and \(\mathsf{I}_0\) are \(d\)-dimensional vector spaces (\(d\) denoting the distance of both the Ising model and the toric code), with bases corresponding to edges \(e_i \in \mathsf{B}(\mathsf{I}_1)\) and vertices \(v_i \in \mathsf{B}(\mathsf{I}_0)\) of the graph, respectively, as illustrated.
\begin{equation}
    \begin{tikzpicture}[scale = .8, baseline={([yshift=-.5ex]current bounding box.center)}]
    \foreach \i in {0, ..., 4}{
        \filldraw[color = black] (\i, 0) circle (2 pt);
        \draw[color = black] (\i, 0) -- (\i + 1, 0) ;
        \node at(\i, -0.3) {$v_{\i}$};
        \node at(\i + 0.5, 0.3) {$e_{\i}$};
    }
    \draw[color = black] (-1, 0) -- (0,0);
    \node at(-1 + 0.5, 0.3) {$e_{d -1}$};
    \node at (5.5, 0) {$\cdots$};
    \node at (-1.5, 0) {$\cdots$};
    \end{tikzpicture}
\end{equation}
The boundary map sends an edge to the sum of two vertices on either side of it,
\begin{equation}
    \partial e_i = v_i + v_{i+1},
\end{equation}
where \(i+1\) is regarded modulo \(d\). 

The corresponding cochain complex for the Ising model, \(\mathsf{I}^{\bullet}\), has a coboundary map that sends vertices to edges,
\begin{equation}
    \exd \mathbf{v}_i = \mathbf{e}_{i-1} + \mathbf{e}_i.
\end{equation}
To match Definition~\ref{def:SparseBasedChainComplex}, we organize the dual cochain complex
\(\mathsf{I}^{\bullet}\) such that coboundary maps still go to the right,
\begin{equation}
    \mathsf{I}^\bullet = \left( 0 \to \mathsf{I}^0 \xrightarrow{\exd} \mathsf{I}^1\right).
\end{equation}
That is, \(\mathsf{I}^1\) is regarded as being in the \((-1)\)-component.

Geometrically, \(\mathsf{I}^\bullet\) is the same as \(\mathsf{I}_\bullet\).
Indeed, there is a degree-\(1\) self-duality (Kramers-Wannier) isomorphism \(\kappa_\bullet: \mathsf{I}^\bullet \xrightarrow{\sim} \mathsf{I}_\bullet\),
\begin{equation} \label{eqn:KWDiagram}
\kappa_\bullet = \left(
\begin{tikzpicture}[scale=0.5, baseline={([yshift=-.5ex]current bounding box.center)}]
  \node (C1)     at (2, 0)    {\small $0$};
  \node (C0)     at (4, 0)    {\small $\mathsf{I}^0$};
  \node (Cm1)    at (6, 0)    {\small $\mathsf{I}^{1}$};

  \node (D1)     at (2, -2.2) {\small $\mathsf{I}_1$};
  \node (D0)     at (4, -2.2) {\small $\mathsf{I}_0$};
  \node (Dm1)    at (6, -2.2) {\small $0$};

  \draw[-stealth] (C1)     --  (C0);
  \draw[-stealth] (C0)     -- node[above]{\small $\exd$} (Cm1);

  \draw[-stealth] (D1)     -- node[below]{\small $\partial$} (D0);
  \draw[-stealth] (D0)     -- (Dm1);

  \draw[-stealth] (C0)     -- node[pos=0.4, left]{\small $\kappa_0$}    (D1);
  \draw[-stealth] (Cm1)    -- node[pos=0.4, left]{\small $\kappa_{1}$} (D0);
\end{tikzpicture}
\right)
\end{equation}
given by
\begin{equation}\label{eqn:KWDuality}
    \kappa_0(\mathbf{v}_i) = e_i,\quad
    \kappa_1(\mathbf{e}_i) = v_{i+1}.
\end{equation}
That is, \(\kappa_\bullet\) sends every dual vertex to the edge to its left, and similarly for dual edges and the vertices to their left.
Via this isomorphism, we say that dual vertices are edge-like, in that applying \(\exd\) to a dual vertex gives a sum of two basis elements from the component to the right.
Likewise, dual edges are vertex-like, as implied by Eq.~\eqref{eqn:KWDuality}.

The Ising model has only one non-trivial homology (logical) class in each of \(H_1(\mathsf{I})\) and \(H_0(\mathsf{I})\), and similarly in \(H^1(\mathsf{I})\) and \(H^0(\mathsf{I})\).
The cycles (logicals) of \(\mathsf{I}_\bullet\) are line-like in the \(H_1(\mathsf{I})\) component, and point-like and mobile in \(H_0(\mathsf{I})\).
We denote the Ising logical class in \(H_i(\mathsf{I})\) by \(l_i\) and that in \(H^i(\mathsf{I})\) by \(\mathbf{{l}}^i\).
Explicitly, we have
\begin{alignat}{3}
    l_{1} &\ni \sum_{e \in \mathsf{B}(\mathsf{I}_1)} e, \qquad & l_0 & \ni w,\\
    \mathbf{l}^1 &\ni \mathbf{f}, \qquad & \mathbf{l}^0 &\ni \sum_{\mathbf{v} \in \mathsf{B}(\mathsf{I}^0)} \mathbf{v},
\end{alignat}
where \(w \in \mathsf{B}(\mathsf{I}_0)\) and \(\mathbf{f} \in \mathsf{B}(\mathsf{I}^1)\) are an arbitrary point and edge in the Ising model.
The equivalence class contains all such choices.
Geometrically, these (co)cycles are represented by lines of edges (vertices), or single isolated vertices (edges).
\begin{subequations}\label{eqn:IsingCyclesTikz}
\begin{align}
    l_1 &\ni
    \begin{tikzpicture}[scale = .8, baseline={([yshift=-.5ex]current bounding box.center)}]
    \foreach \i in {-1, ..., 3}{
        \draw[color = red, line width=1pt] (\i, 0) -- (\i + 1, 0) ;
    }
    \foreach \i in {0, ..., 3}{
        \filldraw[color = black] (\i, 0) circle (2 pt);
    }
    \node at (4.5, 0) {$\cdots$};
    \node at (-1.5, 0) {$\cdots$};
    \end{tikzpicture}
    \\
    l_0 &\ni
    \begin{tikzpicture}[scale = .8, baseline={([yshift=-.5ex]current bounding box.center)}]
    \foreach \i in {-1, ..., 3}{
        \draw[color = black] (\i, 0) -- (\i + 1, 0) ;
    }
    \foreach \i in {1, ..., 3}{
        \filldraw[color = black] (\i, 0) circle (2 pt);
    }
    \filldraw[color = red] (0, 0) circle (2.5 pt);
    \node at (4.5, 0) {$\cdots$};
    \node at (-1.5, 0) {$\cdots$};
    \end{tikzpicture}
    \\
    \mathbf{l}^1 &\ni
    \begin{tikzpicture}[scale = .8, baseline={([yshift=-.5ex]current bounding box.center)}]
    \draw[color = red, line width=1pt] (0, 0) -- (1, 0) ;
    \draw[color = black] (-1, 0) -- (0, 0) ;
    \foreach \i in {1, ..., 3}{
        \draw[color = black] (\i, 0) -- (\i + 1, 0) ;
    }
    \foreach \i in {0, ..., 3}{
        \filldraw[color = black] (\i, 0) circle (2 pt);
    }
    \node at (4.5, 0) {$\cdots$};
    \node at (-1.5, 0) {$\cdots$};
    \end{tikzpicture}
    \\
    \mathbf{l}^0 &\ni
    \begin{tikzpicture}[scale = .8, baseline={([yshift=-.5ex]current bounding box.center)}]
    \foreach \i in {-1, ..., 3}{
        \draw[color = black] (\i, 0) -- (\i + 1, 0) ;
    }
    \foreach \i in {0, ..., 3}{
        \filldraw[color = red] (\i, 0) circle (2.5 pt);
    }
    \node at (4.5, 0) {$\cdots$};
    \node at (-1.5, 0) {$\cdots$};
    \end{tikzpicture}
\end{align}
\end{subequations}
In keeping with our observation that dual vertices are edge-like, the 0-cocycle of \(\mathbf{l}^0\) is an extended line-like object.

With these conventions, the tensor product defining the toric code has nonzero components \(\mathsf{TC}_{1,0,-1}\),
\begin{equation}
    \mathsf{TC}_\bullet = \left( \mathsf{I}_{x1} \otimes \mathsf{I}_y^0 \rightarrow \mathsf{I}_{x0} \otimes \mathsf{I}_y^0 \oplus \mathsf{I}_{x1} \otimes \mathsf{I}_y^1 \rightarrow \mathsf{I}_{x0} \otimes \mathsf{I}_y^1 \right).
\end{equation}
Geometrically, this complex corresponds to a two-dimensional square lattice, as expected.
Algebraically, the only difference with the usual presentation of the toric code is that all components are shifted down by one.
Indeed, basis elements of \(\mathsf{TC}_1 = \mathsf{I}_{x1} \otimes \mathsf{I}_y^0\) are a product of two edge-like basis elements \(e \otimes \mathbf{v}\).
As such, these can be identified as plaquettes, and their boundaries consist of four \(\mathsf{B}(\mathsf{TC}_0)\) elements.
\begin{align} 
    \begin{tikzpicture}[scale = .8, baseline={([yshift=-.5ex]current bounding box.center)}]
        \filldraw[color = red, opacity=0.3] (0, 0) rectangle (1, 1);
        \foreach \i in {0, 1}{
        \draw[color = black] (-0.5, \i) -- (1.5, \i) ;
        }
        \foreach \i in {0, 1}{
        \draw[color = black] (\i,-0.5) -- (\i, 1.5) ;
        }
    \end{tikzpicture}
    &\to
    \begin{tikzpicture}[scale = .8, baseline={([yshift=-.5ex]current bounding box.center)}]
        \foreach \i in {0, 1}{
        \draw[color = black] (-0.5, \i) -- (1.5, \i) ;
        }
        \foreach \i in {0, 1}{
        \draw[color = black] (\i,-0.5) -- (\i, 1.5) ;
        }
        \foreach \i in {0, 1}{
        \draw[color = red, line width=1 pt] (0, \i) -- (1, \i) ;
        }
        \foreach \i in {0, 1}{
        \draw[color = red, line width=1 pt] (\i, 0) -- (\i, 1) ;
        }
    \end{tikzpicture} 
\end{align}
As an equation,
\begin{equation}
    e_i \otimes \mathbf{v}_j \to v_i \otimes \mathbf{v}_j + v_{i+1} \otimes \mathbf{v}_j + e_i \otimes \mathbf{e}_{j-1} + e_i \otimes \mathbf{e}_j.
\end{equation}
Similarly, the boundary of an edge-like chain \(v \otimes \mathbf{v}'\) or \(e \otimes \mathbf{e}'\) in \(\mathsf{B}(\mathsf{TC}_0)\) consits of two vertex-like points in \(\mathsf{B}(\mathsf{TC}_{-1})\).
\begin{equation}
    \begin{tikzpicture}[scale = .8, baseline={([yshift=-.5ex]current bounding box.center)}]
        \foreach \i in {0, 1}{
        \draw[color = black] (-0.5, \i) -- (1.5, \i) ;
        }
        \foreach \i in {0, 1}{
        \draw[color = black] (\i,-0.5) -- (\i, 1.5) ;
        }
        \draw[color = red, line width=1 pt] (0, 0) -- (1, 0) ;
    \end{tikzpicture}
    \to
    \begin{tikzpicture}[scale = .8, baseline={([yshift=-.5ex]current bounding box.center)}]
        \foreach \i in {0, 1}{
        \draw[color = black] (-0.5, \i) -- (1.5, \i) ;
        }
        \foreach \i in {0, 1}{
        \draw[color = black] (\i,-0.5) -- (\i, 1.5) ;
        }
        \foreach \i in {0, 1}{
        \filldraw[color = red] (\i, 0) circle (2.5 pt);
        }
    \end{tikzpicture}
\end{equation}
Again as an equation,
\begin{equation}
    e_i \otimes \mathbf{e}_j \to v_i \otimes \mathbf{e}_j + v_{i+1} \otimes \mathbf{e}_j.
\end{equation}
By organizing the components of the Ising chain complex and its dual as we have, the qubits of the CSS realization of the toric code naturally land at component \(0\), conforming with our convention.

The logical content of the toric code is specified by the homology \(H_0(\mathsf{TC})\).
A basis for \(H_0(\mathsf{TC})\)---encoding \(\bar{Z}\) logicals of the toric code---is obtained from products of classes from the Ising factors:
\begin{equation}
    L_x = l_{x1} \otimes \mathbf{l}^1_y,
    \quad
    L_y = l_{x0} \otimes \mathbf{l}^0_y.
\end{equation}
Here, \(x\) and \(y\) subscripts indicate the \(\mathsf{I}_{x\bullet} \otimes  \mathsf{I}_y^\bullet\) Ising factors, which geometrically correspond to the \(x\) and \(y\) axes of the torus. 
The cycles representing \(l_{1}\) and \(\mathbf{l}^{0}\) are extended line-like objects, while \(\mathbf{l}^{1}\) and \(l_{0}\) have pointline representatives.
Thus, \(L_x\) represents a line-like \(\bar{Z}_x\) logical that is extended in the \(x\) direction, and similarly for \(L_y\) in the \(y\) direction.
We label qubits by these \(\bar{Z}\) operators, so that
the conjugate \(\bar{X}_x\) logical, represented by the dual cohomology class \(\mathbf{L}_x = \mathbf{l}_x^1 \otimes l_{y1}\), is extended in the \(y\) direction.

\subsubsection{Interblock Gates}

\begin{figure}
    \centering
    \includegraphics[width=\linewidth]{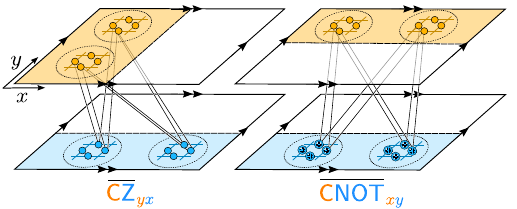}
    \caption{\textbf{Logical \(\CZLogicalColor{y}{x}\) and \(\CNOTLogicalColor{x}{y}\) gates in the toric code.} By acting with \(\mathsf{CZ}\) or \(\mathsf{CNOT}\) gates on only one half of the toric code, it is possible to perform an addressable \(\CZLogical{}{}\) (left) or \(\CNOTLogical{}{}\) (right) gate between only one pair of logical qubits.}
    \label{fig:2TC_CZyx}
\end{figure}

Continuing to label qubits by \(x\) and \(y\) and two code blocks by orange and blue, the four possible addressable \(\CZLogical{}{}\) gates are:
\begin{equation}
    \CZLogicalColor{x}{x}, \quad
    \CZLogicalColor{x}{y}, \quad
    \CZLogicalColor{y}{x}, \quad
    \CZLogicalColor{y}{y}.
\end{equation}
Our formalism finds transversal implementations of all of these, and similarly for \(\CNOTLogicalColor{}{}\).
We describe a representative set of these below.

\paragraph{\texorpdfstring{\(\CZLogicalColor{y}{x}\)}{CZyx} gate}
The first application of the recipe is to an addressable logical \(\CZLogicalColor{y}{x}\) between two code blocks of the toric code, represented by orange and blue colors.
One realization of this gate has already been discussed in Sec.~\ref{subsec:Intuition_Application}, and is shown in Fig.~\ref{fig:Addressable_CZ}(a).
We find an inequivalent implementation of the same logical action, shown in Fig.~\ref{fig:2TC_CZyx}.

Following the steps of the recipe in order, we first specify the logical action. 
The logical action of \(\CZLogicalColor{y}{x}\) on orange \(\textcolor{orange}{\bar{X}}\) operators is
\begin{equation}
    [\CZLogicalColor{y}{x}, \textcolor{orange}{\bar{X}_x}]_{\mathrm{grp}}= 1,
    \quad
    [\CZLogicalColor{y}{x}, \textcolor{orange}{\bar{X}_y}]_{\mathrm{grp}} = \textcolor{dodgerblue}{\bar{Z}_x},
\end{equation}
which, as an action on orange cohomology classes, is equivalent to
\begin{equation}\label{eqn:TCCZyxLogicalAction}
    \textcolor{orange}{\mathbf{L}_x} \mapsto 0, \quad
    \textcolor{orange}{\mathbf{L}_y} \mapsto \textcolor{dodgerblue}{L_x}.
\end{equation}
Under the isomorphism \(\Hom[H^0(\textcolor{orange}{\mathsf{TC}}), H_0(\textcolor{dodgerblue}{\mathsf{TC}})] \isom H_0(\textcolor{dodgerblue}{\mathsf{TC}}) \otimes H_0(\textcolor{orange}{\mathsf{TC}})\), 
Eq.~\eqref{eqn:TCCZyxLogicalAction} corresponds to the logical tensor
\begin{equation}\label{eqn:TCCZyxLogicalCLass}
    (\phii_{*})_0 = \textcolor{dodgerblue}{L_x} \otimes \textcolor{orange}{L_y}
    = \textcolor{dodgerblue}{l_{x1} \otimes \mathbf{l}^1_y} \otimes \textcolor{orange}{l_{x0} \otimes \mathbf{l}_y^0}.
\end{equation}
The next step of the recipe is to specify the rest of the homological action.
It is valid to have the rest of the homological action be trivial, so we set \((\phii_{*})_0 = \phii_*\).

The second, and most important, stage of the recipe is to find a sparse representative of the homological action.
As the toric code is itself a tensor product, one can attempt to find sparse gates between the factor codes and simply multiply them.
To this end, we reshape the Ising logical factors in Eq.~\eqref{eqn:TCCZyxLogicalCLass} into pairs of orange and blue Ising logicals that are highly deformable, and so easy to sparsify.
In particular, there is the following pairing of orange with blue such that each pair represents a deformable line-like logical,
\begin{equation}
    \phii_*
    = (\textcolor{dodgerblue}{l_{x1}} \otimes \textcolor{orange}{l_{x0}}) \otimes ( \textcolor{dodgerblue}{\mathbf{l}^1_y} \otimes \textcolor{orange}{\mathbf{l}_y^0} ) \in [\textcolor{orange}{\mathsf{I}_{x}^\bullet}, \textcolor{dodgerblue}{\mathsf{I}_{x\bullet}}] \otimes [\textcolor{orange}{\mathsf{I}_{y\bullet}}, \textcolor{dodgerblue}{\mathsf{I}_{y}^\bullet}].
\end{equation}

Now it is possible to identify representatives \(\phi^{(x)}_\bullet \in \textcolor{dodgerblue}{l_{x1}} \otimes \textcolor{orange}{l_{x0}}\) and \(\phi^{(y)}_\bullet \in \textcolor{dodgerblue}{\mathbf{l}^1_y} \otimes \textcolor{orange}{\mathbf{l}_y^0} \) which are sparse.
For simple chain complexes like the Ising model, we can rely on geometric intuition to find these representatives.
Indeed, the chain map complexes \([\textcolor{orange}{\mathsf{I}_{x}^\bullet}, \textcolor{dodgerblue}{\mathsf{I}_{x\bullet}}] \isom \mathsf{TC}_\bullet[-1]\) and \([\textcolor{orange}{\mathsf{I}_{y\bullet}}, \textcolor{dodgerblue}{\mathsf{I}_{y}^\bullet}] \isom \mathsf{TC}_\bullet[1]\) are just (degree-shifted) toric codes,
and the desired \(\phi^{(x)}_\bullet\) and \(\phi^{(y)}_\bullet\) are line-like cycles therein.
A poor candidate cycle would be the flat toric code logicals obtained as tensor products of the Ising model cycles from Eq.~\eqref{eqn:IsingCyclesTikz}, which we denote \(\td\phi^{(x)}_\bullet\) and \(\td\phi^{(y)}_\bullet\).
Geometrically, these cycles can be represented as follows:
\begin{equation}\label{eqn:2TCSGateCyclesFlat}
    \includegraphics[valign=c]{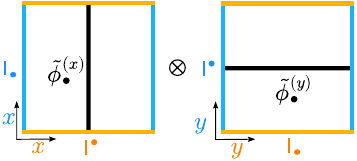}.
\end{equation}

Each of \(\td\phi^{(x)}_\bullet\) and \(\td\phi^{(y)}_\bullet\) are a function from an orange Ising model to a blue Ising model, and the cycles in 
Eq.~\eqref{eqn:2TCSGateCyclesFlat} can be interpreted as a plot of these functions. 
Indeed, \(\td\phi^{(x)}_\bullet\) corresponds to the tensor
\begin{equation}
    \td\phi^{(x)}_\bullet = \sum_{\textcolor{dodgerblue}{e} \in \mathsf{B}(\textcolor{dodgerblue}{\mathsf{I}_1})} \textcolor{dodgerblue}{e} \otimes \textcolor{orange}{w},
\end{equation}
where \(\textcolor{orange}{w} \in \mathsf{B}(\textcolor{orange}{\mathsf{I}_0})\) is arbitrary and fixed.
To evaluate \(\td\phi^{(x)}_\bullet\) as a function we contract this tensor with some orange chain, say a basis element,
\begin{equation}
    \td\phi^{(x)}_\bullet(\mathbf{\textcolor{orange}{v}}) = \delta_{\textcolor{orange}{w v}} \sum_{\textcolor{dodgerblue}{e} \in \mathsf{B}(\textcolor{dodgerblue}{\mathsf{I}_1})} \textcolor{dodgerblue}{e}.
\end{equation}
We observe that the image of \(\mathbf{\textcolor{orange}{v}}\) under \(\td\phi^{(x)}_\bullet\) is exactly the plot in Eq.~\eqref{eqn:2TCSGateCyclesFlat}.
The same holds for \(\td\phi^{(y)}_\bullet\).
We can also see that neither \(\td\phi^{(x)}_\bullet\) nor \(\td\phi^{(y)}_\bullet\) are sparse: in \(\td\phi^{(x)}_\bullet\), there are single points that map to entire lines, and \(\td\phi^{(y)}_\bullet\) maps every edge to the same point.

Sparsifying \(\td\phi^{(x)}_\bullet\) requires deforming the cycles in Eq.~\eqref{eqn:2TCSGateCyclesFlat} such that any horizontal or vertical line intersects the plots at only finitely many points.
The following \emph{tent-shaped} deformations accomplish this, completing step \hyperref[par:sparsify]{(2)} of the recipe. 
\begin{equation}\label{eqn:2TCSGateCyclesTent}
    \includegraphics[valign=c]{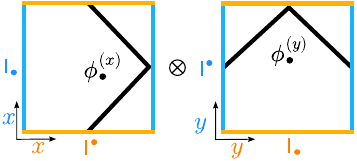}
\end{equation}
While we have drawn these cycles as smooth, they are actually composed of discrete steps in the square lattice.
For completeness, we state the explicit formula corresponding to these chain maps.
Define a \emph{tent map}
\(\tau_\bullet : \mathsf{I}^\bullet \to \mathsf{I}_\bullet\) for odd distances \(d\),
\begin{subequations}
\begin{alignat}{2}
    \tau_1(\mathbf{e}_i) &= \tau_1(\mathbf{e}_{d-1-i}) = v_i \quad& (0\leq i \leq \tfrac{d-1}{2}), \\
    \tau_0(\mathbf{v}_i) &= \tau_0(\mathbf{v}_{d-i}) = e_{i-1} \quad&(1\leq i \leq \tfrac{d-1}{2}),
\end{alignat}
where \(\tau_0(\mathbf{v}_0) = 0\).
For even distance,
\begin{alignat}{2}
    \tau_1(\mathbf{e}_i) &= \tau_1(\mathbf{e}_{d-i}) = v_i \quad& (0\leq i \leq d/2), \\
    \tau_0(\mathbf{v}_i) &= \tau_0(\mathbf{v}_{d+1-i}) = e_{i-1} \quad& (1\leq i \leq d/2),
\end{alignat}
\end{subequations}
where subscripts are modulo \(d\), so that \(\tau_0(\mathbf{v}_0) = \tau_0(\mathbf{v}_{1}) = e_{0}\).
Also using the self-duality \(\kappa_\bullet\), we have
\begin{equation}
    \phi^{(x)}_\bullet =  \tau^{\mathsf{T}}_\bullet, \quad \text{and}\quad
    \phi^{(y)}_\bullet = \kappa_\bullet^{-1} \tau_\bullet \kappa_\bullet^{-1}.
\end{equation}

Third and finally, a circuit of physical \(\mathsf{\textcolor{orange}{C}\textcolor{dodgerblue}{Z}}\) gates is synthesized from \(\phii_\bullet\) through Eq.~\eqref{eq:productofCZ}.
The resulting \(\overline{\mathsf{\textcolor{orange}{C}\textcolor{dodgerblue}{Z}}}_{\textcolor{orange}{y},\textcolor{dodgerblue}{x}}\) circuit is depicted in Fig.~\ref{fig:2TC_CZyx}.

With the method established and illustrated in this example, further toric code gates are presented with less detail.

\paragraph{Another \texorpdfstring{\(\CZLogicalColor{y}{x}\)}{CZyx} gate}
Note that Fig.~\ref{fig:2TC_CZyx} is not the same circuit as in Fig.~\ref{fig:Addressable_CZ}.
The circuit in Fig.~\ref{fig:Addressable_CZ} is obtained through the same recipe, but making a different choice in the extension of the logical action to a homological action.
If one chooses
\begin{equation}\label{eqn:TCCZyxLogicalCLass2}
    \phii_* = \textcolor{dodgerblue}{l_{x1} \otimes \mathbf{l}^1_y} \otimes \textcolor{orange}{l_{x0} \otimes \mathbf{l}_y^0} + \textcolor{dodgerblue}{l_{x1} \otimes \mathbf{l}^0_y} \otimes \textcolor{orange}{l_{x0} \otimes \mathbf{l}_y^1},
\end{equation}
then there is a reshaping and sparsification
\begin{equation}\label{eqn:TCCZyxLogicalCLass2}
    \phii_* = (\textcolor{dodgerblue}{l_{x1}} \otimes \textcolor{orange}{l_{x0}}) \otimes (\textcolor{dodgerblue}{\mathbf{l}^0_y}\otimes \textcolor{orange}{\mathbf{l}_y^1} + \textcolor{dodgerblue}{\mathbf{l}^1_y} \otimes \textcolor{orange}{\mathbf{l}_y^0}) \ni \tau^{\transpose}_\bullet \otimes \kappa^{-1}_\bullet.
\end{equation}
This map has the same logical action as Fig.~\ref{fig:2TC_CZyx}, but has a different action on the homology group \(H_{-1}(\textcolor{orange}{\mathsf{TC}})\).

\paragraph{\texorpdfstring{\(\CZLogicalColor{x}{x}\)}{CZxx} gate}
Starting with a logical action \(\textcolor{dodgerblue}{L_x} \otimes \textcolor{orange}{L_x} \) (or \( \textcolor{dodgerblue}{L_y} \otimes \textcolor{orange}{L_y} \)) yields a \(\CZLogicalColor{x}{x}\) gate (a \(\CZLogicalColor{y}{y}\) gate).
A reshaping into sparsifiable factors is
\begin{multline}
    \textcolor{dodgerblue}{L_x} \otimes \textcolor{orange}{L_x} 
    = (\textcolor{dodgerblue}{\mathbf{l}^1_y} \otimes \textcolor{orange}{l_{x1}}) \otimes (\textcolor{dodgerblue}{l_{x1}} \otimes \textcolor{orange}{\mathbf{l}_y^1} ) \in [\textcolor{orange}{\mathsf{I}_{x}^\bullet}, \textcolor{dodgerblue}{\mathsf{I}_{y}^\bullet}] \otimes [\textcolor{orange}{\mathsf{I}_{y\bullet}},  \textcolor{dodgerblue}{\mathsf{I}_{x\bullet}}],
\end{multline}
with a particular sparsification again being a tent shape.
\begin{equation}\label{eqn:2TCCZCyclesSymmetric}
    \includegraphics[valign=c]{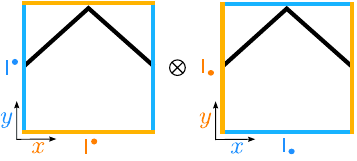}
\end{equation}
The gate is explicitly (suppressing bullets on the right hand side)
\begin{equation}\label{eqn:2TCCZxxFormula}
    \phii_\bullet = \kappa^{-1} \tau  \otimes  \tau^\mathsf{T} \kappa^{-\transpose},
\end{equation}
where the first factor maps \(\textcolor{orange}{x} \to \textcolor{dodgerblue}{y}\) and the second maps \(\textcolor{orange}{y} \to \textcolor{dodgerblue}{x}\).

\paragraph{\texorpdfstring{\(\CNOTLogicalColor{x}{y}\)}{CNOTxy} and \texorpdfstring{\(\CNOTLogicalColor{x}{x}\)}{CNOTxx} gate}
We could obtain a logical \(\CNOTLogical{}{}\) by applying the recipe again, starting from a logical like \(\textcolor{orange}{L_x} \otimes \textcolor{dodgerblue}{\mathbf{L}_y}\) for a \(\CNOTLogicalColor{x}{y}\) gate.
However, it is easier to use the self-duality isomorphism \(\kappa^{-1} \otimes \kappa :\textcolor{dodgerblue}{\mathsf{TC}_\bullet} \xrightarrow{\sim} \textcolor{dodgerblue}{\mathsf{TC}^\bullet}\) to map the target of the chain maps we already found to cochains:
\begin{equation}
    \textcolor{orange}{\mathsf{TC}^\bullet} \to \textcolor{dodgerblue}{\mathsf{TC}_\bullet}\xrightarrow{\sim} \textcolor{dodgerblue}{\mathsf{TC}^\bullet}.
\end{equation}
This composition then maps  \(\textcolor{orange}{\mathsf{TC}^\bullet} \to \textcolor{dodgerblue}{\mathsf{TC}^\bullet}\), and so defines a \(\mathsf{CNOT}\) gate.
Physically, this procedure corresponds to conjugating the \(\textcolor{dodgerblue}{\mathsf{TC}_\bullet}\) side of the \(\mathsf{CZ}\) circuit by Hadamards and making a half-unit-cell translation of the target of all gates.
Such a gate is shown in Fig.~\ref{fig:2TC_CZyx}.

Explicitly, a chain map corresponding to a \(\CNOTLogicalColor{x}{y}\) gate is
\begin{equation}
    \phii_\bullet =  \tau  \otimes  \kappa^{-1}\tau^\mathsf{T} \kappa^{-1},
\end{equation}
where the first factor maps \(\textcolor{orange}{x} \to \textcolor{dodgerblue}{y}\) and the second maps \(\textcolor{orange}{y} \to \textcolor{dodgerblue}{x}\).

\subsubsection{Intra-block Gates}

As discussed in Section~\ref{sec:IntrablockGates}, chain maps can also specify commuting intra-block gates, acting on only one copy of the toric code.
However, in our construction, these chain maps must obey some additional constraints, depending on whether they describe \(\CZLogical{}{}\), \(\bar{\mathsf{S}}\), or \(\CNOTLogical{}{}\) gates.
Here, we describe intra-block gates of these kinds, and also the addressable Hadamard from Section~\ref{subsubsec:2TCExhaustiveness}.

\paragraph{\texorpdfstring{\(\CZLogical{x}{y}\)}{CZxy} gate}
Intra-block diagonal Clifford gates (\(\CZLogical{}{}\) and \(\bar{\mathsf{S}}\)) are obtained from sparse, symmetric \(\ZZ_4\)-valued chain maps \(\widetilde{\phii}_\bullet\) acting on a sparse, nice lift of the chain complex \(\mathsf{TC}_\bullet\).
However, for \(\CZLogical{}{}\) gates, we can essentially ignore the \(\ZZ_4\) lift and instead just directly use the chain map \(\phii_\bullet\) for \(\CZLogicalColor{y}{x}\) that we already found, via Eq.~\eqref{eq-intrablockClZ}.
The resulting logical action becomes symmetrized (up to a \(\bar{Z}\) logical).
For the case of our \(\phii_\bullet\), we get a symmetrized logical action
\begin{equation}
    \mathbf{L}_x \mapsto L_y, \quad \mathbf{L}_y \mapsto L_x.
\end{equation}
This corresponds to an intra-block \(\CZLogical{x}{y}\) gate, as expected.

The procedure to construct the corresponding physical gate is to move the controls of the physical \(\mathsf{\textcolor{orange}{C}\textcolor{dodgerblue}{Z}}\) gates in the \(\CZLogicalColor{y}{x}\) circuit down to the blue target block, and replace gates by \(Z\) operators when controls and targets collide.
Such an intra-block logical \(\overline{\mathsf{CZ}}_{x,y}\) gate is shown in Fig.~\ref{fig:Addressable_CZ}(b).

\paragraph{\texorpdfstring{\(\bar{\mathsf{S}}_x\)}{Sx} gate}
Lifting the \(\CZLogicalColor{x}{x}\) gate to an \(\bar{\mathsf{S}}_x\) requires paying attention to the \(\ZZ_4\) lift.
However, as we have constructed the chain map corresponding to \(\CZLogicalColor{x}{x}\) from a tensor product of chain maps on classical codes that have obvious lifts, finding the \(\ZZ_4\) version of the chain map reduces to repeating the calculations we already made, but now using the \(\ZZ_4\) Ising model and \(\ZZ_4\) toric code.
Additionally, we must ensure that the chain map we find is symmetric.

The \(\ZZ_4\) Ising model \(\tilde{\mathsf{I}}_\bullet\) has the same basis of edges \(\tilde{e}_i \in \mathsf{B}(\tilde{\mathsf{I}}_1)\) and vertices \(\tilde{v}_i \in \mathsf{B}(\tilde{\mathsf{I}}_0)\), and a boundary map \(\partial\tilde{e}_i = \tilde{v}_{i+1} - \tilde{v}_i\).
The corresponding coboundary map is \(\exd \tilde{\mathbf{v}}_i = \tilde{\mathbf{e}}_{i-1} - \tilde{\mathbf{e}}_{i}\).
We have a degree-1 self-duality \(\tilde{\kappa}_\bullet : \tilde{\mathsf{I}}^\bullet \xrightarrow{\sim} \tilde{\mathsf{I}}_\bullet\),
\begin{equation}
    \tilde{\kappa}_0(\tilde{\mathbf{v}}_i) = \tilde{e}_i,
    \quad
    \tilde{\kappa}_1(\tilde{\mathbf{e}}_i) = -\tilde{v}_{i+1},
\end{equation}
and tent map \(\tilde{\tau}_\bullet:\tilde{\mathsf{I}}^\bullet \to \tilde{\mathsf{I}}_\bullet\)
for odd distances \(d\),
\begin{subequations}
\begin{alignat}{2}
    \td\tau_1(\td{\mathbf{e}}_i) &= \td\tau_1(\td{\mathbf{e}}_{d-1-i}) = \td{v}_i \quad& (0\leq i \leq \tfrac{d-1}{2}), \\
    \td\tau_0(\td{\mathbf{v}}_i) &= -\td\tau_0(\td{\mathbf{v}}_{d-i}) = -\td{e}_{i-1} \quad& (1\leq i \leq \tfrac{d-1}{2}),
\end{alignat}
and even distance
\begin{alignat}{2}
    \td{\tau}_1(\td{\mathbf{e}}_i) &= \td{\tau}_1(\td{\mathbf{e}}_{d-i}) = \td{v}_i \quad& (0\leq i \leq d/2), \\
    \td{\tau}_0(\td{\mathbf{v}}_i) &= -\td{\tau}_0(\td{\mathbf{v}}_{d+1-i}) = -\td{e}_{i-1} \quad& (1\leq i \leq d/2).
\end{alignat}
\end{subequations}
As before, \(\tilde{\tau}\) only acts nontrivially in the \(H^1(\tilde{\mathsf{I}}) \to H_0(\tilde{\mathsf{I}})\) component of (co)homology.

The tensor product \(\widetilde{\mathsf{TC}}_\bullet = \tilde{\mathsf{I}}_\bullet \otimes \tilde{\mathsf{I}}^\bullet\) is the \(\ZZ_4\) toric code.
All the chain maps we previously built from \(\kappa\) and \(\tau\) for the \(\ZZ_2\) toric code lift to \(\ZZ_4\) just by replacing factors by \(\tilde{\kappa}\) and \(\tilde{\tau}\), and using the \(\ZZ_4\) tensor product (Appendix~\ref{app:CMCDetails}).

In particular, to arrive at an \(\bar{\mathsf{S}}_x\) gate, we should consider the map \(\phii_\bullet\) from Eq.~\eqref{eqn:2TCCZCyclesSymmetric}, which encoded a \(\CZLogicalColor{x}{x}\).
Note that, in Eq.~\eqref{eqn:2TCCZxxFormula}, the \(\kappa^{-1}\tau\) factor maps \(\mathsf{I}_x^{\bullet} \to \mathsf{I}_y^{\bullet}\) while the \(\tau^\transpose\kappa^{-\transpose}\) factor maps \(\mathsf{I}_{y\bullet} \to \mathsf{I}_{x\bullet}\).
When acting on a single code it is helpful to write the lift \(\widetilde{\phii}_\bullet\) as
\begin{equation}\label{eqn:2TCSGateChainMap}
    \widetilde{\phii}_\bullet = \sigma\circ(\td{\kappa}^{-1} \td{\tau}  \otimes  \td{\tau}^\mathsf{T} \td{\kappa}^{-\transpose}),
\end{equation}
where \(\sigma\) swaps the \(\ZZ_4\)-Ising \(x\) and \(y\) factors. This way, the Ising factors continue to be ordered \((x,y)\) in the output.

This \(\widetilde{\phii}_\bullet\) has the correct action on homology to be an \(\bar{\mathsf{S}}_x\) gate. by the same calculation that showed that \(\phii_\bullet\) was a \(\CZLogicalColor{x}{x}\) gate.
It remains to show that \(\widetilde{\phii}_\bullet = \widetilde{\phii}_\bullet^\transpose\) is symmetric.
Indeed, this is straightforward:
\begin{equation}
    \widetilde{\phii}_\bullet = \sigma \circ(\td{\kappa}^{-1} \td{\tau}  \otimes  \td{\tau}^\mathsf{T} \td{\kappa}^{-\transpose}) = (\td{\tau}^\mathsf{T} \td{\kappa}^{-\transpose} \otimes  \td{\kappa}^{-1} \td{\tau}  ) \circ \sigma = \widetilde{\phii}_\bullet^\transpose.
\end{equation}
We used that \(\sigma\circ (\alpha \otimes \beta) = (\beta \otimes \alpha) \circ \sigma\).
Graphically, this is represented as follows, where \(\bar{\kappa} \equiv \td{\kappa}^{-1}\).
\begin{equation}
    \includegraphics[valign=c]{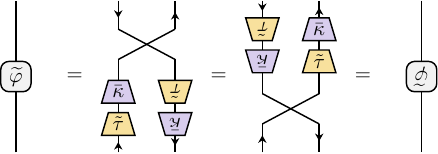}
\end{equation}

Thus, an \(\bar{\mathsf{S}}_x\) gate is immediately obtained as a circuit of \(\mathsf{CZ}\) and \(\mathsf{S}^{(\dagger)}\) gates via Eq.~\eqref{eq-intrablockdiaggate}, illustrated as follows.
\begin{equation}\label{eqn:2TCSIllustration}
    \includegraphics[valign=c]{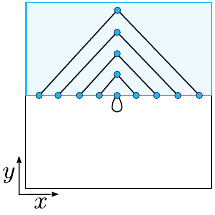}
\end{equation}
(Many gates have not been illustrated for clarity.)
\(\mathsf{S}^{(\dagger)}\) gates occur along an inverted-tent shaped path within the blue region, with only the \(\mathsf{S}\) at tip of the tent being drawn as a self-targetting \(\mathsf{CZ}\).
Gates on \(x\)-aligned segments of the path are \(\mathsf{S}\), while gates on \(y\)-aligned segments are \(\mathsf{S}^\dagger\).
This pattern of \(\mathsf{S}^{(\dagger)}\) is specified from the pattern of \(\pm1\) in \(\td\tau_\bullet\) and \(\td\kappa_\bullet\), and ensures that a stabilizer \(X_{\exd \mathbf{v}}\) is mapped to a stabilizer \(+Z_{\partial b}\), rather than potentially \(-Z_{\partial b}\).

An explicit \(3\times 3\) instance of the \(\bar{\mathsf{S}}_x\) gate is shown below.
\begin{equation}
    \bar{\mathsf{S}}_x^{(3\times 3)}=\includegraphics[valign=c]{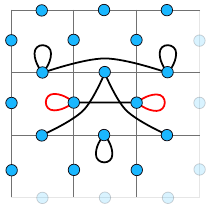}
\end{equation}
\(\mathsf{S}\) gates are again drawn as self-targeting \(\mathsf{CZ}\)s, and \(\mathsf{S}^\dagger\) is drawn in red.

Either with this explicit example, or with the general gate, it can be checked that a vertical \(\bar{X}_x\) logical triggers a line of \(\mathsf{CZ}\) gates and one \(\mathsf{S}\) gate in Eq.~\eqref{eqn:2TCSIllustration}, which together conjugate \(\bar{X}_x\) to itself multiplied by a horizontal \(\bar{Z}_x\),
\begin{equation}
    \bar{X}_x\bar{\mathsf{S}}^\dagger_x \bar{X}_x \bar{\mathsf{S}}_x = -i \bar{Z}_x.
\end{equation}
A horizontal \(\bar{X}_y\) logical can be moved outside the support of the gates, so the gate affects it trivially.

\paragraph{\texorpdfstring{\(\CNOTLogical{x}{y}\)}{CNOTxy} gate}
A sufficient condition for being able to realize a previously-interblock logical \(\CNOTLogical{x}{y}\) as an intra-block gate is that, when controls and targets are moved to the same block, no control acts on the same qubit as any target.
This is a fairly harsh requirement, but we see from inspecting Fig.~\ref{fig:2TC_CZyx} that it is almost satisfied by the \(\CNOTLogicalColor{x}{y}\) gate of the toric code.
Or more correctly, this disjointness condition is almost satisfied by a translated chain map:
\begin{equation}
    \phii_\bullet =  \sigma \circ (T_{d/2} \tau  \otimes  \kappa^{-1}\tau^\mathsf{T} \kappa^{-1}),
\end{equation}
where \(T_\delta : \mathsf{I}_{\bullet} \xrightarrow{\sim} \mathsf{I}_{\bullet}\) is a translation by \(\delta\).
Both controls and targets only act on half the system, so translating by half the distance results in nearly non-overlapping gates.
Any intersections at the boundaries can be resolved by not using a complete tent for \(\tau_\bullet\), but instead a tent with a flat top of \(O(1)\) width, as illustrated.
\begin{equation}\label{eqn:2TCCNOTTruncated}
    \includegraphics[valign=c]{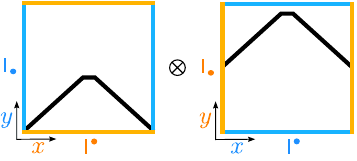}
\end{equation}
This achieves the non-overlapping condition, and allows an intra-block \(\CNOTLogical{x}{y}\) gate.

The physical circuit for the \(\CNOTLogical{x}{y}\) gate is schematically illustrated as follows (with many gates omitted).
\begin{equation}
    \includegraphics[valign=c]{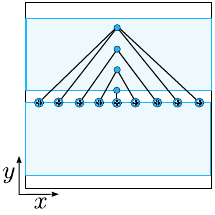}
\end{equation}
As in the \(\bar{\mathsf{S}}_x\) gate, a vertcal \(\bar{X}_x\) logical triggers a line of \(\mathsf{CNOT}\) gates, producing a perpendicular \(\bar{X}_y\) logical.
An \(\bar{X}_y\) logical can be moved into the target region, and is affected trivially.
Due to the compression of the support of the circuit through the flat top tent, some qubits control more gates than they did in the interblock \(\CNOTLogicalColor{x}{y}\) gate.

\paragraph{\texorpdfstring{\(\bar{\mathsf{H}}_x\)}{Hx} gates}
While we have not constructed a concise description of Hadamard gates as chain maps, the proof of Lemma~\ref{lem:CliffordGeneration} in Sec.~\ref{subsubsec:2TCExhaustiveness} shows how an addressable Hadamard can be constructed from an addressable \(\bar{\mathsf{S}}\) and a Hadamard-conjugated \(\bar{\mathsf{S}}\) (a gate of the signature \((-1)^{m/2}\)).
Either of these are handled by the chain map formalism, and so we also have an addressable Hadamard as a product of the synthesized physical circuits.
In particular, an addressable \(\bar{\mathsf{H}}_x\) gate is realized as follows.
\begin{equation}\label{eqn:2TC_HGate}
    \includegraphics[valign=c]{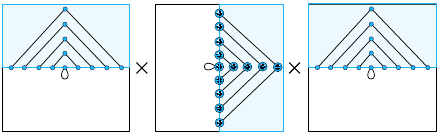}
\end{equation}
The outer gates of the product are \(\bar{\mathsf{S}}_x\) and the inner gate is \(\mathsf{\bar{S}}_x\) conjugated by the global logical Hadamard of the toric code, \(\mathsf{\bar{H}\otimes\bar{H}}\).

\subsection{Diagonal Non-Cliffords and Hadamard-free Cliffords in the 3D Toric Code}
\label{subsec:3DToricGate}

The recipe of Sec.~\ref{subsec:GateRecipe} also applies to higher levels of the chain map hierarchy.
This section applies the recipe to the three-dimensional toric code~\cite{Hamma20053dToric}, finding an addressable \(\mathsf{CCZ}\) gate, as well as Hadamard-free Cliffords (generated by \(\mathsf{S}\), \(\mathsf{CNOT}\)).

The three dimensional toric code can be expressed as a tensor product of three Ising models,
\begin{equation}
    \mathsf{3TC}_\bullet \isom \mathsf{I}_{x\bullet} \otimes \mathsf{I}_{y\bullet} \otimes \mathsf{I}_z^\bullet.
\end{equation}
When placing qubits at the degree-0 component, this pattern of chain and cochain complexes corresponds to a three-dimensional toric code with one-dimensional \(Z\) logicals.
We label logical qubits by \(\{x,y,z\}\), corresponding to the extended direction of their \(\bar{Z}\) operators, so that \(\bar{X}_\alpha\) logicals span planes orthogonal to \(\alpha\).

Note that, unlike the two-dimensional toric code, the three-dimensional toric code is not self-dual, \(\mathsf{3TC}_\bullet \not\isom \mathsf{3TC}^\bullet\).
As such, the three-dimensional toric code cannot have a transversal Hadamard.
Indeed, the \(Z\)-distance of \(\mathsf{3TC}_\bullet\) is \(d\) (the linear dimension of the three-dimensional torus), while the \(X\)-distance is \(d^2\).
Any circuit which conjugates a line-like \(\bar{Z}\) to a plane-like \(\bar{X}\) must, by the pigeonhole principle, act on a single qubit with at least \(d\) two-body gates.

\subsubsection{Interblock Gates}

\newcommand{\CCZLogicalColor}[3]{\overline{\mathsf{\textcolor{red}{C}\textcolor{orange}{C}\textcolor{dodgerblue}{Z}}}_{\textcolor{red}{#1}\textcolor{orange}{#2}\textcolor{dodgerblue}{#3}}}
\newcommand{\CCZLogical}[3]{\overline{\mathsf{CCZ}}_{#1#2#3}}

\newcommand{\CCNOTLogicalColor}[3]{\overline{\mathsf{\textcolor{red}{C}\textcolor{orange}{C}\textcolor{dodgerblue}{NOT}}}_{\textcolor{red}{#1}\textcolor{orange}{#2}\textcolor{dodgerblue}{#3}}}
\newcommand{\CCNOTLogical}[3]{\overline{\mathsf{CCNOT}}_{#1#2#3}}

\paragraph{\texorpdfstring{\(\CNOTLogicalColor{x}{z}\)}{CNOTxz} gate}
The lack of a transversal Hadamard means that we cannot derive an addressable \(\CNOTLogicalColor{x}{z}\) by conjugating a \(\CZLogicalColor{x}{z}\).
Instead, we have to find the gate directly.
Following the recipe, we start with the tensor product representation of the desired logical action,
\begin{equation}
    \textcolor{dodgerblue}{\mathbf{L}^z} \otimes \textcolor{orange}{L_x}
    = \textcolor{dodgerblue}{\mathbf{l}_{x}^0 \otimes \mathbf{l}^0_{y} \otimes l_{z0}} \otimes \textcolor{orange}{l_{x1} \otimes l_{y0} \otimes \mathbf{l}_z^1},
\end{equation}
which we regroup into deformable factors,
\begin{equation}
    \textcolor{dodgerblue}{\mathbf{L}^z} \otimes \textcolor{orange}{L_x}
    = (\textcolor{dodgerblue}{l_{z0}} \otimes \textcolor{orange}{l_{x1}} ) \otimes (\textcolor{dodgerblue}{\mathbf{l}_{y}^0} \otimes \textcolor{orange}{l_{y0}} ) \otimes (\textcolor{dodgerblue}{\mathbf{l}^0_x} \otimes \textcolor{orange}{\mathbf{l}_z^1}),
\end{equation}
and sparsify as tent-shaped deformations of these cycles.
\begin{equation}
    \includegraphics[valign=c]{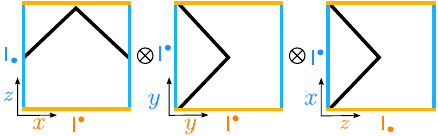}
\end{equation}
Explicitly, this map is
\begin{equation}
    \phii_\bullet = (\tau)_{\textcolor{orange}{x}\to \textcolor{dodgerblue}{z}} 
    \otimes (\kappa^{-1} \tau^\transpose)_{\textcolor{orange}{y}\to \textcolor{dodgerblue}{y}} 
    \otimes (\kappa^{-1} \tau^\transpose \kappa^{-1})_{\textcolor{orange}{z}\to \textcolor{dodgerblue}{x}}.
\end{equation}

This tensor product of Ising chain maps encodes a transversal logical \(\overline{\mathsf{\textcolor{orange}{C}\textcolor{dodgerblue}{NOT}}}_{\textcolor{orange}{x}, \textcolor{dodgerblue}{z}}\) gate between two copies of the three-dimensional toric code, as shown.
\begin{equation}\label{eqn:3TCCNOT}
    \includegraphics[valign=c]{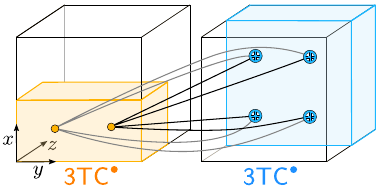}
\end{equation}
One quadrant of a plane-like \(\textcolor{orange}{\bar{X}_x}\) operator is mapped under this gate to an entire plane-like \(\textcolor{dodgerblue}{\bar{X}_z}\) operator.
We can similarly find other addressable \(\CNOTLogicalColor{\alpha}{\beta}\) gates.

\begin{figure}
    \centering
    \includegraphics[width=\linewidth]{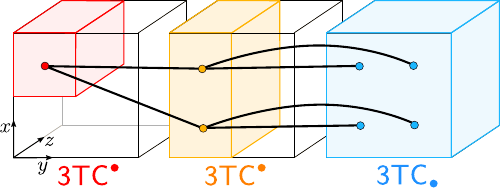}
    \caption{\textbf{Addressible \(\CCZLogicalColor{z}{x}{y}\) gate in the 3D toric code.}
    Physical \(\mathsf{CCZ}\) gates act on one \(xy\) quadrant in the red code, one \(y\) half of the orange code, and the entire blue code.
    The three qubits on which each \(\mathsf{CCZ}\) gate acts are (i)~a red qubit, (ii)~either the orange qubit in the same position as the red, or at a position reflected through the \(yz\) plane, and (iii) either a blue qubit at a position shifted by a half-unit-cell translation (in all axes) from the orange qubit, or at a position reflected through the \(zx\) plane.}
    \label{fig:3TCGate}
\end{figure}

\paragraph{\texorpdfstring{\(\CCZLogicalColor{z}{x}{y}\)}{CCZzxy} gate}
Proceeding to the second level of the chain map hierarchy, we derive the \(\overline{\mathsf{\textcolor{red}{C}\textcolor{orange}{C}\textcolor{dodgerblue}{Z}}}_{\textcolor{red}{z},\textcolor{orange}{x},\textcolor{dodgerblue}{y}}\) gate in Fig.~\ref{fig:3TCGate}.

Following the recipe, we first specify the logical action as a logical tensor. To save space, we will omit most ``\(\otimes\)'' symbols.
\begin{equation}
    (\phii_{*})_0 = \textcolor{dodgerblue}{L_x} \otimes \textcolor{orange}{L_y} \otimes \textcolor{red}{L_z} 
    = 
    \textcolor{dodgerblue}{l_{x0} \, l_{y1} \, \mathbf{l}_z^1} \otimes 
    \textcolor{orange}{l_{x1} \, l_{y0} \, \mathbf{l}_z^1} \otimes \textcolor{red}{l_{x0}\, l_{y0}\, \mathbf{l}_z^0} .
\end{equation}
While one can choose a trivial extension of \((\phii_{*})_0\) to \(\phii_{*}\), a slightly lower circuit depth can be achieved by including additional nontrivial homological actions.
Simulataneously making this extension and reshaping the result gives
\begin{align}
    \phii_* = \,&(\textcolor{dodgerblue}{l_{x0}} \, \textcolor{orange}{l_{x1}}\, \textcolor{red}{l_{x0}}+ 
    \textcolor{dodgerblue}{l_{x1}}\, \textcolor{orange}{l_{x0}}\,\textcolor{red}{l_{x0}}) \nnb \\
    &\otimes (\textcolor{dodgerblue}{l_{y1}}\,\textcolor{orange}{l_{y0}}\,\textcolor{red}{l_{y0}}) \nnb \\
    &\otimes (\textcolor{dodgerblue}{\mathbf{l}_{z}^1}\,\textcolor{orange}{\mathbf{l}_{z}^1}\,\textcolor{red}{\mathbf{l}_{z}^0}
    + \textcolor{dodgerblue}{\mathbf{l}_{z}^1}\,\textcolor{orange}{\mathbf{l}_{z}^0}\, \textcolor{red}{\mathbf{l}_{z}^1}
    + \textcolor{dodgerblue}{\mathbf{l}_{z}^0}\,\textcolor{orange}{\mathbf{l}_{z}^1}\,\textcolor{red}{\mathbf{l}_{z}^1}).
\end{align}%
These homology classes represent line-like logicals in a three-dimensional code---the three-dimensional toric code again.
They can be sparsified as
\begin{equation}\label{eqn:3TCGateFactors}
    \includegraphics[valign=c]{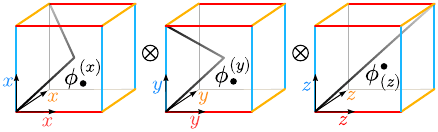}
\end{equation}
which describes the gate of Fig.~\ref{fig:3TCGate}.

\begin{figure}
    \centering
    \includegraphics[width=\linewidth]{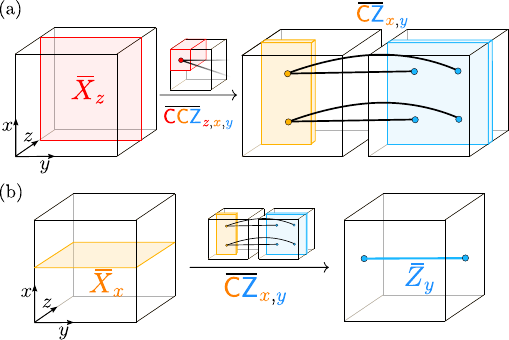}
    \caption{\textbf{Action of the \(\CCZLogicalColor{z}{x}{y}\) gate on logicals.}
    (a)~The \(\CCZLogicalColor{z}{x}{y}\) gate from Fig.~\ref{fig:3TCGate} conjugates a red \textcolor{red}{\(\bar{X}_z\)} logical (oriented in the \(xy\) plane) to itself multiplied by a \(\CZLogicalColor{x}{y}\) gate acting on the orange and blue codes.
    (b)~Indeed, this gate itself conjugates an orange \textcolor{orange}{\(\bar{X}_x\)} logical (oriented in the \(yz\) plane) to itself multiplied by a blue \textcolor{dodgerblue}{\(\bar{Z}_y\)} logical (oriented along \(y\)).}
    \label{fig:3TCGateAction}
\end{figure}

To digest the structure of this gate, we examine its action on plane-like \(\bar{X}\) logicals of three stacks of the three-dimensional toric code (Fig.~\ref{fig:3TCGateAction}).
Considering red \(\textcolor{red}{\bar{X}}\) logicals first, we observe that only \(\textcolor{red}{\bar{X}_z}\) logicals oriented in the \(xy\)-plane cannot be deformed to completely avoid the gate, so only this qubit is affected by the gate in the red code.
Conjugating such an \(\textcolor{red}{\bar{X}_z}\) logical by the gate produces a circuit of \(\mathsf{\textcolor{orange}{C}\textcolor{dodgerblue}{Z}}\) gates supported in the same \(xy\)-plane as the \(\textcolor{red}{\bar{X}_z}\) logical [Fig.~\ref{fig:3TCGateAction}(a)].
The controls are supported only on the left hand side of the orange code, but the targets act on the entire plane in the blue code.
This is simply the generalization of the addressable \(\overline{\mathsf{\textcolor{orange}{C}\textcolor{dodgerblue}{Z}}}\) of the two-dimensional toric code (Fig.~\ref{fig:Addressable_CZ}) to the three-dimensional case.
The only orange \(\textcolor{orange}{\bar{X}}\) logical which cannot avoid the gate is that in the \(yz\)-plane, \(\textcolor{orange}{\bar{X}_x}\), which intersects the half-plane in a line.
The microscopic \(\textcolor{orange}{X}\) operators along this line are each acted upon by two \(\mathsf{\textcolor{orange}{C}\textcolor{dodgerblue}{Z}}\) gates, with targets in the blue code related by a reflection of \(y\).
The result of conjugating the \(\textcolor{orange}{\bar{X}_x}\) logical is thus a blue \(\textcolor{dodgerblue}{\bar{Z}_y}\) logical running along the \(y\) direction, perpendicular to the \(\textcolor{dodgerblue}{\bar{X}_y}\) logical in the \(zx\)-plane [Fig.~\ref{fig:3TCGateAction}(b)].
The logical action of the gate is thus
\begin{equation}
    \CCZLogicalColor{z}{x}{y},
\end{equation}
consistent with the guarantee of the recipe.

\paragraph{No Toffoli gates}
It is likely impossible to find a transversal logical \(\CCNOTLogicalColor{\alpha}{\beta}{\gamma}\) gate (Toffoli gate) in the three-dimensional toric code.
Indeed, if it were possible to implement any logical \(\CCNOTLogicalColor{\alpha}{\beta}{\gamma}\), these could be used to compile any \(\overline{\mathsf{C}^\ell\mathsf{NOT}}\) gate from higher in the Clifford hierarchy, violating a (still hypothetical, though see Section~\ref{sec:discussion}) extension of Bravyi-K\"onig to nonlocal circuits.

In terms of the recipe, while one can easily find dense logical circuits which implement \(\CCNOTLogicalColor{\alpha}{\beta}{\gamma}\), there is an obstruction to sparsifying their associated chain tensors.
Such tensors correspond geometrically to four-dimensional logicals in a nine-dimensional toric code.
In order to be transversal, only \(O(1)\) gates can affect a fixed qubit \(\textcolor{red}{\mathbf{e}}\) in the red code.
This means that the intersection of the four-dimensional cycle with any of the six-dimensional cardinal planes formed by fixing three of the coordinates to correspond to the qubit \(\textcolor{red}{\mathbf{e}}\) must be \(O(1)\).
This is not possible: the intersection of these four- and six-dimensional surfaces in an ambient nine-dimensional space will generically be one-dimensional, and the number of gates acting on \(\textcolor{red}{\mathbf{e}}\) will grow with the linear dimension of the codes.
This holds even for nonlocal gates.
Compare this to the \(\CCZLogicalColor{\alpha}{\beta}{\gamma}\) case: the logical corresponds to a three-dimensional surface [the product of curves in Eq.~\eqref{eqn:3TCGateFactors}], which can be made to only have point-like intersection with any cardinal six-dimensional plane in nine ambient dimensions.
Indeed, each of the factors in Eq.~\eqref{eqn:3TCGateFactors} has at most two intersections with any cardinal plane.
This argument is given more generally in Section~\ref{sec:discussion}.

\subsubsection{Intra-block gates}

\paragraph{\texorpdfstring{\(\CCZLogical{z}{x}{y}\)}{CCZzyx} gate}
As discussed in Section~\ref{subsubsec:IntraBlockClZandCNOT}, to realize the multi-block logical \(\CCZLogicalColor{z}{x}{y}\) within a single block, we can just apply Eq.~\eqref{eq-intrablockClZ}.
The result is logical and, up to corrections by \(\CZLogical{}{}\) and \(\bar{Z}\) operators (all of which we know how to perform in \(\mathsf{3TC}\)), acts on \(\bar{X}_\alpha\) operators with the symmetrized logical action of \(\phii_\bullet\) from Eq.~\eqref{eqn:3TCGateFactors}.

We claim this action is simply
\begin{equation}
    \CCZLogical{x}{y}{z}.
\end{equation}
Indeed, the unsymmetrized \(\phii_*\) is rank-1: it takes \emph{a particular} logical from each code (red, orange, or blue) to the \(\CZLogicalColor{}{}\) between the other two logicals in the other two codes.
The symmetrization takes \emph{any} logical to the symmetrized \(\CZLogical{}{}\) between the other two logical qubits---precisely the logical action of \(\CCZLogical{x}{y}{z}\).

It is instructive to compare this result to what we would get with a non-addressable \(\CCZLogicalColor{}{}{}\) gate.
There, the symmetrization of the logical action would include an even number of \(\CCZLogical{x}{y}{z}\) actions, arising from different combinations of the interblock gates.
The resulting gate would thus be trivial logically.

\paragraph{\texorpdfstring{\(\bar{\mathsf{S}}_{x}\)}{Sx} gate}
Single-logical-qubit gates, such as a Clifford \(\bar{\mathsf{S}}_{x}\) gate, cannot be constructed by symmetrizing some non-symmetric \(\phii_\bullet\).
Rather, we must find a symmetric, sparse chain map acting on a sparse, nice \(\ZZ_4\) lift.
To find an \(\bar{\mathsf{S}}_{x}\) gate in the 3D toric code, we can extend the corresponding lifted chain map for the 2D toric code.

Indeed, the \(\ZZ_2\)-logical tensor corresponding to a \(\CZLogicalColor{x}{x}\) action is
\begin{align}
    \textcolor{dodgerblue}{L_x} \otimes \textcolor{orange}{L_x}
    &=  \textcolor{dodgerblue}{l_{x1} \otimes l_{y0} \otimes \mathbf{l}_z^1} \otimes \textcolor{orange}{l_{x1} \otimes l_{y0} \otimes \mathbf{l}_z^1}, \\
    &= ( \textcolor{dodgerblue}{l_{x1} \otimes \mathbf{l}_z^1} \otimes \textcolor{orange}{l_{x1} \otimes \mathbf{l}_z^1}) \otimes ( \textcolor{dodgerblue}{l_{y0}} \otimes \textcolor{orange}{l_{y0}}).
\end{align}
The first factor is the logical tensor corresponding to a \(\CZLogicalColor{x}{x}\) in the 2D toric code, and the second factor is pointlike.
Thus, letting \(\widetilde{\phii}_\bullet\) be the symmetric \(\ZZ_4\)-chain map from Eq.~\eqref{eqn:2TCSGateChainMap}, and choosing any vertex representative \(\widetilde{w}_y \in \tilde{\mathsf{I}}_0\), we get another symmetric \(\ZZ_4\)-chain map
\begin{equation}
    \td{e}_x \otimes \td{v}_y \otimes \td{\mathbf{e}}_z \mapsto \widetilde{\phii}_\bullet(\td{e}_x \otimes \td{\mathbf{e}}_z) \otimes \widetilde{w}_y \widetilde{\mathbf{w}}_y(\td{v}_y).
\end{equation}
Synthesizing this as a gate, we see that this is simply an embedding of the 2D toric code \(\bar{\mathsf{S}}_{x}\) gate in the 3D toric code.

\paragraph{\texorpdfstring{\(\CNOTLogical{x}{z}\)}{CNOTxz} gate}
As in the 2D toric code, intrablock \(\mathsf{CNOT}\) gates can be constructed from interblock circuits where control qubits and target qubits are disjoint.
The \(\CNOTLogicalColor{\alpha}{\beta}\) gates for \(\alpha \neq \beta\) illustrated in Eq.~\eqref{eqn:3TCCNOT} can satisfy this constraint.
Indeed, we illustrated the \(\CNOTLogicalColor{x}{z}\) in Eq.~\eqref{eqn:3TCCNOT} with controls on \(z \leq d/2\) (say) but all targets are on \(z > d/2\), corresponding to a chain map
\begin{equation}
    \phii'_\bullet = (T_{d/2}\tau)_{\textcolor{orange}{x}\to \textcolor{dodgerblue}{z}} 
    \otimes (\kappa^{-1} \tau^\transpose)_{\textcolor{orange}{y}\to \textcolor{dodgerblue}{y}} 
    \otimes (\kappa^{-1} \tau^\transpose \kappa^{-1})_{\textcolor{orange}{z}\to \textcolor{dodgerblue}{x}}.
\end{equation}
where \(T_\delta : \textcolor{dodgerblue}{\mathsf{I}_{\bullet}} \xrightarrow{\sim} \textcolor{dodgerblue}{\mathsf{I}_{\bullet}}\) is a translation by \(\delta\).
As in the 2D toric code, precisely avoiding collisions of controls and targets may require using a tent map with a flat top, Eq.~\eqref{eqn:2TCCNOTTruncated}.

\subsection{Addressable Cliffords in a Non-Manifold Fracton Code}
\label{subsec:ALPGate}

The chain map hierarchy formalism, and the general recipe of Sec.~\ref{subsec:GateRecipe}, extend beyond manifold codes such as the two- and three-dimensional toric codes.
In this section, we demonstrate the construction of addressable Clifford gates in the anisotropic-linon planon (ALP) code: a geometrically three-dimensional code with fractonic excitations confined to either planes or lines~\cite{xu2004strong,Vijay_2016,Shirley_2019}.

\subsubsection{The Anisotropic Linon-Planeon code}
\label{subsubsec:ALPDefinition}

The ALP code is also a tensor product of classical codes, namely the \emph{plaquette Ising model}~\cite{xu2004strong}, \(\mathsf{PI}_\bullet\), and the (dual) Ising model, \(\mathsf{I}^\bullet\),
\begin{equation}\label{eqn:ALPDef}
    \mathsf{ALP}_\bullet \isom \mathsf{PI}_{xy\bullet} \otimes \mathsf{I}_z^\bullet.
\end{equation}
Subscripts \(x,y,z\) again denote axes labels.

The code properties of the ALP code descend from its classical factor codes, of which the plaquette Ising code is the less familiar.
The plaquette Ising model is described by a two-term chain complex with basis vectors \(\mathsf{B}(\mathsf{PI}_0)\) placed on the vertices of a two-dimensional \(d_x \times d_y\) square lattice, labeled by \(v_{ij}\).
The \(\mathsf{B}(\mathsf{PI}_1)\) basis vectors are labeled by plaquettes, and the boundary map sends a plaquette to the sum of four vertices bordering it,
\begin{equation}
    \mathsf{PI}_\bullet = \left(\includegraphics[valign=c]{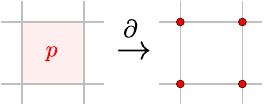}\right).
\end{equation}
As an equation,
\begin{equation}
    \partial p_{ij} = v_{ij} + v_{(i+1)j} + v_{i(j+1)} + v_{(i+1)(j+1)},
\end{equation}
with \(i \in \ZZ_{d_x}\) and \(j \in \ZZ_{d_y}\) regarded with periodic boundary conditions.

Logicals in the plaquette Ising model---cycles in \(\mathsf{PI}^0\)---are given by lines of flipped bits running horizontally or vertically.
Unlike in the toric code, these lines cannot be deformed: they must remain straight.
Similarly, cycles in \(\mathsf{PI}_1\) correspond to sums of plaquettes along horizonal or vertical lines.
\begin{equation}
    \includegraphics[valign=c]{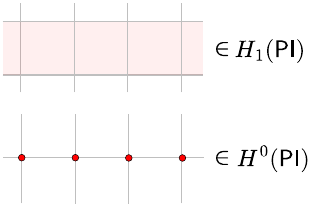}
\end{equation}
As the structural similarity of these cycles suggests, the plaquette Ising model is self-dual with periodic boundary conditions, \(\mathsf{PI}_\bullet \isom \mathsf{PI}^\bullet[-1]\).

As a code, the plaquette Ising model has many encoded bits, corresponding to the many horizontal and vertical lines in the plane.
However, an important feature of the plaquette Ising model is that the natural line-like collections of cycles in \(H^0(\mathsf{PI})\) and \(H_1(\mathsf{PI})\) actually do not form a basis.
The full collection of lines is overcomplete: flipping all horizontal rows of bits is the same as flipping all vertical rows of bits, so the sum of all lines is zero.
The total number of encoded bits (the dimension of \(H^0(\mathsf{PI})\)) is \(d_x + d_y - 1\), the number of lines less one linear dependency.

We fix a basis for the (co)homology classes as follows.
Let \(i \in \ZZ_{d_x} \cup \ZZ_{d_y}\) be a label for a column (\(i \in \ZZ_{d_x}\) with fixed \(x\)) or row (\(i \in \ZZ_{d_y}\) with fixed \(y\)) of vertices in the lattice.
For \(i \neq 0_x \in \ZZ_{d_x}\), denote the cohomology class spanned by the row (column) of vertices \(i\) by \(\mathbf{a}_{i}^0 \in H^0(\mathsf{PI})\), 
\begin{equation}
    \includegraphics[valign=c]{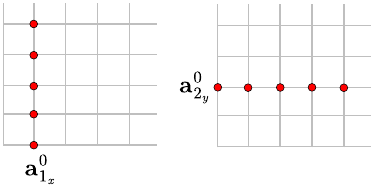}
\end{equation}
and similarly for the plaquettes \(a_{1i} \in H_1(\mathsf{PI})\), where a row (column) of plaquettes \(i\) is bordered by the row (column) of vertices \(i\) and \(i+1\).
These form a basis for \(H^0(\mathsf{PI})\) and \(H_1(\mathsf{PI})\), respectively.
Via the linear dependency, the cycles along the excluded column can be expanded as
\begin{equation}
    \mathbf{a}^0_{0_x} = \sum_{i \neq 0_x} \mathbf{a}^0_i,
    \quad
    a_{10_x} = \sum_{i \neq 0_x} a_{1i}.
\end{equation}
A conjugate basis for \(H_0(\mathsf{PI})\) is \(a_{0i} = [v_{0_x i}]\) for \(i \in \ZZ_{d_y}\)---that is, the homology class of a single vertex in the excluded column \(0_x\)---and \(a_{0i} = [v_{0_x 0_y} + v_{i 0_y}]\) for \(i \in \ZZ_{d_x}\)---a sum of vertices from the row \(0_y\), with one vertex in the excluded column \(0_x\) and the other at the specified column \(i\).
\begin{equation}
    \includegraphics[valign=c]{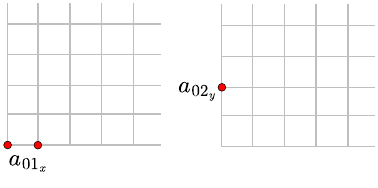}
\end{equation}
Similarly, a conjugate basis for \(H^1(\mathsf{PI})\) is \(\mathbf{a}_{i}^1 = [\mathbf{p}_{0_x i}]\) for \(i \in \ZZ_{d_y}\), and \(\mathbf{a}_{0i} = [\mathbf{p}_{0_x 0_y} + \mathbf{p}_{i 0_y}]\) for \(i \in \ZZ_{d_x}\).

The rigidity of the line-like logicals is what makes the plaquette Ising code, and its descendant the ALP code, a fracton model.
Indeed, consider implementing a partial logical from \(H^0(\mathsf{PI})\) by flipping bits along a straight line segment.
This produces four violated stabilizers, two at each endpoint of the segment.
Such pairs of excitations can be moved by partial logicals, but only along straight lines, and only if those lines are perpendicular to the displacement of the excitations.
These line-mobile pairs are called linons.
Single excitations are immobile, in that they cannot be moved without increasing their energy.%
\footnote{Single excitations can actually be moved arbitrarily far away with \(O(1)\) energy cost, but moving them a distance \(r\) requires flipping roughly \(r^2\) bits.}
In condensed matter language: the line-like symmetries are not deformable 1-form symmetries, but rather subsystem symmetries.

The ALP code, resulting from the tensor product Eq.~\eqref{eqn:ALPDef}, is geometrically three dimensional, with qubits arranged on the vertical edges and horizontal plaquettes (that is, in the \(xy\) plane) of a cubic lattice.
The \(Z\) and \(X\) stabilizers are
\begin{equation}\label{eqn:ALPStabilizers}
    \includegraphics[valign=c]{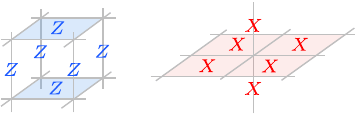}.
\end{equation}
ALP logical operators are horizontal (in the \(x\) or \(y\) directions) or vertical (in the \(z\) direction) lines of \(X\) or \(Z\),
\begin{equation}
    \includegraphics[valign=c]{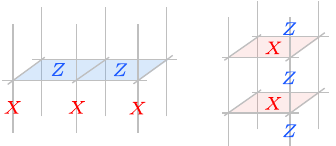}.
\end{equation}

As in plaquette Ising, the line-like logicals of ALP are overcomplete.
The product of all horizontal \(Z\) logicals is trivial, and similarly for \(X\) logicals.
As such, each of the vertical \(X\) and \(Z\) logicals anticommutes with two of the horizontal logicals.
Similarly, the product of four vertical logicals at the corners of a square is a product of stabilizers.
Accounting for all these redundancies---or using the K\"unneth formula---the total number of encoded logical qubits in the ALP code is \(2(L_x + L_y -1)\).
A basis for the homology classes in \(H_0(\mathsf{ALP})\) is as follows,
\begin{equation}
    A^{xy}_{i} = a_{1i} \otimes \mathbf{l}_{z}^1,
    \quad
    A^z_{i} = a_{0i} \otimes \mathbf{l}_{z}^0,
\end{equation}
where \(\mathbf{l}_z^{0,1}\) label Ising cohomology classes, with \(\mathbf{l}_z^{0}\) being extended in the \(z\) direction, and \(i \in (\ZZ_{d_x} \setminus \{0_x\}) \cup \ZZ_{d_y}\).
Similarly label cohomology classes
\begin{equation}
    \mathbf{A}^{xy}_{i} = \mathbf{a}_{i}^1 \otimes l_{z1},
    \quad
    \mathbf{A}^z_{i} = \mathbf{a}_{i}^0 \otimes l_{0z}.
\end{equation}

Horizontally aligned logicals \(\bar{Z}^{xy}_i\) and \(\bar{X}^{xy}_i\) can be deformed in the Ising direction (\(z\) direction), but not within the plaquette Ising plane (\(x\) and \(y\) directions), a fact which descends from the non-deformability of the plaquette Ising model logicals.
Truncating a horizontal logical produces two excited stabilizers at each boundary, and these exitations can move in the plane spanned by the logical direction (one of \(x\) or \(y\)) and the Ising direction \(z\). These are called planons.
On the other hand, vertical logicals \(\bar{Z}^{z}_i\) and \(\bar{X}^{z}_i\) cannot be deformed (much) transversally, and truncating them produces single-stabilizer excitations which can only move vertically without growing in energy.
These are the linons.

\subsubsection{Gates in the ALP code}

\paragraph{Toric code embedding}
Following the recipe verbatim for the ALP code would require us to work with cycles in \(\mathsf{PI}_\bullet \otimes \mathsf{PI}_\bullet\), which is naturally embedded in four dimensions.
This is already hard to visualize, so it is helpful to introduce a reduction from \(\mathsf{PI}_\bullet\) to \(\mathsf{I}_\bullet\).
Essentially, rows and columns in \(\mathsf{PI}_\bullet\) are isomorphic to \(\mathsf{I}_\bullet\), which in turn gives that \(xz\)- and \(yz\)-planes in \(\mathsf{ALP}_\bullet\) are isomorphic to \(\mathsf{TC}_\bullet\).
This observation allows gates in ALP to be immediately obtained from the gates in the two-dimensional toric code.

For any row or column \(i \in \ZZ_{d_x} \cup \ZZ_{d_y}\), there is an embedding chain map from the Ising model to the plaquette Ising model,
\begin{subequations}\label{eqn:IotaDef}
\begin{equation}
    \iota^{(i)}_\bullet : \mathsf{I}_\bullet \to \mathsf{PI}_\bullet,
\end{equation}
defined as
\begin{align}
    \iota_1^{(i)}(e_j) &= \left\{
    \begin{array}{ll}
        p_{ij} \qquad\qquad\quad& \text{for }i \in \ZZ_{d_x}, \\
        p_{ji} \qquad\qquad\quad& \text{for }i \in \ZZ_{d_y},
    \end{array}
    \right. \\
    \iota_0^{(i)}(v_j) &= 
    \left\{
    \begin{array}{ll}
        v_{ij} + v_{(i+1)j} \quad& \text{for }i \in \ZZ_{d_x}, \\
        v_{ji} + v_{j(i+1)} \quad& \text{for }i \in \ZZ_{d_y}.
    \end{array}
    \right.
\end{align}
\end{subequations}
Visualized as a cycle in the chain map complex \(\mathsf{PI}_\bullet \otimes \mathsf{I}^\bullet\)---which is just \(\mathsf{ALP}_\bullet\) again---\(\iota^{(i)}_\bullet\) corresponds to a diagonal ramp along row (column) \(i\), spanning a single line of plaquettes and the two lines of vertices bordering it.
\begin{equation}\label{eqn:IotaCycle}
    \includegraphics[valign=c]{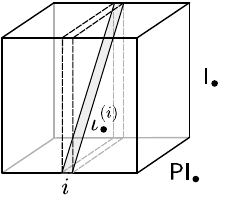}
\end{equation}
The action of \(\iota^{(i)}_*\) on homology is
\begin{equation}\label{eqn:IotaHomologyAction}
    \iota^{(i)}_{*}(l_1) = a_{1i}, \quad
    \iota^{(i)}_{*}(l_0) = a_{0i} + a_{0(i+1)}.
\end{equation}
As \(i\) ranges over \(\ZZ_{d_x} \cup \ZZ_{d_y}\), the only logical which is not in the image of some linear combination of \(\iota^{(i)}_{*}\) is \(\sum_{j \in \ZZ_{d_y}} a_{0 j}\).
The ALP gates we find with this embedding will thus be addressable on all but the two qubits descending from this cycle.

ALP gates are constructed from toric code gates and the embedding \(\iota^{(i)}_\bullet\) as follows.
If \(\phii_\bullet : \textcolor{orange}{\mathsf{TC}^\bullet} \to \textcolor{dodgerblue}{\mathsf{TC}_\bullet}\) is a \(\mathsf{CZ}\)-type (say) chain map for the toric code, then
\begin{equation}
    \phii^{(ij)}_\bullet: \textcolor{orange}{\mathsf{ALP}^\bullet}
    \xrightarrow{\iota_{(i)}^{\transpose} \otimes \id}
    \textcolor{orange}{\mathsf{TC}^\bullet}
    \xrightarrow{\phii}
    \textcolor{dodgerblue}{\mathsf{TC}_\bullet}
    \xrightarrow{\iota^{(j)}\otimes \id}
    \textcolor{dodgerblue}{\mathsf{ALP}_\bullet}
\end{equation}
is an addressable \(\mathsf{CZ}\)-type chain map for ALP.
By addressable, we mean that the gate affects an \(O(1)\) number of qubits of ALP.
Indeed, from Eq.~\eqref{eqn:IotaHomologyAction}, one reads off the homology action of the \(\mathsf{TC}_\bullet \to \mathsf{ALP}_\bullet\) inclusion,
\begin{subequations}\label{eqn:EmbeddingLogicalAction}
\begin{align}
    ({\iota^{(i)\transpose}_{*}} \otimes \id)(\textcolor{orange}{\mathbf{A}^{xy}_{k}}) &= \delta_{ik} \textcolor{orange}{\mathbf{L}_{x}},\\
    ({\iota^{(i)\transpose}_{*}} \otimes \id)(\textcolor{orange}{\mathbf{A}^z_{k}}) &= (\delta_{ik} + \delta_{(i+1)k})\textcolor{orange}{\mathbf{L}_{y}}, \\
    (\iota_*^{(j)} \otimes \id)(\textcolor{dodgerblue}{L_x}) &= \textcolor{dodgerblue}{A^{xy}_j},\\
    (\iota_*^{(j)} \otimes \id)(\textcolor{dodgerblue}{L_y}) &= \textcolor{dodgerblue}{A^z_{j}} + \textcolor{dodgerblue}{A^z_{j+1}}.
\end{align}
\end{subequations}
Even if \(\phii_*\) is full rank, \(\phii_*^{(ij)}\) has rank at most \(2\).

\begin{figure}
    \centering
    \includegraphics[width=\linewidth]{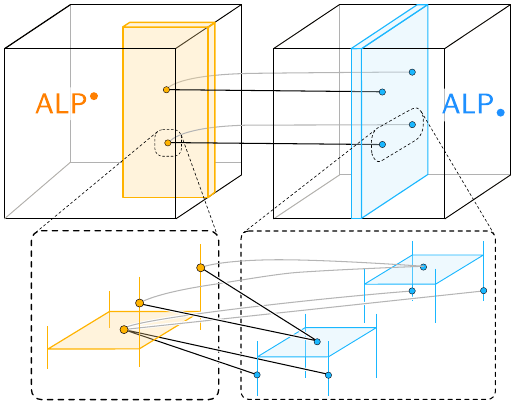}
    \caption{\textbf{Addressable \(\overline{\mathsf{CZ}}\) in the ALP model.}
    Transversal gates for the anisotropic linon-planon (ALP) model, a fracton model, can also be identified using the chain map complex.
    The gate consists of physical \(\mathsf{\textcolor{orange}{C}\textcolor{dodgerblue}{Z}}\) gates between a half-slab of qubits in the orange \(\textcolor{orange}{\mathsf{ALP^\bullet}}\) and a full slab in the blue  \(\textcolor{dodgerblue}{\mathsf{ALP_\bullet}}\).
    The logical action of the gate is determined by which slabs are chosen.}
    \label{fig:ALPGate}
\end{figure}

\paragraph{Addressable \texorpdfstring{\(\CZLogicalColor{}{}\)}{CZ} gates}
Applying the toric code embedding to \(\phii_\bullet = \tau^{\transpose} \otimes \kappa^{-1}\) from Eq.~\eqref{eqn:TCCZyxLogicalCLass2}, which gave a \(\CZLogicalColor{y}{x}\) in the toric code, we get a chain map \(\phii^{(ij)}_\bullet\) in ALP.
The corresponding circuit is shown in Fig.~\ref{fig:ALPGate}.

To find the logical action of the circuit in Fig.~\ref{fig:ALPGate}, we compute the action of \(\phii^{(ij)}_*\) on homology.
Recalling that \(\phii_*\) acts on homology as \(\textcolor{orange}{\mathbf{L}_x} \mapsto 0\), \(\textcolor{orange}{\mathbf{L}_y} \mapsto \textcolor{dodgerblue}{L_x}\) and using Eq.~\eqref{eqn:EmbeddingLogicalAction}, the nontrivial action of \(\phii^{(ij)}_*\) is
\begin{equation}
    \phii^{(ij)}_*(\textcolor{orange}{\mathbf{A}^z_{i}}) = \phii^{(ij)}_*(\textcolor{orange}{\mathbf{A}^z_{i+1}}) = \textcolor{dodgerblue}{A_j^{xy}}
\end{equation}
only.
Thus, this is a product of two logical gates,
\begin{equation}
    \CZLogicalColor{iz}{jxy} \CZLogicalColor{(i+1)z}{jxy}.
\end{equation}
When \(i = 0_x\) or \(i+1 = 0_x\), one of these gates becomes controlled on the parity of all the logicals in our chosen basis.
However, a different choice of basis would reduce this to two gates again.

Alternatively, we can take \(\phii_\bullet\) to be the symmetric gate from Eq.~\eqref{eqn:2TCCZxxFormula}, implementing \(\CZLogicalColor{x}{x}\).
The homology action of this \(\phii_\bullet\) is \(\textcolor{orange}{\mathbf{L}_x} \mapsto \textcolor{dodgerblue}{L_x}\), \(\textcolor{orange}{\mathbf{L}_y} \mapsto 0\), which gives a nontrivial action of
\begin{equation}
    \phii_*^{(ij)}(\textcolor{orange}{\mathbf{A}^{xy}_{i}}) = \textcolor{dodgerblue}{A_j^{xy}}
\end{equation}
only.
Thus, we obtain a singly-addressable 
\begin{equation}
    \CZLogicalColor{ixy}{jxy}
\end{equation}
gate.

Similarly, using toric code \(\CZLogicalColor{y}{y}\) and \(\CZLogicalColor{x}{y}\) gates, we obtain
\begin{equation}
    \CZLogicalColor{iz}{jz} \CZLogicalColor{iz}{(j+1)z} \CZLogicalColor{(i+1)z}{jz} \CZLogicalColor{(i+1)z}{(j+1)z}
\end{equation}
and
\begin{equation}
    \CZLogicalColor{ixy}{jz} \CZLogicalColor{ixy}{(j+1)z}
\end{equation}
gates, respectively.

\paragraph{Intra-block \texorpdfstring{\(\CZLogical{}{}\)}{CZ} and \texorpdfstring{\(\bar{\mathsf{S}}\)}{S} gates}
Any of the above \(\CZLogicalColor{}{}\) gates can be realized in a single code block by applying Eq.~\eqref{eq-intrablockClZ} to the chain maps for the inter-block gates.
From this, we obtain gates like 
\begin{equation}
    \CZLogical{iz,}{jxy} \CZLogical{(i+1)z,}{jxy}
\end{equation}
for any \(i,j\), and
\begin{equation}
    \CZLogical{ixy,}{jxy}
\end{equation}
for \(i \neq j\).

To get an \(\bar{\mathsf{S}}\) gate, we need to make a \(\ZZ_4\) lift.
A sparse, nice \(\ZZ_4\) lift of the plaquette Ising model has a basis of plaquettes \(\td{p}_{ij} \in \widetilde{\mathsf{PI}}_1\) and vertices \(\td{v}_{ij} \in \widetilde{\mathsf{PI}}_1\) as before, with boundary map,
\begin{equation}
    \partial p_{ij} = v_{ij} - v_{(i+1)j} - v_{i(j+1)} + v_{(i+1)(j+1)}.
\end{equation}
A lift of the embedding map \(\td{\iota}^{(i)}_\bullet : \td{\mathsf{I}}_\bullet \to \widetilde{\mathsf{PI}}_\bullet\) is given by
\begin{subequations}\label{eqn:TildeIotaDef}
\begin{align}
    \td{\iota}_1^{(i)}(\td{e}_j) &= \left\{
    \begin{array}{ll}
        \td{p}_{ij} \qquad\qquad\quad& \text{for }i \in \ZZ_{d_x}, \\
        \td{p}_{ji} \qquad\qquad\quad& \text{for }i \in \ZZ_{d_y},
    \end{array}
    \right. \\
    \td{\iota}_0^{(i)}(\td{v}_j) &= 
    \left\{
    \begin{array}{ll}
        \td{v}_{(i+1)j} - \td{v}_{ij} \quad& \text{for }i \in \ZZ_{d_x}, \\
        \td{v}_{j(i+1)} - \td{v}_{ji} \quad& \text{for }i \in \ZZ_{d_y}.
    \end{array}
    \right.
\end{align}
\end{subequations}
Then, embedding the \(\widetilde{\phii}_\bullet\) responsible for the toric code \(\bar{\mathsf{S}}_x\) gate in the \(\ZZ_4\) ALP code gives,
\begin{equation}
    \widetilde{\phii}^{(ii)}_\bullet = (\td{\iota}^{(i)}\otimes\id) \circ \widetilde{\phii} \circ (\td{\iota}^{(i)\transpose}\otimes\id).
\end{equation}
Further, this chain map is clearly symmetric and sparse if \(\widetilde{\phii}_\bullet\) is.
Thus, using the symmetric \(\widetilde{\phii}_\bullet\) from Eq.~\eqref{eqn:2TCSGateChainMap}, we obtain a
\begin{equation}
    \bar{\mathsf{S}}_{ixy}
\end{equation}
gate in ALP.
Similarly, using the chain map for \(\bar{\mathsf{S}}_{y}\), we obtain a logical
\begin{equation}
    \bar{\mathsf{S}}_{iz} \bar{\mathsf{S}}_{(i+1)z} \CZLogical{iz,}{(i+1)z}.
\end{equation}

\paragraph{Addressable \texorpdfstring{\(\CNOTLogicalColor{}{}\)}{CNOT} gates}
As in the toric code, the ALP code has a self-duality \(\textcolor{dodgerblue}{\mathsf{ALP}_\bullet} \xrightarrow{\sim} \textcolor{dodgerblue}{\mathsf{ALP}^\bullet}\).
This allows any \(\CZLogicalColor{}{}\) gate to be converted to a \(\CNOTLogicalColor{}{}\) by composing its encoding chain map with the self-duality.

\paragraph{Intra-block off-diagonal gates}
The ALP code has a transversal Hadamard, implemented by a Hadamard on every qubit followed by a swap of qubits across a diagonal \(x=y\).
Using this, we can compile addressable Hadamard gates using addressable \(\bar{\mathsf{S}}\) gates, as in Lemma~\ref{lem:CliffordGeneration} and Eq.~\eqref{eqn:2TC_HGate}.
For instance, this gives us an addressable \(\bar{\mathsf{H}}_{ixy}\) gate.
In turn, these Hadamards can be used to construct \(\CNOTLogical{}{}\) gates.

Alternatively, we can observe that the targets and controls of the \(\CNOTLogicalColor{}{}\) gates frequently do not intersect.
This certainly happens when the planes where the gates act are parallel and separated by at least one unit cell.
It also occurs when they are orthogonal: in the \(\CNOTLogicalColor{}{}\) analogue of Fig.~\ref{fig:ALPGate}, the orange controls can be translated parallel to their plane without changing the action of the gate.
Thus, the controls can be made to avoid the plane of the targets.
With disjoint control and target qubits, \(\CNOTLogicalColor{}{}\) gates can be realized as intra-block \(\CNOTLogical{}{}\) gates, simply by moving all controls and targets to a single code block.

\subsection{Addressable Magic Gates in a Non-Manifold Code}
\label{subsec:ALPIsiGate}

The chain map hierarchy also delivers magic (non-Clifford) gates in non-manifold codes.
In this final section of examples, we consider a code \(\mathsf{3ALP}_\bullet\), which is defined as of any of the following equivalent tensor products,
\begin{subequations}
\begin{align}
    \mathsf{3ALP}_\bullet &:= \mathsf{ALP}_{x\td x z\bullet} \otimes \mathsf{I}_{y\bullet} \\
    &\isom \mathsf{PI}_{x \td x\bullet} \otimes \mathsf{I}_{y\bullet} \otimes \mathsf{I}_z^\bullet \\
    &\isom \mathsf{PI}_{x \td x\bullet} \otimes \mathsf{TC}_{yz\bullet}.
\end{align}
\end{subequations}
Geometrically, this is naturally a four-dimensional code, and we label the dimensions by \((x,\td x,y,z)\).
However, \(\mathsf{3ALP}_\bullet\) is only a three-dimensional (four-term) chain complex, with correspondingly restricted mobility of excitations.

A single code block of \(\mathsf{3ALP}_\bullet\) encodes \(3(d_x + d_{\td x}-1)\) qubits, where \(d_x\) and \(d_{\td x}\) are the linear distances of the \(x\) and \(\td x\) dimensions.
A basis for the homology classes of \(H_0(\mathsf{3ALP})\) is
\begin{subequations}
\begin{align}
    B_i^{x} &= a_{1i} \otimes l_{0y} \otimes \mathbf{l}^{1}_z,\\
    B_i^{y}&= a_{0i} \otimes l_{1y} \otimes \mathbf{l}^{1}_z, \\
    B_i^{z} &= a_{0i} \otimes l_{0y} \otimes \mathbf{l}^{0}_z,
\end{align}
\end{subequations}
where \(i \in (\ZZ_{d_x} \setminus\{0_x\}) \cup \ZZ_{d_{\td x}}\).
We similarly label the line-like logicals \(\bar{Z}_i^{x}\) (and so on), and the plane-like \(\bar{X}_i^{x}\) logicals.

As in the ALP code, rather than directly follow the recipe of Sec.~\ref{subsec:GateRecipe}, we can use the map \(\iota^{(i)}_\bullet: \mathsf{I}_\bullet \to \mathsf{PI}_\bullet\) to embed the three-dimensional toric code into \(\mathsf{3ALP}_\bullet\), and thus vastly simplify the construction of gates.

We will not discuss the realization of most Clifford gates in this code, but they can also be obtained from the three-dimensional toric code Cliffords by embedding.
As in the three-dimensional toric code, it is impossible for \(\mathsf{3ALP}_\bullet\) to have a transversal Hadamard.

\paragraph{Addressable \texorpdfstring{\(\CCZLogicalColor{}{}{}\)}{CCZ} gates}
The \(\CCZLogicalColor{z}{x}{y}\) of the three-dimensional toric code can be embedded into \(\mathsf{3ALP}_\bullet\) to find another \(\CCZLogicalColor{}{}{}\) gate which only affects \(O(1)\) logical qubits.

As we are now dealing with chain tensors with many indices, it is convenient to adopt a graphical notation.
We represent the chain map \(\iota^{(i)}_\bullet\) from Eq.~\eqref{eqn:IotaDef} as
\begin{equation}
    \iota^{(i)}_\bullet = 
    \includegraphics[valign=c]{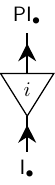}.
\end{equation}
The arrow indicates the direction that chains go through the map.
Then we can realize an addressable \(\CCZLogicalColor{}{}{}\) in \(\mathsf{3ALP}_\bullet\) as the following chain-tensor network.
\begin{equation}\label{eqn:3ALPTensorNetwork}
    \phii_\bullet = \includegraphics[valign=c]{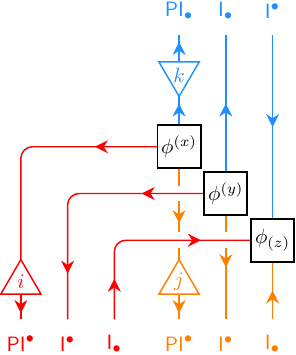}
\end{equation}
The three-leg tensors \(\phi^{(\alpha)} \in [\textcolor{red}{\mathsf{I}^\bullet}, [\textcolor{orange}{\mathsf{I}^\bullet},  \textcolor{dodgerblue}{\mathsf{I}_\bullet}]]\) may be, for instance, those from Eq.~\eqref{eqn:3TCGateFactors}.

Translating Eq.~\eqref{eqn:3ALPTensorNetwork} to a circuit gives the \(\overline{\mathsf{\textcolor{red}{C}\textcolor{orange}{C}\textcolor{dodgerblue}{Z}}}\) gate.
The logical action of this gate can be determined by applying \(\iota_*^{(\cdot)}\) to the \(\phii_{0*}\) action originally used to specify the \(\mathsf{3TC}_\bullet\) gate.
\begin{align}
    (\phii_{*})_0 &=  
    \textcolor{dodgerblue}{l_{x0} \, l_{y1} \, \mathbf{l}_z^1} \otimes 
    \textcolor{orange}{l_{x1} \, l_{y0} \, \mathbf{l}_z^1} \otimes
    \textcolor{red}{l_{x0}\, l_{y0}\, \mathbf{l}_z^0}  \nnb\\
    &\mapsto 
    \textcolor{dodgerblue}{\iota_*^{(k)}(l_{x0}) \, l_{y1} \, \mathbf{l}_z^1} \otimes 
    \textcolor{orange}{\iota_*^{(j)}(l_{x1}) \, l_{y0} \, \mathbf{l}_z^1} \otimes
    \textcolor{red}{\iota_*^{(i)}(l_{x0})\, l_{y0}\, \mathbf{l}_z^0} \nnb \\
    &= \textcolor{dodgerblue}{(a_{0k} + a_{0(k+1)}) \, l_{y1} \, \mathbf{l}_z^1} \otimes
    \textcolor{orange}{a_{1j} \, l_{y0} \, \mathbf{l}_z^1}  \nnb \\
    &\hspace{0.3\linewidth} \otimes\textcolor{red}{(a_{0i} + a_{0(i+1)})\, l_{y0}\, \mathbf{l}_z^0} 
\end{align}
That is, Eq.~\eqref{eqn:3ALPTensorNetwork} implements a
\begin{equation}
    \CCZLogicalColor{i}{j}{k}^{\textcolor{red}{z}\textcolor{orange}{x}\textcolor{dodgerblue}{y}}
    \CCZLogicalColor{(i+1)}{j}{k}^{\textcolor{red}{z}\textcolor{orange}{x}\textcolor{dodgerblue}{y}}
    \CCZLogicalColor{i}{j}{(k+1)}^{\textcolor{red}{z}\textcolor{orange}{x}\textcolor{dodgerblue}{y}}
    \CCZLogicalColor{(i+1)}{j}{(k+1)}^{\textcolor{red}{z}\textcolor{orange}{x}\textcolor{dodgerblue}{y}}
\end{equation}
gate.

Applying the same embedding to a 3D toric code \(\CCZLogicalColor{x}{x}{x}\) gate
analogously gives a 
\begin{equation}\label{eqn:3ALTP3xGate}
    \CCZLogicalColor{i}{j}{k}^{\textcolor{red}{x}\textcolor{orange}{x}\textcolor{dodgerblue}{x}}
\end{equation}
gate.

\paragraph{Intra-block diagonal gates}
Constructing intra-block realizations of \(\CCZLogicalColor{}{}{}\) gates that target distinct logicals is, as before, straightforward by appying Eq.~\eqref{eq-intrablockClZ}, yielding a symmetrized logical action.
Focusing instead on \(\mathsf{S}\) gates, we observe that applying the \(\ZZ_4\)-embedding map \(\td{\iota}^{(i)} \otimes \id \otimes \id\) to each leg of a symmetric tensor \(\widetilde{\phii}_\bullet \in \mathsf{3TC}_\bullet^{\otimes 2}\) 
again produces a symmetric \(\ZZ_4\)-tensor.
Thus, for instance, embedding the symmetric chain tensor for the three-dimensional toric code's \(\bar{\mathsf{S}}_x\) gate gives a \(\bar{\mathsf{S}}_{ix}\) gate on \(\mathsf{3ALP}_\bullet\).

\section{A Universal Form of the Cup Product from the Chain Map Hierarchy}
\label{sec:UniversalCup}

Thus far, we have shown that the chain map hierarchy provides a conceptually clear language for both understanding computation with qLDPC codes, and finding logical operations therein.
We now discuss the relationship between the chain map hierarchy and other approaches to understanding logical gates.

From Section~\ref{sec:ChainMapHierarchy}, it is manifestly clear how several approaches to Clifford computation naturally fit into the chain map hierarchy language.
In particular, homomorphic $\mathsf{CNOT}$ gadgets~\cite{huang_2023_homomorphicCNOT,xu2024constant}, homomorphic measurement gadgets~\cite{horsman2012surface,litinski2019game,Vuillot_2019,cohen2022lowoverhead,cowtan2024ssip,Cowtan2024SurgeryUniversal,swaroop2024universal,williamson2024lowoverhead,cross2024improved,ide2025faulttolerant,zhang2025timeefficient,zhang2025accelerating,zheng2025highrate,baspin2025fast,cowtan_fast_2025,cowtan2025parallel,he2025extractors}, automorphism gadgets~\cite{Calderbank1997gf4, Grassl_2013, Sayginel2025codeautomorphism,hsin2025automorphismgaugetheorieshigher, berthusen2025automorphismgadgetshomologicalproduct}, and fold-transversal gates~\cite{Kubica2015Unfolding,Moussa2016FoldedSurface,Breuckmann2024foldtransversal} on qLDPC codes are regularly phrased in terms of chain maps---which are simply labeled by cycles of chain complexes at the second level of the chain map hierarchy. 
However, known non-Clifford gadgets are not typically understood via chain maps and how they relate to the chain map hierarchy is not as manifest.

A popular approach to understanding these gadgets is through the use of \emph{cup products}---bilinear maps on cochains that induce a product on the cohomology of a given chain complex.
It has been shown that if a chain complex is equipped with such a product, it is possible define sparse cohomology invariants which in turn enables defining certain non-Clifford gates~\cite{zhu2025, breuckmann2025cupsgatesicohomology, lin2024transversalnoncliffordgatesquantum, zhu2025topological}.
Consequently, a flurry of recent work has been devoted to developing constructions of cup products and their cohomology invariants on chain complexes associated with qLDPC codes.

These works take a variety of different approaches.
Some generalize the cup product associated with manifolds to other chain complexes by either exploiting the sheaf structure naturally associated with particular qLDPC codes, or by mapping a qLDPC code of interest to a high dimensional manifold \cite{zhu2025topological, lin2024transversalnoncliffordgatesquantum, li_poincare_2025, li2026theorycohomologicalinvariantsquantum}.
Others bypass building a true product structure on cohomology and instead attempt to define the cohomology invariants directly on codes with certain algebraic requirements on their (co)boundary maps~\cite{breuckmann2025cupsgatesicohomology, menon2025magic, tiew2026copycupgatestensorproducts}.

In this section, we show that these approaches can be understood and extended within the chain map hierarchy.
In particular, after reviewing the notion of a cup product and their associated cohomology invariants in Subsection~\ref{subsec:cupproductreview}, we show how chain maps at the third level of the chain map hierarchy completely parameterize cup products on qLDPC codes---a fact that follows from the \emph{universal property} of the tensor product (Subsection~\ref{subsec:universalcupproduct}).
Moreover, each cohomology invariant defined from these specifies an element of the chain map hierarchy.
Consequently, our formalism can uncover cup products and associated gates that are not accessible through prior methods.

\subsection{Review: Cup Products, Cohomology Invariants, and Logical Gates}
\label{subsec:cupproductreview}

We start by defining precisely what we mean by a cup product on chain complexes, before discussing how it can be used to define gates.
While much of this subsection is review, we remark that we will present these ideas in a way that will naturally lead into our subsequent results.

Broadly speaking, a cup product associated with a chain complex is defined as follows:%
\footnote{We note that some definitions of the cup product \cite{hatcher2002algebraic} require the product to be associative---as it is in its original setting of simplicial complexes---in which case it induces a ring structure on cohomology.
Since we are interested in defining gates from the cup product---for which associativity is not necessary---we omit associativity from the definition.
}
\begin{defn}[(Sparse) Cup Product]\label{def:SparseCup}
Suppose that $C^{\bullet}$ is a cochain complex over \(\FF_2\).
We say that a function
\begin{equation}
\cupp: \bigoplus_{p, q} C^p \times C^q \to \bigoplus_{p+q} C^{p + q}    
\end{equation}
(where \(C^p \cupp C^q \subseteq C^{p+q}\))
defines a cup product on $C^{\bullet}$ if $\cupp$ is bilinear and satisfies the \textit{Leibniz rule}, i.e.
\begin{equation}
        \exd ( \mathbf{a} \cupp \mathbf{b}) = \exd \mathbf{a} \cupp \mathbf{b}  + \mathbf{a} \cupp \exd \mathbf{b}
    \label{eq-Leibniz}
\end{equation}
where $\mathbf{a} \in C^p$ and $\mathbf{b} \in C^q$ and $p$ and $q$ are arbitrary integers.

Further suppose that $C^{\bullet}$ is sparse and based.
We say that $\cupp$ is \textit{input sparse} in degree $p$ and $q$ if for all $\mathbf{x}_{p} \in \mathsf{B}(C^p)$, there are $O(1)$ elements $\mathbf{x}_q \in \mathsf{B}(C^q)$ such that  $\mathbf{x}_p \cupp \mathbf{x}_q \neq 0$ and vice versa.

We say that $\cupp$ is \textit{output sparse} in degree $p$ and $q$, if two conditions are met.
First, for all $\mathbf{x}_{p} \in \mathsf{B}(C^p)$ and $\mathbf{x}_{q} \in \mathsf{B}(C^q)$, $\mathbf{x}_{p} \cupp \mathbf{x}_{q}$ has $O(1)$ support on basis elements in $\mathsf{B}(C^{p + q})$.
Second, for all $\mathbf{x}_{p + q} \in \mathsf{B}(C^{p + q})$, there are $O(1)$ pairs $\mathbf{x}_p \in \mathsf{B}(C^p)$ and $\mathbf{x}_q \in \mathsf{B}(C^q)$ such that $\mathbf{x}_{p + q} \in \text{support}(\mathbf{x}_p \cupp \mathbf{x}_q)$.
Here, \(\text{support}(\mathbf{a}) \subseteq \mathsf{B}(C^{p + q})\) is the set of basis elements \(\mathbf{x}\) for which \(\mathbf{a}(x) \neq 0\).

Any cup product that is both input and output sparse in all degree pairs is simply called sparse.
\end{defn}

The Leibniz rule of Eq.~\eqref{eq-Leibniz} ensures that the cup product induces a bilinear product on the cohomology classes.
This can be seen by observing two facts.
First, the cup product of two cocycles is also a cocycle.
Second, the cohomology class of the cup product of two cocycles depends only on their cohomology class.

We can see the first by noting that if $\mathbf{a}$ and $\mathbf{b}$ are co-cycles, then $\mathbf{a} \cupp \mathbf{b}$ satisfies:
\begin{equation}
    \exd(\mathbf{a} \cupp \mathbf{b}) = \exd \mathbf{a} \cupp \mathbf{b} + \mathbf{a} \cupp \exd \mathbf{b} = 0 
\end{equation}
where the first equality follows from the Leibniz rule and the second equality follows from the fact that (1)~$\mathbf{a}$ and $\mathbf{b}$ are cocycles and so their coboundary must be zero and (2)~bilinearity ensures that $0 \cupp \mathbf{b} = \mathbf{a} \cupp 0 = 0$.
Hence,  the cup product of two cocycles is also a cocycle.

We can see the second by noting that if $\mathbf{a}$ and $\mathbf{b}$ are cocycles and $\exd \mathbf{c}$ is an arbitrary coboundary, then: 
\begin{align}
    \mathbf{a} \cupp (\mathbf{b} + \exd \mathbf{c}) &= \mathbf{a} \cupp \mathbf{b} + \mathbf{a} \cupp \exd \mathbf{c}\\
    &= \mathbf{a} \cupp \mathbf{b} + (\mathbf{a} \cupp \exd \mathbf{c} + \exd \mathbf{a} \cupp \mathbf{c}) \\
    &= \mathbf{a} \cupp \mathbf{b} + \exd (\mathbf{a} \cupp \mathbf{c})
\end{align}
and a similar calculation can be done for $(\mathbf{a} + \exd \mathbf{c}) \cupp \mathbf{b}$.
Consequently, $\mathbf{a} \cupp (\mathbf{b} + \exd \mathbf{c})$ and $(\mathbf{a} + \exd \mathbf{c}) \cupp \mathbf{b}$ are related to $\mathbf{a} \cupp \mathbf{b}$ by a co-boundary.
Therefore, the cohomology class of $[\mathbf{a} \cupp \mathbf{b}]$ depends only on the cohomology classes $[\mathbf{a}]$ and $[\mathbf{b}]$ and we can define on induced bilinear product on cohomology by $[\mathbf{a} \cupp \mathbf{b}] = [\mathbf{a}] \cupp [\mathbf{b}]$.

\subsubsection{Concrete Example: The Ising Model and the Toric Code}

To provide a concrete example of the cup product, we introduce the cup product for the 1D Ising model and show how this can be used with the tensor product to define the cup product for the toric code in 2D and higher dimensions.

Let us recall that the Ising cochain  complex $\mathsf{I}^{\bullet}$ is a sparse based chain complex where $\mathsf{B}(\mathsf{I}^0)$ and $\mathsf{B}(\mathsf{I}^1)$ correspond to the vertices and edges of a one-dimensional line graph with the global topology of a circle:
\begin{equation}
    \begin{tikzpicture}[scale = .8, baseline={([yshift=-.5ex]current bounding box.center)}]
    \foreach \i in {0, ..., 4}{
        \filldraw[color = black] (\i, 0) circle (2 pt);
        \draw[color = black] (\i, 0) -- (\i + 1, 0) ;
        \node at(\i, -0.3) {$v_{\i}$};
        \node at(\i + 0.5, 0.3) {$e_{\i}$};
    }
    \draw[color = black] (-1, 0) -- (0,0);
    \node at(-1 + 0.5, 0.3) {$e_{N -1}$};
    \node at (5.5, 0) {$\cdots$};
    \node at (-1.5, 0) {$\cdots$};
    \end{tikzpicture}
\end{equation}
The cup product on the 1D Ising model may be defined via the following rules on the basis cochains
\begin{align}
    \mathbf{v}_i \cupp \mathbf{v}_j = &\, \delta_{i, j} \mathbf{v}_i,\\
    \mathbf{v}_i \cupp \mathbf{e}_j = \delta_{i, j} \mathbf{e}_j, \quad & \quad \mathbf{e}_{i} \cupp \mathbf{v}_j = \delta_{i, j-1} \mathbf{e}_i 
\end{align}
where $\mathbf{v}_i \in \mathsf{B}(\mathsf{I}^0)$ and $\mathbf{e}_i \in \mathsf{B}(\mathsf{I}^1)$.
Outside of these cases, the cup product is zero.
Intuitively, the rules above communicate that the cup product of two vertices is non-zero if they are the same, the cup product of a vertex with an edge is only non-zero if the edge is directly to the right of it, and the cup product of an edge with a vertex is only non-zero if the edge is directly to the left of the vertex.
This effectively picks an orientation on edges.
The cup product on general cochains is defined by applying the rules above bilinearly.
It is furthermore easy to verify that the cup product above satisfies the Leibniz rule. 

We can use the cup product for the Ising model to define a cup product for the 2D toric code and its higher dimensional variants.
In particular, recall from Section~\ref{sec:DiscoveringGates}, that the 2D toric code is given by
\(
    \mathsf{TC}_{\bullet} \isom \mathsf{I}_{x\bullet} \otimes  \mathsf{I}_y^\bullet
\).
In this section, to make contact with conventions for the cup product, we always use cochain complexes concentrated in non-negative degrees.
As such, we will shift the toric code up a degree relative to Section~\ref{sec:DiscoveringGates}, so that it has non-zero cochains in degrees \(0,1,2\).
This shifted toric is isomorphic to a tensor product of two Ising cochain complexes:
\begin{equation}
    C^\bullet = \mathsf{I}_{x}^\bullet \otimes  \mathsf{I}_y^\bullet \isom \mathsf{TC}[1]^\bullet.
\end{equation}
Basis elements of $C^{0}$ are of the form $\mathbf{v}_i \otimes \mathbf{v}_j$, the basis elements of $C^{1}$ are of the form $\mathbf{v}_i \otimes \mathbf{e}_j$ and $\mathbf{e}_i \otimes \mathbf{v}_j$, and basis elements of $C^{2}$ are given by $\mathbf{e}_i \otimes \mathbf{e}_j$.
The cup product on $C^{\bullet}$ is inherited from the cup product on the consituent Ising factors.
In particular, given two basis cochains of $\mathsf{TC}^{\bullet}$, $(\mathbf{x}_1 \otimes \mathbf{y}_1)$ and $(\mathbf{x}_2 \otimes \mathbf{y}_2)$, the cup products between these elements is given by: 
\begin{equation} \label{eq-tensorproductcupp}
    (\mathbf{x}_1 \otimes \mathbf{y}_1) \cupp (\mathbf{x}_2 \otimes \mathbf{y}_2) = (\mathbf{x}_1 \cupp \mathbf{x}_2 \otimes \mathbf{y}_1 \cupp \mathbf{y}_2).
\end{equation}
Moreover, one can readily check that the Liebniz rule for this cup product is obeyed given that the cup product on the Ising factors obeys the Leibniz rule.

To get an intuitive understanding of the cup product for the toric code, we can draw out the rules implied by Eq.~\eqref{eq-tensorproductcupp} visually (where \(\mathbf{v}, \mathbf{e}, \mathbf{p}\) are product basis elements in the toric code): 
\begin{equation}
    \vspace{-3 mm}
    \begin{tikzpicture}[scale = 0.85, baseline = {([yshift=-.5ex]current bounding box.center)}]
    \foreach \i in {1.5}{
    \filldraw[fill = gray, draw = none, opacity = 0.1] (0 + \i,0) -- (0.75 + \i,0) -- (0.75 + \i, -0.75) -- (0 + \i, -0.75) -- cycle;
    \draw[color = black, line width = 0.8pt](0 + \i, -0.75) -- (0 + \i,0);
    \draw[color = black, line width = 0.8pt](0 + \i, 0) -- (0.75 + \i,0);
    \node at (-0.15 + \i, -0.375) {${\mathbf{e}}$};
    \node at (0.375 + \i, 0.18) {${\mathbf{e}'}$};
    \node at (0.375 + \i, -0.75 - 0.19) {$\textcolor{white}{\mathbf{e}'}$};
    \node at (0.375 + \i, -0.375) {\textcolor{red}{$\mathbf{p}$}};
    }
    \foreach \i in {-0.4}{
    \filldraw[fill = gray, draw = none, opacity = 0.1] (0 + \i,0) -- (0.75 + \i,0) -- (0.75 + \i, -0.75) -- (0 + \i, -0.75) -- cycle;
    \draw[color= black, line width = 0.8pt](0.75 + \i, -0.75) -- (0.75 + \i,0);
    \draw[color = black, line width = 0.8pt](0 + \i ,-0.75) -- (0.75 + \i,-0.75);
    \node at (0.94 + \i, -0.375) {${\mathbf{e}'}$};
    \node at (0.375 + \i, -0.75 - 0.19) {${\mathbf{e}}$};
    \node at (0.375 + \i, -0.375) {\textcolor{red}{$\mathbf{p}$}};
    \node at (0.375 + \i, -0.375) {\textcolor{red}{$\mathbf{p}$}};
    \node at (0.375 + \i, 0.19) {$\textcolor{white}{\mathbf{e}'}$};
    };

    \node at (1, 0.65) {\small $\underline{{\mathbf{e}} \cupp {\mathbf{e}'} = \textcolor{red}{\mathbf{p}}}$};
    \foreach \i in {3.5}{
    \filldraw[fill = gray, draw = none, opacity = 0.1] (0 + \i,0) -- (0.75 + \i,0) -- (0.75 + \i, -0.75) -- (0 + \i, -0.75) -- cycle;
    \filldraw[fill = black] (0 + \i, -0.75) circle (1.5pt);
    \node at (-0.2 + \i, -0.65) {$\mathbf{v}$};
    \node at (0.375 + \i, -0.375) {\textcolor{red}{$\mathbf{p}$}};
    \node at (0.325 + \i, 0.64) {\small $\underline{\mathbf{v} \cupp \textcolor{red}{\mathbf{p}} = \textcolor{red}{\mathbf{p}}}$};
    };

    \foreach \i in {6}{
    \filldraw[fill = gray, draw = none, opacity = 0.1] (0 + \i,0) -- (0.75 + \i,0) -- (0.75 + \i, -0.75) -- (0 + \i, -0.75) -- cycle;
    \filldraw[fill = black] (0.75 + \i,0) circle (1.5pt);
    \node at (1.0 + \i, 0) {$\mathbf{v}$};
    \node at (0.375 + \i, -0.375) {\textcolor{red}{$\mathbf{p}$}};
    \node at (0.325 + \i, 0.64) {\small $\underline{\textcolor{red}{\mathbf{p}} \cupp \mathbf{v} = \textcolor{red}{\mathbf{p}}}$};
    };
    \end{tikzpicture}
    \label{eq-cupproduct}
\end{equation}
Moreover,  for two vertices $\mathbf{v} $ and $\mathbf{w},$ $\mathbf{v} \cupp \mathbf{w} = \mathbf{v} \delta_{\mathbf{v}, \mathbf{w}}$.
Finally, for a vertex and an edge: 
\begin{equation}
    \begin{tikzpicture}[scale = 0.85, baseline = {([yshift=-.5ex]current bounding box.center)}]
        \draw[color = gray] (0,0) -- (0, -0.75);
        \draw[color = gray] (0,0) -- (-0.75, 0);
        \draw[color = gray] (0,0) -- (0.75, 0);
        \filldraw[fill = black] (0,0) circle (1.5pt);
        \node at (0.25, 0.15) {$\mathbf{v}$};
        \node at (-0.25, 0.375) {$\mathbf{e}$};
        \draw[color = black, line width = 0.8pt] (0, 0) -- (0, 0.75);
        \draw[color = gray] (0 + 2,0) -- (0 + 2, -0.75);
        \draw[color = gray] (0 + 2,0) -- (-0.75 + 2, 0);
        \draw[color = gray] (0 + 2,0) -- (0.75 + 2, 0);
        \draw[color = gray] (0 + 2,0) -- (2, 0.75);
        \filldraw[fill = black] (0 + 2,0) circle (1.5pt);
        \node at (0.25 + 2, 0.15) {$\mathbf{v}$};
        \node at (0.375 + 2, -0.25) {$\mathbf{e}$};
        \draw[color = black, line width = 0.8pt] (0 + 2, 0) -- (0.75 + 2, 0.0);
        \draw[color = gray] (0 + 5,0) -- (0 + 5, -0.75);
        \draw[color = gray] (0 + 5,0) -- (-0.75 + 5, 0);
        \draw[color = gray] (0 + 5,0) -- (0.75 + 5, 0);
        \draw[color = gray] (0 + 5,0) -- (5, 0.75);
        \filldraw[fill = black] (0 + 5,0) circle (1.5pt);
        \node at (0.25 + 5, 0.15) {$\mathbf{v}$};
        \node at (-0.375 + 5, -0.25) {$\mathbf{e}$};
        \draw[color = black, line width = 0.8pt] (0 + 5, 0) -- (-0.75 + 5, 0.0);
        \draw[color = gray] (0 + 7,0) -- (0 + 7, -0.75);
        \draw[color = gray] (0 + 7,0) -- (-0.75 + 7, 0);
        \draw[color = gray] (0 + 7,0) -- (0.75 + 7, 0);
        \draw[color = gray] (0 + 7,0) -- (7, 0.75);
        \filldraw[fill = black] (0 + 7,0) circle (1.5pt);
        \node at (0.25 + 7, 0.15) {$\mathbf{v}$};
        \node at (-0.25 + 7, -0.375) {$\mathbf{e}$};
        \draw[color = black, line width = 0.8pt] (0 + 7, 0) -- (7, -0.75);
        \node at (1, 1.4) {$\underline{\mathbf{v} \cupp \mathbf{e} = \mathbf{e}}$};
        \node at (6, 1.4) {$\underline{\mathbf{e} \cupp \mathbf{v} = \mathbf{e}}$};
    \end{tikzpicture}
\end{equation}
Outside of these cases, the cup product is zero.
A bilinear map on \(\mathsf{TC}^\bullet\) is then inherited from \(C^\bullet\) by shifting degrees in the cup product above.
This is not a cup product on the nose, as it maps \(\mathsf{TC}^0 \times \mathsf{TC}^0 \to \mathsf{TC}^1\), but it nonetheless defines a cohomology action.

We remark that the formula in Eq.~\eqref{eq-tensorproductcupp} enables us to further define the cup product for any higher dimensional product of Ising models, and hence higher dimensional toric codes, as well.

\subsubsection{Cohomology Invariants from the Cup Product}

Equipped with the cup product, we now discuss how these products can be used to define \textbf{cohomology invariants}.
These invariants are exactly how the cup product will be used to define logical gates.

To define these invariants, let $\mathcal{M}_{(\ell)} \in C_{\ell}$ be an $\ell$-cycle of the chain complex $C_{\bullet}$.
Then, our cohomology invariant will be defined via a function $\mathcal{I}_{\mathcal{M}_{\ell}}$ that takes in a set of $\ell$ cochains $\{\mathbf{a}_{p_j} \in C^{p_j} \}_{j = 1}^{\ell}$ whose degrees $p_j$ sum to $\ell$, and returns an element of $\mathbb{F}_2$.
In particular, it is concretely defined as: 
\begin{equation} \label{eq-cupproductinvariant}
    \mathcal{I}_{\mathcal{M}_{(\ell)}}(\{\mathbf{a}_{p_j}\}) = \int_{\mathcal{M}_{(\ell)}} \mathbf{a}_{p_1} \cupp (\mathbf{a}_{p_2} \cupp (\cdots \cupp \mathbf{a}_{p_\ell}))) 
\end{equation}
where the integral symbol indicates evaluation, i.e.
\begin{equation}
    \int_{\mathcal{M}_{(\ell)}} \mathbf{c} \equiv \mathbf{c}(\mathcal{M}_{(\ell)}) \qquad \mathbf{c} \in C^{\ell},
\end{equation}
and is used to make contact with notation developed in the context of calculus on manifolds where these invariants originate.
Crucially, when $\mathcal{I}_{\mathcal{M}_{(\ell)}}$ is evaluated on co-cycles, it only depends on their cohomology classes and the homology class of the cycle $\mathcal{M}_{(\ell)}$.
This is established with the following Lemma: 
\begin{lem}[Cohomology Invariants]
    Let $C^{\bullet}$ be a co-chain complex, $\{\boldsymbol{\gamma}_{p_j} \in C^{p_j}\}_{j = 1}^{\ell}$ be a collection of $\ell$ cocycles with degrees $p_j$ that sum to $\ell$, and let $\mathcal{M}_{(\ell)} \in C_{\ell}$ be an $\ell$-cycle in the dual chain complex.
    Then, the invariant: 
    \begin{equation}
        \mathcal{I}_{\mathcal{M}_{(\ell)}}(\{\boldsymbol{\gamma}_{p_j}\})
    \end{equation}
    only depends on the cohomology classes of $\boldsymbol{\gamma}_{p_j}$ and the homology class of the cycle $\mathcal{M}_{(\ell)}$.
\end{lem}
\begin{proof}
    We want to show that: 
    \begin{equation}
        \mathcal{I}_{\mathcal{M}'_{(\ell)}} (\{\boldsymbol{\gamma}'_{p_j}\}) = \mathcal{I}_{\mathcal{M}_{(\ell)}} (\{\boldsymbol{\gamma}_{p_j}\})
    \end{equation}
    where $\boldsymbol{\gamma}_{p_j}' = \boldsymbol{\gamma}_{p_j} + \exd \mathbf{a}_{p_{j} - 1}$ for some $\mathbf{a}_{p_{j} - 1} \in C^{p_j - 1}$, and \(\mathcal{M}'_{(\ell)} = \mathcal{M}_{(\ell)} + \partial b_{(\ell+1)}\) for some \(b_{(\ell+1)} \in C_{\ell+1}\). 
    To do so, it suffices to show that the invariant $\mathcal{I}_{\mathcal{M}_{(\ell)}} (\boldsymbol{\gamma}_{p_j})$ vanishes if any $\boldsymbol{\gamma}_{p_j}$ is a coboundary or if $\mathcal{M}_{(\ell)}$ is a boundary.
    We can check each of these with an explicit calculation.
    Suppose that $\boldsymbol{\gamma}_{p_j} = \exd \mathbf{a}_{p_{j}-1}$ for an arbitrary $j$.
    Then: 
    \begin{align}
        &\mathcal{I}_{\mathcal{M}_{(\ell)}}(\{\boldsymbol{\gamma}_{p_j}\}) = \int_{\mathcal{M}_{(\ell)}} \boldsymbol{\gamma}_{p_1} \cupp (\cdots  \cupp (\exd \mathbf{a}_{p_j-1} \cupp (\cdots \cupp \gamma_{p_{\ell}})))\nonumber \\   
        &= \sum_{k \neq j = 1}^{\ell} \int_{\mathcal{M}_{(\ell)}} \boldsymbol{\gamma}_{p_1} \cupp (\cdots  \cupp (\exd \boldsymbol{\gamma}_{p_k} \cdots \cupp (\mathbf{a}_{p_j-1} \cupp \cdots )))  \nonumber\\
        &+ \int_{\partial \mathcal{M}_{(\ell)}}\boldsymbol{\gamma}_{p_1} \cupp (\cdots  \cupp (\mathbf{a}_{p_j-1} \cupp (\cdots \cupp \gamma_{p_{\ell}}))) \nonumber
    \end{align}
    where the second and third line followed from the fact that the cup product satisfies the Leibniz rule.
    Note that each term in the second and third line vanishes because $\exd \boldsymbol{\gamma}_k = 0$ for all $k$ and $\partial \mathcal{M}_{(\ell)} = 0$.
    This confirms that $\mathcal{I}_{\mathcal{M}_{(\ell)}} (\boldsymbol{\gamma}_{p_j})$ vanishes if any $\boldsymbol{\gamma}_{p_j}$ is a co-boundary.
    Hence, by linearity, $\mathcal{I}_{\mathcal{M}_{(\ell)}}(\{\boldsymbol{\gamma}_{p_j}\})$ depends only on the cohomology classes of the cycles $\boldsymbol{\gamma}_{p_j}$.

    Now suppose that $\mathcal{M}_{(\ell)} = \partial b_{(\ell + 1)}$ for some chain $b_{(\ell + 1)} \in C_{\ell+1}$, then: 
    \begin{align}
        &\mathcal{I}_{\mathcal{M}_{(\ell)}}(\{\boldsymbol{\gamma}_{p_j}\}) = \int_{\partial b_{(\ell + 1)}} \boldsymbol{\gamma}_{p_1} \cupp (\boldsymbol{\gamma}_{p_2} \cupp (\cdots \cupp \boldsymbol{\gamma}_{p_\ell}))) \nonumber \\
        &=\sum_{j  = 1}^{\ell} \int_{ b_{(\ell+1)}} \boldsymbol{\gamma}_{p_1} \cupp (\cdots  \cupp (\exd \boldsymbol{\gamma}_{p_j} \cupp (\cdots \cupp \gamma_{p_{\ell}}))) = 0\nonumber 
    \end{align}
    where the second line follows from the fact that $\int_{\partial b} c = \int_{b} \exd c$ (Stokes' theorem) and the Leibniz rule.
    Thus, $\mathcal{I}_{\mathcal{M}_{(\ell)}}(\{\boldsymbol{\gamma}_{p_j}\})$ vanishes if $\mathcal{M}_{(\ell)}$ is a boundary and by linearity, it only depends only on the homology class of $\mathcal{M}_{(\ell)}$.
\end{proof}

At this point two remarks are in order.
First, we remark that in our proof above, to show that $\mathcal{I}_{\mathcal{M}_{(\ell)}}(\{\boldsymbol{\gamma}_{p_j}\})$ depended only on the cohomology class $\boldsymbol{\gamma}_{p_j}$, we simply had to use the fact that for any cochains $\mathbf{a}_{p_j} \in C^{p_j}$ (where the $p_j$'s sum to $\ell -1$): 
\begin{equation} \label{eq-integratedLeibnizrule}
    \sum_{k = 1}^{\ell} \int_{\mathcal{M}_{(\ell)}} \mathbf{a}_{p_1} \cupp (\cdots \cupp \exd \mathbf{a}_{p_{k}} \cupp (\cdots \cupp \mathbf{a}_{p_{\ell}}))) = 0 
\end{equation}
This fact follows from the Leibniz rule but does not require it.
As such, a recent work~\cite{breuckmann2025cupsgatesicohomology} introduced a \emph{pseudo-cup product} that obeys Eq.~\eqref{eq-integratedLeibnizrule}, dubbed the $\ell$-fold \textbf{integrated Leibniz rule}, enabling them to find certain cohomology invariants without finding a product structure on cohomology.
In the next subsection, we will connect the invariants found from the integrated Leibniz rule to the chain map hierarchy.

Second, the invariants above are connected to the geometric concept of intersection, which is most naturally seen by introducing the concept of the \textit{cap product}.
While we do not discuss this interpretation in the main text, a discussion of this is provided in Appendix~\ref{app:CupsAndPoincare}.

\subsubsection{Logical Gates from Cohomology Invariants}

We now connect the invariants of Eq.~\eqref{eq-cupproductinvariant} to logical gates. 
In particular, we first prove a theorem, a version of which appears in  Refs.~\cite{zhu2025, breuckmann2025cupsgatesicohomology, lin2024transversalnoncliffordgatesquantum, zhu2025topological}, showing how to define a logical $\mathsf{C^{\ell-1}Z}$ gate with the invariants above.
Subsequently, we will prove a Lemma connecting the depth of these gates to the input/output sparsity of the cup product, which will be important for our subsequent results.

The cup product defines a logical gate via the following theorem:
\begin{theo}[Copy-Cup Gates]
Let $C^{\bullet}$ be a cochain complex, and consider a qLDPC code $C[-1]$ (equivalently, place qubits on \(\mathsf{B}(C^1)\)). Label $\ell$ codeblocks of $C^\bullet$ as $C^{(1)}, \ldots, C^{(\ell)}$.
Then, the following unitary performs a logical gate on these $\ell$ codeblocks:
\begin{equation} \label{eq-copycup}
    U^{\cupp}_{\mathsf{C^{\ell-1}Z}}[\mathcal{M}_{(\ell)}] = \prod_{\mathbf{e}_1, \ldots, \mathbf{e}_{\ell} \in \mathsf{B}^{1}(C)} (-1)^{n_{e_1}^{(1)} \cdots n_{e_{\ell}}^{(\ell)} \mathcal{I}_{\mathcal{M}_{(\ell)}}(\{\mathbf{e}_i\})}
\end{equation}
where $\mathcal{M}_{(\ell)}$ is an $\ell$-cycle of $C_{\bullet}$ and  $\mathcal{I}_{\mathcal{M}_{(\ell)}}(\{\mathbf{e}_i\})$ is defined in Eq.~\eqref{eq-cupproductinvariant}.
\end{theo}

One can define similar gates on \(C[-k]\) (qubits placed on \(\mathsf{B}(C^k)\)) for any \(k\), where \(\mathcal{M}_{(\ell)}\) must be replaced with an \(\ell k\)-cycle \(\mathcal{M}_{(\ell k)}\).

\begin{proof}
    A simple proof follows by considering the action on any wavefunction in the logical computational basis.
    In particular, note that the code space of the qLDPC codes associated to (shifts of) $C^{(1)}, \ldots, C^{(\ell)}$ is spanned by a basis of states of the form: 
    \begin{equation}
    \ket{\Psi_{\boldsymbol{\gamma}_1, \ldots, \boldsymbol{\gamma}_j}} =  \sum_{\{\mathbf{a}_j \in C^{0}_{(j)}\}}\prod_{j = 1}^{\ell} X^{(j)}_{\boldsymbol{\gamma_j} + \exd \mathbf{a}_j} \ket{\mathbf{0}}^{\otimes \ell}
    \end{equation}
    where $\boldsymbol{\gamma}_j$ is a (potentially non-trivial) $1$-cocycle in $C^1$ and hence $X^{(j)}_{\boldsymbol{\gamma}_j + \exd \mathbf{a}_{j}}$ labels a logical $X$ operator in code $C^{(j)}$.
    Moreover, $\ket{\mathbf{0}}^{\otimes (\ell + 1)}$ is the state where all qubits in codes $C^{(1)}, \cdots, C^{(\ell + 1)}$ are in the $\ket{0}$ state. 
    
    Now, if we act $U_{\mathsf{C}^{\ell-1} \mathsf{Z}}^{\cupp}[\mathcal{M}_{(\ell)}]$ on any term in the sum above, it is easy to verify that we get:
    \begin{align} &U^{\cupp}_{\mathsf{C^{\ell-1}Z}}[\mathcal{M}_{(\ell)}]\prod_{j = 1}^{\ell} X_{\boldsymbol{\gamma}_j + \exd \mathbf{a}_j} \ket{\mathbf{0}}^{\otimes \ell} \\
    &= (-1)^{\mathcal{I}_{\mathcal{M}_{(\ell)}} (\{\boldsymbol{\gamma}'_i\})}\prod_{j = 1}^{\ell} X_{\boldsymbol{\gamma}_j + \exd \mathbf{a}_j} \ket{\mathbf{0}}^{\otimes \ell} \nonumber
    \end{align}
    where $\boldsymbol{\gamma}_j' = \boldsymbol{\gamma}_j + \exd \mathbf{a}_j$.
    Now, because  $\mathcal{I}_{\mathcal{M}_{(\ell)}}(\{\boldsymbol{\gamma}_j'\})$ is  a cohomology invariant, it is insensitive to shifting $\boldsymbol{\gamma}_j'$ by a coboundary and hence
    $\mathcal{I}_{\mathcal{M}_{(\ell)}}(\{\boldsymbol{\gamma}_j'\}) = \mathcal{I}_{\mathcal{M}_{(\ell)}}(\{\boldsymbol{\gamma}_j\})$.
    Consequently, 
    \begin{equation}
        U^{\cupp}_{\mathsf{C^{\ell-1}Z}}[\mathcal{M}_{(\ell)}] \ket{\Psi_{\boldsymbol{\gamma}_1, \ldots, \boldsymbol{\gamma}_j}} = (-1)^{\mathcal{I}_{\mathcal{M}_{(\ell)}}(\{\boldsymbol{\gamma}_j\})} \ket{\Psi_{\boldsymbol{\gamma}_1, \ldots, \boldsymbol{\gamma}_j}}
    \end{equation}
    Thus, $U^{\cupp}_{\mathsf{C^{\ell-1}Z}}[\mathcal{M}_{(\ell)}]$ defines a logical, possibly trivial, gate of \(\mathsf{C^{\ell-1}Z}\) type.
\end{proof}

We can connect the depth of the gate of Eq.~\eqref{eq-copycup} to the sparsity of the cup product.
Crucially for the gate to be constant depth we do not require the cup product to be sparse in all degrees! 
In particular:
\begin{lem}
    If the cup product is input sparse in degree $1$ and $q$ for integers $q$ between $1$ and $\ell-1$ and output sparse in degree $1$ and $q'$ for integers $q'$ between $1$ and $\ell - 2$, then the gate of Eq.~\eqref{eq-copycup} can be implemented in constant depth.
\end{lem}

\begin{proof}
We start by remarking that the gate of Eq.~\eqref{eq-copycup},
$U^{\cupp}_{\mathsf{C^{\ell-1} Z}}[\mathcal{M}_{(\ell)}]$, is a product of
commuting $\mathsf{C^{\ell-1} Z}$ gates.
As a consequence, it suffices to show that the number of gates acting on any
particular qubit is $O(1)$.
Pick an arbitrary qubit at location $\mathbf{e}_j$ in code $C^{(j)}$.
For a $\mathsf{C^{\ell-1} Z}$ gate containing $\mathbf{e}_j$ to appear in
Eq.~\eqref{eq-copycup}, we must have that
\begin{equation}
\mathcal{I}_{\mathcal{M}_{(\ell)}}\left(\{\mathbf{e}_i\}\right)
=
\int_{\mathcal{M}_{(\ell)}}
\mathbf{e}_1 \cupp
\left(
\mathbf{e}_2 \cupp
\left(
\cdots
\left(
\mathbf{e}_{\ell-1} \cupp \mathbf{e}_{\ell}
\right)
\right)
\right)
\end{equation}
is non-zero and our goal is to show that the number of tuples
$(\mathbf{e}_i)$ containing a fixed $\mathbf{e}_j$ and satisfying this
property is $O(1)$.

The proof proceeds by considering partial products.
For $j \leq k \leq \ell-1$, define
$\mathbf{r}_k
=
\mathbf{e}_k \cupp
\left(
\mathbf{e}_{k+1} \cupp \left(\cdots \mathbf{e}_\ell \right)
\right)
\in C^{\ell+1-k}$ and $\mathbf{r}_{\ell}=\mathbf{e}_{\ell}$.
By output sparsity, each $\mathbf{r}_k$ has $O(1)$ support (unless \(k =j =1\)).
Indeed, this holds for $\mathbf{r}_{\ell}$ trivially, while if
$\mathbf{r}_{k+1}$ has $O(1)$ support, then
$\mathbf{r}_k=\mathbf{e}_k\cupp\mathbf{r}_{k+1}$ is a sum of $O(1)$ cup
products of basis elements, each of which has $O(1)$ support.

First, we show that there are only \(O(1)\) choices of $\mathbf{e}_{j+1},\ldots,\mathbf{e}_{\ell}$ that make \(\mathcal{I}_{\mathcal{M}_{(\ell)}}\left(\{\mathbf{e}_i\}\right) \neq 0\) with \(\mathbf{e}_j\) fixed.
As \(\mathcal{I}_{\mathcal{M}_{(\ell)}}\left(\{\mathbf{e}_i\}\right) \neq 0\), we have that $\mathbf{r}_j\neq 0$, and thus there is some
$\mathbf{x}\in\operatorname{support}\left(\mathbf{r}_{j+1}\right)$ such
that $\mathbf{e}_j\cupp\mathbf{x}\neq 0$.
By input sparsity, there are only $O(1)$ possible choices for $\mathbf{x}$.
For any fixed
$\mathbf{x}\in\operatorname{support}\left(\mathbf{r}_{j+1}\right)$, there
must be some
$\mathbf{y}\in\operatorname{support}\left(\mathbf{r}_{j+2}\right)$ such
that
\begin{equation}
\mathbf{x}
\in
\operatorname{support}\left(
\mathbf{e}_{j+1}\cupp\mathbf{y}
\right).
\end{equation}
The second condition of output sparsity gives only $O(1)$ possible pairs
$(\mathbf{e}_{j+1},\mathbf{y})$.
Applying the same argument to each such $\mathbf{y}$, and continuing up to
$\mathbf{r}_{\ell}=\mathbf{e}_{\ell}$, gives only $O(1)$ choices of
$\mathbf{e}_{j+1},\ldots,\mathbf{e}_{\ell}$.

Second, we similarly show that there are only \(O(1)\) choices of $\mathbf{e}_{j-1},\ldots,\mathbf{e}_1$ that make \(\mathcal{I}_{\mathcal{M}_{(\ell)}}\left(\{\mathbf{e}_i\}\right) \neq 0\) with \(\mathbf{e}_j\) fixed.
As $\mathbf{r}_j$ has $O(1)$ support, if $\mathbf{e}_{j-1}\cupp\mathbf{r}_j\neq 0$, then there is some
$\mathbf{x}\in\operatorname{support}\left(\mathbf{r}_j\right)$ such that
$\mathbf{e}_{j-1}\cupp\mathbf{x}\neq 0$.
Input sparsity therefore gives only $O(1)$ choices for
$\mathbf{e}_{j-1}$.
The resulting partial product again has $O(1)$ support by output sparsity,
so repeating this argument gives only $O(1)$ choices of
$\mathbf{e}_{j-2},\ldots,\mathbf{e}_1$.
Thus, any qubit participates in only $O(1)$ of the $\mathsf{C^{\ell-1} Z}$ gates
appearing in Eq.~\eqref{eq-copycup} and hence the gate can be implemented in constant depth.

We note in particular that the total product \(\mathbf{e}_1 \cupp(... \cupp \mathbf{e}_\ell)\) does \emph{not} have to be low weight, which is why output sparsity is not required between degrees \(1\) and \(\ell-1\).
The only data from the total product that is used is binary: whether or not it pairs with the cycle \(\mathcal{M}_{(\ell)}\).
Thus, a large total product does not result in a circuit with many gates.
\end{proof}

\subsection{Cup Products and Cohomology Invariants from the Chain Map Hierarchy} \label{subsec:universalcupproduct}

Thus far, we have discussed cup products, cohomology invariants derived from them, and the role of these invariants in defining logical gates.
Given these connections, numerous approaches have been developed to define cup products or their cohomology invariants on the chain complexes associated with qLDPC codes.
These approaches take one of two forms:
\begin{enumerate}
    \item \textit{Full Leibniz Rule Cup Products}: In some works \cite{lin2024transversalnoncliffordgatesquantum, zhu2025topological, li_poincare_2025, li2026theorycohomologicalinvariantsquantum}, cup products---i.e. bilinear co-chain products that satisfy the Leibniz Rule---are defined on chain complexes associated with certain qLDPC codes.
    These products are typically defined by either exploiting the sheaf structure naturally associated with particular qLDPC codes or by mapping qLDPC codes to high-dimensional manifolds.

    \item \textit{Integrated Leibniz Rule Cup Products}: In other works \cite{breuckmann2025cupsgatesicohomology, menon2025magic, tiew2026copycupgatestensorproducts}, a pseudo-cup product is developed that doesn't satisfy the Leibniz rule.
    These pseudocup products satisfy a relaxation of the Leibniz rule called the integrated Leibniz rule [c.f. discussion around Eq.~\eqref{eq-integratedLeibnizrule}], which enables them to define certain cohomology invariants without providing a product structure on cohomology.

\end{enumerate}

We now show how these approaches connect to the chain map hierarchy language built up in our work.
In particular, we point out that cup products that satisfy the Leibniz rule are fully parameterized by chain maps (i.e.\ cycles) at the third level of the chain map hierarchy and hence product structures on cohomology are encompassed within our formalism.
Moreover, we further point out that the cohomology invariants developed from cup products or pseudo-cup products further correspond to various cycles in the chain map hierarchy.
Both observations follow from well-known facts in homological algebra.

These observations highlight the role of the chain map hierarchy formalism in relation to prior constructions.
In particular, prior work on cup products and cohomology invariants for qLDPC codes provide routes for explicitly finding sparse constructions of gates on qLDPC codes that preserve their code space.
Such approaches do not explicitly specify a target logical action of these cup products and do not necessarily exhaust the logical non-Clifford gates that can be present in a code.
These explicit constructions also come at the cost of a more restrictive structure.
As a concrete example, the construction of Ref.~\cite{breuckmann2025cupsgatesicohomology} has the following features:
(1)~the \(\ell\)-cycle \(\mathcal{M}_{(\ell)}\) is always a global sum of all elements of the basis \(\mathsf{B}(C_\ell)\),
(2)~\(\mathbf{a\cupp b} \neq 0\) only when \(\mathbf{a}\) and \(\mathbf{b}\) are neighbors on the qLDPC code's hypergraph (defined by the checks), and
(3)~if \(\mathbf{a\cupp b} = \mathbf{c}\), then \(\mathbf{c}\) is close on the hypergraph to \(\mathbf{a}\) and \(\mathbf{b}\).
The chain map hierarchy language complements these constructions with greater generality, at the cost of sparsity no longer being automatic.
It provides an exhaustive search space in which to look for cup products and cohomology invariants, and target logical actions can be specified as part of the search.
It also relaxes all the constraints on locality above.
As an explicit example for classical codes, given a basis of cohomology classes \([\mathbf{c}^{(i)}] \in \mathsf{B}(H^1(C))\), we can pick a collection of chain maps \(\phii^{(i)}_\bullet : C^\bullet \to C_\bullet\) indexed by \(\mathsf{B}(H^1(C))\) to define the following valid cup product on \(C^\bullet\).
\begin{equation}
    \mathbf{a}_p \cupp \mathbf{b}_q = 
    \left\{
    \begin{array}{ll}
        \sum_{i} \mathbf{c}^{(i)} \times  \mathbf{a}_p(\phii^{(i)}(\mathbf{b}_q))\quad & p+q=1, \\
        0 \quad & \text{otherwise.}
    \end{array}
    \right.
\end{equation}
This escapes all the conditions mentioned in the previous paragraph:
(1)~integrating against distinct elements of a dual basis of cycles \(\mathcal{M}_{(1)}^{(i)}\) yields distinct invariants and gates,
(2)~\(\mathbf{a} \cupp \mathbf{b}\) can be highly non-local if \(\phii^{(i)}\) is, and
(3)~the image contains a fixed term \(\mathbf{c}_i\), which can be arbitrarily distant from the location of \(\mathbf{a}\) and \(\mathbf{b}\) in the hypergraph.
Despite these relaxations, this cup product can still yield a transversal gate.
Indeed, these relaxations are necessary to find both the \(\mathsf{CCZ}\) gate from Subsection~\ref{subsec:ALPIsiGate}, and the gates found in a companion work~\cite{Christos2026nonabelian}.

The above construction is natural in the chain map hierarchy formalism for cup products, which we now describe.

\subsubsection{Cup Products and the Chain Map Hierarchy}

To establish a connection between the cup product and the chain map hierarchy, we start by showing that elements of the chain map hierarchy of $C^{\bullet}$ can be used to provide a universal parameterization of the space of cup products, i.e. bilinear products that satisfy the Leibniz rule.

In particular, let us start by defining the following family of cup products before proving that this family is universal.
This family of cup products is defined for a chain complex $C^{\bullet}$ and is parameterized by the $0$-cycles of the chain map complex $[C^{\bullet} \otimes C^{\bullet}, C^{\bullet}] \simeq C^{\bullet} \otimes C_{\bullet} \otimes C_{\bullet}$,  which crucially lives at the third level of the chain map hierarchy of $C^{\bullet}$.
Specifically, for a given $0$-cycle  $\Delta_{\bullet} \in [C^{\bullet} \otimes C^{\bullet}, C^{\bullet}]$ with components \(\Delta_r: (C^\bullet \otimes C^\bullet)^r \to C^r\), we can define the cochain product associated with $\Delta$, denoted $\cupp_{\Delta}$, as:
\begin{equation}\label{eq-universal_cup}
    \mathbf{a}_p \cupp_{\Delta} \mathbf{b}_q \equiv \Delta_{p + q}(\mathbf{a} \otimes \mathbf{b})
\end{equation}
where $\mathbf{a}_p \in C^p$ and $\mathbf{b}_q \in C^q$.

We now prove that this cochain product defines a cup product by proving it is bilinear and satisfies the Leibniz rule.
In particular,

\begin{lem}[Cup Products from the Chain Map Hierarchy] Given a chain complex $C^{\bullet}$ and a $0$-cycle $\Delta_{\bullet} \in [C^{\bullet} \otimes C^{\bullet}, C^{\bullet}]$, the cochain product of Eq.~\eqref{eq-universal_cup} defines a cup product.

Moreover, the sparsity of this cup product is related to the sparsity of $\Delta_{\bullet}$ when decomposed into a direct sum of maps: 
\begin{equation}
    \Delta_{\bullet} = \bigoplus_{r} \underbrace{\left(\bigoplus_{p + q = r} \Delta_{p, q}\right)}_{\Delta_r}
\end{equation}
where $\Delta_{p, q}: C^p \otimes C^q \to C^{p + q}$.
In particular, $\cupp_{\Delta}$ is output sparse in degree $p$ and $q$ if $\Delta_{p, q}$ is sparse as a matrix [with rows and columns labeled by \(\mathsf{B}(C^{p+q})\) and \(\mathsf{B}(C^{p} \otimes C^{q})\)]. 
Moreover, viewing $\Delta_{p, q}$ as a three-component tensor $(\Delta_{p, q})^{\alpha \beta}_{\gamma}$  [with $\alpha$, \(\beta\), $\gamma$ labeling $\mathsf{B}(C^p)$, $\mathsf{B}(C^q)$, and $\mathsf{B}(C^{p + q})$, respectively], $\cupp_{\Delta}$ is input sparse if this tensor has an $O(1)$ number of non-zero rows for fixed $\alpha$ and $O(1)$ number of non-zero rows for fixed $\beta$.

\end{lem}

\begin{proof}
    Note that Eq.~\eqref{eq-universal_cup} is manifestly bilinear (due to linearity of \(\Delta\) and bilinearity of \(\otimes\)) so it suffices to show that it satisfies the Leibniz rule.
    To do so, let us note that: 
    \begin{align}
        \exd (\mathbf{a}_p \cupp_{\Delta} \mathbf{b}_q) &= \exd \Delta_{p + q} (\mathbf{a} \otimes \mathbf{b}) = \Delta_{p + q + 1}\left( \exd (\mathbf{a} \otimes \mathbf{b}\right)) \nonumber \\
        &=\Delta_{p + q + 1}( \mathbf{a} \otimes \exd \mathbf{b} + \exd \mathbf{a} \otimes \mathbf{b}) \nonumber \\
        &= \mathbf{a} \cupp_{\Delta} \exd \mathbf{b} + \exd \mathbf{a} \cupp_{\Delta} \mathbf{b}
    \end{align}
where the second equality on the first line follows from the fact that $\Delta_{\bullet}$ is a chain map and hence $\exd \circ \Delta_{p + q} = \Delta_{p + q + 1} \circ \exd$, the second line follows from the definition of the co-boundary map on the tensor product, and the last line follows by linearity and the definition of $\cupp_{\Delta}$. 
Consequently, $\cupp_{\Delta}$ satisfies the Leibniz rule and defines a cup product on $C^{\bullet}$.

Moreover, showing that the sparsity of $\cupp_{\Delta}$ is a consequence of the sparsity of $\Delta$ immediately follows from the definition of input and output sparsity.
\end{proof}

The parameterization of Eq.~\eqref{eq-universal_cup} nicely establishes a way to define a cup product given a certain element of the chain map hierarchy.
However, we now show that this parameterization is, in fact, \textit{universal}---i.e. we show that any bilinear co-chain product that satisfies the Leibniz rule must take the form of Eq.~\eqref{eq-universal_cup}.

\begin{theo}[Universality of Cup Product Parameterization]
Let $\cupp$ be a cup product defined for a chain complex $C^{\bullet}$.
Then, there exists a $\Delta \in [C^{\bullet} \otimes C^{\bullet}, C^{\bullet}]$ such that $\cupp = \cupp_{\Delta}$.
\end{theo}
\begin{proof}
The proof follows by applying the so-called \emph{universal property} of the tensor product of chain complexes: any bilinear map that satisfies the Leibniz rule factors through a chain map \(\Delta : C^\bullet \otimes C^\bullet \to C^\bullet\).
This property itself follows immediately from the usual universal property of the tensor product for vector spaces.

More explicitly, let us start by noting that we can decompose the $\cupp$ product into a direct sum of bilinear maps: 
\begin{equation}
    \cupp\,  = \bigoplus_{r}  \underbrace{\left(\bigoplus_{p + q = r} \cupp_{p, q} \right)}_{\cupp_{r}}
\end{equation}
where $\cupp_{p, q}: C^p \times C^q \to C^{p + q}$.
Now, note that since $\cupp_{p, q}$ is a bilinear map, the universal property of the tensor product for vector spaces is that any bilinear map on vector spaces must factor through the tensor product.
\begin{equation}
    \begin{tikzpicture}[scale = 0.85, baseline = {([yshift=-.5ex]current bounding box.center)}]
        \node at (0, 0) {$C^p \times C^q$};
        \draw[-stealth] (0.8, 0) --node[pos=0.4, above]{$\otimes$} (1.65, 0);
        \draw[-stealth] (0.25, -0.4) -- node[pos=0.55, left] {$\cupp\ \ $} (1.9, -1.7);
        \draw[-stealth, dashed] (2.5, -0.4) -- node[pos=0.5, right] {$\Delta_{p, q}$} (2.5, -1.7);
        \node at (2.5, 0) {$C^p \otimes C^q$};
        \node at (2.5, -2) {$C^{p+q}$};
    \end{tikzpicture}
\end{equation}
In other words, any bilinear map must satisfy $\mathbf{a}_p\cupp_{p, q} \mathbf{b}_q = \Delta_{p, q}(\mathbf{a}_p\otimes \mathbf{b}_q)$ for some $\Delta_{p, q}$ and any choice of inputs $\mathbf{a}_p \in C^p$ and $\mathbf{b}_q \in C^q$.
Now, note that the collection of maps $\Delta_{\bullet} = \{\Delta_{p, q}\}$ forms a $0$-chain of the chain map complex of $[C^{\bullet} \otimes C^{\bullet}, C^{\bullet}]$ by definition.
Consequently, it suffices to show that $\Delta_{\bullet}$ is a $0$-cycle of $[C^{\bullet} \otimes C^{\bullet}, C^{\bullet}]$.
To do so, let us recognize that since $\cupp$ satisfies the Leibniz rule, it is the case that: 
\begin{equation}
    \mathbf{a}_p \cupp_{p, q+1} \exd \mathbf{b}_{q} + \exd \mathbf{a}_p \cupp_{p +1, q}  \mathbf{b}_{q} = \exd (\mathbf{a}_p \cupp_{p, q} \mathbf{b}_q)  
\end{equation}
As a consequence, we have that: 
\begin{equation}
    \exd \Delta_{p + q} (\mathbf{a} \otimes \mathbf{b}) = \Delta_{p + q + 1} (\exd (\mathbf{a} \otimes \mathbf{b}))
\end{equation}
Since this holds for all $\mathbf{a} \otimes \mathbf{b}$, we find that  $\exd \circ \Delta_{p + q} = \Delta_{p + q + 1} \circ \exd$, which cements the collection of maps $\Delta_{p, q}$ as a $0$-cycle of $[C^{\bullet} \otimes C^{\bullet}, C^{\bullet}]$.
\end{proof}

\subsubsection{Cohomology Invariants and the Chain Map Hierarchy}

The above shows that the chain map hierarchy provides a natural language for characterizing the space of cup products.
We now show that it further encompasses cohomology invariants derived from the cup product and even those described from pseudo-cup products that satisfy the integrated Leibniz rule of Eq.~\eqref{eq-integratedLeibnizrule}.
In particular, given the fact that the cohomology invariant $\mathcal{I}_{\mathcal{M}_{(\ell)}}$ of Eq.~\eqref{eq-cupproductinvariant} defines a logical $\mathsf{C}^{\ell-1} \mathsf{Z}$ gate via Eq.~\eqref{eq-copycup}, we may anticipate that $\mathcal{I}_{\mathcal{M}_{(\ell)}}$ can be used to define a cycle in the chain map complex $[C^{\bullet}, [C^{\bullet}, [\cdots, C_{\bullet}]]]$ at the $\ell$-th level of the chain map hierarchy.
We show that this is indeed the case.

In particular, the cohomology invariant $\mathcal{I}_{\mathcal{M}_{(\ell)}}$ defines a tensor: 
\begin{equation} \label{eq-invarianttensor}    \hat{\mathcal{I}}_{\mathcal{M}_{(\ell)}} = \sum_{p_1 + \cdots+ p_{\ell} = \ell} \sum_{\{\mathbf{x}_{p_j} \in \mathsf{B}(C^{p_j})\}} \mathcal{I}_{\mathcal{M}_{(\ell)}}(\{\mathbf{x}_{p_j}\})\, x_{p_1} \otimes \cdots \otimes x_{p_\ell}
\end{equation}
which we show is an $\ell$-cycle of the $\ell$-fold tensor product chain complex $C_{\bullet}^{\otimes \ell} \isom [C^{\bullet}, [C^{\bullet}, [\cdots, C_{\bullet}]]]$.
Note that in the above equation, $x_{p_j} = \mathbf{x}_{p_j}^{\mathsf{T}}$.

We show this via the following Lemma: 
\begin{lem}[Cohomology Invariants in the Chain Map Hierarchy] Suppose that $\cupp$ is a bilinear cochain product on the chain complex $C_{\bullet}$ that satisfies the $\ell$-fold integrated Leibniz rule of Eq.~\eqref{eq-integratedLeibnizrule} with respect to an $\ell$-cycle $\mathcal{M}_{(\ell)}$.
Then, the invariant tensor of Eq.~\eqref{eq-invarianttensor} defines a $\ell$-cycle of $C_{\bullet}^{\otimes \ell} \isom [C^{\bullet}, [C^{\bullet}, [\cdots, C_{\bullet}]]]$.
    
\end{lem}

\begin{proof}
    To show this is the case, it suffices to show that any cochain $\mathbf{y}_{q_1} \otimes \cdots \otimes \mathbf{y}_{q_\ell}$, where $\mathbf{y}_{q_j} \in C^{q_j}$ and the $q_j$'s sum to $\ell-1$, evaluates to zero on $\partial \hat{\mathcal{I}}_{\mathcal{M}_{(\ell)}}$.
    To show this, note that it is easy to verify that: 
    \begin{align}
        &\mathbf{y}_{q_1} \otimes \cdots \otimes \mathbf{y}_{q_\ell} (\partial \hat{\mathcal{I}}_{\mathcal{M}_{(\ell)}}) = \exd (\mathbf{y}_{q_1} \otimes \cdots \otimes \mathbf{y}_{q_\ell} ) (\hat{\mathcal{I}}_{\mathcal{M}_{(\ell)}}) \nonumber \\
        &= \sum_{k = 1}^{\ell} \int_{\mathcal{M}_{(\ell)}} \mathbf{y}_{q_1} \cupp (\cdots \cupp \exd \mathbf{y}_{q_{k}} \cupp (\cdots \cupp \mathbf{y}_{q_{\ell}}))
    \end{align}
    The above definitionally equals zero for any bilinear co-chain product that satisfies the integrated Leibniz rule.
    Since the above holds for any basis $\ell-1$ chain of $C_{\bullet}^{\otimes \ell}$, the above shows that $\hat{I}_{\mathcal{M}_{(\ell)}}$ is a cycle.
\end{proof}

We conclude by remarking that the cohomology invariant tensors derived from the cup product can be nicely understood via the diagrammatic language built up earlier in our paper.
In particular, note that the cup product can be represented diagrammetically as:
\begin{equation}
     \mathbf{a} \cupp_{\Delta} \mathbf{b} = \Delta(\mathbf{a} \otimes \mathbf{b}) = \begin{tikzpicture}[scale = 1, baseline = {([yshift=-.5ex]current bounding box.center)}] 
        \draw[color = black] (-0.565, 0.565) -- (0, 0) -- (0.565, 0.565);
        \draw[color = black] (0, 0) -- (0, -0.8);
        \draw[fill = lightdodgerblue] (0,0) circle (0.3);
        \node at (-0.0,0) {\small $\Delta$};
        \draw[fill = lightgray] (-0.565,0.565) circle (0.2);
        \draw[fill = lightgray] (0.565,0.565) circle (0.2);
        \node at (-0.565,0.565) {$\mathbf{a}$};
        \node at (0.565,0.565) {$\mathbf{b}$};
    \end{tikzpicture}
\end{equation}
Consequently, a cohomology invariant of the form Eq.~\eqref{eq-cupproductinvariant} can be represented diagrammatically as: 
\begin{equation}
\mathcal{I}_{\mathcal{M}_{(\ell)}}(\{\mathbf{a}_{i}\}) = \hspace{-7 mm} \begin{tikzpicture}[scale = 0.8, baseline = {([yshift=-.5ex]current bounding box.center)}] 
        \draw[color = black] (-0.565, 0.565) -- (0, 0) -- (0.565, 0.565);
        \draw[color = black] (0, 0) -- (-0.565, -0.565);
        \draw[fill = lightdodgerblue] (0,0) circle (0.3);
        \node at (-0.0,0) {\scriptsize $\Delta$};
        \draw[fill = lightgray] (-0.565,0.565) circle (0.25);
        \draw[fill = lightgray] (0.565,0.565) circle (0.25);
        \node at (-0.565,0.965) {\scriptsize $\ell-1$};
        \node at (0.565,0.965) {\scriptsize $\ell$};

        \draw[color = black] (-0.6, -0.6) -- (-1.695, 0.565);
        \draw[color = black] (-0.6, -0.6) -- (-1, -1);
        \node at (-1.1, -1.1) {$\cdot$};
        \node at (-1.2, -1.2) {$\cdot$};
        \node at (-1.3, -1.3) {$\cdot$};
        \draw[fill = lightdodgerblue] (-0.6,-0.6) circle (0.3);
        \node at (-0.6,-0.6) {\scriptsize $\Delta$};
        \draw[fill = lightgray] (-1.695,0.565) circle (0.25);
        \node at (-1.695,0.965) {\scriptsize $\ell-2$};
        \draw[color = black] (-1.4, -1.4) -- (-1.8, -1.8);
        \draw[color = black] (-1.8, -1.8) -- (-3.955, 0.565);
        \draw[color = black] (-1.3, -1.1) -- (-2.8, 0.465);
        \draw[fill = lightgray] (-3.955,0.565) circle (0.25);
        \node at (-3.955,0.965) {\scriptsize $1$};
        \node at ( -2.9, 0.565) {$\cdots$};
        \draw[color = black] (-1.8, -1.8) -- (-1.8, -2.6);
        \draw[fill = lightdodgerblue] (-1.8,-1.8) circle (0.3);
        \node at (-1.8, -1.8) {\scriptsize $\Delta$};
        \draw[fill = white] (-1.8, -2.6) circle (0.3); 
        \node at (-1.8, -2.6) {\small $\mathcal{M}$};
    \end{tikzpicture} = \begin{tikzpicture}[scale = 0.8, baseline = {([yshift=-.5ex]current bounding box.center)}]
    \draw[color = black] (-0.2,0) -- (-1.2, 1.2);
    \draw[color = black] (0,0) -- (-0.0, 1.2);
    \draw[color = black] (0.1,0) -- (0.6, 1.2);
    \draw[color = black] (0.2,0) -- (1.2, 1.2);
    \node at (-0.45, 1.) {$\cdots$};
    \draw[fill = lightgray] (-1.2,1.2) circle (0.2);
    \draw[fill = lightgray] (0,1.2) circle (0.2);
    \draw[fill = lightgray] (0.6,1.2) circle (0.2);
    \draw[fill = lightgray] (1.2,1.2) circle (0.2);
    \draw[fill = lightdodgerblue] (0,0) circle (0.4);
    \node at (0, 0) {$\mathcal{I}$};
    \end{tikzpicture}
\end{equation}
which makes manifest that $\mathcal{I}_{\mathcal{M}_{(\ell)}}(\{\mathbf{e}_i\})$ defines a chain map at the $\ell$-th level of the chain map hierarchy.

\section{Discussion}
\label{sec:discussion}

\paragraph*{Summary}
The understanding of fault tolerant logical operations becomes increasingly limited as one climbs up the Clifford hierarchy.
Pauli logical gates are very well understood~\cite{Kitaev_2003,Freedman2001,Bombin_2007}: the chain complex formalism for CSS codes has provided a framework for working with Paulis that ultimately resulted in the construction of good qLDPC codes~\cite{Breuckmann_2021_rev, MacKay_2004, Kovalev_2013, Tillich_2014, Leverrier_2015, Hastings_2021, Breuckmann_2021,panteleev_asymptotically_2022, leverrier2022quantumtannercodes, dinur_good_2023}.
Moving to Cliffords, let alone higher levels of the hierarchy, even extremely basic questions regarding the existence of transversal logical gates apparently become very difficult to address.

The chain map hierarchy puts transversal gates from any level of the Clifford hierarchy  on the same footing as Pauli logicals.
Using concepts from homological algebra, we show that logical gates between two codes are themselves logicals of an auxiliary code---the chain map complex---which is isomorphic to the tensor product of the two codes.
Iteratively forming chain map complexes produces gates that climb the Clifford hierarchy.
Thus, the discovery of both Clifford and non-Clifford logical operations is reduced to a question about Pauli logicals of simple product codes.
Indeed, deformations of the Pauli logicals can be used to find transversal implementations of candidate logical actions on the codes of interest.

The chain map hierarchy is both a useful method for gate discovery, and a unifying framework capturing most prior methods of constructing logical operations in qLDPC codes.
Indeed, we have demonstrated the utility and breadth of this framework by discovering several previously unknown transversal gates.
For instance, we demonstrate that the entire Clifford group can be implemented in the toric code in finite depth.
To our knowledge, this is the first constant-depth implementation of the full Clifford group on an asymptotic family of connected\footnote{That is, not disentanglable in finite depth into a direct sum of two smaller codes.} qLDPC codes with multiple logical qubits.
We also find constant-depth addressable magic gates in the three-dimensional toric code and large families of addressable gates in non-manifold codes with many encoded qubits.
Unifying other perspectives on qLDPC logic, we have also demonstrated how code surgery~\cite{horsman2012surface,litinski2019game,cohen2022lowoverhead} and the construction of gates through cup products~\cite{zhu2025, breuckmann2025cupsgatesicohomology, lin2024transversalnoncliffordgatesquantum, zhu2025topological} can be understood and extended with the chain map hierarchy. 

The remainder of this section is dedicated to outlook.

\paragraph*{Discovery of gates}
Our work opens immediate opportunities for gate discovery in a plethora of qLDPC codes.
If even the toric code had many undiscovered gates prior to this work, one can anticipate that many other relevant code families similarly have transversal gates waiting to be found.
As qLDPC codes currently form the core of several proposed routes towards fault-tolerant quantum computing~\cite{yoder2025tourgross,webster2026pinnacle, cain2026shor, tham2026breakeven,menssen2026strategicplanneutralatom}, there is a pressing need for theoretical shortcuts that can circumvent experimentally difficult operations, and avoiding code surgery with transversal gates is a promising avenue in this direction.
The chain map hierarchy provides a generally applicable formalism for transversal gate discovery, and its application to modern, experimentally relevant~\cite{Periwal_2021, bluvstein_logical_2024,yoder2025tourgross,chiu_continuous_2025, bluvstein_fault-tolerant_2026,moses_race-track_2023,ransford_helios_2025,tham2026breakeven,zhao2026ultrahighratequantumerrorcorrection,menssen2026strategicplanneutralatom} qLDPC codes is an important avenue for future work.

Finding a transversal implementation of some logical action requires sparsifying a logical representative in the chain map complex, a process which in general is not easy.
However, in small code instances it may be possible to perform a numerical optimization to find sparse gates.
Our framework provides a clear search space: any transversal gate can be decomposed into a fixed (dense) representative, and a deformation by a boundary in the chain map complex.
By minimizing the sparsity degree over the choice of boundary, it is possible to find transversal gates, if they exist.
Of course, this minimization is a computationally hard problem,  analogous to finding the lowest weight logical of a code to identify the distance.
However, heuristics are available to approximate this solution.
Further, while the existence of a low-weight logical can be a disaster for fault tolerance, the existence of a lower weight transversal gate is only a benefit.
If the heuristically identified gate is already good enough, the existence of a better gate is not a danger to its fault tolerance.

Beyond small codes, an important code family in which it may be possible to find general criteria for the existence of transversal gates are balanced/lifted product codes, which include some good qLDPC codes~\cite{Hastings_2021,Breuckmann_2021,Panteleev_2021, panteleev_asymptotically_2022,leverrier2022quantumtannercodes, dinur_good_2023,Bravyi_2024, rakovszky2024physicsgoodldpccodes,yoder2025tourgross,cain2026shor,bhardwaj2026high, hong2026quantum, Lee2026logicalspectroscopyliftedproductcodes,Zheng2026logicalcomputationcanonicallifted}.
Indeed, a balanced product is simply a tensor product followed by a quotient by some symmetry of the code.
The chain map hierarchy interacts naturally with the tensor product, and any identified gate which is invariant under the symmetry has a well-defined action on the quotient (though it may be trivial).
These observations may enable a sufficiently concrete recipe to find transversal gates of balanced products of simple codes, as we did for tensor products.

The construction of the chain map hierarchy and its interpretation as a tensor product extends to chain complexes over other rings (Appendix~\ref{app:CMCDetails}), which encode non-qubit codes.
There is thus an immediate opportunity for gate discovery in many simple non-qubit codes.

Finally, the chain map hierarchy can also provide insight into the design of new code families with many transversal gates.
The structure of gates on a code is expressed by the structure of logicals in repeated tensor products of the code.
Thus, the chain map hierarchy shows that, in order for a code to have transversal gates, the logicals in its tensor products must be easily deformed and sparsifiable.
This is a design criterion which guarantees that a code has transversal gates.

\paragraph*{(Im)possibility of sparsification}
As is already implicit above, the ability or inability to sparsify gates is the core link between the specification of a logical action and the design of a fault-tolerant transversal gate.
In addition to finding code families with many gates, it is also important to ask which code families can \emph{never} have many transversal gates~\cite{Eastin2009TransversalRestrictions,Bravyi2013GatesforLocalStabilizer,Beverland2016TQFTGates}.
The chain map complex also provides a handle on this question, which will be the topic of upcoming work~\cite{Sahay2026GoNoGo}.
Indeed, finding no-go theorems of this kind is something that our general framework can enable that, for instance, other methods of numerical search cannot.

For instance, the chain map complex provides a simple dimension-counting argument for a generalization of the Bravyi-K\"onig theorem---that a geometrically local circuit on a \(D\)-dimensional code can never implement a gate beyond the \(D\)th level of the Clifford hierarchy---to nonlocal transversal circuits in manifold CSS codes.
Consider a \(D\)-dimensional manifold CSS code \(C_\bullet\) with \(D_Z\)-dimensional \(Z\) logicals. 
We seek a logical \(\mathsf{C}^{\ell-1}\mathsf{Z}\) gate, and ask what constraints are placed on \(D\) and \(D_Z\) by the gate's existence.
This gate corresponds to a \(D_Z \ell\)-dimensional cycle in \(C_\bullet^{\otimes \ell}\), which itself describes an ambient \(D\ell\)-dimensional manifold.
The cycle must be deformed to mostly avoid intersection with the \(D(\ell-1)\)-dimensional planes formed by fixing a coordinate in one code, as the number of such intersections is precisely the number of gates that act on a qubit at the fixed coordinate.
The dimension of a transverse intersection%
\footnote{``Transverse intersection'' here refers to the differential geometry notion of transversality: at any point of intersection, the tangent spaces of the two submanifolds generate the tangent space of the ambient manifold. 
This is not to be confused with transversal gates.}
of a \(D_Z \ell\)- and a \(D(\ell-1)\)-dimensional surface in an ambient \(D\ell\)-dimensional space is
\begin{equation}
    D_Z \ell + D(\ell-1) - D\ell = D_Z \ell - D.
\end{equation}
For an \(O(1)\) intersection we need \(D_Z \ell-D \leq 0\), which is satisfied when
\begin{equation}
    \ell \leq D/D_Z.
\end{equation}
That is, only diagonal gates from the \(\lfloor D/D_Z \rfloor\)th level of the chain map hierarchy can be sparsified.
The best case scenario is to consider \(D_Z  = 1\) dimensional \(Z\) logicals (the code has a \(Z\)-distance of \(1\) with \(D_Z=0\)), in which case only logical actions from the \(D\)th level of the hierarchy are available: the Bravyi-K\"onig theorem.

A similar argument rules out a sparsification of a transversal Toffoli gate in any manifold CSS code~\cite{JochymOConnor2018}: the cycle corresponding to a logical \(\mathsf{CCNOT}\) in \(C_\bullet \otimes C_\bullet \otimes C^\bullet\) is a tensor product of cycles encoding two \(\bar{Z}\) and one \(\bar{X}\), which has dimension \(2D_Z + (D-D_Z)\). 
An \(O(1)\) intersection with fixed planes of dimension \(2D\) in \(3D\) ambient dimensions requires
\begin{equation}
    2D_Z  + (D-D_Z ) + 2D - 3D = D_Z  \leq 0,
\end{equation}
which is only satisfied by classical codes with \(D_Z=0\) and a \(Z\)-distance of \(1\).

These arguments can be generalized to gates between distinct codes, and likely can be formalized into a rigorous theorem.
Extending these arguments to non-manifold codes requires identifying a meaningful notion of both the dimension of chain complex and cycle, and also of transverse intersections of cycles, such that the dimension and intersection notions are compatible. 
In particular, the dimension of the transverse intersection of cycles of dimension \(p\) and \(q\) in an ambient \(D\)-dimensional complex should be \(p+q-D\).
[A cup product would provide such a notion of intersection (Appendix~\ref{app:CupsAndPoincare}), but one expects the theorem to hold without additional algebraic constraints on the code.]

\paragraph*{Generalizations}
There are several enticing generalizations of the chain map hierarchy formalism that are likely feasible.
One of the most pressing is understanding non-commuting intra-block gates.
We have provided a comprehensive description of interblock gates between distinct codes, and subsets of intra-block gates on a single code that commute with one another, but have not constructed a framework that can simultaneously describe non-commuting gates, and their composition rules.
As we have seen in the toric code examples, our setting is frequently already enough to identify a generating set of all possible gates, but to sparsify products of several generators into a low-depth implementation requires understanding the whole group of gates simultaneously.

The extension of the chain map hierarchy to non-CSS codes is also likely to be possible.
The codes that are most concisely described by a chain complex are CSS codes (either qubit or non-qubit, Appendix~\ref{app:CMCDetails}), but non-CSS codes can also be described by a chain complex with an additional symplectic structure~\cite{Kitaev1997quantumcomputation,Calderbank1997gf4,Bombin_2007,Haah2013latticequantumcodesexotic}.
Chain maps that preserve the symplectic structure---symplectomorphisms---correspond to logical gates.
The full set of symplectomorphisms encodes many non-commuting gates, and it is not clear how to assemble it into a chain complex analogous to the chain map complex.
However, making progress on describing non-commuting gates is also likely to unlock a chain map hierarchy description of gates in non-CSS codes.

While code surgery gadgets are encoded in the chain map hierarchy, the criteria for their fault-tolerance is far less concrete.
Indeed, while sufficient criteria for the fault-tolerance of code surgery are known~\cite{williamson2024lowoverhead,ide2025faulttolerant,baspin2025fast}, the necessary and sufficient criteria for fault tolerance are generally unclear.
Thus, while the chain map complex provides deformations of code surgery gadgets, there is not a clear criteria like sparsity to distinguish when these deformations have arrived at a fault-tolerant implementation.
Clarifying such a criterion would allow a fuller application of the chain map hierarchy to code surgery.
One route to this could be to formulate an extension of our methods for subsystem codes, and then treat the code surgery ancillas and the code of interest as both belonging to a parent subsystem code.

Moving slightly beyond quantum error correction,
the graphical chain tensor formalism for working with (non)-Clifford logical operations may be a useful tool for the co-design of quantum error correcting codes and quantum algorithms.
The legs of these chain tensors are chain complexes, which contain both data about logical qubits as well as how these logical qubits unravel into physical qubits.
As such, the formalism provides a kind of physical-qubit-aware logical circuit: the structure of the tensor network encodes a logical circuit, while the legs of the tensor carry microscopic information.
It is plausible that this dual description allows the interaction between algorithmic primitives and code primitives to be examined transparently.

\paragraph*{Connections to Condensed Matter/High-Energy Physics}
Moving even further from quantum error correction, our formalism can be useful for condensed matter and high-energy physics applications.
The first application is to the understanding of emergent symmetries.
In the language of QFT, the transversal gates we find are emergent 0-form symmetries~\cite{McGreevy2023GenSymm}.
They are emergent because they fix the low-energy effective theory (the code space), without fixing the microscopic Hamiltonian (the stabilizers).
Being 0-form is sometimes described as being non-deformable, but the deformations of the chain map complex show that this description is imprecise.
As the symmetry is emergent, it actually can be deformed while preserving the low-energy action, despite  the fact that the action on the microscopic degrees of freedom changes.
These deformations are reminiscent of the space-time picture of 0-form symmetries---they are deformable codimension-1 sheets in \(D+1\) dimensional spacetime.
The spacetime deformations, realized microscopically in the lattice, are exactly the chain map complex boundaries.

The same experimental advances that have bolstered QEC research in qLDPC codes with non-local connectivity have also prompted interest in \emph{non-local matter}---the emergent properties of systems with non-local, but sparse, connectivity~\cite{rakovszky2023physicsgoodldpccodes,rakovszky2024physicsgoodldpccodes,placke_topological_2024,placke_expansion_2025,zhu_non-abelian_2026, Christos2026nonabelian, yin_low-density_2025, de_roeck_low-density_2025, mcdonough2026calderbankshorsteanecodesgroupvaluedqudits}.
Our formalism is adapted to this setting, and can extend the same insights regarding, say, emergent symmetries in more traditional systems to non-local or expanding geometries.

Another application is to the construction of new exactly solvable models of fractonic topological orders with novel phenomenology.
A companion work~\cite{Christos2026nonabelian} pointed out how mapping cones of higher levels in the chain map hierarchy encode stabilizer models of nonabelian topological orders.
This construction extends to any chain complex, not just two-dimensional manifold codes, providing nonabelian analogues of fractonic~\cite{Chamon2005fracton,Haah_2011,Vijay_2016, Tantivasadakarn_2021} and topological spin-glass~\cite{placke_topological_2024} order.

\subparagraph*{Note added:}
During the preparation of this manuscript, Ref.~\cite{Holmes2026quantumlogiccodes} appeared, which also reports a logical \(\CZLogical{}{}\) between the two encoded qubits of the toric code. This appeared shortly after our logical \(\CZLogical{}{}\) gate was presented at QEC 2026. There is no overlap between the approach of Ref.~\cite{Holmes2026quantumlogiccodes} and our formalism, and none of the other novel gates we report were identified in Ref.~\cite{Holmes2026quantumlogiccodes}.
Also after QEC 2026, a blog post~\cite{He2026RoundRobin} observed that (generalized) \emph{round robin} \(\mathsf{C^{\ell-1}Z}\) gates~\cite{Yoder2016RoundRobin} act as logical (possibly trivial) \(\overline{\mathsf{C^{\ell-1}Z}}\) gates, and that these generate all logical circuits of physical \(\mathsf{C^{\ell-1}Z}\)s.
Through the lens of our work, trivial round robin  gates coincide with the \(\mathsf{C^{\ell-1}Z}\)-type chain map hierarchy boundaries, and logically nontrivial round robin gates coincide with dense canonical logical representatives from \(\CMH{\ell}\).
However, our general formalism which organizes  these gates into a chain complex does not appear in Ref.~\cite{He2026RoundRobin}, nor do any of our new gates.

\begin{acknowledgements}
    We thank Nikolas Breuckmann, Margarita Davydova, Nik Gjonbalaj, Jeongwan Haah, Gideon Lee, John Preskill, Dominic Williamson, and Willers Yang for useful discussions.
    We especially thank Fiona Burnell, Maine Christos, Chiu Fan Bowen (Leo) Lo, and Tibor Rakovszky for collaborations on related work. 
    R.S. thanks Stanford Q-FARM and the Leinweber Institute for Theoretical Physics at Stanford for their hospitality during this work.
    This work was supported by the Harvard Graduate School of Arts and Sciences Merit Fellowship Award (RS),
    a Stanford Q-FARM Bloch postdoctoral fellowship (DML),
    a Packard Fellowship in Science and Engineering (DML, VK), 
    the US Department of Energy, Office of Science under Award No.\ DE-SC0019380 (DML), 
    and the Office of Naval Research Young Investigator Program (ONR YIP) under Award Number N00014-24-1-2098 (VK).
\end{acknowledgements}

\appendix
\renewcommand{\theequation}{\thesection\arabic{equation}}

\section{Deferred proofs}
\label{app:DeferredProofs}

In Section~\ref{sec:ChainMapHierarchy}, several technical or lengthy proofs were deferred.
In this appendix, we collect those proofs.

\subsection{Pauli measurements}
\label{subapp:PauliMeasurements}

Existing perspectives on Pauli measurement using code surgery can be reinterpreted in the chain map formalism.
In this subappendix, we review two key results demonstrating that code surgery Pauli measurement gadgets are encoded by chain maps.
The first is Theorem~\ref{thm:CodeSurgeryChainMaps}.
\begin{shaded*}
    \CodeSurgeryChainMaps*
\end{shaded*}
En route to Theorem~\ref{thm:CodeSurgeryChainMaps}, we will encounter a more traditional perspective on code surgery in terms of mapping cones.
\begin{shaded*}
    \MappingConeCode*
\end{shaded*}
Both Theorem~\ref{thm:CodeSurgeryChainMaps} and Lemma~\ref{lem:MappingConeCode} follow from the same calculation, which we now present.

\begin{proof}
Theorem~\ref{thm:CodeSurgeryChainMaps} and Lemma~\ref{lem:MappingConeCode} can be proved by tracking the stabilizer group for the code throughout the process of Eq.~\eqref{eq-measurementcircuit}.
That is, under the successive operations of \(U_{\mathsf{CZ}}^\phii\) and the measurement of the ancillas from code \(C\),
\begin{equation} \label{eq:measurementcircuit_appendix}
M^{\phii}_{Z} = \left(\prod_{e \in \mathsf{B}(C_0)} M^C_{e, Z} \right) \times  \left( \prod_{\mathbf{v} \in \mathsf{B}(C^{-1}) } M^{C}_{\mathbf{v}, X} \right) U_{\mathsf{CZ}}^{\phii}.
\end{equation}

Before the entangling circuit \(U_{\mathsf{CZ}}^\phii\), the codes \(C\) and \(D\) have disjoint stabilizers:
\begin{equation}\label{eqn:InitialStabilizers}
    X_{\mathbf{v}}^C X_{\exd_0 \mathbf{v}}^C,\quad
    Z_{\partial_0 e}^C Z_{e}^C,\quad
    X^D_{\exd_0 \mathbf{w}}, \quad
    Z^D_{\partial_1 q}.
\end{equation}
Here, (co)boundary maps in the subscripts of Pauli operators from code \(C\) are the boundaries for the chain complex \(C_\bullet\), and similarly for \(D\).

The circuit \(U_{\mathsf{CZ}}^\phii\) affects the qubits from \(\mathsf{B}(C^{-1})\) and \(\mathsf{B}(D_0)\) only.
Conjugating all stabilizers from Eq.~\eqref{eqn:InitialStabilizers} by \(U_{\mathsf{CZ}}^\phii\) gives
\begin{equation}
    X_{\mathbf{v}}^C X_{\exd_0 \mathbf{v}}^C Z^D_{\phii_0(\mathbf{v})},\quad
    Z_{\partial_0 e}^C Z_{e}^C,\quad
    Z^C_{\partial_0 \phii_1^{\transpose}(\mathbf{w})} X^D_{\exd_0 \mathbf{w}} , \quad
    Z^D_{\partial_1 q}.
\end{equation}
We used \(\phii_1^{\transpose}(\exd_0 \mathbf{w}) = \partial_0 \phii_1^{\transpose}(\mathbf{w})\).
An equivalent choice of generators is obtained by multiplying \(Z^C_{\partial_0 \phii_1^{\transpose}(\mathbf{w})} X^D_{\exd_0 \mathbf{w}}\) by any other stabilizer. In particular, we multiply by \(Z_{\partial_0 \phii_1^{\transpose}(\mathbf{w})}^C Z_{\phii_1^{\transpose}(\mathbf{w})}^C\), which is a product of \(Z_{\partial_0 e}^C Z_{e}^C\) stabilizers, obtaining
\begin{equation}
    X_{\mathbf{v}}^C X_{\exd_0 \mathbf{v}}^C Z^D_{\phii_0(\mathbf{v})},\quad
    Z_{\partial_0 e}^C Z_{e}^C,\quad
    Z^C_{\phii_1^{\transpose}(\mathbf{w})} X^D_{\exd_0 \mathbf{w}}, \quad
    Z^D_{\partial_1 q}.
\end{equation}

The next operation in \(M^{\phii}_{Z}\) is the measurement of the \(\mathsf{B}(C^{-1})\) qubits in the \(X\) basis.
In the stabilizer group, this appends generators \(x^C_{\mathbf{v}} X^C_{\mathbf{v}}\), where the joint distribution of the random scalars \(x^C_{\mathbf{v}} \in \{\pm1\}\) is determined by the specific code state. This leaves the  \(\mathsf{B}(C^{-1})\) in an unentangled product state, so we neglect them.
Further, only stabilizers with \(Z^C_{v}\) parts that commute with all \(x^C_{\mathbf{v}}X^C_{\mathbf{v}}\) remain in the stabilizer group, which in this case means that no stabilizer with \(Z\)-support on \(\mathsf{B}(C^{-1})\) survives.
The new stabilizer group is generated by
\begin{equation}\label{eqn:AlmostMappingCone}
    x_{\mathbf{v}}^C X_{\exd_0 \mathbf{v}}^C Z^D_{\phii_0(\mathbf{v})},\quad
    Z_{c}^C,\quad
    Z^C_{\phii_1^{\transpose}(\mathbf{w})} X^D_{\exd_0 \mathbf{w}}, \quad
    Z^D_{\partial_1 q},
\end{equation}
where \(c \in \mathsf{ker}(\partial_0^C)\) is a cycle in \(C_0\) (including both boundaries and nontrivial cycles).

At this point, we observe that the values of the logicals \(Z^D_{\phii_0(\mathbf{c})}\) have been revealed, where \(\mathbf{c} \in \mathsf{ker}(\exd^C_0)\) is a cocyle of \(C^{-1}\).
Indeed,
\begin{equation}
    x_{\mathbf{c}}^C X_{\exd_0 \mathbf{c}}^C Z^D_{\phii_0(\mathbf{c})} = x_{\mathbf{c}}^C Z^D_{\phii_0(\mathbf{c})}
\end{equation}
is in the stabilizer group, so the code state is a \((+1)\)-eigenstate of this operator.
Thus, it is an \((x_{\mathbf{c}}^C)\)-eigenstate of \(Z^D_{\phii_0(\mathbf{c})}\).
In particular, if \(\mathbf{c} = \exd^C_{-1} \mathbf{a}\) is a boundary, then \(x_{\exd_{-1} \mathbf{a}}^C = 1\), as it measures the stabilizer \(Z^D_{\phii_0(\exd_{-1} \mathbf{a})} = Z^D_{\partial_1 \phii_{-1}(\mathbf{a})}\).
While we have now measured \(Z^D_{\phii_0(\mathbf{c})}\), the code \(D\) is still entangled with the ancilla code \(C\), so stopping here would only give a destructive measurement.
Performing more measurements and feedback returns the system to a codestate of \(D\) alone.

Before proceeding to the \(Z^C\) measurement step, we complete the proof of Lemma~\ref{lem:MappingConeCode}.
The stabilizers of Eq.~\eqref{eqn:AlmostMappingCone} are almost the mapping cone code stabilizers of Lemma~\ref{lem:MappingConeCode} (with the \(c\) logicals initialized in the \(\ket{0}\) state), except that the \(x_{\mathbf{v}}^C X_{\exd_0 \mathbf{v}}^C Z^D_{\phii_0(\mathbf{v})}\) stabilizers don't have the right signs.
These signs can be flipped by multiplying the code state (conjugating the stabilizers) by Pauli \(Z^C_e\) or \(X^D_{\mathbf{e}}\) operators.
First, multiply by the conjugate \(X^D\) logical for any logical \(Z^D_{\phii_0(\mathbf{c})}\) for which \(x^C_{\mathbf{c}} = -1\). Absorbing the sign flips into new coefficients \(\td{x}^C_{\mathbf{v}}\), we now have that \(\td{x}^C_{\mathbf{c}} = 1\) for all cocycles \(\mathbf{c}\).
Consider the chain \(r \in C_{-1}\) such that \(\td{x}^C_{\mathbf{v}} = (-1)^{\mathbf{v}(r)}\).
Then \(\td{x}^C_{\mathbf{c}} = 1\) implies that \(\mathbf{c}(r) = 0\) for all cocycles \(\mathbf{c}\), which implies that \(r = \partial b\) is a boundary---it is a cycle as \(\exd_{-1}\mathbf{a}(r) = \mathbf{a}(\partial_{-1} r) = 0\) for all \(\mathbf{a} \in C^{-2}\), which is only possible if \(\partial_{-1} r = 0\); and it is homologically trivial as it does not pair with any nontrivial cocyle.
Acting by \(Z^C_b\) on the code state then flips the stabilizers \(\td{x}_{\mathbf{v}}^C X_{\exd_0 \mathbf{v}}^C Z^D_{\phii_0(\mathbf{v})}\) for which \(\exd_0 \mathbf{v}(b) = 1\), which is exactly those with \(\td{x}^C_{\mathbf{v}}=-1\), as desired.
Thus, after the \(X^C\) measurement, the state of the system is related to the mapping cone code code state by a product of Paulis, as claimed.

Finally, we analyze the effect of the \(Z\) measurement of the \(\mathsf{B}(C_0)\) qubits on the stabilizer group.
Continuing from Eq.~\eqref{eqn:AlmostMappingCone}, the stabilizers become
\begin{equation}
    x_{\mathbf{c}}^C Z^D_{\phii_0(\mathbf{c})},\quad
    z^C_{\phii_1^{\transpose}(\mathbf{w})} X^D_{\exd_0 \mathbf{w}}, \quad
    Z^D_{\partial_1 q},
\end{equation}
where we have removed single-qubit stabilizers and \(z^C_{v} \in \{\pm 1\}\) is a measurement outcome.
Similarly to the previous paragraph, we have that \(z^C_{\phii_1^{\transpose}(\mathbf{c})} X^D_{\exd_0 \mathbf{c}} = z^C_{\phii_1^{\transpose}(\mathbf{c})} = 1\) for every cocycle \(\mathbf{c} \in \mathsf{ker}(\exd_0)\), so that if we define the cochain \(\mathbf{r}\in C^0\) satisfying \(z^C_{\phii_1^{\transpose}(\mathrm{w})} = (-1)^{\mathbf{r}(\phii_1^{\transpose}(\mathrm{w}))}\), then we have
\begin{equation}
    \mathbf{r}(\phii_1^\transpose(\mathbf{c})) = \mathbf{c}(\phii_1(\mathbf{r})) = 0, \quad\text{for all }\mathbf{c} \in \mathsf{ker}(\exd_0),
\end{equation}
and hence \(\phii_1(\mathbf{r}) = \partial^D_0 a\) is a cycle.
Then multiplying the codestate by \(Z^D_{a}\) flips exactly those stabilizers \(z^C_{\phii_1^{\transpose}(\mathbf{w})} X^D_{\exd_0 \mathbf{w}}\) for which \(z^C_{\phii_1^{\transpose}(\mathrm{w})} = -1\).
Indeed, \(Z^D_{a}\) anticommutes with those \(X^D_{\exd_0 \mathbf{w}}\) for which
\begin{equation}
    \exd^D_0 \mathbf{w}(a) = \mathbf{w}(\partial^D_0 a) = \mathbf{w}(\phii_1(\mathbf{r})) = 1.
\end{equation}
Thus, the final stabilizer group is
\begin{equation}
    x_{\mathbf{c}}^C Z^D_{\phii_0(\mathbf{c})},\quad
    X^D_{\exd_0 \mathbf{w}}, \quad
    Z^D_{\partial_1 q},
\end{equation}
which are the stabilizers for \(D_\bullet\), with additional (known) logical projections.
This completes the proof of Theorem~\ref{thm:CodeSurgeryChainMaps}.
\end{proof}

\subsection{Clifford logic}
\label{subapp:CliffordLogic}

The full encoding of Clifford logic in the chain map complex \([\mathbf{C}^{[2]}]_\bullet = [C^\bullet, D_\bullet]\) is summarized by Theorem~\ref{thm:Cliffordlogic}.
\begin{shaded*}
    \Cliffordlogic*
\end{shaded*}
\begin{proof}
Parts 1 and 2 are restatements of Theorems~\ref{thm:TransversalGatesFromChainMaps} and \ref{thm:CodeSurgeryChainMaps}, respectively, in the language of the chain map complex.
Recalling that a logical class \([\phii]\) specifies an action on homology and that cycles \(\phii_\bullet\) are chain maps, those theorems translate this to a logical circuit or logical measurement.
The statement that two circuits \(U_{\mathsf{CZ}}^\phii\) and \(U_{\mathsf{CZ}}^\psi\) (similarly, that two measurements \(M_Z^\phii\) and \(M_Z^\psi\)) that act differently on logical qubits correspond to distinct \([\phii] \neq [\psi]\) is just the statement that the map from homological action to logical action is well-defined: the same homology class cannot map to multiple logical actions.

Part 3 is the only result we have not previously shown. It follows from the fact that every \(\mathbb{F}_2\)-linear map of homology \(\phii_*: H^\bullet(C) \to H_\bullet(D)\) is realized as the induced action of some chain map \(\phii_\bullet : C^\bullet \to D_\bullet\), and the previous two results.
    
Indeed, the chain map \(\phii_\bullet\) can be constructed explicitly, which we do for a \(\phii_*\) which is only nonzero on the \(H^0(C) \to H_0(D)\) component (which determines the logical action). 
We give a more abstract presentation in Appendix~\ref{app:CMCDetails}.
Choose a basis \(\{\mathbf{L}_j\} = \mathsf{B}(H^0(C))\) and dual basis \(\{L_j\} = \mathsf{B}(H_0(C))\) such that \(\mathbf{L}_j(L_k) = \delta_{jk}\).
Pick any representative cycles \(c_j \in L_j\) and \(d_j \in \phii_*(\mathbf{L}_j)\).
Then the map
\begin{equation}
    \phii_0(\mathbf{c}) = \sum_j \mathbf{c}(c_j) d_j
\end{equation}
is a (usually dense) chain map concentrated on \(C^0 \to D_0\) which has \(\phii_*\) as its induced action.
Indeed, if \(\mathbf{c} = \exd \mathbf{b}\), then
\begin{align}
    \phii_0(\exd \mathbf{b}) = \sum_j \mathbf{b}(\partial c_j) d_j = 0 = \partial \phii_{-1}(\mathbf{b}),
\end{align}
where we used \(\partial c_j = 0\) and \(\phii_{-1} = 0\).
Similarly
\begin{equation}
    \partial \phii_0(\mathbf{c}) = \sum_j \mathbf{c}(c_j) \partial d_j = 0 = \phii_{1}(\exd \mathbf{c}),
\end{equation}
where we also used \(\partial d_j = 0\).
Thus, \(\phii_\bullet\) is a chain map.
It also has the action of \(\phii_*\) on the basis cocycle representatives \(\mathbf{m}_j \in \mathbf{L}_j\),
\begin{align}
    \phii_0(\mathbf{m}_j) = \sum_k \mathbf{m}_j(c_k) d_k = d_j \in \phii_*(\mathbf{L}_j),
\end{align}
and thus agrees with \(\phii_*\) on all of \(H^0(C)\) by linearity.

This calculation extends to \(H_{-1}([\mathbf{C}^{[2]}])\) by replacing \(C^\bullet\) with \(C[-1]^{\bullet}\).
\end{proof}

\subsection{Chain map hierarchy measurements}
\label{subapp:CMHMeasurements}

This subappendix proves that the measurement circuit from Section~\ref{subsubsec:NonAbelianSurgery},
\begin{equation}\tag{\ref{eq-nonAbelianmeasurementcircuit}}
M_{\mathsf{C}^{\ell - 2} \mathsf{Z}}^{\phii}  = \left(\prod_{e \in \mathsf{B}^C_0} M^C_{e, Z} \right) \times  \left( \prod_{\mathbf{v} \in \mathsf{B}_C^{-1} } M^{C}_{\mathbf{v}, X} \right) U_{\mathsf{C^{\ell-1} Z}}^{\phii},
\end{equation}
performs a logical \(\mathsf{C}^{\ell - 2} \mathsf{Z}\) measurement.

\begin{shaded*}
    \NonAbelianSurgery*
\end{shaded*}
\begin{proof}
The proof mirrors that of Theorem~\ref{thm:Cliffordlogic}, directly tracking the evolution of the stabilizer group, with the additional complication that, after the magic gate \(U_{\mathsf{C^{\ell-1} Z}}^{\phii}\), the stabilizer group becomes non-Pauli.

Initially, the stabilizer group is
\begin{equation}
    X_{\mathbf{v}}^{(1)} X_{\exd \mathbf{v}}^{(1)},\quad
    Z_{\partial e}^{(1)} Z_{e}^{(1)},\quad
    X^{(j)}_{\exd \mathbf{w}}, \quad
    Z^{(j)}_{\partial q} \quad(2\leq j\leq \ell)
\end{equation}
where superscripts denote code blocks, and component indices on (co)boundary maps have been suppressed. Similarly, we omit code block labels on chains which already appear as subscripts of Paulis with code block labels.
After the entangling gate \(U_{\mathsf{C^{\ell-1} Z}}^{\phii}\), the \(X\) stabilizer generators are conjugated to
\begin{equation}
    X_{\mathbf{v}}^{(1)} X_{\exd \mathbf{v}}^{(1)} U_{\mathsf{C^{\ell-2} Z}}^{\phii(\mathbf{v}_{(1)})},\quad
    U_{\mathsf{C^{\ell-2} Z}}^{\phii(\exd \mathbf{w}_{(j)})} X^{(j)}_{\exd \mathbf{w}} = U_{\mathsf{C^{\ell-2} Z}}^{\partial^{[\ell-1]}\phii(\mathbf{w}_{(j)})} X^{(j)}_{\exd \mathbf{w}} .
\end{equation}
In the expression \(\phii(\mathbf{v}_{(j)})\) (and similar) we regard \(\phii\) as a map from cochains in \(C^\bullet_{(j)}\) to maps one level lower in the hierarchy with the reordering isomorphism \([C^{\bullet-1}_{(1)}, [\cdots]] \isom [C^{\bullet}_{(j)}, [\cdots]]\).
We also used that \(\phii\) is a cycle to commute it with the boundary map.

The stabilizer generators \(U_{\mathsf{C^{\ell-2} Z}}^{\partial^{[\ell-1]}\phii(\mathbf{w}_{(j)})} X^{(j)}_{\exd \mathbf{w}}\) can be replaced with 
\begin{equation}
    U_{\mathsf{C^{\ell-2} Z}}^{\phii(\mathbf{w}_{(j)})} X^{(j)}_{\exd \mathbf{w}},
\end{equation}
where the \(\mathsf{C^{\ell-2} Z}\) gates now act between qubits in \(\mathsf{B}(C_{(1)}^0)\) and the other codes,
by multiplying by non-Pauli stabilizers from the \(\ell-1\) level of the chain map hierarchy.
Indeed, if \(C_{(1)}^\bullet\) were not a hypergraph cluster code, this would be clear, as \(U_{\mathsf{C^{\ell-2} Z}}^{\phii(\mathbf{w}_{(j)})}\) is exactly a stabilizer in the setting where all qubits are in \(\mathsf{B}(C_{(j)}^0)\).
To deal with the shift, we are more explicit.
Writing \(\phii' = \phii(\mathbf{w}_{(j)})\), we expand the boundary as
\begin{align}
    U_{\mathsf{C^{\ell-2} Z}}^{\partial^{[\ell-1]}\phii'} &= 
    \exp\left[ i\pi (\partial^{[\ell-1]}\phii)^{\alpha_1 \cdots w_j \cdots \alpha_\ell} n_{\alpha_1}\cdots n_{\alpha_\ell} \right] \nnb \\
    &= \prod_{k \neq j} \exp\left[ i\pi \phii^{ \alpha_1 \cdots w_j \cdots \alpha_\ell}  n_{\alpha_1} \cdots n_{\partial \alpha_k} \cdots  n_{\alpha_\ell} \right] \nnb \\
    &= \prod_{k \neq j} (Z^{(k)}_{\partial \alpha_k})^{\phii^{\alpha_1 \cdots w_j \cdots \alpha_\ell} n_{\alpha_1} \cdots  n_{\alpha_\ell} }. \label{eqn:ExplicitStabilizerDecomp}
\end{align}
Here, all products over \(n_{\alpha_k}\) exclude \(n_{\alpha_j}\), and the final line also excludes \(n_{\alpha_k}\).
For \(k \geq 2\), we see that every term in the product over \(k\) is a power of a (product of) code stabilizers \(Z^{(k)}_{\partial q}\) (which are still stabilizers after the \(U_{\mathsf{C^{\ell-1} Z}}^{\phii}\) gate, as they commute with it).
Thus, these terms all act trivially on the ground state, and can be multiplied into the generator \(U_{\mathsf{C^{\ell-2} Z}}^{\partial^{[\ell-1]}\phii(\mathbf{w}_{(j)})} X^{(j)}_{\exd \mathbf{w}}\) without making the code space smaller---any old codeword is still a codeword.
This operation also does not make the code space larger, as all of the diagonal chain map hierarchy stabilizers are independent of the \(X^{(j)}_{\exd \mathbf{w}}\) part of the stabilizer.
For the \(k=1\) term, the factor in Eq.~\eqref{eqn:ExplicitStabilizerDecomp} is \emph{not} a stabilizer.
Rather, the stabilizer involving \(k=1\) would be
\begin{equation}
    (Z^{(1)}_{\partial \alpha_1} Z^{(1)}_{\alpha_1})^{\phii^{\alpha_1 \cdots w_j \cdots \alpha_\ell} n_{\alpha_1} \cdots  n_{\alpha_\ell} }.
\end{equation}
Instead of removing the \(Z^{(k)}_{\partial \alpha_k}\) factor, multiplying by this term only moves it to the \(\mathsf{B}(C^0_{(1)})\) qubits.
Explicitly, the new stabilizer generator is
\begin{equation}
    (Z^{(1)}_{\alpha_1})^{\phii^{\alpha_1 \cdots w_j \cdots \alpha_\ell} n_{\alpha_1} \cdots  n_{\alpha_\ell} } X^{(j)}_{\exd \mathbf{w}} = U_{\mathsf{C^{\ell-2} Z}}^{\phii(\mathbf{w}_{(j)})} X^{(j)}_{\exd \mathbf{w}},
\end{equation}
as claimed.

Summarizing, the stabilizer group after the entangling gate is generated by
\begin{equation}
    X_{\mathbf{v}}^{(1)} X_{\exd \mathbf{v}}^{(1)} U_{\mathsf{C^{\ell-2} Z}}^{\phii(\mathbf{v}_{(1)})},\quad
    Z_{\partial e}^{(1)} Z_{e}^{(1)},\quad
    U_{\mathsf{C^{\ell-2} Z}}^{\phii(\mathbf{w}_{(j)})} X^{(j)}_{\exd \mathbf{w}}, \quad
    Z^{(j)}_{\partial q}.
\end{equation}
Performing a measurement of the \(\mathsf{B}(C^{-1}_{(1)})\) qubits in the \(X\) basis and eliminating the resulting disentangled product state results in new stabilizers
\begin{equation}
    x^{(1)}_{\mathbf{v}} X_{\exd \mathbf{v}}^{(1)} U_{\mathsf{C^{\ell-2} Z}}^{\phii(\mathbf{v}_{(1)})},\quad
    Z_{\mathcal{M}}^{(1)},\quad
    U_{\mathsf{C^{\ell-2} Z}}^{\phii(\mathbf{w}_{(j)})} X^{(j)}_{\exd \mathbf{w}}, \quad
    Z^{(j)}_{\partial q},
\end{equation}
where \(x^{(1)}_{\mathbf{v}}\) are random measurement outcomes and \(\mathcal{M}_(1) \in \ker \partial_{(1)}\) is a cycle (both boundaries and nontrivial cycles).
As in the Pauli measurement case, the eigenvalue of the operators to be measured have now been revealed: if \(\mathbf{c} \in H^{-1}(C_{(1)})\) is a cocycle, and \(x_{\bm{\mu}}^{(1)}\) denotes a product of \(x^{(1)}_{\mathbf{v}}\) values over the cocycle, then we have that
\begin{equation}
    x^{(1)}_{\mathbf{c}} X_{\exd \mathbf{c}}^{(1)} U_{\mathsf{C^{\ell-2} Z}}^{\phii(\mathbf{c}_{(1)})} = x^{(1)}_{\mathbf{c}} U_{\mathsf{C^{\ell-2} Z}}^{\phii(\mathbf{c}_{(1)})}
\end{equation}
is a stabilizer. 
Thus, in the code space, we must have \(x^{(1)}_{\mathbf{c}} = U_{\mathsf{C^{\ell-2} Z}}^{\phii(\mathbf{c}_{(1)})}\).
Translating this to a circuit of logical measurements in some basis gives Eq.~\eqref{eqn:CMHMeasurementLogical}.

Measuring a \(\mathsf{B}(C^{0}_{(1)})\) qubit in the \(Z\) has the following effect on \(\mathsf{C}^{\ell-2}\mathsf{Z}\) gates targeting that qubit.
If the measurement outcome is \(z_{e} = +1\), then the \(\mathsf{C}^{\ell-2}\mathsf{Z}\) gate cannot give a \(-1\) phase no matter the state of the other qubits it acts on, and the gate becomes trivial.
If \(z_{e} = -1\), then all the other qubits in the gate must still have \(z=-1\) to give a \(-1\) phase, so the gate is effectively replaced with a \(\mathsf{C}^{\ell-3}\mathsf{Z}\) on the remaining qubits.
Thus, writing \(r_e = (1-z_e)/2\), the post-\(Z\)-measurement stabilizer group can be written
\begin{equation}
    x^{(1)}_{\mathbf{c}} U_{\mathsf{C^{\ell-2} Z}}^{\phii(\mathbf{c}_{(1)})},\quad
    (-1)^{\phii^{\alpha_1 \cdots w_j \cdots \alpha_\ell} r_{\alpha_1} n_{\alpha_2} \cdots n_{\alpha_\ell}}
    X^{(j)}_{\exd \mathbf{w}}, \quad
    Z^{(j)}_{\partial q}.
\end{equation}
Alternatively, we can define a cochain \(\mathbf{r}_{(1)}\) such that  
\begin{equation}
    (-1)^{\phii^{\alpha_1 \cdots \alpha_\ell} N_{\alpha_1} n_{\alpha_2} \cdots n_{\alpha_\ell}} = U_{\mathsf{C^{\ell-2} Z}}^{\phii(\mathbf{r}_{(1)})},
\end{equation}
so that the stabilizer group is
\begin{equation}
    x^{(1)}_{\mathbf{c}} U_{\mathsf{C^{\ell-2} Z}}^{\phii(\mathbf{c}_{(1)})},\quad
     U_{\mathsf{C^{\ell-3} Z}}^{\phii(\mathbf{r}_{(1)}, \mathbf{w}_{(j)})}
    X^{(j)}_{\exd \mathbf{w}}, \quad
    Z^{(j)}_{\partial q}.
\end{equation}
In an analogous calculation to the proof of Theorem~\ref{thm:CodeSurgeryChainMaps}, the fact that contracting the chain tensor \(\phii(\mathbf{r}_{(1)})\) with any cocycle \(\mathbf{c}_{(j)}\) (or a combination thereof) gives a gate which is trivial on the code space (as \(X^{(j)}_{\exd \mathbf{c}} = 1\)), shows that \(\phii(\mathbf{r}_{(1)}) =\partial^{[\ell-1]} \psi\) is a boundary.
Then one can check that conjugating the stabilizers by
\begin{equation}
    U_{\mathsf{C^{\ell-2} Z}}^{\psi}
\end{equation}
removes the extra \(U_{\mathsf{C^{\ell-3} Z}}^{\phii(\mathbf{r}_{(1)}, \mathbf{w}_{(j)})}\) from the \(X^{(j)}_{\exd \mathbf{w}}\) stabilizers (up to a product of other chain map hierarchy stabilizers).
Indeed, this calculation is essentially homological, so it follows from the analogous calculation in \textbf{}Theorem~\ref{thm:CodeSurgeryChainMaps} when taking the target code to be \([\mathbf{C}']\).

Exhaustiveness is a corollary of exhaustiveness for the unitary gates \(U_{\mathsf{C^{\ell-1} Z}}^{\phii}\).
Any \(\phii_*\) can be chosen in Eq.~\eqref{eqn:CMHMeasurementLogical}, provided that the dimension of \(H^{-1}(C_{(1)})\) is large enough to express the number of measurements to be made, so any measurement can be performed through a circuit of this kind.
\end{proof}

\subsection{Intra-block \texorpdfstring{\(\mathsf{CCZ}\)}{CCZ} and \texorpdfstring{\(\mathsf{CNOT}\)}{CNOT} gates}
\label{subapp:IntraBlockCCZ}

The subappendix proves that \(0\)-cycles \(\phii_\bullet \in \CMH{\ell}_0\) can be used to define intra-block \(\mathsf{C^{\ell-1}Z}\) gates with a symmetrized logical action, and that block-off-diagonal \(0\)-cycles of \(\phii_\bullet \in [C^\bullet, C^\bullet]_0\) can be used to define \(\mathsf{CNOT}\) gates.
We begin with \(\mathsf{C^{\ell-1}Z}\) gates.

\begin{shaded*}
    \CCZIntraBlock*
\end{shaded*}
\begin{proof}
    Using that \(X_\mathbf{a} V_\ell^{\phii\dagger}X_{\mathbf{a}}\) commutes with \(V_\ell^{\phii}\), we can express
    \begin{equation}
        [V_\ell^\phii, X_{\mathbf{a}}]_{\rm grp} = \exp\left[ i\pi  \phii^{\alpha \cdots \gamma}(X_{\mathbf{a}} n_\alpha \cdots n_\gamma X_{\mathbf{a}} - n_\alpha \cdots n _\gamma) \right].
    \end{equation}
    The conjugation of \(n_e\) by \(X_{\mathbf{a}}\) flips \(n_e \to 1-n_e \equiv 1 + n_e \bmod 2\) when \(\mathbf{a}_e = 1\). That is,
    \begin{equation}
        X_{\mathbf{a}} n_e X_{\mathbf{a}} = a_e + n_e \mod 2.
    \end{equation}
    Substituting this and expanding gives
    \begin{align}
        [V_\ell^\phii, X_{\mathbf{a}}]_{\rm grp} = \exp\left[ i\pi \sum_{\emptyset \neq A \subseteq [\ell]}   \phii^{\alpha_1 \cdots \alpha_\ell}\prod_{i \in A} a_{\alpha_i} \prod_{j \in A^c} n_{\alpha_j} \right],
    \end{align}
    where \([\ell] = \{1, ..., \ell\}\) and summation is implied.
    This can be expressed as a product of \(V_{\ell-|A|}\) circuits,
    \begin{equation}\label{eqn:VphiGroupCommutator}
        [V_\ell^\phii, X_{\mathbf{a}}]_{\rm grp} = \prod_{\emptyset \neq A \subseteq [\ell]} V_{\ell-|A|}^{\phii(\mathbf{a}^{\otimes A})}.
    \end{equation}
    Here, \(\mathbf{a}^{\otimes A} = \bigotimes_{i \in A} \mathbf{a}_i\) is the repeated tensor product of the cochain \(\mathbf{a}\) from each of the cochain complex copies \(C^\bullet_{(i)}\).
    Thus, \(\phii(\mathbf{a}^{\otimes A})\) is the contraction of \(\phii\) with \(\mathbf{a}\) on each of the legs corresponding to \(A\).

    From Eq.~\eqref{eqn:VphiGroupCommutator}, we show that \(V^\phii_\ell\) is logical with a logical action that, up to logical gates from the \((\ell-1)\)th level of the Clifford hierarchy, is the symmterization of the logical tensor \(\phii_*\).
    We use induction on \(\ell\).
    The induction base is provided by \(\ell=1\), where \(V^\phii_1\) is just the logical \(\bar{Z}\) encoded by the cycle \(\phii\).
    Suppose that the theorem holds at level \(\ell-1\).
    Then \(V_\ell^\phii\) preserves the code space: if \(\mathbf{a} = \exd \mathbf{b}\), the right hand side of Eq.~\eqref{eqn:VphiGroupCommutator} is a product of \(V_{\ell-|A|}^{\phii(\exd\mathbf{b}^{\otimes A})} = V_{\ell-|A|}^{\partial^{[\ell-1]}\phii(\mathbf{b}_i\otimes\exd\mathbf{b}^{\otimes A\setminus\{i\}})}\) gates, which by the inductive hypothesis are logical and have trivial logical action as \(\partial^{[\ell-1]} \phii\) is a null homotopy.
    Now consider the action of \(V_\ell^\phii\) on \(X_\mathbf{a}\), where \(\mathbf{a} \in \mathbf{L}\in H^0(C)\) is a cycle.
    Ignoring terms in Eq.~\eqref{eqn:VphiGroupCommutator} from levels of the Clifford hierarchy lower than \(\ell-1\), we have 
    \begin{align}
        [V_\ell^\phii, X_{\mathbf{a}}]_{\rm grp} &= \prod_{i \in [\ell]} V_{\ell-1}^{\phii(\mathbf{a}_i)} \times \cdots \\
        &= V_{\ell-1}^{\sum_{i \in [\ell]} \phii(\mathbf{a}_i)} \times \cdots 
    \end{align}
    The logical action for each \(V_{\ell-1}^{\phii(\mathbf{a}_i)}\) is, by the hypothesis, \([\phii_*(\mathbf{L}_i)]^{\mathbb{S}}\), the symmetrized logical action of \(\phii_*(\mathbf{L}_i)\) on the remaining \(\ell-1\) legs.
    The level \(\ell\) part of the logical action of \(V^\phii_\ell\) is thus
    \begin{equation}
        \sum_{i \in [\ell]} [\phii_*(\mathbf{L}_i)]^{\mathbb{S}} = \phii^{\mathbb{S}}_*(\mathbf{L}_1),
    \end{equation}
    corresponding to the full symmetrization of \(\phii_*\).
\end{proof}

We now proceed to Theorem~\ref{thm:CNOTIntraBlock}, describing intra-block \(\mathsf{CNOT}\) gates.

\begin{shaded*}
    \CNOTIntraBlock*
\end{shaded*}
\begin{proof}
    We show this by direct calculation.
    In particular, let us note that since $\phii_0$ is block off-diagonal, 
    \begin{equation}
        V_{\mathsf{CNOT}}^{\phii} = \exp\left(i \pi m^{\alpha} \phii_{\alpha}^{\ \beta} n_{\beta} \right)
    \end{equation}
    The gate above is product of commuting $\mathsf{CNOT}$ gates. 
    As a consequence, we have that: 
    \begin{equation}
        [V_{\mathsf{CNOT}}^{\phii}, X_{\mathbf{a}}] = \exp\left(i\pi m^{\alpha} \phii_{\alpha}^{\ \beta} X_{\mathbf{a}} n_{\beta} X_{\mathbf{a}} \right) V^{\phii}_{\mathsf{CNOT}} \nonumber
    \end{equation}
    Using the fact that $X_{\mathbf{a}} n_{\beta} X_{\mathbf{a}} \overset{2}{=} \mathbf{a}_{\beta} + n_{\beta}$ where $\overset{2}{=}$ is equality modulo two, we have that: 
    \begin{equation}
        [V_{\mathsf{CNOT}}^{\phii}, X_{\mathbf{a}}]  = \exp\left(i \pi m^{\alpha} \phii^{\ \beta}_{\alpha} \mathbf{a}_{\beta} \right) = X_{\phii(\mathbf{a})}
    \end{equation}
    An identical calculation is used to show that $[V_{\mathsf{CNOT}}^{\phii}, Z_{b}] = Z_{\phii^{\mathsf{T}}(b)}$.
\end{proof}

\section{Details and Generalizations of the Chain Map Complex}
\label{app:CMCDetails}

In this appendix we provide additional details regarding the chain map complex.
In Appendix~\ref{subapp:OtherRings} we make immediate generalizations to chain complexes over other commutative rings.
This includes rings other than fields, which are relevant for codes over qudits with non-prime-power numbers of levels
Appendix~\ref{subapp:MapTensorIsom} proves the isomorphism between the chain map complex and the tensor product, and Appendix~\ref{subapp:HomotopyVsLogical} reviews the distinctions between the homology classes of the chain map complex and classes of gates with the same logical action.

\subsection{Rings other than \texorpdfstring{\(\FF_2\)}{F2}}
\label{subapp:OtherRings}

The chain map complex generalizes without important modification to chain complexes over any commutative ring, with all the properties we used in the main text continuing to hold~\cite{Weibel1994HomologicalBook}.
Here, we explicitly work through the construction of \([C, D]_\bullet\) in this more general setting.

The definition of sparse, based chain complexes mirrors the case of fields: 
\begin{defn}[Sparse, based chain complex]\label{def:ChainComplexAppendix}
    A chain complex \(C_\bullet\) over the commutative ring \(R\) is a sequence of \(R\) modules and \(R\)-linear maps (module homomorphisms)
    \begin{equation}
         \cdots \xrightarrow[\partial_2]{} C_1 \xrightarrow[\partial_1]{} C_0 \xrightarrow[\partial_0]{} C_{-1} \xrightarrow[\partial_{-1}]{}  \cdots 
    \end{equation}
    such that \(\partial_j \partial_{j+1} = 0\) for all \(j \in \ZZ\).
    \(C_\bullet\) is \emph{based} if all but finitely many of the \(C_j\) are zero and each nonzero \(C_j\) is free and equipped with finite basis \(\mathsf{B}(C_j)\), and \emph{sparse} if the matrices for \(\partial_j\) in these bases are sparse.
    
    Associated to every chain complex \(C_\bullet\) is a cochain complex \(C^\bullet\),
    \begin{equation}
        \cdots \xrightarrow[]{\exd_{-1}} C^{-1} \xrightarrow[]{\exd_0} C^0 \xrightarrow[]{\exd_1} C^{1} \xrightarrow[]{\exd_2}  \cdots 
    \end{equation}
    where \(C^j = \Hom(C_j, R)\) is the space of \(R\)-linear functionals on \(C_j\) and \(\exd_j \mathbf{x} = \mathbf{x} \partial_j\) for \(\mathbf{x} \in C^{j-1}\).
    \(C^\bullet\) is based and sparse if \(C_\bullet\) is, with respect to the dual basis to \(\mathsf{B}(C_j)\), denoted \(\mathsf{B}(C^j)\).
\end{defn}

Similarly, the definition of the chain map complex is semantically extremely similar to the \(\mathbb{F}_2\) case.

\begin{defn}[Chain map complex]\label{def:ChainMapComplexAppendix}
    The \emph{internal hom} or \emph{chain map complex} associated to two chain complexes \(C_\bullet\) and \(D_\bullet\), written $[C,D]_\bullet$, is a chain complex
    \begin{equation}
         \cdots \xrightarrow[\partial^{[2]}_2]{} [C,D]_1 \xrightarrow[\partial^{[2]}_1]{} [C,D]_0 \xrightarrow[\partial^{[2]}_0]{} [C,D]_{-1} \xrightarrow[\partial^{[2]}_{-1}]{}  \cdots 
    \end{equation}
    with components
    \begin{equation}\label{eqn:CMCComponentsAppendix}
        [C,D]_j = \bigoplus_{i \in \mathbb{Z}} \Hom(C_i, D_{i+j}),
    \end{equation}
    that is, collections of $R$-linear maps from $C_i$ to $D_{i + j}$.
    Given $f_{\bullet} = (f_k) \in [C, D]_{j}$, the boundary maps $\partial^{[2]}: [C, D]_{j} \to [C, D]_{j-1}$ are defined by:
    \begin{equation}\label{eqn:CMCBoundaryIndicesAppendix}
        (\partial^{[2]} f)_{k} = \partial_{j+k}^{(D)}  f_k - (-1)^j f_{k-1}  \partial_{k}^{(C)}.
    \end{equation}
\end{defn}

Writing \(\partial^{(C)} = \partial\) and \(\partial^{(D)} = \td\partial\) and denoting the degree of \(f\) by \(|f|\), Eq.~\eqref{eqn:CMCBoundaryIndicesAppendix} can be more compactly expressed
\begin{equation}\label{eqn:CMCBoundaryImplicit}
    \partial^{[2]} f = \tilde{\partial} f - (-1)^{|f|} f  \partial,
\end{equation}
where component indices must be inferred.

As a remark, if both \(C_\bullet\) and \(D_\bullet\) were unbounded (infinitely many nonzero \(C_j\) and \(D_j\)), then Eq.~\eqref{eqn:CMCComponentsAppendix} would have to be modified to
\begin{equation}
    [C,D]_j = \prod_{i \in \mathbb{Z}} \Hom(C_i, D_{i+j}).
\end{equation}
That is, the sequences of maps \((f_k)\) can have infinitely many nonzero terms.

The only difference in the definition above to the main text is in \(\partial^{[2]}\), where an alternating sign \((-1)^{|f|}\) now appears which was invisible in the \(\mathbb{F}_2\) case. 
This ensures that \((\partial^{[2]})^2 = 0\). Indeed, we have
\begin{subequations}
\begin{align}
    \partial^{[2]} \partial^{[2]} f &= \partial^{[2]} \left( \td\partial f - (-1)^{|f|} f\partial\right) \\
    &= \td\partial^2  f - f \partial^2 +  (-1)^{|f|}  \left( \td\partial  f \partial  - \td\partial f \partial \right)  \\
    &= 0.
\end{align}
\end{subequations}

The 0-cycles of \([C,D]_\bullet\) still have the interpretation of chain maps. They are exactly the collections of maps \(f_\bullet\) such that
\begin{equation}
    \partial^{[2]}_0 f = \td\partial f - f \partial =0,
\end{equation}
which is the commutativity condition for the diagram of Eq.~\eqref{eqn:ChainMapDiagram}. 
Null homotopies continue to be identified with the image of \(\partial^{[2]}_1\), so that the 0-homology \(H_0([\C,\D])\) encodes homotopy classes of gates. 
Null homotopic chain maps continue to have trivial action on homology, so that all chain maps belonging to the same homotopy class have the same induced action on homology.

An important example of the chain map complex is the \emph{dual complex} \(C^\vee_\bullet = [C,R_0]_\bullet\), where
\begin{equation}
    R_0 = \cdots \to 0 \to R \to 0 \to \cdots
\end{equation}
is the \emph{unit} chain complex with a single copy of \(R\) at component zero.
\(C^\vee_\bullet\) is essentially the same as the cochain complex \(C^\bullet\), up to a change of indices and sign conventions.
Indeed, the only nonzero component of \(C^\vee_j\) in Eq.~\eqref{eqn:CMCComponentsAppendix} arises when the target of the map is in degree zero, in which case the domain is \(C_{-j}\). That is
\begin{equation}\label{eqn:CveeIsomCochainComponent}
    C^\vee_j \isom \Hom(C_{-j},R) = C^{-j}.
\end{equation}
Similarly, the boundary map of \(C^\vee_\bullet\) is identical to the coboundary map for \(C^\bullet\) up to a sign. Denoting the dual differential as \(\partial^\vee\) and making indices explicit:
\begin{equation}
    (\partial^\vee_j f)_{1-j} = -(-1)^{j}f_{-j} \partial_{1-j} = -(-1)^{j} \exd_{1-j} f_{-j}.
\end{equation}
That is, \(\partial^\vee_\bullet = -(-1)^{\bullet}\exd_{1-\bullet}\). 
The sign alternation is trivial when \(R = \mathbb{F}_2\), so the main text does not distinguish between \(C^\bullet\) and \(C^\vee_\bullet\). 
Here, we will retain \(C^\vee_\bullet\).

\subsection{Isomorphism with the tensor product}
\label{subapp:MapTensorIsom}

It is a familiar fact that linear maps between vector spaces \(V\) and \(W\) can be regarded as vectors in the space \(W \otimes V^*\), where \(V^*\) is the dual space to \(V\). 
The same turns out to be true for finite dimensional free modules over any commutative ring \(R\), and even for based chain complexes.

\begin{theo}\label{thm:HomTensorIsomAppendix}
    For any based chain complexes \(C_\bullet\) and \(D_\bullet\), there is an isomorphism
    \begin{equation}
         [C,D]_\bullet \isom (D \otimes C^\vee)_\bullet.
    \end{equation}
\end{theo}

The tensor product of chain complexes is symmetric, so we also have \([C,D]_\bullet \isom (C^\vee \otimes D)_\bullet\), but the version stated in Theorem~\ref{thm:HomTensorIsomAppendix} involves fewer minus signs in the proof.

\begin{proof}
    The components of \((D \otimes C^\vee)_\bullet\) are
    \begin{subequations}
    \begin{align}
        (D \otimes C^\vee)_j &:= \bigoplus_{i \in \ZZ} D_{i+j} \otimes C^\vee_{-i} \\
        &\isom \bigoplus_{i \in \ZZ} D_{i+j} \otimes C^{i} \\
        &\isom \bigoplus_{i \in \ZZ} \Hom(C_i, D_{i+j})
        \\
        &=: [C, D]_j.
    \end{align}
    \end{subequations}
    In the second line we used Eq.~\eqref{eqn:CveeIsomCochainComponent}, and in the third line we used that the module of \(R\)-linear maps between finite dimensional free modules \(C_i\) and \(D_{i+j}\) is isomorphic to the tensor product \(D_{i+j} \otimes C^i\). Thus, we have an isomorphism between the components of \((D \otimes C^\vee)_\bullet\) and \([C, D]_\bullet\). It remains to show that this gives a chain map.

    We construct the chain isomorphism \(\gamma\) by components. For a homogeneous tensor \(x \otimes f \in D_{i+j} \otimes C^\vee_{-i}\), we define
    \begin{equation}
        \gamma(x \otimes f) := \left[ c_i \mapsto x f(c_i) \right] \in \Hom(C_i, D_{i+j}).
    \end{equation}
    This is just the canonical isomorphism from \(D_{i+j} \otimes C^i\) to \(\Hom(C_i, D_{i+j})\), and we now verify that \(\gamma\) is in fact a chain map.
    
    We must show that \(\gamma \partial_\otimes = \partial^{[2]} \gamma\), where \(\partial_\otimes\) is the boundary map for the tensor product. Recall that
    \begin{subequations}
    \begin{align}
        \partial_\otimes(x \otimes f) &:=  \td\partial x \otimes f + (-1)^{|x|} x \otimes \partial^\vee f \\
        &= \td\partial x \otimes f - (-1)^{|x|+|f|} x \otimes f\partial,
    \end{align}
    \end{subequations}
    where \(|x| = i+j\) is the degree of \(x\) and \(|f|=-i\) is the degree of \(f\).
    The two terms \(\td\partial x \otimes f\) and \(x \otimes f\partial\) belong to different components of \(( D \otimes C^\vee)_{j-1}\), namely \(D_{i+j-1} \otimes C^\vee_{-i}\) and \(D_{i+j} \otimes C^\vee_{-i-1}\), respectively. Mapping each of these by \(\gamma\) gives
    \begin{subequations}
    \begin{align}
        \gamma( \td\partial x \otimes f ) &= \left[ c_i \mapsto \td\partial x f(c_i) \right], \label{eqn:gd1} \\
        \gamma( x \otimes f \partial) &= \left[ c_{i+1} \mapsto x f(\partial c_{i+1}) \right]. \label{eqn:gd2}
    \end{align}
    \end{subequations}
    The full \(\gamma \partial_\otimes(f\otimes x)\) is the sum of these two maps, the second having a coefficient \(- (-1)^{|x|+|f|}\).
    To compute \(\partial^{[2]} \gamma\), denote \(g = \gamma(x\otimes f)\). Then we have \(\partial^{[2]} g = \td\partial g - (-1)^{|g|}g\partial\), where \(|g| = |x|+|f|\) is the degree of \(g\). Explicitly, the two terms in the differential evaluate to
    \begin{subequations}
    \begin{align}
        \td\partial g &= \left[ c_i \mapsto  \td\partial x f(c_i) \right], \label{eqn:dg1} \\
        g\partial &= \left[ c_{i+1} \mapsto x f(\partial c_{i+1}) \right] \label{eqn:dg2}
    \end{align}
    \end{subequations}
    We manifestly have equality between the pairs Eqs.~(\ref{eqn:gd1},\ref{eqn:dg1}) and Eqs.~(\ref{eqn:gd2},\ref{eqn:dg2}), and the minus sign coefficients on the latter also match.

    Thus, \(\gamma\) is a chain map composed of component-wise isomorphisms. This automatically makes \(\gamma\) a chain isomorphism.
\end{proof}

The proof above generalizes without alteration to any bounded chain complex (finitely many nonzero components) over modules with the property that \(D_{i+j} \otimes C^i \isom \Hom(C_i, D_{i+j})\).
The most general class of modules for which this holds are \emph{finitely generated projective} modules---those which occur as direct summands of free modules.
Such chain complexes have not, to our knowledge, been used for the construction of quantum codes before, but we note the generalization here nonetheless.

Also note that the theorem makes no reference to sparsity.
Nonetheless, if both \(C_\bullet\) and \(D_\bullet\) are sparse, then so are \([C,D]_\bullet\) and \((D \otimes C^\vee)_\bullet\).

As a final aside, we remark that the theorem statement is the defining property of a \emph{compact closed monoidal category}.
A more abstract version of the theorem would be: the category of bounded chain complexes over free finite-dimensional \(R\)-modules (or finitely generated projective modules), when equipped with the usual tensor product \(\otimes\) and internal hom \([-,-]\), becomes a compact closed monoidal category.

\subsection{Homotopy vs logical action}
\label{subapp:HomotopyVsLogical}

The \(X\) logicals of the CSS code associated to \(C_\bullet\) are determined by the homology \(H_0(C^\vee) \isom H^0(C)\).
Similarly, the logical action of a \(\mathsf{CZ}\) gate corresponding to the chain map \(\phii: C^\vee_\bullet \to D_\bullet\) is determined by the induced action of \(\phii\) on homology, \((\phii_{*})_0: H^0(C) \to H_0(D)\).
Two chain maps with the same logical action need not, in general, be related by a chain homotopy.
As such, there are generally more homology classes in \(H_0([C^\vee,D])\) than there are logical \(\mathsf{CZ}\) gates between \(C^\vee_\bullet\) and \(D_\bullet\).
This was encounted in the main text already in the case of \(\mathbb{F}_2\), but the separation between logical action and homology action is still more severe for general commutative rings \(R\).
Here we explain the distinction between homotopy classes of chain maps and classes of chain maps with the same induced action on \(H^0(C)\).

Over a field, such as \(\mathbb{F}_2\) in the main text, the only distinction between homotopy class and logical action is easy to identify.
Namely, that logical action is determined only by the action on \(H^0(C)\), but the homotopy class is determined by the action on \(H_i(C^\vee) \isom H^{-i}(C)\) for all \(i\).
If the code corresponding to \(C_\bullet\) has nontrivial redundancies of the \(X\) checks (\(Z\) checks), then the action of the gate on these redundancies is encoded by the induced action on \(H^{-1}(C)\) (\(H_{1}(C)\)), and similarly for redundancies of these redundancies, and so on.
Two gates with the same logical action but different effects on redundancies will not be homotopy equivalent.

The obstruction from redundancies turns out to be the only distinction between homotopy class and logical action for chain complexes over a field.
There is a bijection between homotopy classes of chain maps and collections of linear maps \(\phii_{i*}: H^{-i}(C) \to H_i(D)\).
This is easy to show from Theorem~\ref{thm:HomTensorIsom}.
We can compute homology classes of \([C^\vee,D]_\bullet\) (the \(0\)th homology will give chain maps up to homotopy) as
\begin{subequations}\label{eqn:Hmpty2Hom}
\begin{align}
    H_j([C^\vee,D]) &\isom H_j(C \otimes D) \\
    &\isom \bigoplus_{i \in \ZZ} H_{-i}(C) \otimes H_{i+j}(D) \label{eqn:Hmpty2Hom_Kunneth} \\
    &\isom \bigoplus_{i \in \ZZ} H^{-i}(C)^* \otimes H_{i+j}(D) \label{eqn:Hmpty2Hom_HomDual2Cohom} \\
    &\isom \bigoplus_{i \in \ZZ} \Hom(H^{-i}(C), H_{i+j}(D)). \label{eqn:Hmpty2Hom_Compact}
\end{align}
\end{subequations}
In addition to Theorem~\ref{thm:HomTensorIsom}, we used 
the K\"unneth formula [Eq.~\eqref{eqn:Hmpty2Hom_Kunneth}], 
that the cohomology group \(H^i(C) \isom H_i(C)^*\) is the dual of the homology group [Eq.~\eqref{eqn:Hmpty2Hom_HomDual2Cohom}], 
and that \(W \otimes V^* \isom \Hom(V,W)\) for vector spaces [Eq.~\eqref{eqn:Hmpty2Hom_Compact}].
The final result is that \(H_j([C^\vee,D])\) is isomorphic to the collection of linear maps \(\phii_*: H^{-\bullet}(C) \to H_{\bullet+j}(D)\).

More fundamentally, every chain complex \(C_\bullet\) over a field is homotopy equivalent to its homology complex \(H_\bullet(C)\), where components are given by \(H_i(C)\) and all boundary maps are zero.
It follows that there is a homotopy equivalence between \([C^\vee,D]_\bullet\) and \([H_\bullet(C^\vee), H_\bullet(D)]\). The latter has zero boundary maps, and components given by Eq.~\eqref{eqn:Hmpty2Hom_Compact}.

Over more general rings, there are more subtle distinctions between logical action and homotopy class than the action on redundancies.
There can be chain maps with identical action on all homology groups which, nonetheless, are not homotopic.
In the calculation Eq.~\eqref{eqn:Hmpty2Hom}, both the steps (\ref{eqn:Hmpty2Hom_HomDual2Cohom}) and (\ref{eqn:Hmpty2Hom_Compact}) generally fail.
It is no longer true that \(H_{-i}(C) \isom H^{-i}(C)^*\), and further \(H^{-i}(C)\) can fail to be a free module (or even projective) so that we may have \(H^{-i}(C)^* \otimes H_{i+j}(D) \not\isom \Hom(H^{-i}(C), H_{i+j}(D))\).
Thus, homotopy classes of chain maps are generally smaller than even equivalence classes determined by the action on homology, let alone classes determined only by the action on \(H^0(C)\).

\section{Interpretation of Cohomology Invariants as Intersection Invariants}
\label{app:CupsAndPoincare}

In this appendix, we review the notion of the cap product, which provides a physical interpretation for the cohomology invariants of the main text.
In particular, given a cup product, the cap product is a bilinear product defined as: 
\begin{defn}[(Sparse) Cap Product] Suppose that $C_{\bullet}$ is a chain complex equipped with a cup product. 
Then, the cap product is defined as: 
\begin{align}
    \capp\, &: \bigoplus_{p, q} C_p \times C^{q} \to C^{p - q} \nonumber \\
    b_p\capp \mathbf{a}_q &\equiv \sum_{\mathbf{x} \in \mathsf{B}(C^{p - q})} \left(\int_{b_p} \mathbf{a}_q \cupp \mathbf{x} \right) x 
\end{align}
where $b_p \in C_p$, $\mathbf{a}_q \in C^q$, $x = \mathbf{x}^{\mathsf{T}}$ is the basis chain that is dual to the basis cochain $\mathbf{x}$, and the integral symbol indicates evaluation $\int_b \mathbf{c} \equiv \mathbf{c}(b)$ and is used to make contact to the theory of calculus on manifolds.
If the cup product is sparse, we say the cap product is sparse.
\end{defn}

Intuitively, the cap product of a chain $b_p$ with a co-chain $\mathbf{a}_q$ constructs a $(p - q)$-chain by ``laying down'' every basis chain $x$ such that $\mathbf{a}_q \cupp \mathbf{x}$ pairs with $b_q$.

We will now prove that the cap product, similar to the cup product maps a cycle-cocycle pair to a cycle and furthermore, defines a product between homology and cohomology.
This product can be viewed as a ``duality'' map.
In particular, fixing a cycle $b_p$, the map from co-cycles $\mathbf{a}_q$ to cycles $c_{p - q} = b_p \capp \mathbf{a}_q$ via the cap product is closely related to the concept of \textbf{Poincar\'e duality} in manifolds.
With this in mind, we now prove the following theorem:
\begin{theo}[Generalized Poincar\'e Duality from the Cap Product] Suppose that $C_{\bullet}$ is a chain complex equipped with a cup product, then for all $p$-cycles $b_p$ and $q$-cocycles $\mathbf{a}_q$, the cap product $b_p \capp \mathbf{a}_q$ is a $(p - q)$-cycle. 

Moreover, the homology class of $b_p \capp \mathbf{a}_q$ only depends on the homology class of $b_p$ and the cohomology class of $\mathbf{a}_q$.
Hence, the cap product defines a product on homology and cohomology: 
\begin{equation}
    [b_p] \capp [\mathbf{a}_q] = [b_p \capp \mathbf{a}_q]
\end{equation}

\end{theo}
\begin{proof}
    Suppose that $b_p$ and $\mathbf{a}_q$ are a $p$-cycle and $q$-cocycle respectively, we want to show that $\partial (b_p \capp \mathbf{a}_q) = 0$.
    It suffices to show that for all $\mathbf{y} \in \mathsf{B}(C^{p - q - 1})$, 
    \begin{equation}
        \mathbf{y}[ \partial (b_p \capp \mathbf{a}_q)] = 0
    \end{equation}
    We can see this by explicit computation:
    \begin{align} \label{eq-capcalc}
        \mathbf{y}[\partial(b_p \capp \mathbf{a}_q)] &= \sum_{\mathbf{x} \in \mathsf{B}(C^{p - q})} \left(\int_{b_p} \mathbf{a}_q \cupp \mathbf{x} \right) \mathbf{y}(\partial x) \nonumber\\ 
        &= \sum_{\mathbf{x} \in \mathsf{B}(C^{p - q})} \left(\int_{b_p} \mathbf{a}_q \cupp \mathbf{x} \right) \exd \mathbf{y}(x) \nonumber \\
        &= \int_{b_p} \mathbf{a}_q \cupp \exd \mathbf{y} 
    \end{align}
where in the last step, we used the fact that $\exd \mathbf{y} = \sum_{x} \exd \mathbf{y}(x) \mathbf{x}$.
Now, note that, by the Leibniz rule:
\begin{equation}
    \mathbf{a}_q \cupp \exd \mathbf{y} = \exd (\mathbf{a}_q \cupp \mathbf{y}) + \exd \mathbf{a}_q \cupp \mathbf{y} = \exd (\mathbf{a}_q \cupp \mathbf{y})
\end{equation}
where in the last step, we used that $\exd \mathbf{a}_q = 0$ because it is a co-cycle.
Hence, we have that: 
\begin{equation}
    \mathbf{y}[\partial(b_p \capp \mathbf{a}_q)] = \int_{b_p} \exd(\mathbf{a}_q \cupp \mathbf{y})  = \int_{\partial b_p} \mathbf{a}_q \cupp \mathbf{y} = 0
\end{equation}
Since this is true for any basis element $\mathbf{y} \in \mathsf{B}(C^{p - q - 1})$, it follows that $\partial(b_p \capp \mathbf{a}_q) = 0$ and hence $b_p \capp \mathbf{a}_q$ is a cycle.

We now show that $b_p \capp \mathbf{a}_q$ depends only on the homology of $b_p$ and $\mathbf{a}_q$.
Indeed, we now show that, if  $\mathbf{c} \in C^{q - 1}$, then:
\begin{equation}
    b_p \capp (\mathbf{a}_q + \exd \mathbf{c}) = b_p \capp \mathbf{a}_q + \partial (b_p \capp \mathbf{c})
\end{equation}
and hence, $b_p \capp (\mathbf{a}_q + \exd \mathbf{c})$ is in the same homology class as $b_p \capp \mathbf{a}_q$.
Similarly, we can show that if $c \in C_{p + 1}$, $\partial c\capp \mathbf{a}_q = \partial (c \capp \mathbf{a}_q)$.
We show the former and the latter follows similarly. 

To show that $b_p \capp \exd \mathbf{c} = \partial (b_p \capp \mathbf{c})$, let $\mathbf{y} \in \mathsf{B}(C^{p - q})$.
Then:
\begin{align}
     \mathbf{y}(b_q \capp \exd \mathbf{c}) &= \int_{b_p} \exd \mathbf{c} \cupp \mathbf{y} = \int_{b_p} \mathbf{c} \cupp \exd \mathbf{y} \\
     &= \exd \mathbf{y} ( b_q \capp \mathbf{c}) = \mathbf{y}[\partial ( b_q \capp \mathbf{c})] 
\end{align}
where the first equality on the first line follows from a calculation similar to Eq.~\eqref{eq-capcalc} and the second equality on the first line follows from the Leibniz rule and the fact that $\partial b_p = 0$.
Since the above holds for all $\mathbf{y}$, we have that $b_q \capp \exd \mathbf{c} = \partial (b_q \capp \mathbf{c})$.
We can repeat a similar calculation to show $\partial c \capp \mathbf{a}_q = \partial (c \capp \mathbf{a}_q)$.
These show that the cap product defines a product on homology and cohomology.
\end{proof}

To get an intuitive feeling of the cap product, we evaluate it for a particular case in the 2D toric code.
This will be particularly helpful for building intuition for how the cup product is used to construct logical gates.

We report the cap product between the unique $2$-cycle of the toric code,%
\footnote{Recall that we are using the convention where $Z$ stabilizers live on $C_2$.} 
denoted $\mathcal{M} \in C_2$, and a co-cycle of the toric code $\mathbf{a} \in C^1$.
Recall that for the toric code $\mathcal{M}$ is simply the sum of all plaquettes on the toroidal lattice: 
\begin{equation} \label{eq-fundamentalclass}
    \mathcal{M} = \sum_{p \in \mathsf{B}(C_2)} p
\end{equation}
which is a $2$-cycle because $\partial \mathcal{M} = 0$ as a consequence of  periodic boundary conditions
Then, the cap product $\mathcal{M} \capp \mathbf{a}$ is given by: 
\begin{equation}
    \textcolor{dodgerblue}{b} = \mathcal{M} \capp \textcolor{red}{\mathbf{a}} =      \begin{tikzpicture}[scale = 0.8, baseline = {([yshift=-.5ex]current bounding box.center)}]
        \foreach \i in {0, 1, 2}{
            \draw[color = lightgray] (\i, -0.5) -- (\i, 1.5);}
        \foreach \i in {0, 1}{
            \draw[color = lightgray] (-0.5, \i) -- (2.5, \i);}
        \draw[color = red, line width = 0.5 pt, dashed] (-0.5, 0.5) -- (2.5, 0.5);
        \foreach \i in {0, 1, 2}{
            \draw[color = red, line width = 1 pt] (\i, 0) -- (\i, 1);
            \draw[color = dodgerblue, line width = 1 pt] (\i - 0.5, 1) -- (\i + 0.5, 1);
        }
    \end{tikzpicture}  
\end{equation}
Note that, in the code language, $\mathbf{a}$ is an $X$-logical, and $\mathcal{M} \capp \mathbf{a}$ is a $Z$-logical that runs parallel to it.
This association between $X$ and $Z$ logical operators is the famous \textit{Poincar\'e duality} of manifolds.

Finally, we can see how the cohomology invariants can be interpreted as measuring a notion of \textbf{intersection} by using the cap product.
In particular, it is easy to verify that: 
\begin{equation}
    \mathcal{I}_{\mathcal{M}_{(2)}}(\mathbf{a}, \mathbf{b}) = \int_{\mathcal{M}_{(2)}\, \capp\,  \mathbf{a}} \mathbf{b},
\end{equation}
i.e. the invariant is the evaluation of the co-cycle $\mathbf{b}$ on the cycle $\mathcal{M}_{(2)} \capp \mathbf{a}$.
A concrete example of this follows from the toric code, if $\mathcal{M}_{(2)} = \mathcal{M}$ written in Eq.~\eqref{eq-fundamentalclass}, then:
\begin{equation}
    \mathcal{I}_{\mathcal{M}_{(2)}}(\textcolor{red}{\mathbf{a}}, \textcolor{orange}{\mathbf{b}}) = \int_{\textcolor{dodgerblue}{c} = \mathcal{M} \capp \textcolor{red}{\mathbf{a}}} \textcolor{orange}{\mathbf{b}} =      \begin{tikzpicture}[scale = 0.8, baseline = {([yshift=-.5ex]current bounding box.center)}]
        \foreach \i in {0, 1, 2}{
            \draw[color = lightgray] (\i, -0.5) -- (\i, 2.5);}
        \foreach \i in {0, 1, 2}{
            \draw[color = lightgray] (-0.5, \i) -- (2.5, \i);}
        \draw[color = red, line width = 0.5 pt, dashed] (-0.5, 0.5) -- (2.5, 0.5);
        \draw[color = orange(ryb), line width = 1 pt, dashed] (0.5, -0.5) -- (0.5, 2.5);
        \foreach \i in {0, 1, 2}{
            \draw[color = red, line width = 1 pt] (\i, 0) -- (\i, 1);
            \draw[color = lightorange(ryb), line width = 1.5 pt] (0, \i) -- (1, \i);
            \draw[color = dodgerblue, line width = 1 pt] (\i - 0.5, 1) -- (\i + 0.5, 1);
        }
    \end{tikzpicture} = 1
\end{equation}
We can see that 
$\mathcal{I}_{\mathcal{M}_{(2)}(b)}$ intuitively captures the ``intersection'' between the dashed lines representing the two co-cycles.
The way it does this is by first shifting co-cycle $\textcolor{red}{\mathbf{a}}$ onto the cycle $\textcolor{dodgerblue}{c} = \mathcal{M} \capp \textcolor{red}{\mathbf{a}}$, which directly intersects with $\textcolor{orange}{\mathbf{b}}$.
This line of reasoning can be generalized to $\ell$-fold cup products.

\bibliography{UnifyingLDPCComputation}

\end{document}